\documentclass[11pt]{article}

\usepackage{amsmath, amsthm}
\usepackage{amsmath, amsfonts}
\usepackage{amsmath, amssymb}
\usepackage{amsmath}
\usepackage{graphics}
\usepackage{graphicx}
\usepackage{color}
\usepackage{hyperref}
\usepackage[all]{xy}

\newtheorem{theorem}{Theorem}[subsection]
\newtheorem{definition}[theorem]{Definition}
\newtheorem{definition-lemma}[theorem]{Definition/Lemma}
\newtheorem{definition-explanation}[theorem]{Definition/Explanation}
\newtheorem{explanation-definition}[theorem]{Explanation/Definition}
\newtheorem{definition-fact}[theorem]{Definition/Fact}
\newtheorem{definition-notation}[theorem]{Definition/Notation}
\newtheorem{definition-conjecture}[theorem]{Definition/Conjecture}
\newtheorem{lemma}[theorem]{Lemma}
\newtheorem{lemma-definition}[theorem]{Lemma/Definition}

\newtheorem{remark}[theorem]{\it Remark}
\newtheorem{remark-notation}[theorem]{\it Remark/Notation}

\newtheorem{application-lemma}[theorem]{Application/Lemma}

\newtheorem{convention}[theorem]{\it Convention}

\newtheorem{example-definition}[theorem]{Example/Definition}

\newtheorem{definition-prototype}[theorem]{Definition-Prototype}

\numberwithin{equation}{subsection}

\newtheorem{sdefinition-lemma}[stheorem]{Definition/Lemma}
\newtheorem{sdefinition-explanation}[stheorem]{Definition/Explanation}
\newtheorem{sexplanation-definition}[stheorem]{Explanation/Definition}
\newtheorem{sdefinition-fact}[stheorem]{Definition/Fact}
\newtheorem{sdefinition-notation}[stheorem]{Definition/Notation}
\newtheorem{sdefinition-conjecture}[stheorem]{Definition/Conjecture}

\newtheorem{slemma-definition}[stheorem]{Lemma/Definition}

\newtheorem{sremark-notation}[stheorem]{\it Remark/Notation}

\newtheorem{sapplication-lemma}[stheorem]{Application/Lemma}

\newtheorem{sexample-definition}[stheorem]{Example/Definition}

\newtheorem{sdefinition-prototype}[stheorem]{Definition-Prototype}

\newtheorem{ssdefinition-lemma}[sstheorem]{Definition/Lemma}
\newtheorem{ssdefinition-explanation}[sstheorem]{Definition/Explanation}
\newtheorem{ssexplanation-definition}[sstheorem]{Explanation/Definition}
\newtheorem{ssdefinition-fact}[sstheorem]{Definition/Fact}
\newtheorem{ssdefinition-notation}[sstheorem]{Definition/Notation}
\newtheorem{ssdefinition-conjecture}[sstheorem]{Definition/Conjecture}

\newtheorem{sslemma-definition}[sstheorem]{Lemma/Definition}

\newtheorem{ssremark-notation}[sstheorem]{\it Remark/Notation}

\newtheorem{ssapplication-lemma}[sstheorem]{Application/Lemma}

\newtheorem{ssexample-definition}[sstheorem]{Example/Definition}

\newtheorem{ssdefinition-prototype}[sstheorem]{Definition-Prototype}

\newcommand{\Der}{\mbox{\it Der}\,}

  \newcommand{\tinyDouble}{\mbox{\it\tiny Double}}

\newcommand{\Endsheaf}{\mbox{\it ${\cal E}\!$nd}\,}

\newcommand{\Exponential}{\mbox{\it Exp}\,}

\newcommand{\Id}{\mbox{\it Id}\,}

\newcommand{\Imaginary}{\mbox{\it Im}\,}

\newcommand{\Log}{\mbox{\it Log}\,}

\newcommand{\parameter}{\mbox{\scriptsize\it parameter}\,}

\newcommand{\Real}{\mbox{\it Re}\,}

\newcommand{\SL}{\mbox{\it SL}}

\newcommand{\SO}{\mbox{\it SO}\,}

\newcommand{\Spin}{\mbox{\it Spin}\,}

\newcommand{\Tr}{\mbox{\it Tr}\,}

  \newcommand{\tinyWZ}{\mbox{\tiny\it WZ}\,}

\newcommand{\scriptsizeach}{\mbox{\scriptsize\it ach}}

\newcommand{\anticommuting}{\mbox{\scriptsize\it anti-c}}
  \newcommand{\tinyanticommuting}{\mbox{\tiny\it anti-c}}

\newcommand{\chd}{\mbox{\it c.h.d}\,}
  
\newcommand{\scriptsizech}{\mbox{\scriptsize\it ch}}

\newcommand{\constrained}{\mbox{\scriptsize\it constrained}\,}
\newcommand{\coordinates}{\mbox{\scriptsize\it coordinates}}

\newcommand{\even}{\mbox{\scriptsize\rm even}}
  
\newcommand{\field}{\mbox{\scriptsize\it field}}

\newcommand{\gaugescriptsize}{\mbox{\scriptsize\it gauge}\,}
 \newcommand{\gaugetiny}{\mbox{\tiny\it gauge}\,}

\newcommand{\odd}{\mbox{\scriptsize\rm odd}}

\newcommand{\physics}{\mbox{\scriptsize\it physics}}
  \newcommand{\tinyphysics}{\mbox{\tiny\it physics}}

\newcommand{\stc}{\mbox{\scriptsize\it stc}\,}

\newcommand{\boldsigma}{\mbox{\boldmath $\sigma$}}

\newcommand{\LARGEdot}{\mbox{\LARGE $\cdot$}}

\newcommand{\tinybullet}{\raisebox{.2ex}{\tiny $\bullet$}}							
							
\begin{document}

\enlargethispage{24cm}

\begin{titlepage}

$ $

\vspace{-1.5cm} 

\noindent\hspace{-1cm}
\parbox{6cm}{\small February 2019}\
   \hspace{7cm}\
   \parbox[t]{6cm}{\small
                arXiv:yymm.nnnnn [hep-th] \\
				D(14.1.Supp.1), SUSY(1):\\  $\mbox{\hspace{4.3em}}$ 
                        $d=3+1$, $N=1$ \\   $\mbox{\hspace{4.3em}}$ 
				        tower construction
				}

\vspace{2cm}

\centerline{\large\bf
 Physicists' $d=3+1$, $N=1$ superspace-time and supersymmetric QFTs}
\vspace{1ex}
\centerline{\large\bf 
 from a tower construction in complexified ${\Bbb Z}/2$-graded $C^\infty$-Algebraic Geometry}
\vspace{1ex}
\centerline{\large\bf 
 and a purge-evaluation/index-contracting map}
 

\bigskip

\vspace{3em}

\centerline{\large
  Chien-Hao Liu       \hspace{1ex} and \hspace{1ex} 
  Shing-Tung Yau
}

\vspace{4em}

\begin{quotation}
\centerline{\bf Abstract}
\vspace{0.3cm}

\baselineskip 12pt  
{\small
 The complexified ${\Bbb Z}/2$-graded $C^\infty$-Algebraic Geometry aspect 
   of a superspace(-time) $\widehat{X}$
   in Sec.\,1 of D(14.1) (arXiv:1808.05011 [math.DG]) 
  together with the Spin-Statistics Theorem in Quantum Field Theory, 
     which requires fermionic components of a superfield be anticommuting, 
 lead us to the notion of towered superspace(-time) $\widehat{X}^{\widehat{\boxplus}}$ 
   and 
     the built-in purely even physics sector $X^{\mbox{\tiny physics}}$ from $\widehat{X}^{\widehat{\boxplus}}$.
 We use this to reproduce
   the $d=3+1$, $N=1$ Wess-Zumino model and
   the $d=3+1$, $N=1$ supersymmetric $U(1)$ gauge theory with matter
    ---  as in, e.g.,  Chap.\,V and Chap.\,VI \& part of Chap.\,VII of the classical Supersymmetry \& Supergravity textbook
	      by Julius Wess and Jonathan Bagger ---
   and, hence, recast physicists' two most basic supersymmetric quantum field theories
   solidly into the realm of (complexified ${\Bbb Z}/2$-graded) $C^\infty$-Algebraic Geometry.
 Some traditional differential geometers' ways of understanding supersymmetric quantum field theories
  are incorporated into the notion of a purge-evaluation/index-contracting map
  ${\cal P}:C^\infty(X^{\mbox{\tiny physics}})\rightarrow C^\infty(\widehat{X})$ in the setting.
 This completes for the current case
   a $C^\infty$-Algebraic Geometry language we sought for in D(14.1), footnote 2,
   that can directly link to the study of supersymmetry in particle physics.
 Once generalized to the nonabelian case in all dimensions and extended $N\ge 2$,
  this prepares us for a fundamental (as opposed to solitonic) description of
  super D-branes parallel to Ramond-Neveu-Schwarz fundamental superstrings
 } 
\end{quotation}

\vspace{2em}

\baselineskip 12pt
{\footnotesize
} 

\bigskip

\noindent {\small MSC number 2010:  58A50, 81T60; 14A22, 17A70,  46L87, 81Q60, 81T30
} 

\bigskip

\baselineskip 10pt
{\scriptsize
\noindent{\bf Acknowledgements.}
 We thank
 Andrew Strominger and Cumrun Vafa
   for influence to our understanding of strings, branes, and gravity.
 C.-H.L.\ thanks in addition
 Girma Hailu
   for his topic course on supersymmetry and discussions, fall 2018, that lead to the current work
  (cf. footnote 1 {\it Special acknowledges});
 Cumrun Vafa
   for consultation;
 Denis Auroux,
 Ashvin Vishwanath,
 Goran Radanovic/Haifeng Xu/Brian Plancher
   for other inspiring topic courses, fall 2018;
 Francesca Pei-Jung Chen
   for the recording of explanation/demonstration of J.S.\,Bach's sonata that accompanies the typing of the current work;
 Jing Gu, Jia-Chi Lee for sharing their life stories;
 Ling-Miao Chou
   for comments that improve the illustration and moral support.
 The project is supported by NSF grants DMS-9803347 and DMS-0074329.
} 

\end{titlepage}

\newpage

\enlargethispage{24cm}
\begin{titlepage}

$ $

\vspace{12em}

\centerline{\small\it
 Chien-Hao Liu dedicates this note}
\centerline{\small\it
 to the loving memory of}
\centerline{\small\it
 Rev.\ \& Mrs.\ R.\ Campbell Willman (1925-2014) and Barbara M.\ Willman (1927-1999).}

\vspace{30em}

\baselineskip 11pt

{\footnotesize
\noindent
(From C.H.L.)
 There are memories that are too abundant to condense and too cherished and personal to reveal.
 A seemingly accidental encounter in my teenage years, which turned out to have profound impact on me.    
 The motherly love to me from Mrs.\ Willman and the friendship this family, including {\it Ann} and {\it Lisa},
 had provided me with
  turned a rebellious, unresting teenager from an almost high-school dropout to a researcher.
 The fact that Rev.\,Willman, graduated from University  of California at Berkeley with a degree in biology, gave up a more earthly pursuit
   to answer what he regarded as a higher call made a great impact on my mind.
 There is not a word with which I can express my gratitude to this family
   just like there is not a word I can use to express my gratitude to my own parents and family.
 The current work has equal weight to D-project as D(14.1) and,
  together, they completed our first step toward dynamical supersymmetric D-branes
      along the line of Ramond-Neveu-Schwarz superstring.
 It is thus dedicated to this family.
} 
 
\end{titlepage}


\newpage
$ $

\vspace{-3em}

\centerline{\sc
 $d=3+1$, $N=1$ Superspace-Time and SQFT from Tower Construction
 } %

\vspace{2em}

\baselineskip 14pt  

\begin{flushleft}
{\Large\bf 0. Introduction and outline}
\end{flushleft}
 In Section 1 of [L-Y1] (D(14.1), arXiv:1808.05011 [math.DG])
  a $d=3+1$, $N=1$ superspace-time $\widehat{X}=(X, \widehat{\cal O}_X)$ was constructed
  from the aspect of complexified ${\Bbb Z}/2$-graded $C^\infty$-Algebraic Geometry.
 In this work\footnote{\makebox[17.6em][l]{\it Special acknowledgements from C.H.L.}
                                                              It's a pure coincidence, and a God-given luck to us, that in fall semester 2018
																{\it Girma Hailu} gave once again the topic course
																{\sl Physics 253cr Quantum Field Theory III}
                                                                  on supersymmetry [Hai]
																  at Department of Physics, Harvard University,
																  just after we had completed [L-Y1] (D(14.1)) in August.
															   This is the second time I sat in his topic courses on supersymmetry,
															   (the first time in fall semester 2013).
                                                               With half of the story
															     (i.e.\ the complexified ${\Bbb Z}/2$-graded $C^\infty$-Algebraic Geometry side)
															    more solidly in mind as the reference point,
																we have some slight intellectual leisure to pay more attention to
																 what we had still missed from the physics side.
															   It is his lectures on supersymmetry, two after-class long discussions, and
															    one email communication
															    that propelled us to re-examine
																 what we had taken as the foundation, i.e.\ superspaces,
																for the construction of a theory of fermionic D-branes
																along the line of  the Ramond-Neveu-Schwarz superstrings
															   and, in the end, led to the tower construction of superspaces in this work,
															    which upgrades and simplifies [L-Y1] (D(14.1)) tremendously.
															   Despite his generosity to stay behind the scene,
															     we regard him as a hidden coauthor who nourished the current work.
															   Sec.\,1.2 is particularly attributed to him.
                                                                             },     
 we extend the construction ibidem
  so that
    both the nature of a superspace-time from complexified ${\Bbb Z}/2$-graded $C^\infty$-Algebraic Geometry
	     and the matching of spin and statistics required of fields from Quantum Field Theory
	 can be consistently taken into account and harmoniously built into the construction.	
 While, as complexified super $C^\infty$-schemes,
  the towered superspace $\widehat{X}^{\widehat{\boxplus}}$ fibers over $\widehat{X}$ in [L-Y1] (D(14.1))
    with a canonical section and a canonical flat connection,
 it has a new sector, i.e.\ the {\it physics sector} $X^{\physics}$ ,
  whose structure sheaf ${\cal O}_X^{\physics}$ comes from superfields in the sense of physicists,
  e.g.,  Abdus Salam \& John Strathdee [S-S] (1978), and
  Julius Wess \& Jonathan Bagger [W-B] (1992),
  polished to take into account the unavoidable mathematical consequences of nilpotency of component fields
  once the spinorial component fields are assumed to be anticommuting and hence nilpotent.
 As a locally-ringed space,
  $X^{\physics}$ turns out to be a complexified $C^\infty$-scheme in its own right;  in particular,  it is purely even.
 For this reason,
  the parity of many physically relevant objects on the towered superspace $\widehat{X}^{\widehat{\boxplus}}$
  become purely even.
 In such situations, 	
   the sign-factor issues one has to deal with in [L-Y1] (D(14.1))
     due to the ${\Bbb Z}/2$-grading are all gone  and
   classical formulae in physics literature on supersymmetry
    become valid when suitably interpreted via a purge-evaluation map
	from ${\cal O}_X^{\physics}$ to $\widehat{\cal O}_X$.
 To demonstrate the validity of  above tower construction from complexified ${\Bbb Z}/2$-graded $C^\infty$-Algebraic Geometry
   to describe particle physicists' notion of superspaces and supersymmetric quantum field theories,
 the notion of $d=3+1$, $N=1$ chiral superfields and antichiral superfields and
 the construction of $d=3+1$, $N=1$ supersymmetric chiral matter theory and $U(1)$ gauge theory
   in Chap.\ V and Chap.\ VI of the standard textbook [W-B] of Wess \& Bagger are re-done
   on $\widehat{X}^{\widehat{\boxplus}}$.
 This completes for the current case
   a $C^\infty$-Algebraic Geometry language we sought for in [L-Y1: footnote 2] (D(14.1))
   that can directly link to the study of supersymmetry and supersymmetric quantum field theory in particle physics.
 
 Once generalized to the nonabelian case in all dimensions with extended $N\ge 2$ supersymmetry,
  this prepares us for a fundamental (as opposed to solitonic) description of
  super D-branes parallel to Ramond-Neveu-Schwarz fundamental superstrings

\bigskip
\bigskip

\noindent
{\bf Convention.}
Same as [L-Y1]  (D(14.1)).
In particular, references for standard notations, terminology, operations and facts are\;\;
 (1) algebraic geometry: [Hart];\; $C^{\infty}$-algebraic geometry: [Joy];\;\;
 (2) spinors and supersymmetry (mathematical aspect):  [Ch], [De], [D-F1], [D-F2], [Fr], [Harv], [S-W];\;\;
 (3) supersymmetry (physical aspect, especially $d=4$, $N=1$ case): [W-B], [G-G-R-S], [We];\,
            also [Argu], [Argy],  [Bi], [St], [S-S].

\bigskip
\bigskip
   
\begin{flushleft}
{\bf Outline}
\end{flushleft}
\nopagebreak
{\small
\baselineskip 12pt  
\begin{itemize}
 \item[0.]
  Introduction
  
 \item[1.]
  The $d=3+1$, $N=1$ towered superspace-time and its physics sector
   \vspace{-.6ex}
   \begin{itemize}
     \item[1.1]
	  Supermanifolds in the sense of complexified ${\Bbb Z}/2$-graded $C^\infty$-Algebraic Geometry
	
     \item[1.2]
	  What should the function ring of a superspace be:
	  naturality from $C^\infty$-Algebraic Geometry aspect vs.\ naturality from Quantum Field Theory aspect

     \item[1.3]
	 The $d = 3+1$, $N = 1$ towered superspace-time $\widehat{X}^{\widehat{\boxplus}_l}$
	 	
	 \item[1.4]
     The physics sector $X^{\tinyphysics}$ of $\widehat{X}^{\widehat{\boxplus}_l}$	
	
     \item[1.5]
	 Purge-evaluation maps and the Fundamental Theorem on supersymmetric action\\ functionals	 	 
   \end{itemize}
   
 \item[2.]
  The chiral/antichiral theory on $X^{\tinyphysics}$ and Wess-Zumino model
   \vspace{-.6ex}
   \begin{itemize}     		
	 \item[2.1]
	 More on the chiral and the antichiral sector of $X^{\tinyphysics}$
	
     \item[2.2]
     Wess-Zumino model on $X$ in terms of $X^{\tinyphysics}$
   \end{itemize}
   
  \item[3.]
  Supersymmetric $U(1)$ gauge theory with matter on $X$ in terms of $X^{\tinyphysics}$
   \vspace{-.6ex}
   \begin{itemize}
     \item[3.1]
	 The bundle/sheaf context underlying a supersymmetric $U(1)$ gauge theory with matter
                         built from $X^{\tinyphysics}$
		
	 \item[3.2]
	 Pre-vector superfields in Wess-Zumino gauge
	
	 \item[3.3]
	 Supersymmetry transformations of a pre-vector superfield in Wess-Zumino gauge
	
	 \item[3.4]
	 From pre-vector superfields to vector superfields
	
	 \item[3.5]
	 Supersymmetric $U(1)$ gauge theory with matter on $X$ in terms of $X^{\tinyphysics}$
   \end{itemize}
 
 \item[]\hspace{-3.8ex}
   Appendix\;\;\;Notations, conventions, and identities in spinor calculus
   %
 %
 %
 %
\end{itemize}
} 

\newpage

\section{The $d = 3+1$, $N = 1$ towered superspace-time and its physics sector}

\subsection{Supermanifolds in the sense of complexified ${\Bbb Z}/2$-graded $C^\infty$-Algebraic Geometry}

The notion of {\sl supermanifolds} in the sense of complexified ${\Bbb Z}/2$-graded $C^\infty$-Algebraic Geometry
  and a few basic objects we need
 are recalled in this subsection for the introduction of terminology and notations.
Details are referred to [L-Y1: Sec.\,1] D(14.1).

\bigskip

\begin{definition} {\bf [supermanifold/superscheme]}\; {\rm 		
 Given a (real smooth) manifold (in general, $C^\infty$-scheme) $M$,
  denote its structure sheaf of smooth functions by ${\cal O}_M$ and its complexification
  ${\cal O}_M^{\,\Bbb C}:= {\cal O}_M\otimes_{\Bbb R}{\Bbb C}$
   (i.e.\ the sheaf of complex-valued smooth functions on $M$).
 For any ${\cal O}_M^{\,\Bbb C}$-module ${\cal F}$ of finite rank,
  one can construct a {\it complexified ${\Bbb Z}/2$-graded $C^\infty$-scheme} (i.e.\ {\it supermanifold})
  $$
     \widehat{M}\;:=\; (M, \widehat{\cal O}_X )
  $$
  from ${\cal F}$ by taking the new {\it structure sheaf} on $M$
    to be the exterior ${\cal O}_M^{\,\Bbb C}$-algebra generated by ${\cal F}$\,:
  $$
    \widehat{\cal O}_M\; :=\;  \mbox{$\bigwedge$}^{\tinybullet}_{{\cal O}_M^{\,\Bbb C}}{\cal F}.
  $$
 This a locally-ringed space with the underlying topology $M$.
 We shall call ${\cal F}$ the {\it generating sheaf} of the supermanifold.
 This is a {${\Bbb Z}/2$-graded} ${\cal O}_M^{\,\Bbb C}$-algebra with
  the {\it even part}  and the {\it odd part} given respectively by
  $$
     \widehat{\cal O}_M^{\,\even}\;
	   :=\;  \mbox{$\bigwedge$}^{\even}_{{\cal O}_M^{\,\Bbb C}}{\cal F}\,,
	    \hspace{2em}
	 \widehat{\cal O}_M^{\,\odd}\;
	   :=\;  \mbox{$\bigwedge$}^{\odd}_{{\cal O}_M^{\,\Bbb C}}{\cal F}\,.
  $$
 The {\it $C^\infty$-hull} of $\widehat{\cal O}_M$
   is given by
  $$
   \mbox{\it $C^\infty$-hull}\,(\widehat{\cal O}_M)\;
    : =\;  {\cal O}_M \oplus  \mbox{$\bigwedge$}_{{\cal O}_M^{\,\Bbb C}}^{\,\even,\,\ge 2}\;
    \subset \;  \widehat{\cal O}_M\,.
  $$
  $\mbox{\it $C^\infty$-hull}\,(\widehat{\cal O}_M)$ is a sheaf of $C^\infty$-rings.
 This, by definition, gives a {\it partial $C^\infty$-ring structure} on $\widehat{\cal O}_M$.
 The set $\Gamma(\tinybullet)$ of global sections of these structure sheaves $(\tinybullet)$ on $M$ are denoted by
  $$
    C^\infty(M)\,,\hspace{2em}
    C^\infty(M)^{\Bbb C}\,,\hspace{2em}
    C^\infty(\widehat{X})\,,\hspace{2em}
    \mbox{\it $C^\infty$-hull}\,(C^\infty(\widehat{X}))
  $$
  respectively.
 Each carries a corresponding complexified, ${\Bbb Z}/2$-graded, $C^\infty$-or-partial-$C^\infty$-algebraic
  (whichever applicable) structure.
}\end{definition}

\medskip

\begin{definition} {\bf [(left) derivation of $C^{\infty}(\widehat{M})$]}\; {\rm
 (Cf.\ [L-Y1: Definition~1.3.2, footnote~7] (D(14.1)).)
 A {\it (left) derivation}   of $C^{\infty}(\widehat{M})$ over ${\Bbb C}$
  is a ${\Bbb Z}/2$-graded ${\Bbb C}$-linear operation\\
  $\xi: C^{\infty}(\widehat{M})\rightarrow C^{\infty}(\widehat{M})$
  on $C^{\infty}(\widehat{M})$
  that satisfies the ${\Bbb Z}/2$-graded Leibniz rule
  $$
    \xi(fg)\;=\; (\xi f)g\,+\, (-1)^{p(\xi)p(f)}f(\xi g)
  $$
  when in parity-homogeneous situations.
 The set $\Der_{\Bbb C}(\widehat{M}):=\Der_{\Bbb C}(C^{\infty}(\widehat{M}))$
     of derivations of $C^{\infty}(\widehat{M})$
	 is a (left) $C^{\infty}(\widehat{M})$-module,
	with $(a\xi)(\,{\LARGE \cdot}\,):= a (\xi(\,\cdot\,))$ and $p(a\xi):= p(a)+p(\xi)$
	for $a\in C^{\infty}(\widehat{M})$ and $\xi\in \Der_{\Bbb C}(\widehat{M})$.
}\end{definition}

\medskip

\begin{definition} {\bf [differential of $C^{\infty}(\widehat{M})$]}\; {\rm
(Cf.\ [L-Y1: Definition~1.3.6, footnote~10] (D(14.1)).)
 The bi-$C^{\infty}(\widehat{M})$-module
   $\Omega_{\widehat{M}}:= \Omega_{C^{\infty}(\widehat{M})}$
   of {\it differentials} of $C^{\infty}(\widehat{M})$ over ${\Bbb C}$
  is the quotient of
    the free bi-$C^{\infty}(\widehat{M})$-module
      	          generated by $d(f)$, $f\in C^{\infty}(\widehat{M})$,
   by the bi-$C^{\infty}(\widehat{M})$-submodule of relators generated by
   \begin{itemize}
    \item[(1)]
	 [{\it ${\Bbb C}$-linearity}\,]\hspace{1em}
	 $d(c_1f_1 + c_2f_2)-c_1 d(f_1) - c_2 d(f_2)$\,,\;
	  for $c_1, c_2\in {\Bbb C}$, $f_1, f_2\in C^{\infty}(\widehat{M})$;
	
    \item[(2)]
	 [{\it Leibniz rule}\,]\hspace{1em}
	  $d(f_1f_2)- (d(f_1))f_2 -f_1 d(f_2)$\,,\; for $f_1, f_2\in C^{\infty}(\widehat{M})$;
	
	\item[(3)]
	 [{\it chain-rule identities from the $C^{\infty}$-hull structure}\,]\hspace{1em}
	 $$
	    d(h(f_1,\,\cdots\,, f_l))- \sum_{k=1}^l (\partial_kh)(f_1,\,\cdots\,, f_l)\,d(f_k)
	 $$
	 for
	   $h\in C^{\infty}({\Bbb R}^l)$,
	   $f_1,\,\cdots\,, f_l
	          \in   C^{\infty}(\mbox{\it $C^\infty$-hull}\,(C^\infty(\widehat{M})))
		      \subset C^{\infty}(\widehat{M})$;
	 here, $\partial_kh$ is the partial derivative of $h\in C^{\infty}({\Bbb R}^l)$
	           with respect to the $k$-th argument.
   \end{itemize}
 The element of $\Omega_{\widehat{M}}$  associated to $d(f)$,
    $f\in C^{\infty}(\widehat{M})$, is denoted by $df$.
 Using Relators (2), one can convert $\Omega_{\widehat{M}}$ to
   either solely a left $C^{\infty}(\widehat{M})$-module
         or solely a right $C^{\infty}(\widehat{M})$-module.
		 		
 A differential of $C^{\infty}(\widehat{M})$  is also called synonymously
  a {\it $1$-form} on $\widehat{M}$.
  
 By construction, there is a built-in map
  $d:C^{\infty}(\widehat{M})\rightarrow \Omega_{\widehat{M}}$ defined by
  $f\mapsto df$.
}\end{definition}

\medskip

\begin{convention} $[$cohomological degree vs.\ parity\,$]$\; {\rm
(Cf.\ [L-Y1: Convention~1.3.5, footnote~9] (D(14.1)).)
 We treat elements $f$ of $C^{\infty}(\widehat{M})$ as of cohomological degree $0$
  and the exterior differential operator $d$ as of cohomological degree $1$ and {\it even}.
 In notation, $\chd(f)=0$ and $\chd(d)=1$, $p(d)=0$.
 Under such $({\Bbb Z}\times ({\Bbb Z}/2))$-bi-grading,
   $$
      ab = (-1)^{c.h.d(a)\,c.h.d(b)}(-1)^{p(a)p(b)}ba
   $$
   for objects $a, b$ homogeneous with respect to the bi-grading.
  Here, $a$ and $b$ are not necessarily of the same type.
 This is the convention that matches with the sign rules in
   [W-B: Chap.\ XII, Eq.'s (12.2), (12.3)]	of Wess \& Bagger.
}\end{convention}

\medskip

\begin{lemma} {\bf [evaluation of $\Omega_{\widehat{M}}$ on
       $\Der_{\Bbb C}(\widehat{M})$ from the right]}\;
{\rm (Cf.\ [L-Y1: Lemma~1.3.7, footnote~11] (D(14.1).))}	
 The specification
   $$
  (df)(\xi) \; :=\;  (\xi)^{\leftarrow}\!\!\!(df)\;   :=\; \xi(f)
   $$
   for $f\in C^{\infty}(\widehat{M})$ and $\xi\in \Der_{\Bbb C}(\widehat{M})$,
  defines an {\it evaluation} of $\Omega_{\widehat{M}}$ on $\Der_{\Bbb C}(\widehat{M})$
  from the right: \\
  for
    $\varpi=\sum_{i=1}^k a_i\,df_i\in \Omega_{\widehat{M}}$, with $a_i$ parity-homogeneous, and
	$\xi\in \Der_{\Bbb C}(\widehat{M})$ parity-homogeneous,
  $$
    \varpi(\xi)\;
	:=\;  (\xi)^{\leftarrow}\!\!\!\varpi\;
	:=\;  \sum_{i=1}^k (-1)^{p(\xi)p(a_i)}  a_i\, \xi(f_i)\,.
  $$
 This evaluation is (left) $C^{\infty}(\widehat{M})$-linear:
  $\varpi(a\xi)= a \varpi(\xi)$, for $a\in C^{\infty}(\widehat{M})$.
\end{lemma}

\bigskip

Higher tensors, in particular $k$-forms, on $\widehat{M}$ can also be defined.
See [L-Y1: Sec.\,1.3] (D(14.1)) for more details on the differential calculus on $\widehat{M}$.

\bigskip

\subsection{What should the function-ring of a superspace be:
         naturality from $C^\infty$-Algebraic Geometry aspect vs.\ naturality from Quantum Field Theory aspect}

Given the $4$-dimensional  Minkowski space-time $X={\Bbb R}^{3+1}$.
Let
 \begin{itemize}
  \item[$\cdot$]
    $P$ be the Lorentzian frame bundle over $X$ with the (flat) Levi-Civita connection,
	
  \item[$\cdot$]	
    $S^{\,\Bbb C}=S^\prime\oplus S^{\prime\prime}$ be the complexified Dirac-spinor bundle
	from the spinor representation of the Lorentz group,
   $S^{{\Bbb C},\vee}= S^{\prime\,\vee}\oplus S^{\prime\prime\,\vee}$ be the dual of $S^{\,\Bbb C}$.
   The corresponding sheave are denoted by
     ${\cal S}^{\,\Bbb C}$, ${\cal S}^\prime$, ${\cal S}^{\prime\prime}$  and
	 ${\cal S}^{\,{\Bbb C}, \vee}$, ${\cal S}^{\prime\,\vee}$, ${\cal S}^{\prime\prime\,\vee}$   respectively.
    These spinor bundles and sheaves are all equipped with a flat connection induced from that on $P$.
 \end{itemize}
A `{\sl $d=3+1$, $N=1$ superspace}' is meant to be a supermanifold/superscheme in the sense of Definition~1.1.1
 with the underlying topology $M=X$ and
 the generating sheaf ${\cal F}$ ``coming from" {\it one copy} of ${\cal S}^{\,\Bbb C}$,
   or ${\cal S}^{\,{\Bbb C}\,\vee}$
  to provide ``{\it one set of fermionic/anticommuting coordinates
  $(\theta^1,\theta^2, \bar{\theta}^{\dot{1}}, \bar{\theta}^{\dot{2}})$ on the superspace}".
In this subsection, we shall re-examine the notion of `{\sl superspace}' based on this definition and guided by the question
 \begin{itemize}
  \item[{\bf Q.}] {\bf [guiding/key]}\;\;
  {\it What should the function ring of a $d=3+1$, $N=1$ superspace be?}
 \end{itemize}
with care not only from the complexified ${\Bbb Z}/2$-graded $C^\infty$-Algebraic Geometry
 but also from the Quantum Field Theory.

\bigskip

\begin{flushleft}
{\bf Naturality from the complexified ${\Bbb Z}/2$-graded $C^\infty$-Algebraic Geometry aspect}
\end{flushleft}
In [L-Y1: Sec.\,1.2] (D(14.1)) a $d=3+1$, $N=1$ superspace $\widehat{X}$
 was constructed as a super $C^\infty$-scheme with complexification.
There we took $M=X$, ${\cal F}={\cal S}^{\,{\Bbb C},\vee}$.\footnote{See
                                                                      [L-Y1: footnote 4 in Sec.\.1.2] for an explanation of
																	    ${\cal S}^{\,\Bbb C}$ vs.\:${\cal S}^{\,{\Bbb C},\vee}$.
																	 In ibidem, to avoid carrying the dual `\,$^\vee$\,' everywhere
																	    and enough for the purpose there,
																	  we actually chose ${\cal F}={\cal S}^{\Bbb C}$
																	  instead of ${\cal S}^{{\Bbb C},\vee}$ for the simplicity of notation.
																	 For the current work, such distinction matters and we resume what  it should be;
																	  cf.\: [L-Y9] (D(11.4.1)) in Reference of ibidem.
																	 It turns out that the choice ${\cal F}={\cal S}^{\,{\Bbb C},\vee}$
																	  also matches better the convention in [W-B: Appendix A] of Wess \& Bagger
																	  since we write fermionic/anticommuting coordinates with upper spinor index
																	   $(\theta^1,\theta^2,
																	         \bar{\theta}^{\dot{1}},\bar{\theta}^{\dot{2}})$,
																	   rather than lower spinor index.
   																	                                              }  
A tuple of constant sections
  $(\theta,\bar{\theta}):= (\theta^1,\theta^2, \bar{\theta}^{\dot{1}}, \bar{\theta}^{\dot{2}})$
 from ${\cal S}^{\prime\,\vee}\oplus{\cal S}^{\prime\prime\,\vee}$
 with respect the built-in flat connection
 was chosen to serve as the fermionic coordinate functions on the superspace $\widehat{X}$.
Together with the standard coordinate functions $x:=(x^0,x^1,x^2,x^3)$ on $X$,
they generate (in the sense of complexified ${\Bbb Z}/2$-graded $C^\infty$-ring ) the function ring
 $$
   C^\infty(\widehat{X})\;
   =\; C^\infty(X)^{\Bbb C}[\theta^1,\theta^2,
          \bar{\theta}^{\dot{1}},\bar{\theta}^{\dot{2}}]^{\anticommuting}
 $$
 of $\widehat{X}$.
Here $[\,\cdots\,]^{\anticommuting}$ means `{\it polynomial ring (with coefficients in $C^\infty(X)^{\Bbb C}$)
    in anticommuting variables $\cdots$}'.
{\it Mathematically} from pure algebra,
 whatever ring that contains
    both $C^\infty(X)^{\Bbb C}$ and $\theta^1,\theta^2,\bar{\theta}^{\dot{1}},\bar{\theta}^{\dot{2}}$
 must contain
    $C^\infty(X)^{\Bbb C}[\theta^1,\theta^2,
                      \bar{\theta}^{\dot{1}},\bar{\theta}^{\dot{2}}]^{\anticommuting}$,
despite the fact that
 as $P$-modules, the corresponding structure sheaf contains not just Lorentz scalars:
 \begin{eqnarray*}
  \widehat{\cal O}_X
   & :=\:  &   \mbox{$\bigwedge$}^{\tinybullet}_{{\cal O}_X^{\,\Bbb C}}
          	({\cal S}^{\prime\,\vee}\oplus {\cal S}^{\prime\prime\,\vee}) \\
   & \simeq
       & {\cal O}_X^{\,\Bbb C}
	         \oplus ({\cal S}^{\prime\,\vee}\oplus {\cal S}^{\prime\prime\,\vee})
		     \oplus ({\cal O}_X^{\,\Bbb C}
					          \oplus {\cal S}^{\prime\,\vee}\otimes_{{\cal O}_X^{\,\Bbb C}}
							                  {\cal S}^{\prime\prime\,\vee}
							   \oplus {\cal O}_X^{\,\Bbb C})
		     \oplus  ({\cal S}^{\prime\,\vee}\oplus {\cal S}^{\prime\prime\,\vee})
			 \oplus {\cal O}_X^{\,\Bbb C}\,.
 \end{eqnarray*}
Note also that an $f\in C^\infty(\widehat{X})$
  in the $(\theta,\bar{\theta})$-expansion\footnote{Here,
                                                                                    since $f_{(\tinybullet)}$ are commuting,
                                                                                     one can write the coefficients either on the right or on the left.
																					We choose to write them on the right in order to match better
																					  with the later setting when we allow anticommuting coefficients as well.
                                                                                } 
   \begin{eqnarray*}
     f & = &  f_{(0)}\,
	        +\, \sum_{\alpha}\theta^\alpha f_{(\alpha)}\,
			+\, \sum_{\dot{\beta}} \bar{\theta}^{\dot{\beta}} f_{(\dot{\beta})}\,
			+\, \theta^1\theta^2  f_{(12)}\,
			+\, \sum_{\alpha,\dot{\beta}}
			       \theta^\alpha\bar{\theta}^{\dot{\beta}} f_{(\alpha\dot{\beta})}\,  \\
        && 				
            +\, \bar{\theta}^{\dot{1}}\bar{\theta}^{\dot{2}}  f_{(\dot{1}\dot{2})}\,
			+\, \sum_{\dot{\beta}}
			        \theta^1\theta^2\bar{\theta}^{\dot{\beta}}  f_{(12\dot{\beta})}\,
			+\, \sum_\alpha
			        \theta^\alpha\theta^{\dot{1}}\theta^{\dot{2}}  f_{(\alpha\dot{1}\dot{2})}\,
		    +\, \theta^1\theta^2\theta^{\dot{1}}\theta^{\dot{2}} f_{(12\dot{1}\dot{2})}			
   \end{eqnarray*}
 has all its coefficients $f_{(\tinybullet)}\in C^\infty(X)^{\Bbb C}$ and, hence, {\it commuting}.

\bigskip

\begin{flushleft}
{\bf Naturality from the Quantum Field Theory aspect}
\end{flushleft}
For distinction,
  denote the would-be superspace by $?\widehat{X}$ and its function ring by $C^\infty(?\widehat{X})$.
Then it follows from the discussion in the previous theme that
 $$
   C^\infty(?\widehat{X})\;\supset\;
   C^\infty(X)^{\Bbb C}[\theta^1,\theta^2,
                      \bar{\theta}^{\dot{1}},\bar{\theta}^{\dot{2}}]^{\anticommuting}\,.
 $$
If we insist that $?\widehat{X}$ be an $N=1$ superspace, i.e.\
 the tuple $(\theta^1,\theta^2,\bar{\theta}^{\dot{1}},\bar{\theta}^{\dot{2}})$
 remains to serve as a maximal tuple of fermionic coordinates on $?\widehat{X}$,
then it is natural to assume that every $\breve{f}\in  C^\infty(?\widehat{X})$ remains to have
 an expansion in $(\theta,\bar{\theta})$:
  \begin{eqnarray*}
   \breve{f}  & = &  \breve{f}_{(0)}\,
	        +\, \sum_{\alpha}\theta^\alpha \breve{f}_{(\alpha)}\,
			+\, \sum_{\dot{\beta}} \bar{\theta}^{\dot{\beta}} \breve{f}_{(\dot{\beta})}\,
			+\, \theta^1\theta^2  \breve{f}_{(12)}\,
			+\, \sum_{\alpha,\dot{\beta}}
			       \theta^\alpha\bar{\theta}^{\dot{\beta}} \breve{f}_{(\alpha\dot{\beta})}\,  \\
        && 				
            +\, \bar{\theta}^{\dot{1}}\bar{\theta}^{\dot{2}}  \breve{f}_{(\dot{1}\dot{2})}\,
			+\, \sum_{\dot{\beta}}
			        \theta^1\theta^2\bar{\theta}^{\dot{\beta}}  \breve{f}_{(12\dot{\beta})}\,
			+\, \sum_\alpha
			        \theta^\alpha\theta^{\dot{1}}\theta^{\dot{2}}  \breve{f}_{(\alpha\dot{1}\dot{2})}\,
		    +\, \theta^1\theta^2\theta^{\dot{1}}\theta^{\dot{2}} \breve{f}_{(12\dot{1}\dot{2})}\,.			
  \end{eqnarray*}
The question is now
 \begin{itemize}
  \item[{\bf Q.}] {\bf [coefficients in $(\theta,\bar{\theta})$-expansion]}\;\;  {\it
   If a coefficient $\breve{f}_{(\tinybullet)}$ of an $\breve{f}\in C^\infty(?\widehat{X})$
	       does not lie in $C^\infty(X)^{\Bbb C}$,
    then in where should it lie?}
 \end{itemize}
In constructing $\widehat{X}$ as a complexified ${\Bbb Z}/2$-graded $C^\infty$-scheme in [L-Y1: Sec.\,1.2],
 we have used all what mathematics can offer.
The answer to the above question thus has to come from insights from physics.
 
From the Quantum Field Theory aspect,
  a physics-relevant element (i.e.\;{\it physics superfield}) $\breve{f}$ in $C^\infty(?\widehat{X})$
 must satisfy the following two basic requirements
 \begin{itemize}
  \item[(1)] [{\sl Lorentz scalar}\,]\hspace{1em}
    \parbox[t]{26em}{{\it $\breve{f}$ be a Lorentz scalar},
    i.e.\;$\breve{f}$ be a section of a trivial $P$-module over $X$.}

   \item[(2)] [{\sl Spin-Statistics Theorem}\,]\hspace{1em}
     \parbox[t]{24em}{\it Bosons be commuting and fermions be anticommuting.}
 \end{itemize}
Since we take
  $\theta^\alpha\in {\cal S}^{\prime\,\vee}$   and
  $\bar{\theta}^{\dot{\beta}}\in {\cal S}^{\prime\prime\,\vee}$,
Condition (1) suggests that
  \begin{itemize}
   \item[$\cdot$] {\it
    $\breve{f}_{(\alpha)}\in {\cal S}^\prime$    and
	$\breve{f}_{\dot{\beta}} \in {\cal S}^{\prime\prime}$; and
	$\breve{f}_{(\alpha\dot{1}\dot{2})}$ and $\breve{f}_{(12\dot{\beta})}$
	 lie in some isomorphic copy of ${\cal S}^\prime$, ${\cal S}^{\prime\prime}$\\ respectively.}
  \end{itemize}
Condition (2)  then comes along to imply that
 \begin{itemize}
  \item[$\cdot$] {\it
    The coefficients
     $\breve{f}_{(\alpha)}$, $\breve{f}_{\dot{\beta}} $,
     $\breve{f}_{(\alpha\dot{1}\dot{2})}$, $\breve{f}_{(12\dot{\beta})}$
	 in the $(\theta,\bar{\theta})$-expansion of $\breve{f}$
	 are themselves anticommuting.}	 	
  In other words, they take values on some Grassmann numbers as well.
  This renders $\breve{f}$ {\it purely even}(!).
  Furthermore, as the notation suggests, we require that
   $$
    \frac{\partial}{\partial\theta^\alpha}  \breve{f}_{(\mbox{\tiny $\bullet$})}\;
     =\;\frac{\partial}{\partial\bar{\theta}^{\dot{\beta}}} \breve{f}_{(\mbox{\tiny $\bullet$})}\;
	 =\; 0
   $$
   for all component fields $f_{(\mbox{\tiny $\bullet$})}$ in the expansion.
  In particular, while the fermionic component fields take values on Grassmann numbers,
   {\it they are independent of the existing fermionic coordinates $(\theta,\bar{\theta})$}.
 \end{itemize}
Clearly, such an $\breve{f}$ cannot lie in $C^\infty(\widehat{X})$:
The sought-for $C^\infty(?\widehat{X})$ is definitely larger than $C^\infty(\widehat{X})$ constructed.
%
%
 %
 \begin{itemize}
  \item[\bf Q.]  {\bf [${\Bbb C}$-${\Bbb Z}/2$-$C^\infty$-AG + QFT=?]}
   {\it Can one take the naturality from Quantum Field Theory aspect also into account
      in the construction of a superspace?}
 \end{itemize}
(Cf.\ [L-Y1: footnote 2] (D(14.1)).)
 
In the next subsection, we will take $\widehat{X}$ as the foundation and the starting point
  --- since it must be there mathematically ---
  and extend over it to answer the above question affirmatively.

\bigskip

\subsection{The $d = 3+1$, $N = 1$ towered superspace-time $\widehat{X}^{\widehat{\boxplus}_l}$}

We now proceed in two steps to answer Question [${\Bbb C}$-${\Bbb Z}/2$-$C^\infty$-AG + QFT=?] in Sec.\,1.2.
For the first step,
 physicists should be aware that when one tries to extend $C^\infty(\widehat{X})$ to another
 complexified ${\Bbb Z}/2$-graded $C^\infty$-ring $C^\infty(?\widehat{X})$ by adding `something',
 not only `something' appears in the new ring but also all those objects from taking $C^\infty$-closure come in as well.
No surprise that not all the elements in $C^\infty(?\widehat{X})$ have physical meaning,
  but their inclusion is a mathematical must.
And, hence, we have to allow them.
This is completed in this subsection.
Then comes the second step: the identification of those elements in $C^\infty(?\widehat{X})$
  that are physical and the justification that they are even and do form a compexified $C^\infty$-ring.
This is completed in the next subsection.

\bigskip

\begin{flushleft}
{\bf The $d=4$, $N=1$ towered superspace $\widehat{X}^{\widehat{\boxplus}_l}$ with $l$ field-theory levels}
\end{flushleft}
Recall the Weyl-spinor sheaves ${\cal S}^\prime$ and ${\cal S}^{\prime\prime}$ and their dual
 ${\cal S}^{\prime\,\vee}$, ${\cal S}^{\prime\prime\,\vee}$
 at the beginning of Sec.\,1.2.

\bigskip

\begin{definition}
{\bf [$d=4$, $N=1$ towered superspace $\widehat{X}^{\widehat{\boxplus}_l}$ with $l$ field-theory levels]}\;
{\rm
 The complexified ${\Bbb Z}/2$-graded $C^\infty$-scheme
   $$
     \widehat{X}^{\widehat{\boxplus}_l}\;
	  :=\; (X, \widehat{\cal O}_X^{\,\widehat{\boxplus}_l})\;
	  :=\; (X, \mbox{$\bigwedge$}_{{\cal O}_X^{\,\Bbb C}}^{\tinybullet}{\cal F})
   $$
 with
   $$
    {\cal F}\;
	    :=\;  ({\cal S}^{\prime\,\vee}_{\coordinates}\oplus {\cal S}^{\prime\prime\,\vee}_{\coordinates})
		        \oplus
				({\cal S}^{\prime\,\vee}_{\parameter}\oplus {\cal S}^{\prime\prime\,\vee}_{\parameter})
			    \oplus
				  \mbox{$\bigoplus$}_{i=1}^l
				     ({\cal S}^\prime_{\field, i}\oplus {\cal S}^{\prime\prime}_{\field, i})
   $$
  is called the {\it $d=4$, $N=1$ towered superspace with $l$ field-theory levels}.
 Here,
   all ${\cal S}^\prime_{\tinybullet}$
    (resp.\ ${\cal S}^{\prime\prime}_{\tinybullet}$,
	               ${\cal S}^{\prime\, \vee}_{\tinybullet}$, ${\cal S}^{\prime\prime\,\vee}_{\tinybullet}$)
   are copies\footnote{Mathematically
                                           this means that
                                          ${\cal S}^{\prime\,\vee}_{\mbox{\tiny\it coordinates }}$ is isomorphic to
										  ${\cal S}^{\prime\,\vee}$ with a fixed isomorphism; and similarly for all other spinor sheaves
										  that appear as direct summands of ${\cal F}$.
										      } 
      of ${\cal S}^\prime$
    (resp.\  ${\cal S}^{\prime\prime}$,
	               ${\cal S}^{\prime\,\vee}$, ${\cal S}^{\prime\prime\,\vee}$).
 When $l$ is implicit in the problem, we will denote
  $\widehat{X}^{\widehat{\boxplus}_l}$ simply by $\widehat{X}^{\widehat{\boxplus}}$.
}\end{definition}

\bigskip

Note that, as an ${\cal O}_X^{\,\Bbb C}$-modules,
 \begin{eqnarray*}
  \lefteqn{
   \widehat{\cal O}_X^{\,\widehat{\boxplus}}\;
    =\;   \mbox{$\bigwedge$}_{{\cal O}_X^{\,\Bbb C}}^{\tinybullet}
	             ({\cal S}^{\prime\,\vee}_{\coordinates}
				         \oplus{\cal S}^{\prime\prime\,\vee}_{\coordinates})   }\\
   && \hspace{3em}						
			  \otimes_{{\cal O}_X^{\,\Bbb C}}
              \mbox{$\bigwedge$}_{{\cal O}_X^{\,\Bbb C}}^{\tinybullet}
	             ({\cal S}^{\prime\,\vee}_{\parameter}
				         \oplus{\cal S}^{\prime\prime\,\vee}_{\parameter})
			  \otimes_{{\cal O}_X^{\,\Bbb C}}
			   \mbox{$\bigotimes_{{\cal O}_X^{\,\Bbb C}}$}_{i=1}^l
                 \mbox{$\bigwedge$}_{{\cal O}_X^{\,\Bbb C}}^{\tinybullet}
	             ({\cal S}^\prime_{\field, i}\oplus{\cal S}^{\prime\prime}_{\field, i})\,.
 \end{eqnarray*}

\bigskip

\begin{definition-explanation} {\bf [levels of $\widehat{X}^{\widehat{\boxplus}_l}$]}\; {\rm
 (1)
 The $d=4$, $N=1$ superspace $\widehat{X}$ constructed in [L-Y1: Definition 1.2.5] is given by
  $$
   \widehat{X}\;
     :=\;  (X, \widehat{\cal O}_X
	                    := \mbox{$\bigwedge$}^{\tinybullet}
                                  ({\cal S}^{\prime\,\vee}_{\coordinates}
								         \oplus {\cal S}^{\prime\prime\,\vee}_{\coordinates}))
  $$
  in the current context; (cf.\ [L-Y1: footnote 4] (D(14.1))).
 Any collection of complex conjugate generating constant global sections,
  denoted collectively as
   $(\theta, \bar{\theta})
      :=(\theta^1, \theta^2; \bar{\theta}^{\dot{1}}, \bar{\theta}^{\dot{2}})$,
  of ${\cal S}^{\prime\,\vee}_{\coordinates}\oplus {\cal S}^{\prime\prime\,\vee}_{\coordinates}$
  serves as the {\it fermionic coordinate functions} on $\widehat{X}$,
  ([L-Y1: Definition 1.2.5] (D(14.1))).
 This explains the subscript `$\coordinates$'  in
   ${\cal S}^{\prime\,\vee}_{\coordinates}
       \oplus {\cal S}^{\prime\prime\,\vee}_{\coordinates}$.
 $\widehat{X}$ is called the {\it fundamental level} or
    the {\it ground level} of $\widehat{X}^{\widehat{\boxplus}_l}$.
	
 (2)
  When physicists working on supersymmetry introduce `Grassmann number' parameter
    $(\eta,\bar{\eta}):= (\eta^1, \eta^2, \bar{\eta}^{\dot{1}}, \bar{\eta}^{\dot{2}})$
    in their computation,
   these `Grassmann number' parameter are meant to be independent of anything else.
  Thus, they should be thought of as constant sections
    of another copy of ${\cal S}^{\prime\,\vee}\oplus {\cal S}^{\prime\prime\,\vee}$.\footnote{
	                                                                   Algebraic-Geometry-oriented readers may find something uncomfortable here.
                                                                       In Algebraic Geometry one usually begins with a choice of
																		 a ground field (in the sense of algebra) $\Bbbk$,
																		 for example, $\Bbbk={\Bbb C}$
																		    when studying Calabi-Yau spaces in string theory.
                                                                        Naively, here since we are dealing with
																		   complexified ${\Bbb Z}/2$-graded geometry,
																		   one would choose
																		   $\Bbbk={\Bbb C}[\eta,\bar{\eta}]^{\tinyanticommuting}$
																		   to begin with and all the rings are $\Bbbk$-algebras.
                                                                        However, while this has nothing wrong mathematically,
                                                                           it is misleading from the perspective of the physics side.	
                                                                        That has to do with how such Grassmann numbers/parameters
																		   are used in physics.
																		Most often they are to be paired with supersymmetry generators.
                                                                        In such case, supersymmetry generators are realized as a spinor representation
																		 of Lorentz group;  thus these parameters have to be realized as dual spinors.
																		It is for this reason we choose $(\eta,\bar{eta})$ from constant sections
																		  of (a copy of )
																		  ${\cal S}^{\prime\,\vee}\oplus {\cal S}^{\prime\prime\,\vee}$.
																		In particular, though one intends to think of
																		  ${\Bbb C}[\eta,\bar{\eta}]^{\tinyanticommuting}$
																		  as the ground ring in the problem, these parameters themselves are not
																		  Lorentz scalars.
																		 As we will see, nor is the function ring
																	$C^\infty(\widehat{X}^{\widehat{\boxplus}})^{\tinyphysics}$
																		   of the physics sector a
																		   ${\Bbb C}[\eta,\bar{\eta}]^{\tinyanticommuting}$-algebra
																		   since the former is purely even.
																		
																		For a topologically twisted theory,
																		   the twisted supersymmetry generators become Lorentz scalars.
																		 To match with this, the Grassmann numbers/parameters
																		    become Lorentz scalars as well.
																		  Only in such cases,
																		     the mathematical notion of ${\Bbb C}[\eta,\bar{\eta}]$
																			    as the ground ring    and
																			 physicists' use of $(\eta,\bar{\eta})$ match well.
																			
                                                                           This is why we take $\widehat{X}$ as the ground level
																		     and  the Grassmann number/parameter level comes next over it,
																			 not the other round.
	                                                                                 } 
  This explains the subscript `$\parameter$'  in
   ${\cal S}^{\prime\,\vee}_{\parameter}
       \oplus {\cal S}^{\prime\prime\,\vee}_{\parameter}$.
 We say that
    ${\cal S}^{\prime\,\vee}_{\parameter}
          \oplus {\cal S}^{\prime\prime\,\vee}_{\parameter}$
	 contributes to the {\it Grassmann parameter level} of $\widehat{X}^{\widehat{\boxplus}_l}$.

 (3)	
 The fermionic coefficients of a physical superfield are themselves anticommuting and must correspond to sections
   of a copy of spinor sheaves.
 Physical superfields that  are associated to different types/classes/generations of particles should be thought of
  as sections of the Grassmann algebra generated by different copies of spinor sheaves.
 This explains the subscript `$\field$'  in
   ${\cal S}^\prime_{\field, i}
       \oplus {\cal S}^{\prime\prime}_{\field, i}$.
 We say that
    ${\cal S}^\prime_{\field, i}
          \oplus {\cal S}^{\prime\prime}_{\field, i}$
	 contributes to the {\it $i$-th field-theory level} of $\widehat{X}^{\widehat{\boxplus}_l}$.			
 The total level number $l$ is the number of distinct types of (particle, its superpartner)
   in a $d=3+1$, $N=1$ supersymmetric field theory one wants to construct.
 It can be different theory by theory.
}\end{definition-explanation}

\bigskip
	
By construction, there is a commutative diagram of
  ${\Bbb Z}/2$-grading-preserving ${\cal O}_X^{\,\Bbb C}$-algebra-homomorphisms
  $$
    \xymatrix{
    & \widehat{\cal O}_X^{\,\widehat{\boxplus}_l} \ar@{->>}[ld]_-{\iota^\sharp}  \\
    \widehat{\cal O}_X \ar@{=}[rr]
	  &&  \widehat{\cal O}_X \rule{0ex}{1.2em} \ar@{_{(}->}[lu]_-{\pi^\sharp}
   }
  $$
  that respect the partial $C^\infty$-ring structures.
 This gives a commutative diagram of morphisms of complexified super $C^\infty$-schemes
  $$
   \xymatrix{
    &&  \widehat{X}^{\widehat{\boxplus}_l} \ar@{->>}[rd]^-{\hat{\pi}}  \\
    & \;\widehat{X}\; \ar@{^{(}->}[ru]^-{\hat{\iota}} \ar@{=}[rr]
	    && \widehat{X}   &.
    }
  $$
 
The built-in flat connection on the ${\cal S}^\prime$, ${\cal S}^{\prime\prime}$ and their dual
  induces a flat connection on $\widehat{X}^{\widehat{\boxplus}_l}$ over $\widehat{X}$.
This defines a canonical inclusion
 $$
   \Der_{\Bbb C}(\widehat{X}) \;
   \hookrightarrow\; \Der_{\Bbb C}(\widehat{X}^{\widehat{\boxplus}_l})\,.
 $$
And any flow on $\widehat{X}$ lifts canonically to a  flow on $\widehat{X}^{\widehat{\boxplus}_l}$.
 
\bigskip

\begin{definition}
{\bf [derivation on $\widehat{X}$ applied to $C^\infty(\widehat{X}^{\widehat{\boxplus}_l})$]}\;
{\rm
 Let
   $\xi\in \Der_{\Bbb C}(\widehat{X})$  be a derivation on $\widehat{X}$ and
   $\breve{f}\in C^\infty(\widehat{X}^{\widehat{\boxplus}_l})$.
 Then we define $\xi\breve{f}\in C^\infty(\widehat{X}^{\widehat{\boxplus}_l})$
   via the canonical inclusion $\Der_{\Bbb C}(\widehat{X})
                                   \hookrightarrow \Der_{\Bbb C}(\widehat{X}^{\widehat{\boxplus}_l})$.
}\end{definition}

\bigskip

\begin{definition} {\bf [complex conjugation vs.\;twisted complex conjugation]}\; {\rm
 The complex conjugation
    $ \bar{\hspace{1.2ex}}:{\cal O}_X^{\,\Bbb C}\rightarrow {\cal O}_X^{\,\Bbb C}$ and
		${\cal S}^\prime\rightarrow {\cal S}^{\prime\prime}\,,\;
	       {\cal S}^{\prime\prime}\rightarrow {\cal S}^\prime$,
    of Weyl spinors
   extends canonically to a {\it complex conjugation}
   $$
     \bar{\hspace{1.2ex}}\;:\; \widehat{\cal O}_X^{\widehat{\boxplus}_l}\;
	    \longrightarrow\; \widehat{\cal O}_X^{\widehat{\boxplus}_l}\,,
   $$
   by setting
   \begin{itemize}	
	\item[(1)]
	 $\overline{\breve{f}+\breve{g}}\;=\; \bar{\breve{f}} + \bar{\breve{g}}$\,;
	
    \item[(2)]  $\overline{\breve{f}\breve{g}}\;=\; \bar{\breve{g}}\bar{\breve{f}}$\,.
   \end{itemize}
  and a	{\it twisted complex conjugation}
  $$
    ^\dag\;:\;  \widehat{\cal O}_X^{\widehat{\boxplus}_l}\;
	                    \longrightarrow\; \widehat{\cal O}_X^{\widehat{\boxplus}_l}\,,
  $$
  by setting
   \begin{itemize}
    \item[(0$^\prime$)]
	 $^\dag = \bar{\hspace{1.2em}}:
	       {\cal O}_X^{\,\Bbb C}\rightarrow {\cal O}_X^{\,\Bbb C}\,;
	      {\cal S}^\prime \rightarrow {\cal S}^{\prime\prime}\,,\;
	      {\cal S}^{\prime\prime}\rightarrow {\cal S}^\prime$\,;
	
	\item[(1$^\prime$)]
	 $(\breve{f}+\breve{g})^\dag\;=\; \breve{f}^\dag + \breve{g}^\dag$\,;
	
    \item[(2$^\prime$)]  $(\breve{f}\breve{g})^\dag\;=\; \breve{g}^\dag \breve{f}^\dag$\,.
   \end{itemize}
 Caution that the order of multiplication is preserved under the complex conjugate $\bar{\hspace{1.2ex}}$
   but is reversed under the twisted complex conjugate $^\dag$.
}\end{definition}
 
\medskip

\begin{definition} {\bf [standard coordinate functions on $\widehat{X}^{\widehat{\boxplus}_l}$]}\;
{\rm
 The standard coordinate functions $(x,\theta,\bar{\theta})$ on $\widehat{X}$
 extends uniquely to a tuple of coordinate functions
 $$
  (x^\mu, \theta^{\alpha},
      \bar{\theta}^{\dot{\beta}};
	  \eta^{\alpha^\prime}, \bar{\eta}^{\dot{\beta}^\prime};
	  \vartheta^1_{\gamma_1}, \bar{\vartheta}^1_{\dot{\delta}_1} ;\,
	   \cdots\,;
	  \vartheta^l_{\gamma_l}, \bar{\vartheta}^l_{\dot{\delta}_l}  )\;
	  =:\;   (x, \theta, \bar{\theta}, \eta, \bar{\eta},
	               \vartheta, \bar{\vartheta})\;
 $$
 on $\widehat{X}^{\widehat{\boxplus}_l}$ via
    the $\varepsilon$-tensor
     $\varepsilon:
 	    {\cal S}^\prime\otimes_{{\cal O}_X^{\,\Bbb C}}{\cal S}^\prime
 		      \rightarrow {\cal O}_X^{\,\Bbb C} $,\,
        ${\cal S}^{\prime\prime}\otimes_{{\cal O}_X^{\,\Bbb C}}{\cal S}^{\prime\prime}
 		      \rightarrow {\cal O}_X^{\,\Bbb C} $, 			
   and
   the fixed isomorphisms
  $ {\cal S}^\prime_{\tinybullet}\simeq {\cal S}^\prime$,
  ${\cal S}^{\prime\prime}_{\tinybullet}\simeq {\cal S}^{\prime\prime}$.

 Explicitly, regard
  ${\cal S}^{\prime\,\vee}_{\parameter}$ as a copy of ${\cal S}^{\prime\,\vee}_{\coordinates}$,
  ${\cal S}^{\prime\prime\,\vee}_{\parameter}$
      as a copy of ${\cal S}^{\prime\prime\,\vee}_{\coordinates}$,
  ${\cal S}^\prime_{\field, i}$
    as a copy of $({\cal S}^{\prime\,\vee}_{\coordinates})^\vee = {\cal S}^\prime_{\coordinates}$, and
  ${\cal S}^{\prime\prime}_{\field, i}$
   as a copy of $({\cal S}^{\prime\prime\,\vee}_{\coordinates})^\vee
      = {\cal S}^{\prime\prime}_{\coordinates}$
  under the fixed isomorphisms.
 Then,
   $(\eta^{\alpha^\prime},\bar{\eta}^{\dot{\beta}^\prime})
     = (\theta^{\alpha^\prime},\bar{\theta}^{\dot{\beta}^\prime})$ and
   $(\vartheta^i_{\alpha_i},\bar{\vartheta}^i_{\dot{\beta}_i})
       = (\theta_{\alpha_i}, \bar{\theta}_{\dot{\beta}_i})$
    for all $i$, 	 	
 where $\theta_\alpha = \sum_{\gamma}\varepsilon_{\alpha\gamma}\theta^\gamma$,  	
	         $\theta_{\dot{\beta}}
			   = \sum_{\dot{\delta}}\varepsilon_{\dot{\beta}\dot{\delta}}\bar{\theta}^{\dot{\delta}}$\,.  
			
 We shall call $(x, \theta, \bar{\theta}, \eta, \bar{\eta}, \vartheta, \bar{\vartheta})$		
   the {\it standard coordinate functions} on $\widehat{X}^{\widehat{\boxplus}_l}$.
}\end{definition}	

\bigskip

In terms of this,
 $$
   C^\infty(\widehat{X}^{\widehat{\boxplus}_l})\;
     =\;  C^\infty(X)^{\Bbb C}
	        [\theta, \bar{\theta}, \eta, \bar{\eta}, \vartheta, \bar{\vartheta}]^{\anticommuting}
 $$
and an $\breve{f}\in C^\infty(\widehat{X}^{\widehat{\boxplus}_l})$
 has a $(\theta,\bar{\theta})$-expansion
   \begin{eqnarray*}
     \breve{f} & = &  \breve{f}_{(0)}\,
	        +\, \sum_{\alpha}\theta^\alpha \breve{f}_{(\alpha)}\,
			+\, \sum_{\dot{\beta}} \bar{\theta}^{\dot{\beta}} \breve{f}_{(\dot{\beta})}\,
			+\, \theta^1\theta^2  \breve{f}_{(12)}\,
			+\, \sum_{\alpha,\dot{\beta}}
			       \theta^\alpha\bar{\theta}^{\dot{\beta}} \breve{f}_{(\alpha\dot{\beta})}\,  \\
        && 				
            +\, \bar{\theta}^{\dot{1}}\bar{\theta}^{\dot{2}}  \breve{f}_{(\dot{1}\dot{2})}\,
			+\, \sum_{\dot{\beta}}
			        \theta^1\theta^2\bar{\theta}^{\dot{\beta}}  \breve{f}_{(12\dot{\beta})}\,
			+\, \sum_\alpha
			        \theta^\alpha\theta^{\dot{1}}\theta^{\dot{2}}  \breve{f}_{(\alpha\dot{1}\dot{2})}\,
		    +\, \theta^1\theta^2\theta^{\dot{1}}\theta^{\dot{2}} \breve{f}_{(12\dot{1}\dot{2})}			
   \end{eqnarray*}
 with coefficients
   $\breve{f}_{(\tinybullet)}\in
       C^\infty(X)^{\Bbb C}[\eta,\bar{\eta},\vartheta,\bar{\vartheta}]^{\anticommuting}$.

\bigskip

\begin{flushleft}
{\bf The chiral, the antichiral, and the self-twisted-conjugate sector of $\widehat{X}^{\widehat{\boxplus}_l}$}
\end{flushleft}
Recall from [L-Y1: Sec.\,1.4] (D(14.1))
 the standard infinitesimal supersymmetry generators
 $$
   Q_{\alpha}\;
    =\;  \frac{\partial}{\partial \theta^{\alpha}}
            -\, \sqrt{-1}\,\sum_{\mu=0}^3 \sum_{\dot{\beta}=\dot{1}}^{\dot{2}}
			          \sigma^{\mu}_{\alpha\dot{\beta}}\bar{\theta}^{\dot{\beta}}
					   \frac{\partial}{\partial x^{\mu}}
    \hspace{2em}\mbox{and}\hspace{2em}
  \bar{Q}_{\dot{\beta}}	\;
    =\; -\, \frac{\partial}{\rule{0ex}{.8em}\partial \bar{\theta}^{\dot{\beta}}}\,
	      +\, \sqrt{-1}\sum_{\mu=0}^3 \sum_{\alpha=1}^2
		              \theta^{\alpha} \sigma^\mu_{\alpha\dot{\beta}}\frac{\partial}{\partial x^{\mu}}
 $$
 and derivations that are invariant under the flow that generate supersymmetries
 $$
   e_{\alpha^{\prime}}\;
    =\;  \frac{\partial}{\partial \theta^{\alpha}}
            +\, \sqrt{-1}\,\sum_{\mu=0}^3 \sum_{\dot{\beta}=\dot{1}}^{\dot{2}}
			          \sigma^{\mu}_{\alpha\dot{\beta}}\bar{\theta}^{\dot{\beta}}
					   \frac{\partial}{\partial x^{\mu}}
    \hspace{2em}\mbox{and}\hspace{2em}
  e_{\beta^{\prime\prime}}	\;
    =\; -\, \frac{\partial}{\rule{0ex}{.8em}\partial \bar{\theta}^{\dot{\beta}}}\,
	      -\, \sqrt{-1}\sum_{\mu=0}^3 \sum_{\alpha=1}^2
		              \theta^{\alpha} \sigma^\mu_{\alpha\dot{\beta}}\frac{\partial}{\partial x^{\mu}}\,.
 $$
Since
 $$
   \xi \eta^{\alpha^\prime}\;=\; \xi \bar{\eta}^{\dot{\beta}^\prime}\;
    =\; \xi\vartheta_{\gamma}\;=\; \xi\bar{\vartheta}_{\dot{\delta}}\;=\; 0\,,
 $$
 for all $\xi\in \Der_{\Bbb C}(\widehat{X})$,
the notion of
 {\it chiral functions}, {\it antichiral functions},  properties of chiral function ring, and antichiral function ring
 studied in [L-Y1: Sec.\,1.4] (D(14.1)) for $C^\infty(\widehat{X})$
 can be generalized immediately to parallel notions/objects
 for $C^\infty(\widehat{X}^{\widehat{\boxplus}_l})$
 with $f_{(\tinybullet)}$ ibidem replaced by $\breve{f}_{(\tinybullet)}$
 in the $(\theta,\bar{\theta})$-expansion
 for $\breve{f}\in C^\infty(\widehat{X}^{\widehat{\boxplus}_l})$.
 
\bigskip

\begin{definition-lemma}
{\bf [chiral sector $\widehat{X}^{\widehat{\boxplus}_l, \scriptsizech}$
           of $\widehat{X}^{\widehat{\boxplus}_l}$]}\; {\rm
 (1)
 $\; \breve{f}\in C^{\infty}(\widehat{X}^{\widehat{\boxplus}_l})$ is called {\it chiral}\,
       if $\;e_{1^{\prime\prime}}\breve{f}= e_{2^{\prime\prime}}\breve{f} =0\,$.
 The set of chiral functions on $\widehat{X}^{\widehat{\boxplus}_l}$ is a ${\Bbb C}$-subalgebra of
   $C^{\infty}(\widehat{X}^{\widehat{\boxplus}_l})$,
 called the {\it chiral function-ring } of $\widehat{X}^{\widehat{\boxplus}_l}$,
  denoted by $C^{\infty}(\widehat{X}^{\widehat{\boxplus}_l})^{\scriptsizech}$.
  
 (2)
 Replacing
   $\widehat{X}$ by $\widehat{U}$ for $U\subset X$ open   and
   $e_{1^{\prime\prime}}$, $e_{2^{\prime\prime}}$
     by   $e_{1^{\prime\prime}}|_{\widehat{U}}$, $e_{2^{\prime\prime}}|_{\widehat{U}}$
     in Item (1),
 one obtains the {\it sheaf  of chiral functions}
  $\widehat{\cal O}_X^{\,\widehat{\boxplus}_l,\scriptsizech}
    \subset \widehat{\cal O}_X^{\,\widehat{\boxplus}_l}$,
  also called the {\it chiral structure sheaf} of $\widehat{X}^{\widehat{\boxplus}_l}$.
 Denote the locally-ringed space $(X, {\cal O}_X^{\,\widehat{\boxplus}_l,\scriptsizech})$
   by $\widehat{X}^{\widehat{\boxplus}_l, \scriptsizech}$.
  
 (3) {\it
   The $C^\infty$-hull of $C^\infty(\widehat{X}^{\widehat{\boxplus}_l})$
    restricts to the $C^\infty$-hull of $C^\infty(\widehat{X}^{\widehat{\boxplus}_l})^{\scriptsizech}$
	and similarly for the localized version.
  This renders $\widehat{X}^{\widehat{\boxplus}_l,\scriptsizech}$	
    a complexified ${\Bbb Z}/2$-graded $C^\infty$-scheme.
	 }

 (4)	
 Similar to Definition~1.3.5
  under the fixed isomorphisms of spinor sheaves and the $\varepsilon$-tensor,
  the standard chiral coordinate functions $(x^\prime,\theta,\bar{\theta})$ on $\widehat{X}$,
   where
   $$
     x^{\prime\mu} \;
	 =\; x^\mu
	       + \sqrt{-1}\sum_{\alpha,\dot{\beta}}
		        \theta^\alpha \sigma^\mu_{\alpha\dot{\beta}}\bar{\theta}^{\dot{\beta}}\,,
   $$
  extends canonically to the {\it standard chiral coordinate functions}
			   $(x^\prime, \theta, \bar{\theta}, \eta, \bar{\eta},\vartheta,\bar{\vartheta})$
	 on $\widehat{X}^{\widehat{\boxplus}_l}$.
 {\it In terms of the standard chiral coordinate functions,
    an $\breve{f}\in C^\infty(\widehat{X}^{\widehat{\boxplus}_l})^{\scriptsizech}$
	 can be expressed as
    \begin{eqnarray*}
     \breve{f}
      & =  &  \breve{f}^\prime(x^\prime, \theta, \bar{\theta}, \eta,\bar{\eta}, \vartheta, \bar{\vartheta}) \\
      & =
	    &   \breve{f}^\prime_{(0)}(x^\prime, \eta, \bar{\eta}, \vartheta, \bar{\vartheta})\,
	        +\, \sum_\alpha    \theta^\alpha\,			
			   \breve{f}^\prime_{(\alpha)}(x^\prime, \eta,\bar{\eta}, \vartheta, \bar{\vartheta})
			+\, \theta^1\theta^2\,
			       \breve{f}^\prime_{(12)}(x^\prime, \eta,\bar{\eta}, \vartheta, \bar{\vartheta})\,,
    \end{eqnarray*}
     where
     $f^\prime_{(\tinybullet)}(x^\prime, \eta,\bar{\eta}, \vartheta, \bar{\vartheta})$
       are polynomials in
	     $\eta^\alpha$'s, $\bar{\eta}^{\dot{\beta}}$'s, $\vartheta^i_{\alpha_i}$'s ,
		 $\bar{\vartheta}^i_{\dot{\beta}_i}$'s
	   with coefficients smooth functions in $x^{\prime\mu}$'s.
    }			
}\end{definition-lemma}

\medskip

\begin{definition-lemma}
{\bf [antichiral sector $\widehat{X}^{\widehat{\boxplus}_l, \scriptsizeach}$
           of $\widehat{X}^{\widehat{\boxplus}_l}$]}\; {\rm
 (1)
 $\; \breve{f}\in C^{\infty}(\widehat{X}^{\widehat{\boxplus}_l})$ is called {\it antichiral}\,
       if $\;e_{1^\prime}\breve{f}= e_{2^\prime}\breve{f} =0\,$.
 The set of antichiral functions on $\widehat{X}^{\widehat{\boxplus}_l}$ is a ${\Bbb C}$-subalgebra of
   $C^{\infty}(\widehat{X}^{\widehat{\boxplus}_l})$,
 called the {\it antichiral function-ring } of $\widehat{X}^{\widehat{\boxplus}_l}$,
  denoted by $C^{\infty}(\widehat{X}^{\widehat{\boxplus}_l})^{\scriptsizeach}$.
  
 (2)
 Replacing
   $\widehat{X}$ by $\widehat{U}$ for $U\subset X$ open   and
   $e_{1^\prime}$, $e_{2^\prime}$
     by   $e_{1^\prime}|_{\widehat{U}}$, $e_{2^\prime}|_{\widehat{U}}$
     in Item (1),
 one obtains the {\it sheaf  of antichiral functions}
  $\widehat{\cal O}_X^{\,\widehat{\boxplus}_l,\scriptsizeach}
    \subset \widehat{\cal O}_X^{\,\widehat{\boxplus}_l}$,
  also called the {\it antichiral structure sheaf} of $\widehat{X}^{\widehat{\boxplus}_l}$.
 Denote the locally-ringed space $(X, {\cal O}_X^{\,\widehat{\boxplus}_l,\scriptsizeach})$
   by $\widehat{X}^{\widehat{\boxplus}_l, \scriptsizeach}$.
  
 (3) {\it
   The $C^\infty$-hull of $C^\infty(\widehat{X}^{\widehat{\boxplus}_l})$
    restricts to the $C^\infty$-hull of $C^\infty(\widehat{X}^{\widehat{\boxplus}_l})^{\scriptsizeach}$ 
	and similarly for the localized version.
  This renders $\widehat{X}^{\widehat{\boxplus}_l,\scriptsizeach}$	
    a complexified ${\Bbb Z}/2$-graded $C^\infty$-scheme.
	 }

 (4)	
 Similar to Definition~1.3.5
  under the fixed isomorphisms of spinor sheaves and the $\varepsilon$-tensor,
  the standard antichiral coordinate functions $(x^{\prime\prime},\theta,\bar{\theta})$ on $\widehat{X}$,
   where
   $$
     x^{\prime\prime\mu} \;
	 =\; x^\mu
	       - \sqrt{-1}\sum_{\alpha,\dot{\beta}}
		        \theta^\alpha \sigma^\mu_{\alpha\dot{\beta}}\bar{\theta}^{\dot{\beta}}\,,
   $$
  extends canonically to the {\it standard antichiral coordinate functions}
    $(x^{\prime\prime}, \theta, \bar{\theta}, \eta, \bar{\eta},\vartheta,\bar{\vartheta})$
	on $\widehat{X}^{\widehat{\boxplus}_l}$.
 {\it In terms of the standard antichiral coordinate functions,
    an $\breve{f}\in C^\infty(\widehat{X}^{\widehat{\boxplus}_l})^{\scriptsizeach}$
	 can be expressed as
    \begin{eqnarray*}
     \breve{f}
      & =  &  \breve{f}^{\prime\prime}
	                (x^{\prime\prime}, \theta, \bar{\theta}, \eta,\bar{\eta}, \vartheta, \bar{\vartheta}) \\
      & =
	    &   \breve{f}^{\prime\prime}_{(0)}
		        (x^{\prime\prime}, \eta, \bar{\eta}, \vartheta, \bar{\vartheta})\,
	        +\, \sum_{\dot{\beta}}    \bar{\theta}^{\dot{\beta}}\,			
			   \breve{f}^{\prime\prime}_{(\dot{\beta})}
			     (x^{\prime\prime}, \eta,\bar{\eta}, \vartheta, \bar{\vartheta})
			+\, \bar{\theta}^{\dot{1}}\bar{\theta}^{\dot{2}}\,
			       \breve{f}^{\prime\prime}_{(\dot{1}\dot{2})}
				    (x^{\prime\prime}, \eta,\bar{\eta}, \vartheta, \bar{\vartheta})\,,
    \end{eqnarray*}
     where
     $f^{\prime\prime}_{(\tinybullet)}(x^{\prime\prime}, \eta,\bar{\eta}, \vartheta, \bar{\vartheta})$
       are polynomials in
	     $\eta^\alpha$'s, $\bar{\eta}^{\dot{\beta}}$'s, $\vartheta^i_{\alpha_i}$'s,
		 $\bar{\vartheta}^i_{\dot{\beta}_i}$'s
	   with coefficients smooth functions in $x^{\prime\prime\mu}$'s.
	 }			
}\end{definition-lemma}
 
\bigskip

\noindent
We refer readers to ibidem for details.

By construction, one has the following diagram of inclusions of complexified ${\Bbb Z}/2$-graded $C^\infty$-rings
 $$
  \xymatrix{
     & C^\infty(\widehat{X}^{\widehat{\boxplus}_l}) &  \\
   C^\infty(\widehat{X}^{\widehat{\boxplus}_l})^{\scriptsizech} \ar@{^{(}->}[ur]
     & C^\infty(\widehat{X}) \ar@{^{(}->}[u]
	 &  C^\infty(\widehat{X}^{\widehat{\boxplus}_l})^{\scriptsizeach}
                 \ar@{_{(}->}[ul]	 \\
   C^\infty(\widehat{X})^{\scriptsizech}\rule{0ex}{1.2em}
          \ar@{^{(}->}[u]  \ar@{^{(}->}[ur]
     &&  C^\infty(\widehat{X})^{\scriptsizeach}\rule{0ex}{1.2em}
       	          \ar@{^{(}->}[u]       \ar@{_{(}->}[ul]\,. 	
	}
 $$
Which gives rise to
 the following diagram of dominant morphisms of complexified ${\Bbb Z}/2$-graded $C^\infty$-schemes
 $$
   \xymatrix{
     & \widehat{X}^{\widehat{\boxplus}_l}
	       \ar@{->>}[dl]     \ar@{->>}[d] \ar@{->>}[dr]     &  \\
    \widehat{X}^{\widehat{\boxplus}_l, \scriptsizech} \ar@{->>}[d]
	  & \widehat{X}   \ar@{->>}[dl]     \ar@{->>}[dr]
	  &  \widehat{X}^{\widehat{\boxplus}_l, \scriptsizeach} \ar@{->>}[d]	 \\
    \widehat{X}^{\scriptsizech}
     && \;\;\;\;\widehat{X}^{\scriptsizeach}\,.
	}
 $$

The following sector of $\widehat{X}^{\widehat{\boxplus}_l}$ is introduced
 for the purpose of studying supersymmetric $U(1)$ gauge theory on $X$ in Sec.\,3.
In some sense, it is the ``{\it real sector}"  of $\widehat{X}^{\widehat{\boxplus}_l}$:
 
\bigskip

\begin{definition} {\bf [self-twisted-conjugate sector of $\widehat{X}^{\widehat{\boxplus}_l}$]}\;
{\rm
 An element $\breve{f}\in C^\infty(\widehat{X}^{\widehat{\boxplus}_l})$ is
  called {\it self-twisted-conjugate } if $\breve{f}^\dag=\breve{f}$.
 The set of all self-twisted-conjugate elements in $ C^\infty(\widehat{X}^{\widehat{\boxplus}_l})$
   is denoted by $C^\infty(\widehat{X}^{\widehat{\boxplus}_l})^{\stc}$.
 Note that this is only a $C^\infty(X)$-module, not a $C^\infty(X)$-algebra.
}\end{definition}

\bigskip

\subsection{The physics sector $X^{\physics}$ of $\widehat{X}^{\widehat{\boxplus}_l}$}

We now proceed to identify the physical elements in $C^\infty(\widehat{X}^{\widehat{\boxplus}_l})$
 and construct the physics sector $X^{\physics}$ under $\widehat{X}^{\widehat{\boxplus}_l}$.

\bigskip

\begin{flushleft}	
{\bf The guide from [W-B] of Wess \& Bagger to identify the physical elements
           in each field-theory level of $C^\infty(\widehat{X}^{\widehat{\boxplus}_l})$}
\end{flushleft}
Assume that $\widehat{X}^{\widehat{\boxplus}}$ has only one field-theory level (i.e.\;$l=1$)  and
 suppress the parameter level.
We begin with the following question:\footnote{Notations
                                                                             and conventions here follow [W-B] of Wess \& Bagger
																			 as much as we can for the convenience of the illuminations.
																			  }  
 \begin{itemize}
  \item[{\bf Q.}]
  {\it What should the chiral functions in the sought-for physics sector be?\\
          Similarly, for antichiral functions?}
 \end{itemize}
 The answer, if taken from [W-B: Chap.\,V, Eq.\,(5.3) ] Wess \& Bagger, would be
   $$
     \Phi =A(y)+\sqrt{2} \theta\psi(y) +\theta\theta F(y)\,,
   $$
  where $y=x+ \sqrt{-1}\theta\sigma\bar{\theta}$ is the chiral coordinate.
 In this expression,
  $A$ and $F$  are {\it scalar} functions (and hence are commuting)
    while $\psi$ is a two-component spinor (whose components are thus anticommuting
	--- which renders $\theta\psi$ and hence $\Phi$ even).
 So far so good, until one starts to think deeper.
 Chiral functions are required to form a ring.
 In particular, the multiplication of two chiral functions should be also chiral:
 (cf.\;[W-B: Chap.\,V, Eq.\,(5.7)])
  {\small
  \begin{eqnarray*}
    \tilde{\Phi}\;\;:=\;\;\Phi_1 \Phi_2
	 & = & A_1(y)A_2(y)
          + \sqrt{2}\theta\,
		       \mbox{\large $($}
			      \phi_1(y)A_2(y) + A_1(y)\psi_2(y)
		       \mbox{\large $)$}  \\
     &&  \hspace{2em}		
		  + \theta\theta
		       \mbox{\Large $($}
			    A_1(y)F_2(y)+ A_2(y)F_1(y) -\psi_1(y)\psi_2(y)
			   \mbox{\Large $)$}   \\
     &\:=:&
       	\tilde{A}(y)+\sqrt{2} \theta \tilde{\psi}(y) +\theta\theta \tilde{F}(y)\,.
  \end{eqnarray*}}When   
 physicists say that $F$	is a scalar, it is referred to the fact that $F$ is a Lorentz scalar
   (i.e.\; sections of the associated bundle from the trivial representation of the principal Lorentz-group bundle).
 This says nothing about the {\it nilpotency}\footnote{\makebox[9em][l]{\it Note for physicists}
                                                                        An element $r\ne 0$ of a (either commutative or noncommutative) ring $R$
    																	  is called {\it nilpotent}
																		 if $r^{l+1}=0$ for some $l\ge 1$.
																		The minimal such $l$ is called the {\it nilpotency} of $r$.
																		In particular, any odd element of a ${\Bbb Z}/2$-graded ring is nilpotent.
																	  }  
	of $F$.
 The coefficient $\tilde{F}$ in this product reveals something peculiar:
   \begin{itemize}
    \item[\LARGE $\cdot$]
	 The summand $\psi_1(y)\psi_2(y)$ in the coefficient $\tilde{F}$ of the product $\Phi_1\Phi_2$
	  is nilpotent.
   \end{itemize}
 If we want to include $\Phi_1\Phi_2$ in the chiral ring, then
  we must allow $\tilde{F}$ to have nilpotent summand.
 Now, {\it if $F_1$ and $F_2$ themselves are not nilpotent, then
  this poses an issue}:
  \begin{itemize}
   \item[\LARGE $\cdot$]  {\it
    $\tilde{F}$ now has a non-nilpotent summand $A_1(y)F_2(y)+ A_2(y)F_1(y)$
 	  and a nilpotent summand $-\,\psi_1(y)\psi_2(y)$.
    These two summands are not like terms and
	  hence should be treated as different degrees of freedom.}\footnote{\makebox[9em][l]{\it Note for physicists}
	                                                                                 This is completely analogous to the statement
																					   that the expression $f+\sqrt{-1}g$, $f, g\in C^\infty({\Bbb R}^n)$,
                           															        has two degrees of freedom
																					  or that the ring ${\Bbb C}[t]/(t^{l+1})$ is $(l+1)$-dimensional
																					        as a ${\Bbb C}$-vector space.
																									  }  
  \end{itemize}
 But the chiral multiplet from the representation of $d=3+1$, $N=1$ supersymmetry algebra requires that
   the coefficient $F$  of $\theta\theta$ contribute as same degree of freedom as the total-$\theta$-degree-zero term $A$.
 Thus, the only way to avoid such a contradiction is demand that
   \begin{itemize}
    \item[\LARGE $\cdot$]
   {\it $F$ must be a nilpotent Lorentz scalar.}
   \end{itemize}
   
  It is in this way, one deduces that physical chiral functions must be of the form (in our notation)
   $$
    \breve{f}\;
	  =\; f^{(0)}_{(0)}(x^\prime)
	       + \sum_\alpha \theta^\alpha(\vartheta_\alpha f^{(\alpha)}_{(\alpha)}(x^\prime))
		   + \theta^1\theta^2 (\vartheta_1\vartheta_2 f^{(12)}_{(12)}(x^\prime))\,,
   $$
   where $(x^\prime,\theta,\bar{\theta})$ are chiral coordinate functions on $\widehat{X}$.
 Such expressions are now closed under multiplications and hence form a ring.
 They now agree with the chiral multiplet representation of the $d=3+1$, $N=1$ supersymmetry algebra.
    
 Similar reasoning for antichiral functions implies that physical antichiral functions must be of the form (in our notation)
   $$
    \breve{g}\;
	  =\; g^{(0)}_{(0)}(x^{\prime\prime})
	       + \sum_{\dot{\beta}} \bar{\theta}^{\dot{\beta}}
		        (\bar{\vartheta}_{\dot{\beta}}
				   g^{(\dot{\beta})}_{(\dot{\beta})}(x^{\prime\prime}))
		   + \bar{\theta}^{\dot{1}}\bar{\theta}^{\dot{2}}
  		     (\bar{\vartheta}_{\dot{1}}\bar{\vartheta}_{\dot{2}}
			     g^{(\dot{1}\dot{2})}_{(\dot{1}\dot{2})}(x^{\prime\prime}))\,,
   $$
   where $(x^{\prime\prime},\theta,\bar{\theta})$ are antichiral coordinate functions on $\widehat{X}$.
 
 The physical function-ring must contain both the physical chiral function ring and the physical antichiral function ring,
   and hence their product as well.
 Multiplying a physical chiral $\breve{f}$ and a physical antichiral $\breve{g}$ as given above
  and re-expressing the product $\breve{f}\breve{g}$ in the standard coordinate functions
   $(x,\theta,\bar{\theta})$ on $\widehat{X}$
  leads to the notion of {\it physical superfields} in Definition~1.4.1 in the next theme.
																
 From the above illumination,
   one sees that the construction is a minimal one:
 {\it We only include those that are required to make a ring, beginning with the following two demands}
      \begin{itemize}
	   \item[(1)]
	    {\it Chiral functions must from a ring; so does antichiral functions.}
	
	   \item[(2)]
	   {\it A chiral superfield must match the chiral multiplet and an antichiral superfield must match antichiral multiplet
	         in representations of $d=3+1$, $N=1$ supersymmetry algebra.}
      \end{itemize}	
 Thus, as long as physical relevance is concerned,
  $C^\infty(X^{\physics})$ in Definition~1.4.1 is unique.

\bigskip
  
The above reasoning and construction works field-theory level by field-theory level.
Once the physical elements of $C^\infty(\widehat{X}^{\widehat{\boxplus}})$  from each field-theory level
  are identified,
the complexified $C^\infty$-subring in $C^\infty(\widehat{X}^{\widehat{\boxplus}})$ generated by them
 should be the sought-for $C^\infty(\widehat{X}^{\widehat{\boxplus}})^{\physics}$.

\bigskip

\begin{flushleft}
{\bf The purely even physical structure sheaf ${\cal O}_X^{\physics}$ on $X$}
\end{flushleft}
For the simplicity of notations,
 we assume that $\widehat{X}^{\widehat{\boxplus}}$ has only one field-theory level (i.e.\;$l=1$)
  and suppress the parameter level of $\widehat{X}^{\widehat{\boxplus}}$.

Recall then the standard coordinate functions
 $$
   (x^0, x^1, x^2, x^3;
	       \theta^1, \theta^2, \bar{\theta}^{\dot{1}}, \bar{\theta}^{\dot{2}};
		   \vartheta_1,\vartheta_2, \bar{\vartheta}_{\dot{1}}, \bar{\vartheta}_{\dot{2}} )
 $$
 on $\widehat{X}^{\widehat{\boxplus}}$,
 denoted collectively by
  $(x, \theta,\bar{\theta},\vartheta,\bar{\vartheta})$ or
  $(x^\mu,
        \theta^\alpha, \bar{\theta}^{\dot{\beta}},
		\vartheta_{\gamma}, \bar{\vartheta}_{\dot{\delta}})
		                                             _{\mu, \alpha, \dot{\beta}, \gamma, \dot{\delta}}$.
  
\bigskip

\begin{definition} {\bf [superfield in physical sector/physical superfield]}\; {\rm
 An $\breve{f}\in C^{\infty}(\widehat{X}^{\widehat{\boxplus}})$ is called
   a {\it superfield in the physical sector} of $\widehat{X}^{\widehat{\boxplus}}$
  if, as a polynomial in the anticommuting coordinate-functions
   $(\theta,\bar{\theta},\vartheta,\bar{\vartheta})$,  it is of the following form
 {\small
 \begin{eqnarray*}
  \breve{f}
   & =  &
   \breve{f}_{(0)}
   + \sum_{\alpha}\theta^\alpha\breve{f}_{(\alpha)}
   + \sum_{\dot{\beta}}\bar{\theta}^{\dot{\beta}}\breve{f}_{(\dot{\beta})}
   + \theta^1\theta^2 \breve{f}_{(12)}
   + \sum_{\alpha,\dot{\beta}} \theta^\alpha\bar{\theta}^{\dot{\beta}}
           \breve{f}_{(\alpha\dot{\beta})}
   + \bar{\theta}^{\dot{1}}\bar{\theta}^{\dot{2}} \breve{f}_{(\dot{1}\dot{2})} \\
 && \hspace{4em}
   + \sum_{\dot{\beta}}\theta^1\theta^2\bar{\theta}^{\dot{\beta}}
           \breve{f}_{(12\dot{\beta})}
   + \sum_\alpha \theta^\alpha\bar{\theta}^{\dot{1}}\bar{\theta}^{\bar{\dot{2}}}
           \breve{f}_{(\alpha\dot{1}\dot{2})}
   + \theta^1\theta^2\bar{\theta}^{\dot{1}}\bar{\theta}^{\dot{2}}
           \breve{f}_{(12\dot{1}\dot{2})} \\
  &= &
    f^{(0)}_{(0)}
	+ \sum_{\alpha}\theta^\alpha\vartheta_\alpha f^{(\alpha)}_{(\alpha)}
	+ \sum_{\dot{\beta}}
	       \bar{\theta}^{\dot{\beta}}\bar{\vartheta}_{\dot{\beta}}
		     f^{(\dot{\beta})}_{(\dot{\beta})}  \\
  && \hspace{1em}			
    +\; \theta^1\theta_2\vartheta_1\vartheta_2 f^{(12)}_{(12)}
	+ \sum_{\alpha,\dot{\beta}}\theta^\alpha \bar{\theta}^{\dot{\beta}}
	      \left(\rule{0ex}{1.2em}\right.\!
		    \sum_\mu \sigma^\mu_{\alpha\dot{\beta}} f^{(0)}_{[\mu]}\,
			 +\, \vartheta_\alpha\bar{\vartheta}_{\dot{\beta}}
			        f^{(\alpha\dot{\beta})}_{(\alpha\dot{\beta})}
		  \!\left.\rule{0ex}{1.2em}\right)
    + \bar{\theta}^{\dot{1}}\bar{\theta}^{\dot{2}}
	   \bar{\vartheta}_{\dot{1}}\bar{\vartheta}_{\dot{2}}
	    f^{(\dot{1}\dot{2})}_{(\dot{1}\dot{2})}  \\
  && \hspace{1em}		
	+ \sum_{\dot{\beta}}\theta^1\theta^2\bar{\theta}^{\dot{\beta}}
	     \left(\rule{0ex}{1.2em}\right.\!
		   \sum_\alpha \vartheta_\alpha f^{(\alpha)}_{(12\dot{\beta})}\,
		    +\, \vartheta_1\vartheta_2\bar{\vartheta}_{\dot{\beta}}
			            f^{(12\dot{\beta})}_{(12\dot{\beta})}
		 \!\left.\rule{0ex}{1.2em}\right)
    + \sum_\alpha \theta^\alpha\bar{\theta}^{\dot{1}}\bar{\theta}^{\dot{2}}
	   \left(\rule{0ex}{1.2em}\right.\!
	     \sum_{\dot{\beta}} \bar{\vartheta}_{\dot{\beta}}
		     f^{(\dot{\beta})}_{(\alpha\dot{1}\dot{2})}\,
         +\, \vartheta_\alpha \bar{\vartheta}_{\dot{1}}\bar{\vartheta}_{\dot{2}}	
		           f^{(\alpha\dot{1}\dot{2})}_{(\alpha\dot{1}\dot{2})}
	   \!\left.\rule{0ex}{1.2em}\right)\\
  && \hspace{1em}
	+\; \theta^1\theta^2\bar{\theta}^{\dot{1}}\bar{\theta}^{\dot{2}}
	     \left(\rule{0ex}{1.2em}\right.\!
		  f^{(0)}_{(12\dot{1}\dot{2})}
		  + \sum_{\alpha,\dot{\beta}} \vartheta_\alpha\bar{\vartheta}_{\dot{\beta}}
		         f^{(\alpha\dot{\beta})}_{(12\dot{1}\dot{2})}
		  + \vartheta_1\vartheta_2\bar{\vartheta}_{\dot{1}}\bar{\vartheta}_{\dot{2}}
		        f^{(12\dot{1}\dot{2})}_{(12\dot{1}\dot{2})}
		 \!\left.\rule{0ex}{1.2em}\right)\,,
 \end{eqnarray*}
  }where 
  $\alpha=1,2$; $\dot{\beta}=\dot{1}, \dot{2}$; $\mu=0,1,2,3$; and
  the {\it thirty-three} coefficients $f^{\tinybullet}_{\tinybullet}$
   of the $(\theta,\bar{\theta},\vartheta,\bar{\vartheta})$-monomial summands of $\breve{f}$
   are complex-valued functions on $X$:
 {\small
  $$
   f^{(0)}_{(0)};\,
   f^{(\alpha)}_{(\alpha)};\,  f^{(\dot{\beta})}_{(\dot{\beta})};\,
   f^{(12)}_{(12)};\,
   f^{(0)}_{([\mu])},  f^{(\alpha\dot{\beta})}_{(\alpha\dot{\beta})};\,
   f^{(\dot{1}\dot{2})}_{(\dot{1}\dot{2})};
   f^{(\alpha)}_{(12\dot{\beta})},  f^{(12\dot{\beta})}_{(12\dot{\beta})};
   f^{(\dot{\beta})}_{(\alpha\dot{1}\dot{2})},
      f^{(\alpha\dot{1}\dot{2})}_{(\alpha\dot{1}\dot{2})};
   f^{(0)}_{(12\dot{1}\dot{2})},  f^{(\alpha\dot{\beta})}_{12\dot{1}\dot{2}},
      f^{(12\dot{1}\dot{2})}_{(12\dot{1}\dot{2})}\;
    \in\;  C^\infty(X)^{\Bbb C}\,.
  $$}For  
 simplicity, such an $\breve{f}$ is also called a {\it physical superfield}
  on $\widehat{X}^{\widehat{\boxplus}}$.
}\end{definition}

\bigskip

\begin{lemma} {\bf [physical sector of $\widehat{X}^{\widehat{\boxplus}}$]}\;
 The collection of  physical superfields on $\widehat{X}^{\widehat{\boxplus}}$ as defined in Definition~1.4.1
  is an even subring of the complexified ${\Bbb Z}/2$-graded $C^\infty$-ring
  $C^\infty(\widehat{X}^{\widehat{\boxplus}})$.
 Denote this subring
   (also  a $C^{\infty}(X)^{\Bbb C}$-subalgebra
   of $C^{\infty}(\widehat{X}^{\widehat{\boxplus}})$)
   by $C^\infty(\widehat{X}^{\widehat{\boxplus}})^{\physics}$.
 Then, the $C^\infty$-hull of $C^\infty(\widehat{X}^{\widehat{\boxplus}})$ restricts to the $C^\infty$-hull
   of $C^\infty(\widehat{X}^{\widehat{\boxplus}})^{\physics}$, which is given by
   $$
     \mbox{$C^\infty$-hull}\,(C^\infty(\widehat{X}^{\widehat{\boxplus}})^{\physics})\;
	   =\; \{\breve{f}\in C^\infty(\widehat{X}^{\widehat{\boxplus}})^{\physics}
	                                                                            \,|\, \breve{f}_{(0)}\in C^\infty(X)\}\,.
   $$																					
\end{lemma}

\begin{proof}
 From the $(\theta,\bar{\theta}, \vartheta,\bar{\vartheta})$-expansion of a physical superfield $\breve{f}$,
  one concludes that it is even.
 That the set $C^\infty(\widehat{X}^{\widehat{\boxplus}})^{\physics}$ of physical superfields
     is a subring of $C^\infty(\widehat{X}^{\widehat{\boxplus}})$
   follows from the observation that
    as a $C^\infty(X)^{\Bbb C}$-module,
	$C^\infty(\widehat{X}^{\widehat{\boxplus}})^{\physics}$
 	is generated by thirty-three monomials\footnote{A
	                                                                    characterization of these
																		 $(\theta, \bar{\theta}, \vartheta, \bar{\vartheta})$-monomials
																		 is given as follows.
																		First, define a {\it balanced monomial} to be one whose
																		 $(\vartheta, \bar{\vartheta})$-factor matches exactly
																		    with the $(\theta,\bar{\theta})$-factor.
																		There are fifteen of them:
																		  $$
       1\,,\;																		
       \theta^\alpha\vartheta_\alpha\,,\;
       \bar{\theta}^{\dot{\beta}}\bar{\vartheta}_{\dot{\beta}}\,,\;
       \theta^1\theta^2\vartheta_1\vartheta_2\,,\;
       \theta^\alpha \bar{\theta}^{\dot{\beta}}\vartheta_\alpha\bar{\vartheta}_{\dot{\beta}}\,,\;
       \theta^1\theta^2\bar{\theta}^{\dot{\beta}}
	                                                                 \vartheta_1\vartheta_2\bar{\vartheta}_{\dot{\beta}}\,,\;																 
       \theta^\alpha\bar{\theta}^{\dot{1}}\bar{\theta}^{\dot{2}}
		                      \vartheta_\alpha \bar{\vartheta}_{\dot{1}}\bar{\vartheta}_{\dot{2}}\,,\;						  		 															
       \theta^1\theta^2\bar{\theta}^{\dot{1}}\bar{\theta}^{\dot{2}}	
		                                                              \vartheta_1\vartheta_2
			    													  \bar{\vartheta}_{\dot{1}}\bar{\vartheta}_{\dot{2}}\,,
																		  $$
																		     $\alpha=1,\,2,\;\dot{\beta}=\dot{1},\,\dot{2}$\,.
																		 Then, {\it reduce} from them
																		   by dropping a
																		   $\vartheta_{\gamma}\bar{\vartheta}_{\dot{\delta}}$-factor
																		   until there are no more such factors.
																		 For example,
																		   $$
             \theta^\alpha\bar{\theta}^{\dot{1}}\bar{\theta}^{\dot{2}}
		                           \vartheta_\alpha \bar{\vartheta}_{\dot{1}}\bar{\vartheta}_{\dot{2}}\;\;
                 \rightsquigarrow\;\;		
                       \theta^\alpha\bar{\theta}^{\dot{1}}\bar{\theta}^{\dot{2}}
		                                               \bar{\vartheta}_{\dot{\beta}}
		      \hspace{2em}\mbox{and}\hspace{2em}
             \theta^1\theta^2\bar{\theta}^{\dot{1}}\bar{\theta}^{\dot{2}}	
		                                                                \vartheta_1\vartheta_2
					        											  \bar{\vartheta}_{\dot{1}}\bar{\vartheta}_{\dot{2}}\;\;
				 \rightsquigarrow\;\;
		       \theta^1\theta^2\bar{\theta}^{\dot{1}}\bar{\theta}^{\dot{2}}		
			                                                     \vartheta_\alpha\bar{\vartheta}_{\dot{\beta}}\;\;
                 \rightsquigarrow\;\;
               \theta^1\theta^2\bar{\theta}^{\dot{1}}\bar{\theta}^{\dot{2}}\,.  		 				
																		   $$
																		    } 
	$$
	 \left\{\begin{array}{c}
	   1\,,\;
	      \theta^\alpha\vartheta_\alpha\,, \;\;\;\;
	      \bar{\theta}^{\dot{\beta}}\bar{\vartheta}_{\dot{\beta}}\,,\;\;\;\;
          \theta^1\theta^2\vartheta_1\vartheta_2\,,\;\;\;\;
		  \theta^\alpha \bar{\theta}^{\dot{\beta}}\vartheta_\alpha\bar{\vartheta}_{\dot{\beta}}\,,\;\;\;\;
	   	     \theta^\alpha \bar{\theta}^{\dot{\beta}}\,,\;\;\;\;
		  \bar{\theta}^{\dot{1}}\bar{\theta}^{\dot{2}}
	                                                             \bar{\vartheta}_{\dot{1}}\bar{\vartheta}_{\dot{2}}\,,     \\
	   \theta^1\theta^2\bar{\theta}^{\dot{\beta}}
	                                                             \vartheta_1\vartheta_2\bar{\vartheta}_{\dot{\beta}}\,,\;\;\;\;
             \theta^1\theta^2\bar{\theta}^{\dot{\beta}}\vartheta_\alpha\,,\;\;\;\;
		  \theta^\alpha\bar{\theta}^{\dot{1}}\bar{\theta}^{\dot{2}}
		                  \vartheta_\alpha \bar{\vartheta}_{\dot{1}}\bar{\vartheta}_{\dot{2}}\,,\;\;\;\;
	         \theta^\alpha\bar{\theta}^{\dot{1}}\bar{\theta}^{\dot{2}}
			                                                     \bar{\vartheta}_{\dot{\beta}}\,,                 \\
       \theta^1\theta^2\bar{\theta}^{\dot{1}}\bar{\theta}^{\dot{2}}	
		                                                          \vartheta_1\vartheta_2
																  \bar{\vartheta}_{\dot{1}}\bar{\vartheta}_{\dot{2}}\,,\;\;\;\;
		  \theta^1\theta^2\bar{\theta}^{\dot{1}}\bar{\theta}^{\dot{2}}		
			                                                     \vartheta_\alpha\bar{\vartheta}_{\dot{\beta}}\,,\;\;\;\;													  
          \theta^1\theta^2\bar{\theta}^{\dot{1}}\bar{\theta}^{\dot{2}}  		
         \end{array}																
		 \right\}_{\alpha=1,\,2;\, \dot{\beta}=\dot{1},\,\dot{2}}
	$$
	and, up to a sign factor, this set is closed under multiplications.
  Finally, since these monomials are even, they commute with each other.
  Furthermore, except the monomial $1$, they are all nilpotent.
  This implies that
   the $C^\infty$-hull of $C^\infty(\widehat{X}^{\widehat{\boxplus}})$ restricts to the $C^\infty$-hull
    of $C^\infty(\widehat{X}^{\widehat{\boxplus}})^{\physics}$ and
   that the latter is given by
    $$
      \mbox{\it $C^\infty$-hull}\,(C^\infty(\widehat{X}^{\widehat{\boxplus}})^{\physics})\;
	   =\; \{\breve{f}\in C^\infty(\widehat{X}^{\widehat{\boxplus}})^{\physics}
	                                                                            \,|\, \breve{f}_{(0)}\in C^\infty(X)\}\,.
    $$																					
  This completes the proof.
  
\end{proof}

\bigskip

By localizing all the constructions and discussions to open sets of $X$, one obtains
a new complexified $C^\infty$-scheme supported on $X$:

\bigskip

\begin{definition} {\bf [$X^{\physics}$ as (purely even) complexified $C^\infty$-scheme]}\; {\rm
 Let ${\cal O}_X^{\physics}$ be the sheaf on $X$ associated to the assignment
  $U\mapsto C^{\infty}({\widehat{U}^{\widehat{\boxplus}}})^{\physics}$ for open sets $U$ of $X$.
 This is a sheaf of complexified $C^\infty$-rings on $X$.
 Denote by $X^{\physics}$
   the associate complexified $C^\infty$-scheme $(X,{\cal O}_X^{\physics})$.
 Note that $X^{\physics}$ is purely even\footnote{Here,
                                               we denote this scheme by $X^{\tinyphysics}$,
											     rather than $\widehat{X}^{\tinyphysics}$,
												 to emphasize that it is purely even.}.
}\end{definition}

\bigskip

By construction, one has the following commutative diagram of dominant morphisms
  (of complexified ${\Bbb Z}/2$-graded $C^\infty$-schemes,
     where both the odd part of $X^{\physics}$ and $X^{\Bbb C}$  are zero)
$$
  \xymatrix{
    \,&     & \;\;\widehat{X}^{\widehat{\boxplus}} \ar@{->>}[dl] \ar@{->>}[dd]  \ar@{->>}[dr]  \\
	& \widehat{X}\ar@{->>}[dr]    && \hspace{-1em}X^{\physics} \ar@{->>}[dl]  \\
	 & & \;\;X^{\Bbb C}&&.
 }
$$

\bigskip

The chiral sector of $\widehat{X}^{\widehat{\boxplus}}$ restricts to the chiral sector of $X^{\physics}$ and
the antichiral sector of $\widehat{X}^{\widehat{\boxplus}}$ restricts to the antichiral sector of $X^{\physics}$.
We will look at them more closely in Sec.\,2.1.

\bigskip

For $l\ge 2$,  as a $C^\infty(X)^{\Bbb C}$-module,
 the complexified $C^\infty$-subring $C^\infty(\widehat{X}^{\widehat{\boxplus}_l})^{\physics}$
  of $C^\infty(\widehat{X}^{\widehat{\boxplus}_l})$
 is generated by elements from the product of the sets from each field-theory level of
   $\widehat{X}^{\widehat{\boxplus}_l}$
 $$
	 \left\{\begin{array}{c}
	   1\,,\;
	      \theta^\alpha\vartheta^i_\alpha\,, \;\;\;\;
	      \bar{\theta}^{\dot{\beta}}\bar{\vartheta}^i_{\dot{\beta}}\,,\;\;\;\;
          \theta^1\theta^2\vartheta^i_1\vartheta^i_2\,,\;\;\;\;
		  \theta^\alpha \bar{\theta}^{\dot{\beta}}
		                               \vartheta^i_\alpha\bar{\vartheta}^i_{\dot{\beta}}\,,\;\;\;\;
	   	     \theta^\alpha \bar{\theta}^{\dot{\beta}}\,,\;\;\;\;
		  \bar{\theta}^{\dot{1}}\bar{\theta}^{\dot{2}}
	                                                             \bar{\vartheta}^i_{\dot{1}}\bar{\vartheta}^i_{\dot{2}}\,,     \\
	   \theta^1\theta^2\bar{\theta}^{\dot{\beta}}
	                                                             \vartheta^i_1\vartheta^i_2\bar{\vartheta}^i_{\dot{\beta}}\,,\;\;\;\;
             \theta^1\theta^2\bar{\theta}^{\dot{\beta}}\vartheta^i_\alpha\,,\;\;\;\;
		  \theta^\alpha\bar{\theta}^{\dot{1}}\bar{\theta}^{\dot{2}}
		                  \vartheta^i_\alpha \bar{\vartheta}^i_{\dot{1}}\bar{\vartheta}^i_{\dot{2}}\,,\;\;\;\;
	         \theta^\alpha\bar{\theta}^{\dot{1}}\bar{\theta}^{\dot{2}}
			                                                     \bar{\vartheta}^i_{\dot{\beta}}\,,                 \\
       \theta^1\theta^2\bar{\theta}^{\dot{1}}\bar{\theta}^{\dot{2}}	
		                                                          \vartheta^i_1\vartheta^i_2
																  \bar{\vartheta}^i_{\dot{1}}\bar{\vartheta}^i_{\dot{2}}\,,\;\;\;\;
		  \theta^1\theta^2\bar{\theta}^{\dot{1}}\bar{\theta}^{\dot{2}}		
			                                                     \vartheta^i_\alpha\bar{\vartheta}^i_{\dot{\beta}}\,,\;\;\;\;													  
          \theta^1\theta^2\bar{\theta}^{\dot{1}}\bar{\theta}^{\dot{2}}  		
         \end{array}																
		 \right\}_{\alpha=1,\,2;\, \dot{\beta}=\dot{1},\,\dot{2}}\,,
	$$
   $i=1,\,\ldots\,, l$\,.
 Its $C^\infty$-hull remains of the form
    $$
      \mbox{\it $C^\infty$-hull}\,(C^\infty(\widehat{X}^{\widehat{\boxplus}_l})^{\physics})\;
	   =\; \{\breve{f}\in C^\infty(\widehat{X}^{\widehat{\boxplus}_l})^{\physics}
	                                                                            \,|\, \breve{f}_{(0)}\in C^\infty(X)\}\,.
    $$			
However, clearly it becomes extremely messy to express
  a $(\theta,\bar{\theta},\vartheta,\bar{\vartheta})$-expansion
  of a general element $\breve{f}\in C^\infty(\widehat{X}^{\widehat{\boxplus}_l})^{\physics}$.
For the purpose of this work, we assume therefore for the rest of the work that $l=1$ for the simplicity of notations
  though there is no technical difficulty to generalize to the $l\ge 2$ case.
 Cf.\;{\sc Figure}~1-4-1; note that all the schemes involved
   $\widehat{X}\,,\; \widehat{X}^{\widehat{\boxplus}_l}\,,\; X^{\physics}	$
   have the same underlying topology $X$, which is indicated by the ``foggy/fuzzy" nature of the fermionic cloud
   carried by these schemes in the illustration.
  \begin{figure}[htbp]
 \bigskip
  \centering
  \includegraphics[width=0.80\textwidth]{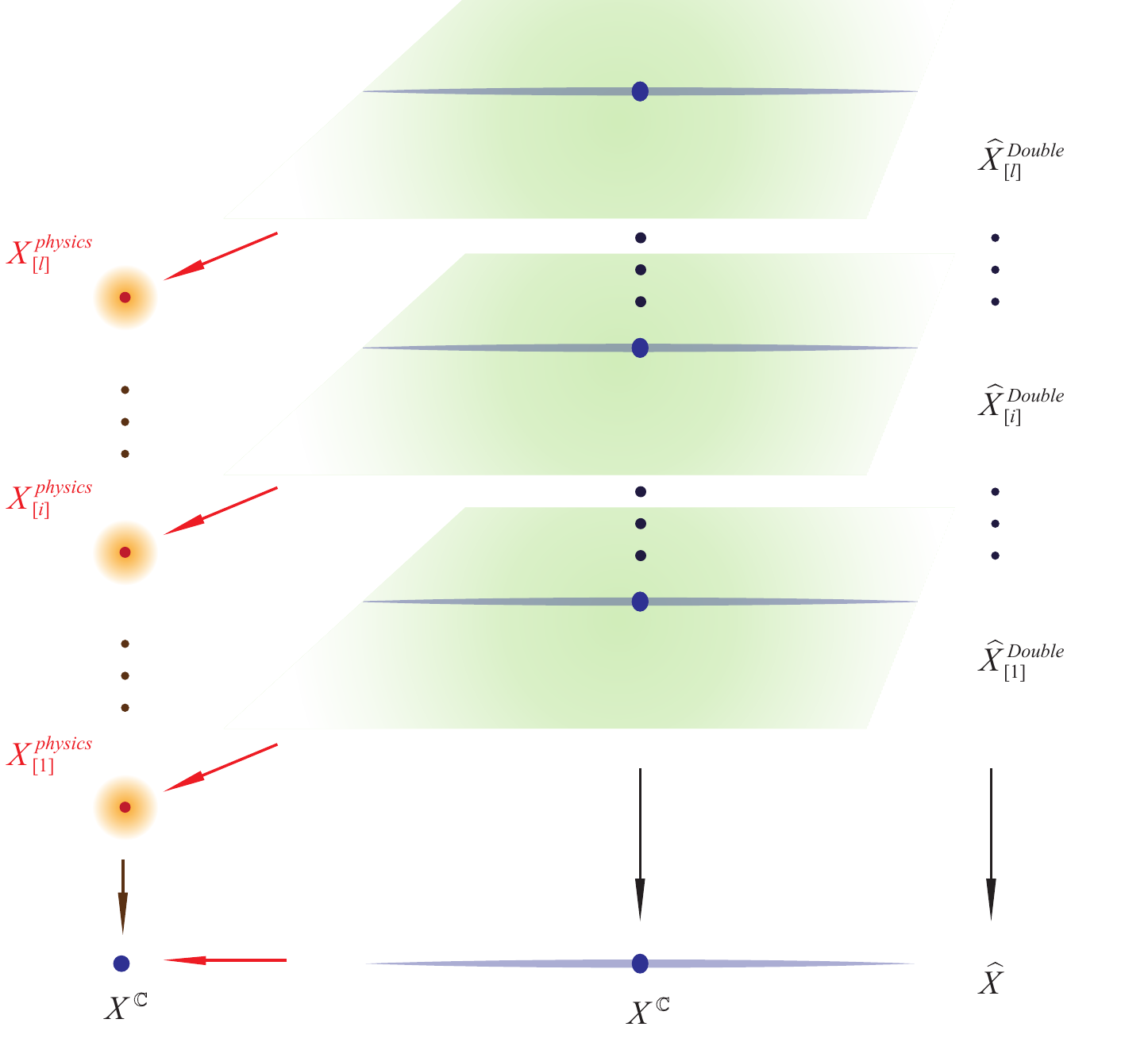}

  \bigskip
  \bigskip
 \centerline{\parbox{13cm}{\small\baselineskip 12pt
  {\sc Figure}~1-4-1.
  The space-time coordinate functions $x^\mu$, $\mu=0,1,2,3$,
     and the fermionic coordinate functions $\theta^\alpha$, $\bar{\theta}^{\dot{\beta}}$,
     $\alpha=1,2$, $\dot{\beta}=\dot{1}, \dot{2}$, 	
   generate the function ring of the fundamental superspace $\widehat{X}$
     as a complexified ${\Bbb Z}/2$-graded $C^\infty$-scheme.
  Over it sits a supertower with Grassmann-number level and other field-theory levels that are needed
    for the construction of supersymmetric quantum field theories.
  From the direct-sum expression of the generating sheaf 	
	\begin{eqnarray*}
    {\cal F}
	   & :=\:
	      &  ({\cal S}^{\prime\,\vee}_{\coordinates}\oplus {\cal S}^{\prime\prime\,\vee}_{\coordinates})
		        \oplus
		({\cal S}^{\prime\,\vee}_{\parameter}\oplus {\cal S}^{\prime\prime\,\vee}_{\parameter})\\
        &&\hspace{2em} 		
			    \oplus\,
				  \mbox{$\bigoplus$}_{i=1}^l
				     ({\cal S}^\prime_{\field, i}\oplus {\cal S}^{\prime\prime}_{\field, i})\\[1ex]
      & =
	      &  ({\cal S}^{\prime\,\vee}_{\coordinates}\oplus {\cal S}^{\prime\prime\,\vee}_{\coordinates})
		        \oplus
		({\cal S}^{\prime\,\vee}_{\parameter}\oplus {\cal S}^{\prime\prime\,\vee}_{\parameter})\\[.6ex]
        &&\hspace{2em} 		
			    \oplus\,
				  \mbox{$\bigoplus$}_{i=1}^l				
				    \frac{
					 ({\cal S}^{\prime\,\vee}_{\coordinates}\oplus {\cal S}^{\prime\prime\,\vee}_{\coordinates}
					       \oplus  {\cal S}^\prime_{\field, i}\oplus {\cal S}^{\prime\prime}_{\field, i})}
					{{\cal S}^{\prime\,\vee}_{\coordinates}\oplus {\cal S}^{\prime\prime\,\vee}_{\coordinates}}
   \end{eqnarray*}
   of the structure sheaf $\widehat{\cal O}_X^{\,\widehat{\boxplus}}$
     of $\widehat{X}^{\widehat{\boxplus}}$,
   one may think of each field-theory level as contributing
   a floor-$[i]$
    $$
	   \widehat{X}_{[i]}^{\tinyDouble}\; :=\;
	     \mbox{\Large $($}
		  X, \mbox{$\bigwedge$}_{{\cal O}_X^{\,\Bbb C}}^{\tinybullet}
		     ({\cal S}^{\prime\,\vee}_{\coordinates}\oplus {\cal S}^{\prime\prime\,\vee}_{\coordinates}
					       \oplus  {\cal S}^\prime_{\field, i}\oplus {\cal S}^{\prime\prime}_{\field, i})
		 \mbox{\Large $)$}
	$$
   over $\widehat{X}$  and
 these field-theory floors are glued by the ${\Bbb Z}/2$-graded version of
    fibered product over $\widehat{X}$ to give $\widehat{X}^{\widehat{\boxplus}}$.	
 Each field-theory floor $\widehat{X}_{[i]}^{\tinyDouble}$ has its own physics sector
   $X_{[i]}^{\tinyphysics}$ that is purely even.
 They generate the physics sector $X^{\tinyphysics}$	of $\widehat{X}^{\widehat{\boxplus}}$
    that is also purely even.
 This physics sector is where most of physics-relevant superfields lie.
  }}
\end{figure}
  
Also, recall that the Grassmann parameter level is introduced to serve the need
  when a discussion/computation/expression requires (or becomes more convenient to present with)
    the Grassmann parameter $(\eta,\bar{\eta})$.
This level will thus be suppressed from now on unless it is needed.

\bigskip
 
\subsection{Purge-evaluation maps and the Fundamental Theorem on supersymmetric action functionals}
 
In this subsection, we recast a fundamental theorem on supersymmetric action functionals
 (e.g.\ [Bi: Sec.\,4.3]) into a general form
 under the current (complexified, ${\Bbb Z}/2$-graded) $C^\infty$-Algebraic Geometry setting.

\bigskip
 
\begin{flushleft}
{\bf The need to get rid of nilpotency}
\end{flushleft}
So far so good.
But when one presses on to construct a supersymmetric action functional using the fermionic integration
 of the form
 $$
  \int_{\widehat{X}^{\widehat{\boxplus}}}
     d^4x\,
	 d\bar{\theta}^{\dot{2}}d\bar{\theta}^{\dot{1}}d\theta^2 d\theta^1 \,\breve{f}
 $$
 for $\breve{f}\in C^\infty(\widehat{X}^{\widehat{\boxplus}})^{\physics}$
  a derived physical superfield\footnote{See
                                                                         Sec.\,2.2 and Sec.\,3.5 for concrete examples.
													                      } 
    constructed from more basic physical superfields in the problem,
  the result is
    $$
	  \int_X d^4x\,
	    \mbox{\LARGE$($}
		 f^{(0)}_{(12\dot{1}\dot{2})}
		  + \sum_{\alpha,\dot{\beta}} \vartheta_\alpha\bar{\vartheta}_{\dot{\beta}}
		         f^{(\alpha\dot{\beta})}_{(12\dot{1}\dot{2})}
		  + \vartheta_1\vartheta_2\bar{\vartheta}_{\dot{1}}\bar{\vartheta}_{\dot{2}}
		        f^{(12\dot{1}\dot{2})}_{(12\dot{1}\dot{2})}
		\mbox{\LARGE $)$}\,.
	$$
Mathematically, there is nothing wrong:
The above integral is nothing but
 $$
     \int_X d^4x\, f^{(0)}_{(12\dot{1}\dot{2})}
     + \sum_{\alpha,\dot{\beta}} \vartheta_\alpha\bar{\vartheta}_{\dot{\beta}}
	       \int_X d^4x\,f^{(\alpha\dot{\beta})}_{(12\dot{1}\dot{2})}
	 + \vartheta_1\vartheta_2\bar{\vartheta}_{\dot{1}}\bar{\vartheta}_{\dot{2}}
	 	   \int_X d^4x\, f^{(12\dot{1}\dot{2})}_{(12\dot{1}\dot{2})}\,.
 $$
Thus, for example, when one applies calculus of variations to it to derive the equations of motions of the component fields on $X$,
  since the six integrals are not like terms, each variation has to be set to zero:
 \begin{eqnarray*}
  \lefteqn{
    \delta
	  \int_X d^4x\,
	    \mbox{\LARGE$($}
		 f^{(0)}_{(12\dot{1}\dot{2})}
		  + \sum_{\alpha,\dot{\beta}} \vartheta_\alpha\bar{\vartheta}_{\dot{\beta}}
		         f^{(\alpha\dot{\beta})}_{(12\dot{1}\dot{2})}
		  + \vartheta_1\vartheta_2\bar{\vartheta}_{\dot{1}}\bar{\vartheta}_{\dot{2}}
		        f^{(12\dot{1}\dot{2})}_{(12\dot{1}\dot{2})}
		\mbox{\LARGE $)$}\;=\; 0    }\\
  && \Longrightarrow\;\;
           \delta \int_X d^4x\, f^{(0)}_{(12\dot{1}\dot{2})}\;=\; 0\,,\hspace{2em}
           \delta \int_X d^4x\,f^{(\alpha\dot{\beta})}_{(12\dot{1}\dot{2})}\;=\;0\,,\hspace{2em}
	 	   \delta \int_X d^4x\, f^{(12\dot{1}\dot{2})}_{(12\dot{1}\dot{2})}\;=\; 0\,.
 \end{eqnarray*}
{\it But this is not what physicists do!}
 In order to match what physicists do,
  the nilpotency ---  though necessary from the perspective of complexified ${\Bbb Z}/2$-graded
  $C^\infty$Algebraic Geometry --- has to be ``purged"  away in the end.
 So that after applying such a purge to
  $f^{(0)}_{(12\dot{1}\dot{2})}
		  + \sum_{\alpha,\dot{\beta}} \vartheta_\alpha\bar{\vartheta}_{\dot{\beta}}
		         f^{(\alpha\dot{\beta})}_{(12\dot{1}\dot{2})}
		  + \vartheta_1\vartheta_2\bar{\vartheta}_{\dot{1}}\bar{\vartheta}_{\dot{2}}
		        f^{(12\dot{1}\dot{2})}_{(12\dot{1}\dot{2})}$,
  all the nilpotency disappears and the result is in $C^\infty(X)^{\Bbb C}$.
Calculus of variation is then applied to one single integral, rather then six independent integrals.

\bigskip
 
\begin{flushleft}
{\bf Purge-evaluation maps and the Fundamental Theorem}
\end{flushleft}
Let\;\;
 ${\cal M}_{\field}\;
       :=\;  \mbox{$\bigwedge$}_{{\cal O}_X^{\,\Bbb C}}^{\tinybullet}
	             ({\cal S}^\prime_{\field}\oplus{\cal S}^{\prime\prime}_{\field})\;
       \subset\;  \widehat{\cal O}_X^{\,\widehat{\boxplus}}$.
 
\bigskip

\begin{definition} {\bf [purge-evaluation map]}\; {\rm
 A $(\theta,\bar{\theta})$-degree-preserving  ${\cal O}_X^{\,\Bbb C}$-module homomorphism
  ${\cal P}: \widehat{\cal O}_X^{\,\widehat{\boxplus}}\rightarrow   \widehat{\cal O}_X$
  that restricts to the identity map
     $\Id_{\widehat{\cal O}_X}: \widehat{\cal O}_X\rightarrow \widehat{\cal O}_X$    and
	 takes the fifteen $(\vartheta,\bar{\vartheta})$-monomials
	     $\vartheta_\alpha$, $\bar{\vartheta}_{\dot{\beta}}$, $\vartheta_1\vartheta_2$,
         $\vartheta_\alpha\bar{\vartheta}_{\dot{\beta}}$,
		 $\bar{\vartheta}_{\dot{1}}\bar{\vartheta}_{\dot{2}}$,
         $\vartheta_1\vartheta_2\bar{\vartheta}_{\dot{\beta}}$, 		
		 $\vartheta_\alpha\bar{\vartheta}_{\dot{1}}\bar{\vartheta}_{\dot{2}}$, and
		 $\vartheta_1\vartheta_2\bar{\vartheta}_{\dot{1}}\bar{\vartheta}_{\dot{2}}$		
	   in each component of fixed $(\theta,\bar{\theta})$-degree to constants in ${\Bbb C}$
  is called a {\it purge-evaluation map}\footnote{\makebox[20em][l]{\it
                                                                                             Remark on the naming and the formulation}\;\;
                            ${\cal P}$ has the effect of removing the second set of fermionic coordinate-functions
		    			    $(\vartheta,\vartheta)$ on $\widehat{X}^{\widehat{\boxplus}}$,
							 hence the name `{\it purge}'.
					        Note that the coefficients of each $(\theta,\bar{\theta})$-monomial summand 
							 of a function on $\widehat{X}^{\widehat{\boxplus}}$ are grouped into
							  sections of sheaves/bundles associated to some irreducible Lorentz representations.
					        There are built-in pairings of these sheaves/bundles in the problem and
					          these pairings define accordingly various natural evaluation maps
							  that in practice either specify or are incorporated into ${\cal P}$
							  to obtain Lorentz-invariant expressions, hence the name `{\it evaluation}'.
							
                            Also note that for the construction of supersymmetric gauge theories,
		                      the purge-evaluation map is defined not on ${\cal O}_X^{\tinyphysics}$,
			                  but another ${\cal O}_X^{\,\Bbb C}$-submodule
						      of $\widehat{\cal O}_X^{\,\widehat{\boxplus}}$.
                            The formulation given here is meant to be as general as possible in order to cover all situations
					          since the proof of the Fundamental Theorem goes the same.
						      }
}\end{definition}

\bigskip

Explicitly, for
 {\small
 \begin{eqnarray*}
  \breve{f}
   & =  &
   \breve{f}_{(0)}
   + \sum_{\alpha}\theta^\alpha\breve{f}_{(\alpha)}
   + \sum_{\dot{\beta}}\bar{\theta}^{\dot{\beta}}\breve{f}_{(\dot{\beta})}
   + \theta^1\theta^2 \breve{f}_{(12)}
   + \sum_{\alpha,\dot{\beta}} \theta^\alpha\bar{\theta}^{\dot{\beta}}
           \breve{f}_{(\alpha\dot{\beta})}
   + \bar{\theta}^{\dot{1}}\bar{\theta}^{\dot{2}} \breve{f}_{(\dot{1}\dot{2})} \\
 && \hspace{4em}
   + \sum_{\dot{\beta}}\theta^1\theta^2\bar{\theta}^{\dot{\beta}}
           \breve{f}_{(12\dot{\beta})}
   + \sum_\alpha \theta^\alpha\bar{\theta}^{\dot{1}}\bar{\theta}^{\bar{\dot{2}}}
           \breve{f}_{(\alpha\dot{1}\dot{2})}
   + \theta^1\theta^2\bar{\theta}^{\dot{1}}\bar{\theta}^{\dot{2}}
           \breve{f}_{(12\dot{1}\dot{2})}
 \end{eqnarray*}}$\in \widehat{\cal O}_X^{\,\widehat{\boxplus}}$,  
 {\small
 \begin{eqnarray*}
  {\cal P}(\breve{f})
   & =  &
   {\cal P}_{(0)}(\breve{f}_{(0)})
   + \sum_{\alpha}\theta^\alpha \cdot {\cal P}_{(\alpha)}(\breve{f}_{(\alpha)})
   + \sum_{\dot{\beta}}\bar{\theta}^{\dot{\beta}}
             \cdot{\cal P}_{(\dot{\beta})}(\breve{f}_{(\dot{\beta})})  \\
  && \hspace{2em}			
   +\, \theta^1\theta^2  \cdot {\cal P}_{(12)}(\breve{f}_{(12)})
   + \sum_{\alpha,\dot{\beta}} \theta^\alpha\bar{\theta}^{\dot{\beta}}
           \cdot{\cal P}_{(\alpha\dot{\beta})}(\breve{f}_{(\alpha\dot{\beta})})
   + \bar{\theta}^{\dot{1}}\bar{\theta}^{\dot{2}}
           \cdot{\cal P}_{(\dot{1}\dot{2})}(\breve{f}_{(\dot{1}\dot{2})}) \\
 && \hspace{2em}
   + \sum_{\dot{\beta}}\theta^1\theta^2\bar{\theta}^{\dot{\beta}}
           \cdot{\cal P}_{(12\dot{\beta})}(\breve{f}_{(12\dot{\beta})})
   + \sum_\alpha \theta^\alpha\bar{\theta}^{\dot{1}}\bar{\theta}^{\bar{\dot{2}}}
          \cdot{\cal P}_{(\alpha\dot{1}\dot{2})}( \breve{f}_{(\alpha\dot{1}\dot{2})})
   + \theta^1\theta^2\bar{\theta}^{\dot{1}}\bar{\theta}^{\dot{2}}
          \cdot{\cal P}_{(12\dot{1}\dot{2})}( \breve{f}_{(12\dot{1}\dot{2})})
 \end{eqnarray*}}$\in \widehat{\cal O}_X$,    
 where ${\cal P}_{(\tinybullet)}: {\cal M}_{\field} \rightarrow  {\cal O}_X^{\,\Bbb C}$
   are ${\cal O}_X^{\,\Bbb C}$-module homomorphisms
   that restricts to the identity map
           $\Id_{{\cal }_X^{\,\Bbb C}}:
		      {\cal O}_X^{\,\Bbb C}\rightarrow {\cal O}_X^{\,\Bbb C}$   and
  takes
	     $\vartheta_\alpha$, $\bar{\vartheta}_{\dot{\beta}}$, $\vartheta_1\vartheta_2$,
         $\vartheta_\alpha\bar{\vartheta}_{\dot{\beta}}$,
		 $\bar{\vartheta}_{\dot{1}}\bar{\vartheta}_{\dot{2}}$,
         $\vartheta_1\vartheta_2\bar{\vartheta}_{\dot{\beta}}$, 		
		 $\vartheta_\alpha\bar{\vartheta}_{\dot{1}}\bar{\vartheta}_{\dot{2}}$, and
		 $\vartheta_1\vartheta_2\bar{\vartheta}_{\dot{1}}\bar{\vartheta}_{\dot{2}}$		
	   to constants in ${\Bbb C}$.
When all ${\cal P}_{(\tinybullet)}$ are identical, we say that ${\cal P}$ is {\it uniform}
  and will denote  all ${\cal P}_{(\tinybullet)}$ simply by ${\cal P}$.
In particular, for any ${\cal P}$,
 each ${\cal P}_{(\tinybullet)}$ defines a uniform purge-evaluation map
  by replacing all other component ${\cal P}_{(\tinybullet^\prime)}$ by ${\cal P}_{(\tinybullet)}$.
We denote the resulting purge-evaluation map simply by ${\cal P}_{(\tinybullet)}$.
  
Note that in practice,  ${\cal P}$ restricts to a non-trivial purge-evaluation map
 ${\cal P}: {\cal O}_X^{\physics}\rightarrow \widehat{\cal O}_X$
 and hence can never be a ${\Bbb Z}/2$-graded-${\cal O}_X^{\,\Bbb C}$-algebra homomorphism.

\bigskip

\begin{lemma} {\bf [property of ${\cal P}$]}\;
 A purge-evaluation map ${\cal P}:
   \widehat{\cal O}_X^{\,\widehat{\boxplus}}\rightarrow \widehat{\cal O}_X$
  satisfies the following properties:
 (1) Let ${\cal P}$ be uniform and
             let $\xi\in \Der_{\Bbb C}(\widehat{X})$ be a derivation on $\widehat{X}$.
            Then\\ ${\cal P}(\xi\breve{f})= \xi {\cal P}(\breve{f})$
			   for $\breve{f}\in \widehat{\cal O}_X^{\,\widehat{\boxplus}} $. \\
 (2) $\int\! d\bar{\theta}^{\dot{2}}d\bar{\theta}^{\dot{1}} d\theta^2d\theta^1
              {\cal P}(\breve{f})
			  = \int\! d\bar{\theta}^{\dot{2}}d\bar{\theta}^{\dot{1}} d\theta^2d\theta^1
                  {\cal P}_{(12\dot{1}\dot{2})}(\breve{f})
			  = {\cal P}_{(12\dot{1}\dot{2})}(
			          \int\! d\bar{\theta}^{\dot{2}}d\bar{\theta}^{\dot{1}} d\theta^2d\theta^1
					    \breve{f})$
			 for $\breve{f}\in \widehat{\cal O}_X^{\,\widehat{\boxplus}} $.\;\;
 (3) For $\breve{f}\in \widehat{\cal O}_X^{\,\widehat{\boxplus}} $ chiral (resp.\ antichiral),
              $\int\! d\theta^2d\theta^1  {\cal P}(\breve{f})
			     = \int\! d\theta^2d\theta^1  {\cal P}_{(12)}(\breve{f})
			     = {\cal P}_{(12)}(\int\! d\theta^2d\theta^1 \breve{f})$\;
             (resp.\;
			      $\int\! d\bar{\theta}^{\dot{2}}d\bar{\theta}^{\dot{1}} {\cal P}(\breve{f})
				     = \int\! d\bar{\theta}^{\dot{2}}d\bar{\theta}^{\dot{1}}
					               {\cal P}_{(\dot{1}\dot{2})}(\breve{f})
			         = {\cal P}_{(\dot{1}\dot{2})}(\int\! d\theta^2d\theta^1 \breve{f})$).					  
\end{lemma}

\begin{proof}
 Statement (1) follows from
    the fact that $\xi\in \Der_{\Bbb C}(\widehat{X})$
    has no $(\vartheta, \bar{\vartheta})$-dependence
    while the uniform ${\cal P}$ applies to $\breve{f}$
	    $(\theta,\bar{\theta})$-degree by $(\theta,\bar{\theta})$-degree
      with ${\cal P}
	              (\vartheta_1^{\epsilon_1}\vartheta_2^{\epsilon_2}
				      \bar{\vartheta}_{\dot{1}}^{\epsilon_3}
					  \bar{\vartheta}_{\dot{2}}^{\epsilon_4})$ constant, $\epsilon_i=0$ or $1$,
   and hence $\xi({\cal P}(\breve{f}_{(\tinybullet)}))
                          = {\cal P}(\xi\, \breve{f}_{(\tinybullet)})$.					
  Statement (2) and Statement (3) follow from the definition of fermionic integration on $\widehat{X}$.
\end{proof}

\bigskip

\begin{theorem} {\bf [fundamental: supersymmetric functional]}\;
 Let ${\cal P}$ be a uniform purge-evaluation map.
 Then, up to a boundary term\footnote{Though
                                                    ignored in the current work, it should be noted that 
													such a boundary term becomes an important part of understanding
													 when one studies supersymmetric quantum field theory with boundary.
                                                       } 
  on $X$,
 (1)
    $S_1(\breve{f})
	  := \int_{\widehat{X}}d^4x\,
	          d\bar{\theta}^{\dot{2}}d\bar{\theta}^{\dot{1}} d\theta^2 d\theta^1
			  {\cal P}(\breve{f})$
      is a functional on $C^\infty(\widehat{X}^{\widehat{\boxplus}})$ that is invariant under supersymmetries;\\
 (2)
    $S_2(\breve{f})
	   := \int_{\widehat{X}}d^4x\, d\theta^2 d\theta^1 {\cal P}(\breve{f})$
	(resp.\
	    $S_3(\breve{f})
		   :=\int_{\widehat{X}}d^4x\, d\bar{\theta}^{\dot{2}} d\bar{\theta}^{\dot{1}}
		         {\cal P}(\breve{f})$)
	 is a functional on\\   $C^\infty(\widehat{X}^{\widehat{\boxplus}})^{\scriptsizech}$
    (resp.\ $C^\infty(\widehat{X}^{\widehat{\boxplus}})^{\scriptsizeach}$)
	  that is invariant under supersymmetries.
\end{theorem}		

\begin{proof}
 For Statement (1),
   since $Q_{\alpha}, \bar{Q}_{\dot{\beta}}\in \Der_{\Bbb C}(\widehat{X})$,
   it follows
     the invariance of
	     $d^4x\,d\bar{\theta}^{\dot{2}}d\bar{\theta}^{\dot{1}}d\theta^2 d\theta^1$,
         $d^4x d\theta^2 d\theta^1$,  and
		 $d^4x d\bar{\theta}^{\dot{2}} d\bar{\theta}^{\dot{1}}$
		 under the flow that generates supersymmetries,
     Lemma~1.5.2, and basic calculus
   that
   \begin{eqnarray*}
     \delta_{Q_\alpha}S_1(\breve{f})
	   & :=\: & \int_{\widehat{X}}d^4x\,
	                         d\bar{\theta}^{\dot{2}} d\bar{\theta}^{\dot{1}} d\theta^2 d\theta^1\,
	                       {\cal P}(Q_\alpha\breve{f})\;
			    =\;  \int_{\widehat{X}}d^4x\,
				        d\bar{\theta}^{\dot{2}} d\bar{\theta}^{\dot{1}}  d\theta^2 d\theta^1\,
	                         Q_{\alpha}{\cal P}(\breve{f})\\
       & = & 	-\sqrt{-1}\,
	               \int_{\widehat{X}}d^4x\,
				     d\bar{\theta}^{\dot{2}} d\bar{\theta}^{\dot{1}} d\theta^2 d\theta^1\,
	                \sum_{\dot{\beta},\,\mu}
					  \sigma^\mu_{\alpha\dot{\beta}}\bar{\theta}^{\dot{\beta}}
					  \partial_\mu ({\cal P}(\breve{f}))		\\
       & = & -\sqrt{-1}\,
                  \int_X	d^4x \sum_{\mu}\partial_\mu
                    \left(\rule{0ex}{1.2em}\right.				
					\int d\bar{\theta}^{\dot{2}} d\bar{\theta}^{\dot{1}}d\theta^2 d\theta^1\,
					   \sum_{\dot{\beta}}
					      \sigma^\mu_{\alpha\dot{\beta}}\bar{\theta}^{\dot{\beta}}
						  {\cal P}(\breve{f})
					\left.\rule{0ex}{1.2em}\right)\;\;
				=\;\; -\sqrt{-1}\int_X dB_{\alpha}\,,
   \end{eqnarray*}
   where
     $B_{\alpha}= B_\alpha^0 dx^1\wedge dx^2\wedge dx^3
	                                   - B_\alpha^1 dx^0\wedge dx^2\wedge dx^3
									  + B_\alpha^2 dx^0\wedge dx^1\wedge dx^3
									   - B_\alpha^3 dx^0\wedge dx^1\wedge dx^2$
      is a $3$-form on $X$ with
    $$
	  B_\alpha^\mu \;
	    =\;  \int d\bar{\theta}^{\dot{2}} d\bar{\theta}^{\dot{1}}d\theta^2 d\theta^1\,
					   \sum_{\dot{\beta}}
					      \sigma^\mu_{\alpha\dot{\beta}}\bar{\theta}^{\dot{\beta}}
						  {\cal P}(\breve{f})\,.
	$$
 The proof that  $\delta_{Q_{\alpha}}S(\breve{f})$ is also a boundary term is similar.

 For Statement (2),
  note that  for $\breve{f}$ chiral,
  $\delta_{Q_\alpha}S_2(\breve{f})=0$ always, for $\alpha=1,2$, and,
  thus one only needs to check the variation
  $\delta_{\bar{Q}_{\dot{\beta}}}S_2(\breve{f})$:
   \begin{eqnarray*}
     \delta_{\bar{Q}_{\dot{\beta}}}S_2(\breve{f})
	   & :=\: & \int_{\widehat{X}}d^4x\, d\theta^2 d\theta^1\,
	                       {\cal P}(\bar{Q}_{\dot{\beta}} \breve{f})\;
                =\; \int_{\widehat{X}}d^4x\, d\theta^2 d\theta^1\,
	                       {\cal P}(
						       (e_{\beta^{\prime\prime}}\,
                                  + \, 2\sqrt{-1}\mbox{$\sum$}_{\alpha, \mu}
								          \theta^\alpha \sigma^\mu_{\alpha\dot{\beta}}\partial_\mu													   
							   )\breve{f})                                                   \\
       & = & 	2 \sqrt{-1}\,
	               \int_{\widehat{X}}d^4x\, d\theta^2 d\theta^1\,
	                \sum_{\alpha,\,\mu}
					   \theta^\alpha \sigma^\mu_{\alpha\dot{\beta}}
					   \partial_\mu ({\cal P}(\breve{f}))		\\
       & = & 2\sqrt{-1}\,
                  \int_X	d^4x \sum_{\mu}\partial_\mu
                    \left(\rule{0ex}{1.2em}\right.				
					\int d\theta^2 d\theta^1\,
					   \sum_\alpha
					      \theta^\alpha \sigma^\mu_{\alpha\dot{\beta}} {\cal P}(\breve{f})
					\left.\rule{0ex}{1.2em}\right)\;\;
				=\;\; 2\sqrt{-1}\int_X dC_{\dot{\beta}}\,,
   \end{eqnarray*}
     where
     $C_{\dot{\beta}}= C_{\dot{\beta}}^0 dx^1\wedge dx^2\wedge dx^3
	                                   - C_{\dot{\beta}}^1 dx^0\wedge dx^2\wedge dx^3
									  + C_{\dot{\beta}}^2 dx^0\wedge dx^1\wedge dx^3
									   - C_{\dot{\beta}}^3 dx^0\wedge dx^1\wedge dx^2$
      is a $3$-form on $X$ with
    $$
	  C_{\dot{\beta}}^\mu \;
	    =\;  \int d\theta^2 d\theta^1\,
					\sum_\alpha
					      \theta^\alpha \sigma^\mu_{\alpha\dot{\beta}}{\cal P}(\breve{f})\,.
	$$
  For $\breve{f}$ antichiral,
   $\delta_{\bar{Q}_{\dot{\beta}}}S_3(\breve{f})=0$ always,
       for $\dot{\beta}=\dot{1}, \dot{2}$,    and
   the variation $\delta_{Q_\alpha}S_3(\breve{f})$, $\alpha=1,2$,
     can be computed similarly to show that it is a boundary term on $X$.
   
  This completes the proof.
  
\end{proof}
  
\bigskip

\begin{remark} $[$Lorentz invariance and
                                 ${\cal P}$-dependence of variations of components under supersymmetry $]$\; {\rm
 (1)
 For applications to particle physics,
  one takes the real-part of the complex-valued functional $S(\tinybullet)$ (if it is not already real)
    to get the action functional for the component fields
    $f^{\tinybullet}_{\tinybullet}$ of $\breve{f}$
  and requires in addition that the action functional be Lorentz-invariant,
      which is usually automatic
	  when ${\cal P}$ comes from natural evaluation maps built-into the problem.
	
 (2)
 Note also that
  from the equalities
    $Q_\alpha\,{\cal P}(\breve{f})= {\cal P}(Q_\alpha \breve{f})$,
	$\bar{Q}_{\dot{\beta}}\,{\cal P}(\breve{f})
	   = {\cal P}(\bar{Q}_{\dot{\beta}} \breve{f})$,
	$\alpha=1,2$, $\dot{\beta}=\dot{1}, \dot{2}$,
  for a uniform purge-evaluation map ${\cal P}$,  											
 the variation under supersymmetry of component fields of $\breve{f}$ for a physics model
  depends on the choice of ${\cal P}$ as well.
}\end{remark}

\bigskip

\section{The chiral/antichiral theory on $X^{\physics}$ and Wess-Zumino model}

Having made the effort to build a platform from complexified ${\Bbb Z}/2$-graded $C^\infty$-Algebraic Geometry
 that incorporates basic requirements from Quantum Field Theory,
one would like to know whether all the well-established supersymmetric quantum field theories in physics 
 fit into the setting.
In this work, we make a humble start to re-look at two earliest constructed, most basic
   supersymmetric quantum field theories in physics:
 the Wess-Zumino model (current section) and the supersymmetric $U(1)$ gauge theory with matter (the next section)
 to justify the validity.

\bigskip

\subsection{More on the chiral and the antichiral sector of $X^{\physics}$}
		
As a preparation to study the Wess-Zumino model,
 some further details of the chiral sector and the antichiral sector of $X^{\physics}$
 are given in this subsection.

\bigskip
 
\begin{flushleft}
{\bf The chiral sector of $X^{\physics}$}
\end{flushleft}
\begin{definition} {\bf [chiral function-ring \& chiral structure sheaf of $X^{\physics}$]}\; {\rm
 (1) An $\breve{f}\in C^{\infty}(X^{\physics})$ is called {\it chiral}
             if $e_{1^{\prime\prime}}\breve{f}= e_{2^{\prime\prime}}\breve{f}=0$.
 (2) As the addition and the multiplication of two chiral functions remain chiral,
             the set of chiral functions on $X^{\physics}$ form a ring, called the {\it chiral function-ring}
			 of $X^{\physics}$, denoted by $C^{\infty}(X^{\physics})^{\scriptsizech}$.
 (3) Localizing Item (1) and Item (2) to open sets of $X$,
             one obtains a sheaf ${\cal O}_X^{\physics, \scriptsizech}$
             of chiral functions on $X^{\physics}$ from the assignment
			 $U\mapsto C^{\infty}(U^{\physics})^{\scriptsizech}$ for $U$ open sets of $X$.
            ${\cal O}_X^{\physics, \scriptsizech}$ is call the {\it chiral structure sheaf} of $X^{\physics}$.
		    By construction, ${\cal O}_X^{\physics, \scriptsizech}\subset {\cal O}_X^{\physics}$ 	
			 as ${\cal O}_X^{\,\Bbb C}$-algebras.
}\end{definition}

\medskip

\begin{lemma} {\bf [$C^\infty$-hull of $C^\infty(X^{\physics})^{\scriptsizech}$]}\;
  The $C^\infty$-hull $\,C^\infty$-{\it hull}\,$(C^\infty(X^{\physics}))$
    of\\ $C^\infty(X^{\physics})$
   restricts to the $C^\infty$-hull $\,C^\infty$-{\it hull}\,$(C^\infty(X^{\physics})^{\scriptsizech})$ 
     of $C^\infty(X^{\physics})^{\scriptsizech}$.
 Thus,\\ $C^\infty(X^{\physics})^{\scriptsizech}$ is a complexified $C^\infty$-ring.
\end{lemma}

\begin{proof}
 Let $h\in C^{\infty}({\Bbb R}^l)$  and
   $\breve{f}_1,\,\cdots\,, \breve{f}_l
      \in C^{\infty}(X^{\physics})^{\scriptsizech}\cap C^{\infty}$-{\it hull}\,$(X^{\physics})$.
 Then
     it follows from Definition/Lemma~1.3.6
    that
	in terms of the chiral  coordinate-functions
    $(x^\prime, \theta^\prime, \bar{\theta}^\prime, \vartheta^\prime, \bar{\vartheta}^\prime )$
   on $\widehat{X}^{\widehat{\boxplus}}$,\\
   $\breve{f}_i=f^{\prime (0)}_{i,(0)}(x^\prime)
             + \sum_\alpha \theta^\prime\,\!^\alpha\vartheta^\prime_{\alpha}
			        f^{\prime\,(\alpha)}_{i,(\alpha)}(x^\prime)
			 + \theta^\prime\,\!^1\theta^\prime\,\!^2\vartheta^\prime_1\vartheta^\prime_2
				    f^{\prime\,(12)}_{i,(12)}(x^\prime)\,$,
   for $i=1,\,\ldots\,,\, l$, and hence\\
   $h(\breve{f}_1,\,\cdots\,,\, \breve{f}_l)
       \in C^{\infty}$-{\it hull}\,$(C^{\infty}(X^{\physics}))$
   can be expressed as
   {\small
   \begin{eqnarray*}
    \lefteqn{
     h(f^{\prime (0)}_{1,(0)}(x^{\prime}),\,
	                                                             \cdots\,, f^{\prime (0)}_{l,(0)}(x^\prime)   )  }\\
   &&																				
     +	 \sum_{k=1}^l
      (\partial_kh)  (f^{\prime (0)}_{1,(0)}(x^{\prime}),\,
	                                      \cdots\,, f^{\prime (0)}_{l,(0)}(x^\prime)   )
		   \cdot
           \left(\rule{0ex}{1.2em}\right.\!
		     \sum_\alpha \theta^\prime\,\!^\alpha\vartheta^\prime_\alpha
		       f^{\prime\,(\alpha)}_{k,(\alpha)}(x^\prime)
		   + \theta^{\prime}\,\!^1 \theta^{\prime}\,\!^2  \vartheta^\prime_1 \vartheta^\prime_2
		         f^{\prime\,(12)}_{k,(12)}(x^\prime)
		   \!\left.\rule{0ex}{1.2em}\right)		     \\
    &&		
	 -\, \theta^{\prime 1}\theta^{\prime 2}\vartheta^\prime_1\vartheta^\prime_2
	    \sum_{k_1, k_2=1}^l
      (\partial_{k_1}\partial_{k_2}h)  (f^{\prime (0)}_{1,(0)}(x^{\prime}),\,
		                                                                                 \cdots\,, f^{\prime (0)}_{l,(0)}(x^\prime) )
        \cdot																						
		     f^{\prime (1)}_{k_1, (1)}(x^\prime)
			 f^{\prime (2)}_{k_2, (2)}(x^\prime)\,,
   \end{eqnarray*}}which   
   is chiral.		
  Here, the following computation and a change of dummy indices $k_1$ and $k_2$ are used to obtain the last term
  {\small
   \begin{eqnarray*}
     \lefteqn{
     \sum_{\alpha,\gamma}
	     \theta^{\prime\alpha}\vartheta^\prime_\alpha
		 \theta^{\prime\gamma}\vartheta^\prime_\gamma
    	 f^{\prime (\alpha)}_{k_1,(\alpha)}(x^\prime)
 		 f^{\prime (\gamma)}_{k_2, (\gamma)}(x^\prime)    \;\;
        =\;\; \theta^{\prime 1}\theta^{\prime 2}\vartheta^\prime_1\vartheta^\prime_2
	             \sum_{\alpha,\gamma}
				     \varepsilon^{\alpha\gamma}\varepsilon_{\alpha\gamma}
					 	 f^{\prime (\alpha)}_{k_1,(\alpha)}(x^\prime)
 		                 f^{\prime (\gamma)}_{k_2, (\gamma)}(x^\prime)  }\\
		&&
		 =\; -\, \theta^{\prime 1}\theta^{\prime 2}\vartheta^\prime_1\vartheta^\prime_2
			            \cdot
						 \left(\rule{0ex}{1.2em}\right.\!
						   f^{\prime (1)}_{k_1,(1)}(x^\prime)
						    f^{\prime (2)}_{k_2, (2)}(x^\prime)
							+f^{\prime (2)}_{k_1,(2)}(x^\prime)
							   f^{\prime (1)}_{k_2, (1)}(x^\prime)
						  \!\left.\rule{0ex}{1.2em}\right)  \hspace{6em}
   \end{eqnarray*}}This  
  proves the lemma.
   
\end{proof}

\medskip

\begin{lemma} {\bf [chiral function on $X^{\physics}$ in terms of
                                            $(x,\theta,\bar{\theta},\vartheta,\bar{\vartheta})$]}\;
 In terms of the standard coordinate functions $(x,\theta, \bar{\theta},\vartheta,\bar{\vartheta})$
   on $\widehat{X}^{\widehat{\boxplus}}$,
  a chiral function $\breve{f}$ on $X^{\physics} $ is determined by the four components
     $f^{(0)}_{(0)}$, $f^{(\alpha)}_{(\alpha)}$, and
	 $f^{(12)}_{(12)}$,  $\alpha=1,2$,
   of $\breve{f}$
   via the following formula
   \begin{eqnarray*}
     \breve{f} & = &
	   f_{(0)}^{(0)}(x)
	   + \sum_\alpha \theta^\alpha\vartheta_\alpha f_{(\alpha)}^{(\alpha)}(x)
	   + \theta^1\theta^2\vartheta_1\vartheta_2 f_{(12)}^{(12)}(x)
       + \sqrt{-1} \sum_{\alpha,\dot{\beta}}
	          \theta^\alpha\bar{\theta}^{\dot{\beta}}	
			   \sum_\mu
			     \sigma^\mu_{\alpha\dot{\beta}}\partial_\mu f_{(0)}^{(0)}(x) \\
      && \hspace{1em}				
       + \sqrt{-1}\sum_{\dot{\beta}, \mu}				
	        \theta^1\theta^2\bar{\theta}^{\dot{\beta}}
			 \mbox{\Large $($}
			    \vartheta_1\sigma^\mu_{2\dot{\beta}}\partial_\mu f_{(1)}^{(1)}(x)
			     - \vartheta_2\sigma^\mu_{1\dot{\beta}}\partial_\mu f_{(2)}^{(2)}(x)
			 \mbox{\Large $)$}
       - \theta^1\theta^2\bar{\theta}^{\dot{1}}\bar{\theta}^{\dot{2}}\,
	        \square f_{(0)}^{(0)}(x)\,,
   \end{eqnarray*}
  where\footnote{In
                                [L-Y1: Lemma 1.4.14], we define the Laplacian $\square$ as
								$-\sum_\mu\partial^\mu\partial_\mu= \partial_0^2-\partial_1^2-\partial_2^2-\partial_3^2$.
                               Here, we recover the convention used in [W-B] of Wess \& Bagger.
								 } 
    $\square := \sum_\mu \partial^\mu\partial_\mu = - \partial_0^2+\partial_1^2+\partial_2^2+\partial_3^2$.
\end{lemma}

\begin{proof}
 Similar to the proof [L-Y1: Lemma 1.4.14], one can prove the statement in two ways.
 The first proof is slick
   while the second proof can be generalized to the situation that involves sections of a bundle with a connection
    when one studies supersymmetric D-branes.
 
 \medskip
 
 \noindent
 $(a)$ {\it First proof}
 
 \smallskip
 
 \noindent
 In terms of the chiral coordinate functions
   $(x^\prime, \theta^\prime, \bar{\theta}^\prime, \vartheta^\prime, \bar{\vartheta}^\prime)
      :=(x+\sqrt{-1}\theta\boldsigma\bar{\theta}^t, \theta,\bar{\theta}, \vartheta, \bar{\theta})$
   on $\widehat{X}^{\widehat{\boxplus}}$	
 a chiral function $\breve{f}$ on $X^{\physics}$ can be written as
 $$
   \breve{f}\;
   =\;  f^{\prime (0)}_{(0)}(x^\prime)\,
	           +\, \sum_\alpha  \theta^\prime\,\!^\alpha \vartheta^\prime_\alpha\,
			            f^{\prime (\alpha)}_{(\alpha)}(x^\prime)\,			                                                                         
			   +\,\theta^\prime\,\!^1 \theta^\prime\,\!^2\vartheta^\prime_1 \vartheta^\prime_2
			           f^{\prime (12)}_{(12)}(x^\prime)
 $$
 where
   $f^\prime_{(0)},\,  f^\prime_{(\alpha)},\,  f^\prime_{(12)}
       \in C^\infty({\Bbb R}^4)^{\Bbb C}$.
 [L-Y1: Lemma~1.1.3]
     applied to the real and the imaginary component of
        $f^{\prime (0)}_{(0)},\,  f^{\prime (\alpha)}_{(\alpha)},\,
		   f^{\prime (12)}_{(12)}$
  gives		
  \begin{eqnarray*}
   f^{\prime (\tinybullet)}_{(\tinybullet)}(x^{\prime})
   & = &
    f^{\prime (\tinybullet)}_{(\tinybullet)}
     (x+\sqrt{-1}\theta\boldsigma\bar{\theta}^t)   \\
   & = & f^{\prime (\tinybullet)}_{(\tinybullet)}(x)\,
        +\,\sqrt{-1}\,  \sum_\mu (\partial_{\mu}f^{\prime (\tinybullet)}_{(\tinybullet)})(x)\,
		         \theta\sigma^\mu \bar{\theta}^t\,
		-\,\frac{1}{2}\,\sum_{\mu,\nu}
		     (\partial_\mu\partial_\nu f^{\prime (\tinybullet)}_{(\tinybullet)})(x)\,
			     (\theta\sigma^\mu\bar{\theta}^t)(\theta\sigma^\nu\bar{\theta}^t)\,.
  \end{eqnarray*}
 The claim follows from
  applying the expansion to
     $f^{\prime (0)}_{(0)}(x^\prime)$,
	 $f^{\prime (\alpha)}_{(\alpha)}(x^\prime)$,  and
     $f^{\prime (12)}_{(12)}(x^\prime)$,
  collecting like terms in $(\theta,\bar{\theta})$,  simplifying them via spinor calculus, and
  re-denoting $f^{\prime (\tinybullet)}_{(\tinybullet)}(x)$
    by $f^{(\tinybullet)}_{(\tinybullet)}(x)$.
    
 \bigskip

 \noindent
 $(b)$  {\it Second proof}
 
 \medskip
 
 \noindent
 This is done by directly solving the system of differential equations obtained from the chiral condition on $\breve{f}$.
 
 Let
 {\small
 \begin{eqnarray*}
  \breve{f}
   & =  &
   \breve{f}_{(0)}
   + \sum_{\alpha}\theta^\alpha\breve{f}_{(\alpha)}
   + \sum_{\dot{\beta}}\bar{\theta}^{\dot{\beta}}\breve{f}_{(\dot{\beta})}
   + \theta^1\theta^2 \breve{f}_{(12)}
   + \sum_{\alpha,\dot{\beta}} \theta^\alpha\bar{\theta}^{\dot{\beta}}
           \breve{f}_{(\alpha\dot{\beta})}
   + \bar{\theta}^{\dot{1}}\bar{\theta}^{\dot{2}} \breve{f}_{(\dot{1}\dot{2})} \\
 && \hspace{4em}
   + \sum_{\dot{\beta}}\theta^1\theta^2\bar{\theta}^{\dot{\beta}}
           \breve{f}_{(12\dot{\beta})}
   + \sum_\alpha \theta^\alpha\bar{\theta}^{\dot{1}}\bar{\theta}^{\bar{\dot{2}}}
           \breve{f}_{(\alpha\dot{1}\dot{2})}
   + \theta^1\theta^2\bar{\theta}^{\dot{1}}\bar{\theta}^{\dot{2}}
           \breve{f}_{(12\dot{1}\dot{2})} \\
  &= &
    f^{(0)}_{(0)}
	+ \sum_{\alpha}\theta^\alpha\vartheta_\alpha f^{(\alpha)}_{(\alpha)}
	+ \sum_{\dot{\beta}}
	       \bar{\theta}^{\dot{\beta}}\bar{\vartheta}_{\dot{\beta}}
		     f^{(\dot{\beta})}_{(\dot{\beta})}  \\
  && \hspace{1em}			
    +\; \theta^1\theta^2\vartheta_1\vartheta_2 f^{(12)}_{(12)}
	+ \sum_{\alpha,\dot{\beta}}\theta^\alpha \bar{\theta}^{\dot{\beta}}
	      \left(\rule{0ex}{1.2em}\right.\!
		    \sum_\mu \sigma^\mu_{\alpha\dot{\beta}} f^{(0)}_{[\mu]}\,
			 +\, \vartheta_\alpha\bar{\vartheta}_{\dot{\beta}}
			        f^{(\alpha\dot{\beta})}_{(\alpha\dot{\beta})}
		  \!\left.\rule{0ex}{1.2em}\right)
    + \bar{\theta}^{\dot{1}}\bar{\theta}^{\dot{2}}
	   \bar{\vartheta}_{\dot{1}}\bar{\vartheta}_{\dot{2}}
	    f^{(\dot{1}\dot{2})}_{(\dot{1}\dot{2})}  \\
  && \hspace{1em}		
	+ \sum_{\dot{\beta}}\theta^1\theta^2\bar{\theta}^{\dot{\beta}}
	     \left(\rule{0ex}{1.2em}\right.\!
		   \sum_\alpha \vartheta_\alpha f^{(\alpha)}_{(12\dot{\beta})}\,
		    +\, \vartheta_1\vartheta_2\bar{\vartheta}_{\dot{\beta}}
			            f^{(12\dot{\beta})}_{(12\dot{\beta})}
		 \!\left.\rule{0ex}{1.2em}\right)
    + \sum_\alpha \theta^\alpha\bar{\theta}^{\dot{1}}\bar{\theta}^{\dot{2}}
	   \left(\rule{0ex}{1.2em}\right.\!
	     \sum_{\dot{\beta}} \bar{\vartheta}_{\dot{\beta}}
		     f^{(\dot{\beta})}_{(\alpha\dot{1}\dot{2})}\,
         +\, \vartheta_\alpha \bar{\vartheta}_{\dot{1}}\bar{\vartheta}_{\dot{2}}	
		           f^{(\alpha\dot{1}\dot{2})}_{(\alpha\dot{1}\dot{2})}
	   \!\left.\rule{0ex}{1.2em}\right)\\
  && \hspace{1em}
	+\; \theta^1\theta^2\bar{\theta}^{\dot{1}}\bar{\theta}^{\dot{2}}
	     \left(\rule{0ex}{1.2em}\right.\!
		  f^{(0)}_{(12\dot{1}\dot{2})}
		  + \sum_{\alpha,\dot{\beta}} \vartheta_\alpha\bar{\vartheta}_{\dot{\beta}}
		         f^{(\alpha\dot{\beta})}_{(12\dot{1}\dot{2})}
		  + \vartheta_1\vartheta_2\bar{\vartheta}_{\dot{1}}\bar{\vartheta}_{\dot{2}}
		        f^{(12\dot{1}\dot{2})}_{(12\dot{1}\dot{2})}
		 \!\left.\rule{0ex}{1.2em}\right)    \\
   & \in & C^\infty(X^{\physics})\,.		
 \end{eqnarray*}}Then 
 a straightforward computation with applications of basic identities in spinor calculus (cf.\ Appendix) gives
 {\small
 \begin{eqnarray*}
  - e_{\beta^{\prime\prime}}\breve{f}
    & = &   \left(\rule{0ex}{1.2em}\right.\!
                  \frac{\partial}{\partial \bar{\theta}^{\dot{\beta}}}
				   + \sqrt{-1}\sum_{\alpha,\mu}
				         \theta^\alpha \sigma^\mu_{\alpha\dot{\beta}} \partial_\mu
	             \!\left.\rule{0ex}{1.2em}\right)
				 \breve{f}\\
    & = &	\breve{f}_{(\dot{\beta})}	
               + \sum_\alpha \theta^\alpha
                    \left(\rule{0ex}{1.2em}\right.\!			
					 -\breve{f}_{(\alpha\dot{\beta})}
					 + \sqrt{-1}\sum_\mu \sigma^\mu_{\alpha\dot{\beta}}
					       \partial_\mu \breve{f}_{(0)}
					\!\left.\rule{0ex}{1.2em}\right)
               - \sum_{\dot{\delta}} \bar{\theta}^{\dot{\delta}} 			
			         \varepsilon_{\dot{\beta}\dot{\delta}}\breve{f}_{(\dot{1}\dot{2})} \\
    &&  +\, \theta^1\theta^2
                  \left(\rule{0ex}{1.2em}\right.\!	
				    \breve{f}_{(12\dot{\beta})}
					 + \sqrt{-1}\sum_{\alpha,\gamma,\mu}
					       \varepsilon^{\alpha\gamma}\sigma^\mu_{\alpha\dot{\beta}}
						      \partial_\mu \breve{f}_{(\gamma)}
				  \!\left.\rule{0ex}{1.2em}\right)	
	       +\, \sum_{\alpha, \dot{\delta}}
	          \theta^\alpha\bar{\theta}^{\dot{\delta}}
			   \left(\rule{0ex}{1.2em}\right.\!
			    \varepsilon_{\dot{\beta}\dot{\gamma}}\breve{f}_{(\alpha\dot{1}\dot{2})}
				+ \sqrt{-1}\sum_\mu \sigma^\mu_{\alpha\dot{\beta}}
				      \partial_\mu\breve{f}_{(\dot{\delta})}
	           \!\left.\rule{0ex}{1.2em}\right)                                       \\
    && +\, \sum_{\dot{\delta}}
	             \theta^1\theta^2\bar{\theta}^{\dot{\delta}}
				   \left(\rule{0ex}{1.2em}\right.\!
	                 -\varepsilon_{\dot{\beta}\dot{\delta}}\breve{f}_{(12\dot{1}\dot{2})}
					  + \sqrt{-1}\sum_{\alpha,\gamma,\mu}
					       \varepsilon^{\alpha\gamma} \sigma^\mu_{\alpha\dot{\beta}}
						   \partial_\mu\breve{f}_{(\gamma\dot{\delta})}
				   \!\left.\rule{0ex}{1.2em}\right)
           +\, \sqrt{-1}\sum_{\alpha,\mu}
	              \theta^\alpha\bar{\theta}^{\dot{1}}\bar{\theta}^{\dot{2}}
				   \sigma^\mu_{\alpha\dot{\beta}}\partial_\mu\breve{f}_{(\dot{1}\dot{2})}   \\				  
    &&  +\,  \sqrt{-1}\,
			    \theta^1\theta^2\bar{\theta}^{\dot{1}}\bar{\theta}^{\dot{2}}
				  \sum_{\alpha,\gamma,\mu}
				   \varepsilon^{\alpha\gamma}\sigma^\mu_{\alpha\dot{\beta}}
				      \partial_\mu\breve{f}_{(\gamma\dot{1}\dot{2})}\,,				
 \end{eqnarray*}}$\beta^{\prime\prime}=1^{\prime\prime}, 2^{\prime\prime}$. 
 Similar to [L-Y1: proof of Lemma 1.4.14],
 solving the overdetermined system of equations
  on components $f^{\tinybullet}_{\tinybullet}$ of $\breve{f}$
  obtained by setting $e_{1^{\prime\prime}}\breve{f}=e_{2^{\prime\prime}}\breve{f}=0$,
  now with applications of basic identities from spinor calculus (cf.\ Appendix),
 proves the statement.
  
\end{proof}

\bigskip

\begin{flushleft}
{\bf The antichiral sector of $X^{\physics}$}
\end{flushleft}
Similar arguments to the previous theme give the results for the antichiral sector of $X^{\physics}$.
 
\bigskip

\begin{definition} {\bf [antichiral function-ring \& antichiral structure sheaf of $X^{\physics}$]}\; {\rm
 (1) An $\breve{f}\in C^{\infty}(X^{\physics})$ is called {\it antichiral}
             if $e_{1^\prime}\breve{f}= e_{2^\prime}\breve{f}=0$.
 (2) As the addition and the multiplication of two antichiral functions remain antichiral,
             the set of antichiral functions on $X^{\physics}$ form a ring, called the {\it antichiral function-ring}
			 of $X^{\physics}$, denoted by $C^{\infty}(X^{\physics})^{\scriptsizeach}$.
 (3) Localizing Item (1) and Item (2) to open sets of $X$,
             one obtains a sheaf ${\cal O}_X^{\physics, \scriptsizeach}$
             of antichiral functions on $X^{\physics}$ from the assignment
			 $U\mapsto C^{\infty}(U^{\physics})^{\scriptsizeach}$ for $U$ open sets of $X$.
            ${\cal O}_X^{\physics, \scriptsizeach}$ is call the {\it antichiral structure sheaf}
			  of $X^{\physics}$.
		    By construction, ${\cal O}_X^{\physics, \scriptsizeach}\subset {\cal O}_X^{\physics}$ 	
			 as ${\cal O}_X^{\,\Bbb C}$-algebras.
}\end{definition}

\medskip

\begin{lemma} {\bf [$C^\infty$-hull of $C^\infty(X^{\physics})^{\scriptsizeach}$]}\;
  The $C^\infty$-hull $\,C^\infty$-{\it hull}\,$(C^\infty(X^{\physics}))$
    of\\ $C^\infty(X^{\physics})$
   restricts to the $C^\infty$-hull $\,C^\infty$-{\it hull}\,$(C^\infty(X^{\physics})^{\scriptsizeach})$ 
     of $C^\infty(X^{\physics})^{\scriptsizeach}$.
 Thus,\\ $C^\infty(X^{\physics})^{\scriptsizeach}$ is a complexified $C^\infty$-ring.
\end{lemma}

\medskip

\begin{lemma} {\bf [antichiral function on $X^{\physics}$ in terms of
                                            $(x,\theta,\bar{\theta},\vartheta,\bar{\vartheta})$]}\;
 In terms of the standard coordinate functions $(x,\theta, \bar{\theta},\vartheta,\bar{\vartheta})$
   on $\widehat{X}^{\widehat{\boxplus}}$,
  an antichiral function $\breve{f}$ on $X^{\physics} $ is determined by the four components
     $f^{(0)}_{(0)}$, $f^{(\dot{\beta})}_{(\dot{\beta})}$, and
	 $f^{(\dot{1}\dot{2})}_{(\dot{1}\dot{2})}$,  $\dot{\beta}=\dot{1}, \dot{2}$,
   of $\breve{f}$
   via the following formula
   \begin{eqnarray*}
     \breve{f} & = &
	   f_{(0)}^{(0)}(x)
	   + \sum_{\dot{\beta}} \bar{\theta}^{\dot{\beta}}\bar{\vartheta}_{\dot{\beta}}
	           f_{(\dot{\beta})}^{(\dot{\beta})}(x)
	   + \bar{\theta}^{\dot{1}}\bar{\theta}^{\dot{2}}
	      \bar{\vartheta}_{\dot{1}}\bar{\vartheta}_{\dot{2}}
		   f_{(\dot{1}\dot{2})}^{(\dot{1}\dot{2})}(x)
       - \sqrt{-1} \sum_{\alpha,\dot{\beta}}
	          \theta^\alpha\bar{\theta}^{\dot{\beta}}	
			   \sum_\mu
			     \sigma^\mu_{\alpha\dot{\beta}}\partial_\mu f_{(0)}^{(0)}(x) \\
      && \hspace{1em}				
       + \sqrt{-1}\sum_{\alpha, \mu}				
	        \theta^\alpha\bar{\theta}^{\dot{1}} \bar{\theta}^{\dot{2}}
			  \mbox{\Large $($}
			     \bar{\vartheta}_{\dot{1}}
			      \sigma^\mu_{\alpha\dot{2}}\partial_\mu f_{(\dot{1})}^{(\dot{1})}(x)
			     - \bar{\vartheta}_{\dot{2}}\sigma^\mu_{\alpha\dot{1}}
				                             \partial_\mu f_{(\dot{2})}^{(\dot{2})}(x)
			  \mbox{\Large $)$}
       - \theta^1\theta^2\bar{\theta}^{\dot{1}}\bar{\theta}^{\dot{2}}\,
	        \square f_{(0)}^{(0)}(x)\,,
   \end{eqnarray*}
   where $\square := -\partial_0^2+\partial_1^2+\partial_2^2+\partial_3^2$.
\end{lemma}
 
\medskip

\noindent
\makebox[7em][l]{\it Sketch of proof.}
 This follows from
 either the expansion of
  $$
   \breve{f}\;
   =\;  f^{\prime\prime (0)}_{(0)}(x^{\prime\prime})\,
	           +\, \sum_{\dot{\beta}}
			          \bar{\theta}^{\prime\prime \dot{\beta}}
					  \bar{\vartheta}^{\prime\prime}_{\dot{\beta}}\,
			            f^{\prime\prime (\dot{\beta})}_{(\dot{\beta})}(x^{\prime\prime})\,
			   +\,\bar{\theta}^{\prime\prime \dot{1} }\bar{\theta}^{\prime\prime \dot{2}}
				   \bar{\vartheta}^{\prime\prime}_{\dot{1}}
				   \bar{\vartheta}^{\prime\prime}_{\dot{2}}
			       f^{\prime\prime (\dot{1}\dot{2})}_{(\dot{1}\dot{2})}(x^{\prime\prime})\,,
  $$
    where
    $(x^{\prime\prime}, \theta^{\prime\prime}, \bar{\theta}^{\prime\prime},
         \vartheta^{\prime\prime}, \bar{\vartheta}^{\prime\prime})
	    :=(x-\sqrt{-1}\theta\boldsigma\bar{\theta}^t, \theta,\bar{\theta}, \vartheta, \bar{\theta})$
     is the antichiral coordinate-functions on $\widehat{X}^{\widehat{\boxplus}}$, 	  	
   via the complexified $C^\infty$-ring structure of $C^\infty(X^{\physics})^{\scriptsizeach}$
 or solving the overdetermined system of equations on components $f^{\tinybullet}_{\tinybullet}$
  of $\breve{f}$ from setting to zero
  {\small
  \begin{eqnarray*}
   e_{\alpha^\prime}\breve{f}
    & = &   \left(\rule{0ex}{1.2em}\right.\!
                  \frac{\partial}{\partial \theta^\alpha}
				   + \sqrt{-1}\sum_{\dot{\beta},\mu}
				         \sigma^\mu_{\alpha\dot{\beta}}\bar{\theta}^{\dot{\beta}}
						 \partial_\mu
	             \!\left.\rule{0ex}{1.2em}\right)
				 \breve{f}\\
	& = & \breve{f}_{(\alpha)}
	         - \sum_\gamma \theta^\gamma \varepsilon_{\alpha\gamma}\breve{f}_{(12)}
			 + \sum_{\dot{\beta}}\bar{\theta}^{\dot{\beta}}
			      \left(\rule{0ex}{1.2em}\right.\!
				    \breve{f}_{(\alpha\dot{\beta})}
					+ \sqrt{-1}\,\sum_\mu\sigma^\mu_{\alpha\dot{\beta}}
					   \partial_\mu\breve{f}_{(0)}
				  \!\left.\rule{0ex}{1.2em}\right) \\
    && - \sum_{\gamma,\dot{\beta}}
              \theta^\gamma\bar{\theta}^{\dot{\beta}}
                   \left(\rule{0ex}{1.2em}\right.\!
                    \varepsilon_{\alpha\gamma}\breve{f}_{(12\dot{\beta})}
					+ \sqrt{-1}\sum_\mu \sigma^\mu_{\alpha\dot{\beta}}
					    \partial_\mu \breve{f}_{(\gamma)}
                   \!\left.\rule{0ex}{1.2em}\right)
           + \bar{\theta}^{\dot{1}}\bar{\theta}^{\dot{2}}
                   \left(\rule{0ex}{1.2em}\right.\!
                    \breve{f}_{(\alpha\dot{1}\dot{2})}
					+ \sqrt{-1}\sum_{\dot{\beta}, \dot{\delta}, \mu}
					    \varepsilon^{\dot{\beta}\dot{\delta}}\sigma^\mu_{\alpha\dot{\beta}}
						 \partial_\mu \breve{f}_{(\dot{\delta})}
                   \!\left.\rule{0ex}{1.2em}\right)                   \\
    && +\,\sqrt{-1}\sum_{\dot{\beta},\mu}
	        \theta^1\theta^2\bar{\theta}^{\dot{\beta}}
			  \sigma^\mu_{\alpha\dot{\beta}}\partial_\mu\breve{f}_{(12)}
		   -\, \sum_{\gamma}
              \theta^\gamma\bar{\theta}^{\dot{1}}\bar{\theta}^{\dot{2}}
                   \left(\rule{0ex}{1.2em}\right.\!
                    \varepsilon_{\alpha\gamma}\breve{f}_{(12\dot{1}\dot{2})}
					+ \sqrt{-1}\sum_{\dot{\beta}, \dot{\delta}, \mu}
             			  \varepsilon^{\dot{\beta}\dot{\delta}}\sigma^\mu_{\alpha\dot{\beta}}
					    \partial_\mu \breve{f}_{(\gamma\dot{\delta})}
                   \!\left.\rule{0ex}{1.2em}\right)                       \\
	&& +\,\sqrt{-1}\, \theta^1\theta^2\bar{\theta}^{\dot{1}}\bar{\theta}^{\dot{2}}
	                \sum_{\dot{\beta}, \dot{\delta}, \mu}
					 \varepsilon^{\dot{\beta}\dot{\delta}}\sigma^\mu_{\alpha\dot{\beta}}
					  \partial_\mu\breve{f}_{(12\dot{\delta})}\,,
  \end{eqnarray*}
  }for 
  $\alpha^{\prime}=1^\prime, 2^\prime$.
  
\noindent\hspace{40.7em}$\square$

\bigskip

\begin{flushleft}
{\bf $C^\infty(X^{\physics})\subset C^\infty(X^{\widehat{\boxplus}})$
           under the twisted complex conjugation}
\end{flushleft}
\begin{lemma} {\bf [$C^\infty(X^{\physics})\subset C^\infty(X^{\widehat{\boxplus}})$
              under twisted complex conjugation]}\\
 The twisted complex conjugation $(\tinybullet)^\dag$ on $C^\infty(\widehat{X}^{\widehat{\boxplus}})$
   leaves $C^\infty(X^{\physics})$ invariant and takes\\
   $C^\infty(X^{\physics})^{\scriptsizech}$ to $C^\infty(X^{\physics})^{\scriptsizeach}$
    and vice versa.
\end{lemma}
 
\medskip
 
\begin{proof}
 Up to a $(-1)$-factor,
  the set
  $$
    \left\{
	 \begin{array}{l}
	   1\,, \; \theta^\alpha\vartheta_\alpha\,,\;
	    \bar{\theta}^{\dot{\beta}}\bar{\vartheta}_{\dot{\beta}}\,,\;
	    \theta^1\theta^2\vartheta_1\vartheta_2\,,\;
		 \theta^\alpha\bar{\theta}^{\dot{\beta}}\,,\;
		     \theta^\alpha\bar{\theta}^{\dot{\beta}}
			    \vartheta_\alpha\bar{\vartheta}_{\dot{\beta}}\,,\;
		 \bar{\theta}^{\dot{1}}\bar{\theta}^{\dot{2}}
		     \bar{\vartheta}_{\dot{1}}\bar{\vartheta}_{\dot{2}}\,, \\
		 \theta^1\theta^2\bar{\theta}^{\dot{\beta}}\vartheta_\alpha\,,\;
		     \theta^1\theta^2\bar{\theta}^{\dot{\beta}}
			     \vartheta_1\vartheta_2\bar{\vartheta}_{\dot{\beta}}\,,\;
         \theta^\alpha\bar{\theta}^{\dot{1}}\bar{\theta}^{\dot{2}}
                \bar{\vartheta}_{\dot{\beta}}\,,\;
             \theta^\alpha\bar{\theta}^{\dot{1}}\bar{\theta}^{\dot{2}}
                \vartheta_\alpha\bar{\vartheta}_{\dot{1}}\bar{\vartheta}_{\dot{2}}\,,\\
		 \theta^1\theta^2\bar{\theta}^{\dot{1}}\bar{\theta}^{\dot{2}}\,,\;		
		     \theta^1\theta^2\bar{\theta}^{\dot{1}}\bar{\theta}^{\dot{2}}
			   \vartheta_\alpha\bar{\vartheta}_{\dot{\beta}}\,,\;
		 	 \theta^1\theta^2\bar{\theta}^{\dot{1}}\bar{\theta}^{\dot{2}}
			   \vartheta_1\vartheta_2\bar{\vartheta}_{\dot{1}}\bar{\vartheta}_{\dot{2}}
	 \end{array} 	
   \right\}_{\alpha, \dot{\beta}}
  $$
  of generating $(\theta,\bar{\theta},\vartheta, \bar{\vartheta})$-monomials
  for $C^\infty(X^{\physics})$ in $C^\infty(X^{\widehat{\boxplus}})$
   as a sub-$C^\infty(X)^{\Bbb C}$-module
  is closed under the twisted complex conjugation $(\tinybullet)^\dag$.
 This implies that $C^\infty(X^{\physics})$ is $\dag$-invariant in $C^\infty(X^{\widehat{\boxplus}})$.
  
 Since
   a chiral function on $X^{\physics}$ is of the form
   {\small
   $$
   \breve{f}\;
   =\;  f^{(0)}_{(0)}(x+\sqrt{-1}\theta\boldsigma\bar{\theta}^t)\,
	           +\, \sum_\alpha  \theta^\alpha \vartheta_\alpha\,
			            f^{(\alpha)}_{(\alpha)}(x+\sqrt{-1}\theta\boldsigma\bar{\theta}^t)\,			                                                                         
			   +\,\theta^1 \theta^2\vartheta_1 \vartheta_2
			           f^{(12)}_{(12)}(x+\sqrt{-1}\theta\boldsigma\bar{\theta}^t)\,,
  $$}for   
  some
    $f^{(0)}_{(0)}, f^{(\alpha)}_{(\alpha)}, f^{(12)}_{(12)}
	  \in C^\infty(X)^{\Bbb C}$, $\alpha=1,2$,
 its twisted complex conjugate must be of the form
  {\small
  \begin{eqnarray*}
   \breve{f}^\dag  & =
     &  \overline{f^{(0)}_{(0)}}(x-\sqrt{-1}(\theta\boldsigma\bar{\theta}^t)^\dag)\,
	           +\, \sum_\alpha  (\theta^\alpha \vartheta_\alpha)^\dag
			          \cdot \overline{f^{(\alpha)}_{(\alpha)}}
						     (x-\sqrt{-1}(\theta\boldsigma\bar{\theta}^t)^\dag)   \\[-1.2ex]
       && \hspace{16em}							
			   +\,(\theta^1 \theta^2\vartheta_1 \vartheta_2)^\dag			
			         \cdot \overline{f^{(12)}_{(12)}}
						  (x-\sqrt{-1}(\theta\boldsigma\bar{\theta}^t)^\dag)      \\
     & = &   \overline{f^{(0)}_{(0)}}(x-\sqrt{-1}\theta\boldsigma\bar{\theta}^t)\,
	           -\, \sum_{\dot{\alpha}}
			          \bar{\theta}^{\dot{\alpha}} \bar{\vartheta}_{\dot{\alpha}}\,
			            \cdot \overline{f^{(\alpha)}_{(\alpha)}}
						    (x-\sqrt{-1}\theta\boldsigma\bar{\theta}^t)\,			                                                                         
			   +\,\bar{\theta}^{\dot{1}} \bar{\theta}^{\dot{2}}
			         \bar{\vartheta}_{\dot{1}} \bar{\vartheta}_{\dot{2}}
			          \cdot
					  \overline{f^{(12)}_{(12)}}(x-\sqrt{-1}\theta\boldsigma\bar{\theta}^t)\,,
  \end{eqnarray*}}which 
   is antichiral.
 Similarly for the converse
   $(\tinybullet)^\dag: C^\infty(X^{\physics})^{\scriptsizeach}
      \rightarrow C^\infty(X^{\physics})^{\scriptsizech}$.

 This proves the lemma.	  	
 
\end{proof}

\bigskip

\subsection{Wess-Zumino model on $X$ in terms of $X^{\physics}$}

We now proceed to construct the Wess-Zumino model ([W-Z]; also [W-B: Chap.\,V]) under the current setting.

\bigskip

\begin{flushleft}
{\bf Relevant basic computations/formulae}
\end{flushleft}
Let
   \begin{eqnarray*}
     \breve{f} & = &
	   f_{(0)}^{(0)}(x)
	   + \sum_\alpha \theta^\alpha\vartheta_\alpha f_{(\alpha)}^{(\alpha)}(x)
	   + \theta^1\theta^2\vartheta_1\vartheta_2 f_{(12)}^{(12)}(x)
       + \sqrt{-1} \sum_{\alpha,\dot{\beta}}
	          \theta^\alpha\bar{\theta}^{\dot{\beta}}	
			   \sum_\mu
			     \sigma^\mu_{\alpha\dot{\beta}}\partial_\mu f_{(0)}^{(0)}(x) \\
      && \hspace{1em}				
       + \sqrt{-1}\sum_{\dot{\beta}, \mu}				
	        \theta^1\theta^2\bar{\theta}^{\dot{\beta}}
			 \mbox{\Large $($}
			   \vartheta_1\sigma^\mu_{2\dot{\beta}}\partial_\mu f_{(1)}^{(1)}(x)
			     - \vartheta_2\sigma^\mu_{1\dot{\beta}}\partial_\mu f_{(2)}^{(2)}(x)
			 \mbox{\Large $)$}
       - \theta^1\theta^2\bar{\theta}^{\dot{1}}\bar{\theta}^{\dot{2}}\,
	        \square f_{(0)}^{(0)}(x)\,,
   \end{eqnarray*}
  be a chiral function on $X^{\physics}$, determined by
   $(f^{(0)}_{(0)},  f^{(\alpha)}_{(\alpha)}, f^{(12)}_{(12)} )_{\alpha}$.
It follows from Lemma~2.1.7
 that its twisted complex conjugate $\breve{f}^\dag$ is
  the antichiral function on $X^{\physics}$
     determined by
	   $(f^{\dag (0)}_{(0)},
	         f^{\dag (\dot{\beta})}_{(\dot{\beta})},
			 f^{\dag (\dot{1}\dot{2})}_{(\dot{1}\dot{2})})_{\dot{\beta}}
		 = (\overline{f_{(0)}^{(0)}}, -\,\overline{f^{(\beta)}_{(\beta)}},
		       \overline{f^{(12)}_{(12)}})$,
		where
		   $\overline{f_{(\tinybullet)}^{(\tinybullet)}}$
     	  	is the complex conjugate of $f^{(\tinybullet)}_{(\tinybullet)}\in C^\infty(X)^{\,\Bbb C}$.
 Explicitly,
   \begin{eqnarray*}
     \breve{f}^\dag & = &
	   \overline{f_{(0)}^{(0)}}(x)
	   - \sum_{\dot{\beta}} \bar{\theta}^{\dot{\beta}}\bar{\vartheta}_{\dot{\beta}}
	           \overline{f_{(\beta)}^{(\beta)}}(x)
	   + \bar{\theta}^{\dot{1}}\bar{\theta}^{\dot{2}}
	      \bar{\vartheta}_{\dot{1}}\bar{\vartheta}_{\dot{2}}
		               \overline{f_{(12)}^{(12)}}(x)
       - \sqrt{-1} \sum_{\alpha,\dot{\beta}}
	          \theta^\alpha\bar{\theta}^{\dot{\beta}}	
			   \sum_\mu
			     \sigma^\mu_{\alpha\dot{\beta}}
				       \partial_\mu \overline{f_{(0)}^{(0)}}(x) \\
      && \hspace{1em}				
       -\, \sqrt{-1}\sum_{\alpha, \mu}				
	        \theta^\alpha\bar{\theta}^{\dot{1}} \bar{\theta}^{\dot{2}}
			 \mbox{\Large $($}
			    \bar{\vartheta}_{\dot{1}}
			      \sigma^\mu_{\alpha\dot{2}} \partial_\mu \overline{f_{(1)}^{(1)}}(x)
			     - \bar{\vartheta}_{\dot{2}}\sigma^\mu_{\alpha\dot{1}}
				                             \partial_\mu \overline{f_{(2)}^{(2)}}(x)
			 \mbox{\Large $)$}
       - \theta^1\theta^2\bar{\theta}^{\dot{1}}\bar{\theta}^{\dot{2}}\,
	        \square \overline{f_{(0)}^{(0)}}(x)\,.
   \end{eqnarray*}

Consequently, (recall that $\square:= -\partial_0^2+\partial_1^2+\partial_2^2+\partial_3^2$\,)
 {\small
 \begin{eqnarray*}
   \breve{f}^\dag \breve{f} & =  & \breve{f} \breve{f}^\dag \\
    & = & \overline{f^{(0)}_{(0)}}(x)f^{(0)}_{(0)}(x)
	            + \sum_{\alpha}\theta^\alpha\vartheta_\alpha
                      \overline{f^{(0)}_{(0)}}(x) f^{(\alpha)}_{(\alpha)}(x)
			    - \sum_{\dot{\beta}}\bar{\theta}^{\dot{\beta}}\bar{\vartheta}_{\dot{\beta}}
				      \overline{f^{(\beta)}_{(\beta)}}(x)f^{(0)}_{(0)}(x)
                + \theta^1\theta^2\vartheta_1\vartheta_2
				      \overline{f^{(0)}_{(0)}}(x) f^{(12)}_{(12)}(x)\\
    &&	+ \sum_{\alpha,\dot{\beta}}
	             \theta^\alpha\bar{\theta}^{\dot{\beta}}
                \left(\rule{0ex}{1.2em}\right.\!
                  \sqrt{-1}\sum_\mu \sigma^\mu_{\alpha\dot{\beta}}
				   \mbox{\Large $($}
 				    \overline{f^{(0)}_{(0)}}(x) \partial_\mu f^{(0)}_{(0)}(x)
				       - \partial_\mu \overline{f^{(0)}_{(0)}}(x) f^{(0)}_{(0)}(x)
				   \mbox{\Large $)$}
					+ \vartheta_\alpha \bar{\vartheta}_{\dot{\beta}}
					    \overline{f^{(\beta)}_{(\beta)}}(x)f^{(\alpha)}_{(\alpha)}(x)
                \!\left.\rule{0ex}{1.2em}\right)	  \\				
    && +\; \bar{\theta}^{\dot{1}}\bar{\theta}^{\dot{2}}
	            \bar{\vartheta}^{\dot{1}}\bar{\vartheta}^{\dot{2}}\,
                  \overline{f^{(12)}_{(12)}}(x)f^{(0)}_{(0)}(x)				\\
	&&   + \sum_{\dot{\beta}}\theta^1\theta^2\bar{\theta}^{\dot{\beta}}
	             \left(\rule{0ex}{1.2em}\right.\!
				  \sqrt{-1}\sum_\mu
				       \left(\rule{0ex}{1.2em}\right.\!					
						  \vartheta^2 \sigma^\mu_{2\dot{\beta}}						
						      \mbox{\Large $($}		
							    - \overline{f^{(0)}_{(0)}}(x) \partial_\mu f^{(1)}_{(1)}(x)
							      + \partial_\mu\overline{f^{(0)}_{(0)}}(x)f^{(1)}_{(1)}(x)
							 \mbox{\Large $)$} \\[-2ex]
    &&	\hspace{16em}						
						 +\; \vartheta^1\sigma^\mu_{1\dot{\beta}}
						     \mbox{\Large $($}
						       - \overline{f^{(0)}_{(0)}}(x)\partial_\mu f^{(2)}_{(2)}(x)
							   + \partial_\mu \overline{f^{(0)}_{(0)}}(x) f^{(2)}_{(2)}(x)
                             \mbox{\Large $)$}						
					   \!\left.\rule{0ex}{1.2em}\right)  \\
        && \hspace{8em}					
			   -\; \vartheta_1\vartheta_2\bar{\vartheta}_{\dot{\beta}}
			           \overline{f^{(\beta)}_{(\beta)}}(x)f^{(12)}_{(12)}(x)
				 \!\left.\rule{0ex}{1.2em}\right)             \\[1.2ex]
	&&   + \sum_{\alpha}\theta^\alpha\bar{\theta}^{\dot{1}}\bar{\theta}^{\dot{2}}
	             \left(\rule{0ex}{1.2em}\right.\!
				  \sqrt{-1}\sum_\mu
				       \left(\rule{0ex}{1.2em}\right.\!					
						  \bar{\vartheta}^{\dot{2}}\sigma^\mu_{\alpha\dot{2}}						     
						      \mbox{\Large $($}		
							    \partial_\mu \overline{f^{(1)}_{(1)}}(x) f^{(0)}_{(0)}(x)
							      - \overline{f^{(1)}_{(1)}}(x)\partial_\mu f^{(0)}_{(0)}(x)
							 \mbox{\Large $)$} \\[-2ex]
    &&	\hspace{16em}						
						 +\; \bar{\vartheta}^{\dot{1}}\sigma^\mu_{\alpha\dot{1}}
						     \mbox{\Large $($}
						       \partial_\mu \overline{f^{(2)}_{(2)}}(x)f^{(0)}_{(0)}(x)
							   - \overline{f^{(2)}_{(2)}}(x) \partial_\mu f^{(0)}_{(0)}(x)
                             \mbox{\Large $)$}						
					   \!\left.\rule{0ex}{1.2em}\right)  \\
        && \hspace{8em}					
			   +\; \vartheta_\alpha \bar{\vartheta}_{\dot{1}}\bar{\vartheta}_{\dot{2}}
			           \overline{f^{(12)}_{(12)}}(x)f^{(\alpha)}_{(\alpha)}(x)
				 \!\left.\rule{0ex}{1.2em}\right)             \\[1.2ex]	
    && +\; \theta^1\theta^2\bar{\theta}^{\dot{1}}\bar{\theta}^{\dot{2}}
	             \left(\rule{0ex}{1.2em}\right.\!
	              \mbox{\Large $($}
				   -\,\square\overline{f^{(0)}_{(0)}}(x)\cdot f^{(0)}_{(0)}(x)
				    -\, \overline{f^{(0)}_{(0)}}(x)\cdot\square f^{(0)}_{(0)}(x)
					+\,2 \sum_\mu
				             \partial_\mu \overline{f^{(0)}_{(0)}}(x)\,
							      \partial^\mu f^{(0)}_{(0)}(x)
				  \mbox{\Large $)$} \\
       && \hspace{6em}				
               + \sum_{\alpha,\dot{\beta}}	   \vartheta_\alpha\bar{\vartheta}_{\dot{\beta}}
			        \mbox{\Large $($}
			          \overline{f^{(\beta)}_{(\beta)}}(x)
					         \cdot \sqrt{-1}\sum_\mu\bar{\sigma}^{\mu,\dot{\beta}\alpha}
							   \partial_\mu f^{(\alpha)}_{(\alpha)}(x)   \\[-2ex]
           && \hspace{16em}							
                     -\, f^{(\alpha)}_{(\alpha)}(x)
					     \cdot	\sqrt{-1}\sum_\mu \bar{\sigma}^{\mu, \dot{\beta}\alpha}
						   \partial_\mu \overline{f^{(\beta)}_{(\beta)}}(x)
					\mbox{\Large $)$}					\\
    && 	\hspace{6em}
	         +\, \vartheta_1\vartheta_2\bar{\vartheta}_{\dot{1}}\bar{\vartheta}_{\dot{2}}
			        \overline{f^{(12)}_{(12)}}(x)f^{(12)}_{(12)}(x)
           	    \!\left.\rule{0ex}{1.2em}\right)    \,;
 \end{eqnarray*}}  
{\small
 \begin{eqnarray*}
  \breve{f}^2 & =
    & f^{(0)}_{(0)}(x)^2
	   + 2 \sum_\alpha\theta^\alpha\vartheta_\alpha
	           f^{(0)}_{(0)}(x) f^{(\alpha)}_{(\alpha)}(x)
	   + 2\, \theta^1\theta^2\vartheta_1\vartheta_2
		      \mbox{\Large $($}
			   f^{(0)}_{(0)}(x) f^{(12)}_{(12)}(x)
			   - f^{(1)}_{(1)}(x) f^{(2)}_{(2)}(x)
			  \mbox{\Large $)$}\\	
    &&  +\, (\mbox{terms of $\bar{\theta}$-degree $\ge 1$})\, ;
 \end{eqnarray*}}     
{\small
 \begin{eqnarray*}
  \breve{f}^3 & =
    & f^{(0)}_{(0)}(x)^3
	   + 3 \sum_\alpha\theta^\alpha\vartheta_\alpha
	           f^{(0)}_{(0)}(x)^2 f^{(\alpha)}_{(\alpha)}(x)
	   + 3\, \theta^1\theta^2\vartheta_1\vartheta_2
		      \mbox{\Large $($}
			   f^{(0)}_{(0)}(x)^2 f^{(12)}_{(12)}(x)
			   - 2\, f^{(0)}_{(0)}(x)f^{(1)}_{(1)}(x)f^{(2)}_{(2)}(x)
			  \mbox{\Large $)$}\\	
    &&  +\, (\mbox{terms of $\bar{\theta}$-degree $\ge 1$})\, ;
 \end{eqnarray*}}     
{\small
 \begin{eqnarray*}
  (\breve{f}^\dag)^2 & =
    & \overline{f^{(0)}_{(0)}}(x)^2
	   - 2 \sum_{\dot{\beta}}\bar{\theta}^{\dot{\beta}}\bar{\vartheta}_{\dot{\beta}}
	           \overline{f^{(0)}_{(0)}}(x)\overline{f^{(\beta)}_{(\beta)}}(x)
	   + 2\, \bar{\theta}^{\dot{1}}\bar{\theta}^{\dot{2}}
	            \bar{\vartheta}_{\dot{1}}\bar{\vartheta}_{\dot{2}}
		      \mbox{\Large $($}
			   \overline{f^{(0)}_{(0)}}(x) \overline{f^{(12)}_{(12)}}(x)
			   - \overline{f^{(1)}_{(1)}}(x) \overline{f^{(2)}_{(2)}}(x)
			  \mbox{\Large $)$}\\	
    &&  +\, (\mbox{terms of $\theta$-degree $\ge 1$})\, ;
 \end{eqnarray*}    } 
{\small
 \begin{eqnarray*}
  (\breve{f}^\dag)^3 & =
    & \overline{f^{(0)}_{(0)}}(x)^3
	   - 3 \sum_{\dot{\beta}}\bar{\theta}^{\dot{\beta}}\bar{\vartheta}_{\dot{\beta}}
	          \overline{f^{(0)}_{(0)}}(x)^2
			  \overline{f^{(\beta)}_{(\beta)}}(x)
	   + 3\, \bar{\theta}^{\dot{1}}\bar{\theta}^{\dot{2}}
	            \bar{\vartheta}_{\dot{1}}\bar{\vartheta}_{\dot{2}}
		      \mbox{\Large $($}
			   \overline{f^{(0)}_{(0)}}(x)^2    \overline{f^{(12)}_{(12)}}(x)
			   - 2\, \overline{f^{(0)}_{(0)}}(x)\overline{f^{(1)}_{(1)}}(x)
			           \overline{f^{(2)}_{(2)}}(x)
			  \mbox{\Large $)$}\\	
    &&  +\, (\mbox{terms of $\theta$-degree $\ge 1$})\,.
 \end{eqnarray*}}     
 
\noindent
Basic identities in spinor calculus used in the above computations to render the spinor indices
 paired up more elegantly are all collected in the Appendix;
see, e.g.,  [W-B: Chap.\,III, Exercises (1) \& (2); Appendices A \& B] for a more complete list.

\bigskip

\begin{flushleft}
{\bf The standard purge-evaluation/index-contracting map ${\cal P}$ with respect to
          $(\theta,\bar{\theta},\vartheta,\bar{\vartheta})$}
\end{flushleft}
Given a uniform purge-evaluation map
 ${\cal P}:C^\infty(X^{\physics})\rightarrow C^\infty(\widehat{X})$,
 let $\breve{f}\in C^\infty(X^{\physics})$.
Then the variation of ${\cal P}(\breve{f})$
  under the infinitesimal supersymmetry generators $Q_\alpha$'s and $\bar{Q}_{\dot{\beta}}$
  takes the form
  {\small
  \begin{eqnarray*}
   Q_\alpha {\cal P}(\breve{f})
    & :=\: &  \mbox{\Large $($}
	                  \mbox{\Large $\frac{\partial}{\partial\theta^\alpha}$}
		                - \sqrt{-1}\sum_{\dot{\beta},\mu}
			          \sigma^\mu_{\alpha\dot{\beta}}\bar{\theta}^{\dot{\beta}}\partial_\mu
	               \mbox{\Large $)$}\,
  		          {\cal P}(\breve{f})   \\
    & = &{\cal P}(\vartheta_\alpha)f^{(\alpha)}_{(\alpha)}(x)
	       - \theta_\alpha {\cal P}(\vartheta_1\vartheta_2)f^{(12)}_{(12)}(x) \\
    &&	+\sum_{\dot{\beta}}\bar{\theta}^{\dot{\beta}}
                \left(\rule{0ex}{1.2em}\right.\!
				  \sum_\mu \sigma^\mu_{\alpha\dot{\beta}}
				     \mbox{\Large $($}
					  f^{(0)}_{[\mu]}(x)   -\sqrt{-1}\partial_\mu f^{(0)}_{(0)}(x)
					 \mbox{\Large $)$}
			     + {\cal P}(\vartheta_\alpha\bar{\vartheta}_{\dot{\beta}})
				        f^{(\alpha\dot{\beta})}_{(\alpha\dot{\beta})}(x)
                \!\left.\rule{0ex}{1.2em}\right)		\\
    && + \sum_{\gamma,\dot{\beta}}
	           \theta^\gamma\bar{\theta}^{\dot{\beta}}
                \left(\rule{0ex}{1.2em}\right.\!
                 - \varepsilon_{\alpha\gamma}
				      \mbox{\Large $($}
					    \sum_{\gamma^\prime}
						  {\cal P}(\vartheta_{\gamma^\prime})
						                f^{(\gamma^\prime)}_{(12\dot{\beta})}(x)
				      + {\cal P}(\vartheta_1\vartheta_2\bar{\vartheta}_{\dot{\beta}})
                              f^{(12\dot{\beta})}_{(12\dot{\beta})}(x)										      
					  \mbox{\Large $)$}\\[-1.8ex]
       && \hspace{18em}					
				 + \sqrt{-1}{\cal P}(\vartheta_\gamma)
				      \sum_\mu\sigma^\mu_{\alpha\dot{\beta}}
					                      \partial_\mu f^{(\gamma)}_{(\gamma)}(x)
                \!\left.\rule{0ex}{1.2em}\right)				\\
    && +\, \bar{\theta}^{\dot{1}}\bar{\theta}^{\dot{2}}				
	            \left(\rule{0ex}{1.2em}\right.\!
				 \sum_{\dot{\delta}}{\cal P}(\bar{\vartheta}_{\dot{\delta}})
				   \mbox{\Large $($}
				     f^{(\dot{\delta})}_{(\alpha\dot{1}\dot{2})}(x)
					- \sqrt{-1}\sum_{\dot{\beta},\mu}
					     \varepsilon^{\dot{\beta}\dot{\delta}}\sigma^\mu_{\alpha\dot{\beta}}
						  \partial_\mu f^{(\dot{\delta})}_{(\dot{\delta})}(x)					
				   \mbox{\Large $)$}
                 + {\cal P}(\vartheta_\alpha
				                           \bar{\vartheta}_{\dot{1}}\bar{\vartheta}_{\dot{2}})
				         f^{(\alpha\dot{1}\dot{2})}_{(\alpha\dot{1}\dot{2})}(x) 				
				\!\left.\rule{0ex}{1.2em}\right) \\
    && -  \sum_{\dot{\beta}}\theta^1\theta^2\bar{\theta}^{\dot{\beta}}
	            {\cal P}(\vartheta_1\vartheta_2)\cdot
				  \sqrt{-1}\sum_\mu
				    \sigma_{\alpha\dot{\beta}}\partial_\mu f^{(12)}_{(12)}(x)\\
    && + \sum_\gamma \theta^\gamma\bar{\theta}^{\dot{1}}\bar{\theta}^{\dot{2}}
	           \left(\rule{0ex}{1.2em}\right.\!
			    -\varepsilon_{\alpha\gamma}
				  \mbox{\Large $($}
				   f^{(0)}_{(12\dot{1}\dot{2})}(x)
				   + \sum_{\gamma^\prime,\dot{\delta}}
				        {\cal P}(\vartheta_{\gamma^\prime}\bar{\vartheta}_{\dot{\delta}})
						     f^{(\gamma^\prime)}_{(12\dot{1}\dot{2})}(x)
				   + {\cal P}(\vartheta_1\vartheta_2
				                             \bar{\vartheta}_{\dot{1}}\bar{\vartheta}_{\dot{2}})
                             f^{(12\dot{1}\dot{2})}_{(12\dot{1}\dot{2})}(x)	
				  \mbox{\Large $)$}  \\
        && \hspace{8em}				
				+\, \sqrt{-1}\sum_{\dot{\beta},\dot{\delta},\mu}
				    \varepsilon^{\dot{\beta}\dot{\delta}}\sigma^\mu_{\alpha\dot{\beta}}
					    \mbox{\Large $($}
						 \sum_\nu\sigma^\nu_{\gamma\dot{\delta}}\partial_\mu f^{(0)}_{[\nu]}(x)
						  + {\cal P}(\vartheta_\gamma \bar{\vartheta}_{\dot{\delta}})
						        \partial_\mu f^{(\gamma\dot{\delta})}_{(\gamma\dot{\delta})}(x)
						\mbox{\Large $)$}
			   \!\left.\rule{0ex}{1.2em}\right)   \\
    && -\, \theta^1\theta^2\bar{\theta}^{\dot{1}}\bar{\theta}^{\dot{2}}\cdot
	            \sqrt{-1}\sum_{\dot{\beta},\dot{\delta},\mu}
				  \varepsilon^{\dot{\beta}\dot{\delta}}\sigma^\mu_{\alpha\dot{\beta}}
				   \mbox{\Large $($}
				     \sum_\gamma {\cal P}(\vartheta_\gamma)
					   \partial_\mu f^{(\gamma)}_{(12\dot{\delta})}(x)
					   + {\cal P}(\vartheta_1\vartheta_2\bar{\vartheta}_{\dot{\delta}})
					         \partial_\mu f^{(12\dot{\delta})}_{(12\dot{\delta})}(x)
				   \mbox{\Large $)$}
  \end{eqnarray*}}and 
 {\small
 \begin{eqnarray*}
  \bar{Q}_{\dot{\beta}} {\cal P}(\breve{f})
    & :=\: &  \mbox{\Large $($}
	                -\,\mbox{\Large $\frac{\partial}{\partial\bar{\theta}^{\dot{\beta}}}$}
					+ \sqrt{-1}\sum_{\alpha,\mu}
					                          \theta^\alpha\sigma^\mu_{\alpha\dot{\beta}}\partial_\mu
	               \mbox{\Large $)$}\,{\cal P}(\breve{f})\\				
    & =
	  & -\,{\cal P}(\bar{\vartheta}_{\dot{\beta}})f^{(\dot{\beta})}_{(\dot{\beta})}
	      + \sum_\alpha \theta^\alpha
               \left(\rule{0ex}{1.2em}\right.\!
                \sum_\mu\sigma^\mu_{\alpha\dot{\beta}}
				  \mbox{\Large $($}
				   f^{(0)}_{[\mu]}
				   + \sqrt{-1}\partial_\mu f^{(0)}_{(0)}(x)
				  \mbox{\Large $)$}\,
				  + {\cal P}(\vartheta_\alpha\bar{\vartheta}_{\dot{\beta}})
				          f^{(\alpha\dot{\beta})}_{(\alpha\dot{\beta})}(x)
               \!\left.\rule{0ex}{1.2em}\right)			   \\
    && +\,\bar{\theta}_{\dot{\beta}}
	           {\cal P}(\bar{\vartheta}_{\dot{1}}\bar{\vartheta}_{\dot{2}})
                  f^{(\dot{1}\dot{2})}_{(\dot{1}\dot{2})}(x)	\\		
    && +\theta^1\theta^2
	          \left(\rule{0ex}{1.2em}\right.\!
			    \sum_\gamma {\cal P}(\vartheta_\gamma)
			      \mbox{\Large $($}
				    -\,f^{(\gamma)}_{(12\dot{\beta})}(x)
				    +\sqrt{-1}\sum_{\alpha,\mu}
				        \varepsilon^{\alpha\gamma}\sigma^\mu_{\alpha\dot{\beta}}
					     \partial_\mu f^{(\gamma)}_{(\gamma)}(x)
				  \mbox{\Large $)$}\,
			  -\, {\cal P}(\vartheta_1\vartheta_2\bar{\vartheta}_{\dot{\beta}})	
			            f^{(12\dot{\beta})}_{(12\dot{\beta})}(x)
			  \!\left.\rule{0ex}{1.2em}\right)\\				
    && +\sum_{\alpha,\dot{\delta}}
	          \theta^\alpha\bar{\theta}^{\dot{\delta}}
              \left(\rule{0ex}{1.2em}\right.\!
                -\,\varepsilon_{\dot{\beta}\dot{\delta}}
				      \mbox{\Large $($}
					   \sum_{\dot{\delta}^\prime}
					      {\cal P}(\bar{\vartheta}_{\dot{\delta}^\prime})
						    f^{(\dot{\delta}^\prime)}_{(\alpha\dot{1}\dot{2})}(x)
							+ {\cal P}(\vartheta_\alpha\bar{\vartheta}_{\dot{1}}
							                          \bar{\vartheta}_{\dot{2}})
								    f^{(\alpha\dot{1}\dot{2})}_{(\alpha\dot{1}\dot{2})}(x)
					  \mbox{\Large $)$}\\[-1.2ex]
         && \hspace{18em}					
				+ \sqrt{-1}{\cal P}(\bar{\vartheta}_{\dot{\delta}})
				    \sum_\mu \sigma^\mu_{\alpha\dot{\beta}}
					                     \partial_\mu f^{(\dot{\delta})}_{(\dot{\delta})}(x)
              \!\left.\rule{0ex}{1.2em}\right)	\\
    && + \sum_{\dot{\delta}}
               \theta^1\theta^2\bar{\theta}^{\dot{\delta}}	
			    \left(\rule{0ex}{1.2em}\right.\!
				  \varepsilon_{\dot{\beta}\dot{\delta}}
				    \mbox{\Large $($}
					  f^{(0)}_{(12\dot{1}\dot{2})}(x)
					  + \sum_{\gamma,\dot{\delta}^\prime}
					       {\cal P}(\vartheta_\gamma\bar{\vartheta}_{\dot{\delta}^\prime})
						     f^{(\gamma\dot{\delta}^\prime)}_{(12\dot{1}\dot{2})}(x)
					  + {\cal P}(\vartheta_1\vartheta_2
					                            \bar{\vartheta}_{\dot{1}}\bar{\vartheta}_{\dot{2}})
							f^{(12\dot{1}\dot{2})}_{(12\dot{1}\dot{2})}(x)
					\mbox{\Large $)$} \\
       && \hspace{8em}					
	         +\sqrt{-1}\sum_{\alpha,\gamma,\mu}
			      \varepsilon^{\alpha\gamma}\sigma^\mu_{\alpha\dot{\beta}}
				    \mbox{\Large $($}
					  \sum_\nu\sigma^\nu_{\gamma\dot{\delta}}
					     \partial_\mu f^{(0)}_{[\nu]}(x)
					  + {\cal P}(\vartheta_\gamma\bar{\vartheta}_{\dot{\delta}})
					        f^{(\gamma\dot{\delta})}_{(\gamma\dot{\delta})}(x)
					\mbox{\Large $)$}
				\!\left.\rule{0ex}{1.2em}\right) \\
    && +\sum_\alpha
	           \theta^\alpha\bar{\theta}^{\dot{1}}\bar{\theta}^{\dot{2}}
	            \cdot
				\sqrt{-1}\sigma^\mu_{\alpha\dot{\beta}}
				 {\cal P}(\bar{\vartheta}_{\dot{1}}\bar{\vartheta}_{\dot{2}})
				  \partial_\mu f^{(\dot{1}\dot{2})}_{(\dot{1}\dot{2})}(x)\\
	&& +\,\theta^1\theta^2\bar{\theta}^{\dot{1}}\bar{\theta}^{\dot{2}}
                \cdot
                \sqrt{-1}\sum_{\alpha,\gamma,\mu}
                 \varepsilon^{\alpha\gamma}\sigma^\mu_{\alpha\dot{\beta}}
  				   \mbox{\Large $($}
				    \sum_{\dot{\delta}}{\cal P}(\bar{\vartheta}_{\dot{\delta}})\,
					   \partial_\mu f^{(\dot{\delta})}_{(\gamma\dot{1}\dot{2})}(x)
					+ {\cal P}(\vartheta_{\gamma}
					                          \bar{\vartheta}_{\dot{1}}\bar{\vartheta}_{\dot{2}})\,
                          \partial_\mu f^{(\gamma\dot{1}\dot{2})}_{(\gamma\dot{1}\dot{2})}(x)
				   \mbox{\Large $)$}\,.				
 \end{eqnarray*}
 }

A comparison of
  $\breve{f}^\dag\breve{f}$
   (resp.\ $\breve{f}^2$,
                $\breve{f}^3$,  and
				$Q_\alpha {\cal P}(\breve{f})$ \&
				    $\bar{Q}_{\dot{\beta}}{\cal P}(\breve{f})$)
  with [W-B: Chap.\,V: Eq.\,(5.9)]
   (resp.\ [ibidem: Eq.\,(5.7)], [ibidem: Eq.\.(5.8)], [ibidem: Eq.\.(3.10)])
     of Julius Wess and Jonathan Bagger
  motivates the following definition:

\bigskip

\begin{definition} {\bf [standard purge-evaluation/index-contracting map]}\; {\rm
 The purge-evaluation map
    ${\cal P}: C^\infty(X^{\physics})\rightarrow C^\infty(\widehat{X})$
  that takes
 {\small
 \begin{eqnarray*}
  \breve{f}
  &= &
    f^{(0)}_{(0)}
	+ \sum_{\alpha}\theta^\alpha\vartheta_\alpha f^{(\alpha)}_{(\alpha)}
	+ \sum_{\dot{\beta}}
	       \bar{\theta}^{\dot{\beta}}\bar{\vartheta}_{\dot{\beta}}
		     f^{(\dot{\beta})}_{(\dot{\beta})}  \\
  && \hspace{1em}			
    +\; \theta^1\theta^2\vartheta_1\vartheta_2 f^{(12)}_{(12)}
	+ \sum_{\alpha,\dot{\beta}}\theta^\alpha \bar{\theta}^{\dot{\beta}}
	      \left(\rule{0ex}{1.2em}\right.\!
		    \sum_\mu \sigma^\mu_{\alpha\dot{\beta}} f^{(0)}_{[\mu]}\,
			 +\, \vartheta_\alpha\bar{\vartheta}_{\dot{\beta}}
			        f^{(\alpha\dot{\beta})}_{(\alpha\dot{\beta})}
		  \!\left.\rule{0ex}{1.2em}\right)
    + \bar{\theta}^{\dot{1}}\bar{\theta}^{\dot{2}}
	   \bar{\vartheta}_{\dot{1}}\bar{\vartheta}_{\dot{2}}
	    f^{(\dot{1}\dot{2})}_{(\dot{1}\dot{2})}  \\
  && \hspace{1em}		
	+ \sum_{\dot{\beta}}\theta^1\theta^2\bar{\theta}^{\dot{\beta}}
	     \left(\rule{0ex}{1.2em}\right.\!
		   \sum_\alpha \vartheta_\alpha f^{(\alpha)}_{(12\dot{\beta})}\,
		    +\, \vartheta_1\vartheta_2\bar{\vartheta}_{\dot{\beta}}
			            f^{(12\dot{\beta})}_{(12\dot{\beta})}
		 \!\left.\rule{0ex}{1.2em}\right)
    + \sum_\alpha \theta^\alpha\bar{\theta}^{\dot{1}}\bar{\theta}^{\dot{2}}
	   \left(\rule{0ex}{1.2em}\right.\!
	     \sum_{\dot{\beta}} \bar{\vartheta}_{\dot{\beta}}
		     f^{(\dot{\beta})}_{(\alpha\dot{1}\dot{2})}\,
         +\, \vartheta_\alpha \bar{\vartheta}_{\dot{1}}\bar{\vartheta}_{\dot{2}}	
		           f^{(\alpha\dot{1}\dot{2})}_{(\alpha\dot{1}\dot{2})}
	   \!\left.\rule{0ex}{1.2em}\right)\\
  && \hspace{1em}
	+\; \theta^1\theta^2\bar{\theta}^{\dot{1}}\bar{\theta}^{\dot{2}}
	     \left(\rule{0ex}{1.2em}\right.\!
		  f^{(0)}_{(12\dot{1}\dot{2})}
		  + \sum_{\alpha,\dot{\beta}} \vartheta_\alpha\bar{\vartheta}_{\dot{\beta}}
		         f^{(\alpha\dot{\beta})}_{(12\dot{1}\dot{2})}
		  + \vartheta_1\vartheta_2\bar{\vartheta}_{\dot{1}}\bar{\vartheta}_{\dot{2}}
		        f^{(12\dot{1}\dot{2})}_{(12\dot{1}\dot{2})}
		 \!\left.\rule{0ex}{1.2em}\right)   		
 \end{eqnarray*}}to 
 {\small
 \begin{eqnarray*}
  {\cal P}(\breve{f})
  &= &
    f^{(0)}_{(0)}
	+ \sum_{\alpha}\theta^\alpha f^{(\alpha)}_{(\alpha)}
	+ \sum_{\dot{\beta}}
	       \bar{\theta}^{\dot{\beta}}
		     f^{(\dot{\beta})}_{(\dot{\beta})}  \\
  && \hspace{1em}			
    +\; \theta^1\theta^2 f^{(12)}_{(12)}
	+ \sum_{\alpha,\dot{\beta}}\theta^\alpha \bar{\theta}^{\dot{\beta}}
	      \left(\rule{0ex}{1.2em}\right.\!
		    \sum_\mu \sigma^\mu_{\alpha\dot{\beta}} f^{(0)}_{[\mu]}\,
			 +\, f^{(\alpha\dot{\beta})}_{(\alpha\dot{\beta})}
		  \!\left.\rule{0ex}{1.2em}\right)
    + \bar{\theta}^{\dot{1}}\bar{\theta}^{\dot{2}}
	    f^{(\dot{1}\dot{2})}_{(\dot{1}\dot{2})}  \\
  && \hspace{1em}		
	+ \sum_{\dot{\beta}}\theta^1\theta^2\bar{\theta}^{\dot{\beta}}
	     \left(\rule{0ex}{1.2em}\right.\!
		   \sum_\alpha  f^{(\alpha)}_{(12\dot{\beta})}\,
		    +\, f^{(12\dot{\beta})}_{(12\dot{\beta})}
		 \!\left.\rule{0ex}{1.2em}\right)
    + \sum_\alpha \theta^\alpha\bar{\theta}^{\dot{1}}\bar{\theta}^{\dot{2}}
	   \left(\rule{0ex}{1.2em}\right.\!
	     \sum_{\dot{\beta}}
		     f^{(\dot{\beta})}_{(\alpha\dot{1}\dot{2})}\,
         +\, f^{(\alpha\dot{1}\dot{2})}_{(\alpha\dot{1}\dot{2})}
	   \!\left.\rule{0ex}{1.2em}\right)\\
  && \hspace{1em}
	+\; \theta^1\theta^2\bar{\theta}^{\dot{1}}\bar{\theta}^{\dot{2}}
	     \left(\rule{0ex}{1.2em}\right.\!
		  f^{(0)}_{(12\dot{1}\dot{2})}
		  + \sum_{\alpha,\dot{\beta}}
		         f^{(\alpha\dot{\beta})}_{(12\dot{1}\dot{2})}
		  +  f^{(12\dot{1}\dot{2})}_{(12\dot{1}\dot{2})}
		 \!\left.\rule{0ex}{1.2em}\right)    \\
 &\:=: &
    f_{(0)}
	+ \sum_{\alpha}\theta^\alpha f_{(\alpha)}
	+ \sum_{\dot{\beta}}
	       \bar{\theta}^{\dot{\beta}}
		     f_{(\dot{\beta})}  \\
  && \hspace{1em}			
    +\; \theta^1\theta^2 f_{(12)}
	+ \sum_{\alpha,\dot{\beta}}\theta^\alpha \bar{\theta}^{\dot{\beta}}
	      \left(\rule{0ex}{1.2em}\right.\!
		    \sum_\mu \sigma^\mu_{\alpha\dot{\beta}} f_{[\mu]}\,
			 +\, f^\prime_{(\alpha\dot{\beta})}
		  \!\left.\rule{0ex}{1.2em}\right)
    + \bar{\theta}^{\dot{1}}\bar{\theta}^{\dot{2}}
	    f_{(\dot{1}\dot{2})}  \\
  && \hspace{1em}		
	+ \sum_{\dot{\beta}}\theta^1\theta^2\bar{\theta}^{\dot{\beta}}
	     \left(\rule{0ex}{1.2em}\right.\!
		   f^\prime_{(12\dot{\beta})}
		   + f^{\prime\prime}_{(12\dot{\beta})}\,
		   + f^{\prime\prime\prime}_{(12\dot{\beta})}
		 \!\left.\rule{0ex}{1.2em}\right)
    + \sum_\alpha \theta^\alpha\bar{\theta}^{\dot{1}}\bar{\theta}^{\dot{2}}
	   \left(\rule{0ex}{1.2em}\right.\!
		  f^\prime_{(\alpha\dot{1}\dot{2})}\,
		 +\, f^{\prime\prime}_{(\alpha\dot{1}\dot{2})}\,
         +\, f^{\prime\prime\prime}_{(\alpha\dot{1}\dot{2})}
	   \!\left.\rule{0ex}{1.2em}\right)\\
  && \hspace{1em}
	+\; \theta^1\theta^2\bar{\theta}^{\dot{1}}\bar{\theta}^{\dot{2}}
	     \left(\rule{0ex}{1.2em}\right.\!
		  f^\prime_{(12\dot{1}\dot{2})}
		  +  f^{\prime\prime}_{(12\dot{1}\dot{2})}		
		  +  f^{\prime\prime\prime}_{(12\dot{1}\dot{2})}
		  +  f^{\prime\prime\prime\prime}_{(12\dot{1}\dot{2})}
		  +  f^{\prime\prime\prime\prime\prime}_{(12\dot{1}\dot{2})}
		  +  f^{\prime\prime\prime\prime\prime\prime}_{(12\dot{1}\dot{2})}
		 \!\left.\rule{0ex}{1.2em}\right)    \\[1.8ex]
 &\:=: &
    f_{(0)}
	+ \sum_{\alpha}\theta^\alpha f_{(\alpha)}
	+ \sum_{\dot{\beta}}
	       \bar{\theta}^{\dot{\beta}}
		     f_{(\dot{\beta})}
    +\; \theta^1\theta^2 f_{(12)}
	+ \sum_{\alpha,\dot{\beta}}\theta^\alpha \bar{\theta}^{\dot{\beta}}
	     f_{(\alpha\dot{\beta})}
    + \bar{\theta}^{\dot{1}}\bar{\theta}^{\dot{2}} f_{(\dot{1}\dot{2})}  \\
  && \hspace{1em}		
	+ \sum_{\dot{\beta}}\theta^1\theta^2\bar{\theta}^{\dot{\beta}}f_{(12\dot{\beta})}
    + \sum_\alpha \theta^\alpha\bar{\theta}^{\dot{1}}\bar{\theta}^{\dot{2}}
	        f_{(\alpha\dot{1}\dot{2})}
	+ \theta^1\theta^2\bar{\theta}^{\dot{1}}\bar{\theta}^{\dot{2}}
  		    f_{(12\dot{1}\dot{2})}
 \end{eqnarray*}}
 
 \noindent
 is called the {\it standard purge-evaluation map}
   with respect to $(\theta,\bar{\theta},\vartheta,\bar{\vartheta})$.\footnote{Here,
                                        for example, we use
										  $f^{\prime}_{(12\dot{\beta})}$,
										  $f^{\prime\prime}_{(12\dot{\beta})}$,
										  $f^{\prime\prime\prime}_{(12\dot{\beta})}$
										 to distinguish
										  ${\cal P}(\vartheta_1f^{(1)}_{(12\dot{\beta})})$,
	                                      ${\cal P}(\vartheta_2f^{(2)}_{(12\dot{\beta})})$,
                                          ${\cal P}(\vartheta_1\vartheta_2\bar{\theta}_{\dot{\beta}}
										       f^{(12\dot{\beta})}_{(12\dot{\beta})})$.
                                        Since we use no further detail of this sum in this work,
										  we leave it unspecified which is which.
										Similarly for
										  $f^\prime_{(12\dot{1}\dot{2})}$,
										  $f^{\prime\prime}_{(12\dot{1}\dot{2})}$,
										  $f^{\prime\prime\prime}_{(12\dot{1}\dot{2})}$,
										  $f^{\prime\prime\prime\prime}_{(12\dot{1}\dot{2})}$,
										  $f^{\prime\prime\prime\prime\prime}_{(12\dot{1}\dot{2})}$,
										  $f^{\prime\prime\prime\prime\prime\prime}_{(12\dot{1}\dot{2})}$.
                                        } 
 Note that in this process, the lower fermionic indices of $(\vartheta,\bar{\vartheta})$
   in each $(\theta,\bar{\theta},\vartheta,\bar{\theta})$-monomial
   are contracted out\footnote{The fermionic indices of a coefficient (e.g.\ $f^{(12)}_{(12)}$)
                                       behave like the indices of a tensor.
                                     When performing the standard purge-evaluation map ${\cal P}$,
									   paired lower-and-upper fermionic indices
									     (e.g.\ the lower $1$ and $2$ in  $\vartheta_1\vartheta_2$ versus
										             the upper $1$ and $2$ in $f^{(12)}_{(12)}$)
									   should drop out from the index structure via evaluation of a dual spinor on a spinor.
									 The true index structure should be the leftover indices.
									 (E.g.\ $\vartheta_1\vartheta_2\, f^{(12)}_{(12)}
									                 \stackrel{\cal P}{\longrightarrow}
													 f^{(12)}_{(12)}=: f_{(12)}$.)													 
									                            }
   with the same upper fermionic indices of the coefficient
   $f^{(\tinybullet)}_{(\tinybullet)}$ and we call ${\cal P}$ also
   the {\it standard index-contracting map}.
}\end{definition}

%

\bigskip

\begin{flushleft}
{\bf A supersymmetric action functional for chiral multiplets via the Fundamental Theorem}
\end{flushleft}
A chiral function $\breve{f}$ on $X^{\physics}$ contains four independent component
 $f^{(0)}_{(0)}, f^{(\alpha)}_{(\alpha)}, f^{(12)}_{(12)}
   \in C^\infty(X)^{\Bbb C}$,
 $\alpha=1,2$.
It follows from Theorem~1.5.3
 that
 {\small
 \begin{eqnarray*}
  \lefteqn{S_1(f^{(0)}_{(0)}; f^{(1)}_{(1)}, f^{(2)}_{(2)};
                                  f^{(12)}_{(12)})\;
     =\; S_1(f_{(0)}; f_{(1)}, f_{(2)};  f_{(12)})								  } \\
   &&    :=\;  \int_{\widehat{X}}d^4x\,
                       d\bar{\theta}^{\dot{2}} d\bar{\theta}^{\dot{1}} d\theta^2 d\theta^1\,
	           {\cal P}(\breve{f}^\dag \breve{f})\,
			+\, \int_{\widehat{X}}d^4x\,d\theta^2 d\theta^1\,
			    {\cal P}\mbox{\Large $($}
				                     \lambda \breve{f}
				                      + \mbox{\large $\frac{1}{2}$}m \breve{f}^2
                                      + \mbox{\large $\frac{1}{3}$}g \breve{f}^3
								   \mbox{\Large $)$}  \\
   && \hspace{16em}
            +\, \int_{\widehat{X}}d^4x\,d\bar{\theta}^{\dot{2}} d\bar{\theta}^{\dot{1}}\,
			    {\cal P}\mbox{\Large $($}
				                     \bar{\lambda} \breve{f}^\dag
				                      + \mbox{\large $\frac{1}{2}$}\bar{m} (\breve{f}^\dag)^2
                                      + \mbox{\large $\frac{1}{3}$}\bar{g} (\breve{f}^\dag)^3
								   \mbox{\Large $)$}
 \end{eqnarray*}}gives  
 a functional of chiral multiplets
 $(f^{(0)}_{(0)}; f^{(1)}_{(1)}, f^{(2)}_{(2)}; f^{(12)}_{(12)})$
 that is invariant under supersymmetries up to a boundary term on $X$.

Explicitly, up to boundary terms on $X$, this is the action functional (cf.\ [W-B: Chap.\,V, Eq.\,(5.11)])
 {\small
 \begin{eqnarray*}
  \lefteqn{S_1(f_{(0)}, f_{(1)}, f_{(2)}, f_{(12)})}\\
  && =\; \int_X d^4x
	             \left(\rule{0ex}{1.2em}\right.\!	
					4\, \sum_\mu
				            \partial_\mu \overline{f_{(0)}}(x)\,  \partial^\mu f_{(0)}(x) \\
       && \hspace{6em}				
               + \sum_{\alpha,\dot{\beta}}	
			        \mbox{\Large $($}
			          -\,\overline{f_{(\beta)}}(x)
					         \cdot \sqrt{-1}\sum_\mu\bar{\sigma}^{\mu,\dot{\beta}\alpha}
							   \partial_\mu f_{(\alpha)}(x)
                     +\, f_{(\alpha)}(x)
					     \cdot	\sqrt{-1}\sum_\mu \bar{\sigma}^{\mu, \dot{\beta}\alpha}
						   \partial_\mu \overline{f_{(\beta)}}(x)
					\mbox{\Large $)$}					\\
    && 	\hspace{6em}
	         +\,\overline{f_{(12)}}(x)f_{(12)}(x)
           	    \!\left.\rule{0ex}{1.2em}\right)                               \\
   && \hspace{1.6em}
       +\, \int_X d^4x
               \left(\rule{0ex}{1.2em}\right.\!
                  \lambda f_{(12)}(x)
                     + m\, \mbox{\Large $($}
			              f_{(0)}(x) f_{(12)}(x)- f_{(1)}(x) f_{(2)}(x)
			                   \mbox{\Large $)$}\\
         && \hspace{6em}
		          + g\, \mbox{\Large $($}
			          f_{(0)}(x)^2 f_{(12)}(x)
			           - 2\, f_{(0)}(x)f_{(1)}(x)f_{(2)}(x)
			          \mbox{\Large $)$}
                + (\mbox{complex conjugate})					
			\left.\rule{0ex}{1.2em}\right)\,.
 \end{eqnarray*}}The  
index structure of this explicit expression implies that this functional is indeed Lorentz invariant.

\bigskip

This is what underlies [W-B: Chap.\,V] of Wess \& Bagger
  from the aspect of $C^\infty$-Algebraic Geometry.

\bigskip

\section{Supersymmetric $U(1)$ gauge theory with matter on $X$ in terms of $X^{\physics}$}

In this section, we reproduce the supersymmetric $U(1)$ gauge theory with matter
 in [W-B: Chap.\,VI \& $U(1)$ part of Chap.\,VII] of Wess \& Bagger
 from the (complexified ${\Bbb Z}/2$-graded) $C^\infty$-Algebraic Geometry setting in Sec.\,1.

\bigskip

\subsection{The bundle/sheaf context underlying a supersymmetric $U(1)$ gauge\\ theory with matter
                         built from $X^{\physics}$}

On the mathematics side,
 a gauge theory usually begins with the setup of a principal bundle, associate bundles from representations,
  connections and their curvature tensor,
 e.g.\;[D-K: Sec.\,2.1].
Such a geometric setup remains there for supersymmetric gauge theories,
  only that there are rarely mentioned or brought to front in physics literature.
When mathematicians (or mathematics-oriented physicists) attempt to set such geometry up,
 the precise setting depends on how the notion of `superspace' is defined in their context.
 
In this subsection, we give such a geometric setup for supersymmetric $U(1)$-gauge theory with matter
 in the language of (complexified ${\Bbb Z}/2$-graded) $C^\infty$-Algebraic Geometry.
As nilpotent objects (e.g.\ $\theta$, $\bar{\theta}$, $\vartheta$, $\bar{\vartheta}$)
 are everywhere in our problem,
 it is more convenient to use the language of sheaves, rather than bundles.
This is the sheaf theory in (complexified ${\Bbb Z}/2$-graded) $C^\infty$-Algebraic Geometry behind the scene
 for [W-B: Chap.\,VI \& $U(1)$ part of Chap.\,VII] of Wess \& Bagger.
It can be generalized to the higher rank, nonabelian case.

\vspace{6em}
						
\begin{flushleft}						
{\bf The built-in principal bundle/sheaf and all that}
\end{flushleft}
The multiplicative group of invertible elements of ${\cal O}_X^{\physics}$
 defines a principal sheaf ${\cal O}_X^{\physics, \times}$ over $X^{\physics}$.
It corresponds to the sheaf of sections of a tautological principal ${\Bbb C}^{\times}$-bundle
 $\boldsymbol{P}^{\,{\Bbb C}^\times}$ over $X^{\physics}$.
Note that an $\breve{f}\in {\cal O}_X^{\physics}$ is invertible if and only if
 its $(\theta,\bar{\theta},\vartheta,\bar{\theta})$-degree-zero component is invertible, i.e.\;
 $f^{(0)}_{(0)}\in {\cal O}_X^{\,{\Bbb C}, \times}$.
Thus,
 $$
  {\cal O}_X^{\physics, \times}\;
    =\;  \{\breve{f}\in {\cal O}_X^{\physics}\,|\,
	              f^{(0)}_{(0)}\in {\cal O}_X^{\,{\Bbb C}, \times}\}\,.
 $$
This is the {\it tautological principal ${\Bbb C}^\times$-sheaf} on $X^{\physics}$,
 where ${\Bbb C}^\times := ({\Bbb C}-\{0\}, \times)$ is the multiplicative group of ${\Bbb C}$.
Since ${\Bbb C}^\times$ is abelian,
 the adjoint representation of of ${\Bbb C}^\times$ on ${\Bbb C}$ as the associated Lie algebra
 is trivial.
This realizes ${\cal O}_X^{\physics}$
 as the associated sheaf of Lie algebras of ${\cal O}_X^{\physics, \times}$.
The {\it exponential map} is given explicitly by
 $$
  \begin{array}{cccccc}
  e:= \Exponential  & :
    & {\cal O}_X^{\physics}  & \longrightarrow  &  {\cal O}_X^{\physics,\times}\\[1.2ex]
  && \breve{f} & \longmapsto
                              &  e^{f^{(0)}_{(0)}}
							            \cdot \sum_{l=0}^4 \frac{1}{l!}(\breve{f}_{(\ge 1)})^l\,.
  \end{array}
 $$
Here,
 $\breve{f}=f^{(0)}_{(0)}+\breve{f}_{(\ge 1)}$,
  with
     $f^{(0)}_{(0)}\in {\cal O}_X^{\,\Bbb C}$  and
	 $\breve{f}_{(\ge 1)}$ the nilpotent part of $\breve{f}$,   and
 note that $(\breve{f}_{(\ge 1)})^5=0$ for $\breve{f}\in {\cal O}_X^{\tinyphysics}$.
Note also that when $f^{(0)}_{(0)}$ is real, this is compatible with the $C^\infty$-hull structure
 of ${\cal O}_X^{\physics}$.
Its local inverse near the identity section $1\in {\cal O}_X^{\physics,\times}$
 defines the {\it Log map}:
  $$
  \begin{array}{cccccc}
   \Log  & :
    & {\cal O}_X^{\physics, \times}  & \longrightarrow  &  {\cal O}_X^{\physics}\\[1.2ex]
  && \breve{f} & \longmapsto
                              & \log f^{(0)}_{(0)}
							      + \sum_{l=1}^4  \frac{(-1)^l}{l}
								      ({f^{(0)}_{(0)}}^{-1}\breve{f}_{(\ge 1)})^l\,.
  \end{array}
 $$
This is also compatible with the $C^\infty$-hull structure of
 ${\cal O}_X^{\physics}\supset {\cal O}_X^{\physics, \times}$
 when $f^{(0)}_{(0)}$ is real.

\bigskip

\begin{definition}  {\bf [physics-related principal subsheaves in ${\cal O}_X^{\physics, \times}$]}\; {\rm
 There are four physics-related principal sheaves of subgroups in the tautological principal
   ${\Bbb C}^\times$-sheaf ${\cal O}_X^{\physics, \times}$ over $X^{\physics}$.
 Each is characterized by its associated sheaf of Lie subalgebras in ${\cal O}_X^{\physics}$:
 \begin{itemize}
  \item[(1)]  \makebox[23.4em][l]{[{\it tautological chiral principal ${\Bbb C}^\times$-sheaf
                    ${\cal O}_X^{\physics,\scriptsizech, \times}$}\,]}
   Note that
    $$
     {\cal O}_X^{\physics, \times,\scriptsizech}\;
		:=\; {\cal O}_X^{\physics,\times}\cap {\cal O}_X^{\physics, \scriptsizech}
		= {\cal O}_X^{\physics, \scriptsizech, \times}\,,
	$$
	where ${\cal O}_X^{\physics, \scriptsizech, \times}$
	   is the sheaf of group of invertible elements in ${\cal O}_X^{\physics, \scriptsizech}$.
   This defines the {\it tautological chiral principal ${\Bbb C}^\times$-sheaf} on $X^{\physics}$,
     whose associated sheaf of Lie algebras is ${\cal O}_X^{\physics, \scriptsizech}$,
     (which is	the same as $\sqrt{-1}\cdot{\cal O}_X^{\physics, \scriptsizech}$).
   The exponential map $\Exponential$ restricts to
   $e=\Exponential: {\cal O}_X^{\physics, \scriptsizech}
          \rightarrow {\cal O}_X^{\physics,\scriptsizech, \times} $.
     
  \item[(2)] \makebox[25.6em][l]{[{\it tautological antichiral principal ${\Bbb C}^\times$-sheaf
                    ${\cal O}_X^{\physics, \scriptsizeach, \times}$}\,]}
   Note that
    $$
     {\cal O}_X^{\physics, \times,\scriptsizeach}\;
		:=\; {\cal O}_X^{\physics,\times}\cap {\cal O}_X^{\physics, \scriptsizeach}
		= {\cal O}_X^{\physics, \scriptsizeach, \times}\,,
	$$
	where ${\cal O}_X^{\physics, \scriptsizeach, \times}$
	   is the sheaf of group of invertible elements in ${\cal O}_X^{\physics, \scriptsizeach}$.
   This defines the {\it tautological antichiral principal ${\Bbb C}^\times$-sheaf} on $X^{\physics}$, 
     whose associated sheaf of Lie algebras is ${\cal O}_X^{\physics, \scriptsizeach}$,
     (which is	the same as $\sqrt{-1}\cdot{\cal O}_X^{\physics, \scriptsizeach}$).
   The exponential map $\Exponential$ restricts to
   $e=\Exponential: {\cal O}_X^{\physics, \scriptsizeach}
          \rightarrow {\cal O}_X^{\physics,\scriptsizeach, \times} $.
 
  \item[(3)] \makebox[30.4em][l]{[{\it tautological principal $U(1)$-sheaf
                        ${\cal O}_X^{\physics,\flat, U(1)}$ and its descendants}\,]}
  Let\footnote{Caution
                           that ${\cal O}_X^{\tinyphysics,\flat}$ is only an ${\cal O}_X^{\,\Bbb C}$-submodule,
                           {\it not} an ${\cal O}_X^{\,\Bbb C}$-subalgebra, 						   
			  			    of ${\cal O}_X^{\tinyphysics}$.
		   				  It s not closed under the multiplication of sections.						   
                             }  
   $$
     {\cal O}_X^{\physics,\flat}\;
	   :=\; \left\{\breve{f}\in {\cal O}_X^{\physics}\left|
	              \begin{array}{l}
				    \mbox{$\breve{f}$ is of the form\;
					   {\small (in the standard coordinate}
					        }   \\
					  \mbox{\small  functions $(x,\theta,\bar{\theta}, \vartheta, \bar{\vartheta})$
						    on $\widehat{X}^{\widehat{\boxplus}}$)
							}\\[1.2ex]
				    \mbox{\small
					  $f^{(0)}_{(0)}
					  + \sum_\alpha \theta^\alpha\vartheta_\alpha\, f^{(\alpha)}_{(\alpha)}
					  + \sum_{\dot{\beta}}
                             \bar{\theta}^{\dot{\beta}}\bar{\vartheta}_{\dot{\beta}}\,
							   f^{(\dot{\beta})}_{(\dot{\beta})}
					  + \theta^1\theta^2\vartheta_1\vartheta_2\, f^{(12)}_{(12)}$
					     }\\[1.2ex] 
					\mbox{\small
					 $+\, \sum_{\alpha,\dot{\beta}}\theta^\alpha\bar{\theta}^{\dot{\beta}}
					          \sum_\mu\sigma^\mu_{\alpha\dot{\beta}}\,f^{(0)}_{[\mu]}
					 +  \bar{\theta}^{\dot{1}}\bar{\theta}^{\dot{2}}
					     \bar{\vartheta}_{\dot{1}}\bar{\vartheta}_{\dot{2}}\,
						   f^{(\dot{1}\dot{2})}_{(\dot{1}\dot{2})}$
					     }\\[1.2ex] 
					\mbox{\small
					  $+\, \sum_{\dot{\beta}}\theta^1\theta^2\bar{\theta}^{\dot{\beta}}
                           \sum_\alpha \vartheta_\alpha\, f^{(\alpha)}_{(12\dot{\beta})}					  
					    + \sum_\alpha \theta^\alpha\bar{\theta}^{\dot{1}}\bar{\theta}^{\dot{2}}
						     \sum_{\dot{\beta}}\bar{\vartheta}_{\dot{\beta}}
							   f^{(\dot{\beta})}_{(\alpha\dot{1}\dot{2})}$
							      } \\[1.2ex]
					 \mbox{\small
					  $+\, \theta^1\theta^2\bar{\theta}^{\dot{1}}\bar{\theta}^{\dot{2}}
						      f^{(0)}_{(12\dot{1}\dot{2})}$\,;
					     }\\[1.8ex]
                   \mbox{namely, {\small
				      $f^{(\alpha\dot{\beta})}_{(\alpha\dot{\beta})}
					   = f^{(12\dot{\beta})}_{(12\dot{\beta})}
					   = f^{(\alpha\dot{1}\dot{2})}_{(\alpha\dot{1}\dot{2})}
					   = f^{(\alpha\dot{\beta})}_{(12\dot{1}\dot{2})}
					   = f^{(12\dot{1}\dot{2})}_{(12\dot{1}\dot{2})}
					   = 0$
				                   }}
				  \end{array}
                                              \right. \right\}
   $$
  and
   $$
  {\cal O}_X^{\physics,\flat,\stc}\;
    :=\; \{\breve{s}\in {\cal O}_X^{\physics, \flat} \,|\, \breve{s}^\dag = \breve{s}  \}\;
	 \subset {\cal O}_X^{\physics}
   $$
   as a sub-${\cal O}_X$-module.\footnote{Here,
                                                                           {\it stc}\, stands for {\it self-twisted-complex-conjugate}
                                                                               } 
  The image of the restriction of $\Exponential$ to $\sqrt{-1}\cdot{\cal O}_X^{\physics,\flat, \stc}$ 
    is a sheaf of subgroups in ${\cal O}_X^{\physics, \times}$
   whose restriction to $X^{\Bbb C}$ is a sheaf of sections of a principal $U(1)$-bundle over $X$.
  Denote this image in ${\cal O}_X^{\physics, \times}$ by ${\cal O}_X^{\physics,\flat, U(1)}$ and
   call it the {\it tautological principal $U(1)$-sheaf on $X^{\physics}$}.
  This corresponds to the principal $U(1)$-bundle in the gauge theory we are going to study.
  The construction realizes
    $\sqrt{-1}\cdot {\cal O}_X^{\physics,\flat,\stc}$
    as the associated sheaf of Lie algebras of ${\cal O}_X^{\physics,\flat, U(1)}$.
  One can impose further ${\Bbb R}$-linear constraints on ${\cal O}_X^{\physics,\flat, \stc}$ to obtain
 {\it descendants} of ${\cal O}_X^{\physics, \flat, U(1)}$
   via the restriction of the exponential map $\Exponential$ on
  $\sqrt{-1}\cdot {\cal O}_X^{\physics,\flat,\stc}$.
  
  \item[(4)] \makebox[30em][l]{[{\it tautological principal ${\Bbb R}^\times$-sheaf
                        ${\cal O}_X^{\physics,\flat, {\Bbb R}^\times}$ and its descendants}\,]}
  The image of the restriction of $\Exponential$ to ${\cal O}_X^{\physics,\flat, \stc}$
    is a sheaf of subgroups in ${\cal O}_X^{\physics, \times}$
   whose restriction to $X^{\Bbb C}$ is a sheaf of sections of a principal ${\Bbb R}^\times$-bundle over $X$.
  Denote this image in ${\cal O}_X^{\physics, \times}$
    by ${\cal O}_X^{\physics,\flat, {\Bbb R}^\times}$ and
   call it the {\it tautological principal ${\Bbb R}^\times$-sheaf on $X^{\physics}$}.
  The construction realizes ${\cal O}_X^{\physics,\flat,\stc}$
    as the associated sheaf of Lie algebras of ${\cal O}_X^{\physics,\flat, {\Bbb R}^\times}$.
  One can impose further ${\Bbb R}$-linear constraints on ${\cal O}_X^{\physics,\flat, \stc}$ to obtain
 {\it descendants} of ${\cal O}_X^{\physics, \flat, {\Bbb R}^\times}$
 via the restriction of the exponential map $\Exponential$ on ${\cal O}_X^{\physics,\flat,\stc}$.
 \end{itemize}
}\end{definition}

\medskip

\begin{definition} {\bf [${\cal O}_X^{\physics, \scriptsizech}$,
              ${\cal O}_X^{\physics, \scriptsizeach}$, ${\cal O}_X^{\physics,\flat,{\Bbb R}^\times}$
			  as ${\cal O}_X^{\physics,\scriptsizech,\times}$-modules]}\; {\rm
 For each $e_m\in {\Bbb R}$,\footnote{\makebox[11.6em][l]{\it Note for mathematicians}
                                                          The number $e_m$ is the {\it electric charge} of the chiral matter fields
  														   realized as global sections of ${\cal O}_X^{\tinyphysics, ch}$
														   in the supersymmetric $U(1)$ gauge theory with matter.
                                                             }  
  we shall consider the following action of the tautological principal ${\Bbb C}^\times$-sheaf
   ${\cal O}_X^{\physics,\scriptsizech,\times}$
  on ${\cal O}_X^{\physics, \scriptsizech}$,
              ${\cal O}_X^{\physics, \scriptsizeach}$, and ${\cal O}_X^{\physics,\flat,U(1)}$.
 This turns them into ${\cal O}_X^{\physics,\scriptsizech,\times}$-modules.
 \begin{itemize}
  \item[(1)] \makebox[7em][l]{[$\,{\cal O}_X^{\physics, \scriptsizech}\,$]}
   \parbox[t]{31.6em}{Left
     multiplication in ${\cal O}_X^{\physics}$ by the section:\;
      $e^{\sqrt{-1}\,e_m \breve{\Lambda}}\cdot {\cal O}_X^{\physics,\scriptsizech}$. \\
	 Note that this leaves
	  ${\cal O}_X^{\physics,\scriptsizech}\subset {\cal O}_X^{\physics}$ invariant.}

  \item[(2)] \makebox[7em][l]{[$\,{\cal O}_X^{\physics, \scriptsizeach}\,$]}
   \parbox[t]{31.6em}{Right
     multiplication in ${\cal O}_X^{\physics}$ by the twisted complex conjugate of the section:
     ${\cal O}_X^{\physics,\scriptsizeach}\cdot e^{-\,\sqrt{-1}\,e_m\breve{\Lambda}^\dag}$.\\
	 Note that this leaves
	  ${\cal O}_X^{\physics,\scriptsizeach}\subset {\cal O}_X^{\physics}$ invariant.}

  \item[(3)] \makebox[7em][l]{[$\,{\cal O}_X^{\physics,\flat, {\Bbb R}^\times}\,$]}
   \parbox[t]{31.6em}{Left
    multiplication in ${\cal O}_X^{\physics}$ by the inverse of the twisted complex conjugate of the section
     in accompany with right multiplication in ${\cal O}_X^{\physics}$ by the inverse of the section:	
     $e^{\sqrt{-1}\,e_m\breve{\Lambda}^\dag}
	   \cdot {\cal O}_X^{\physics,\flat, {\Bbb R}^\times}
	   \cdot e^{-\,\sqrt{-1}\,e_m \breve{\Lambda}}$.\\
    Note that this leaves
	  ${\cal O}_X^{\physics,\flat,{\Bbb R}^\times}\subset {\cal O}_X^{\physics}$ invariant.}
 \end{itemize}
 Here
  a section of ${\cal O}_X^{\physics, ch,\times}$ is expressed
    in terms of its associated sheaf of Lie algebras as $e^{\sqrt{-1}\breve{\Lambda}}$
    with $\breve{\Lambda}\in {\cal O}_X^{\physics, \scriptsizech}$   and
  a section of ${\cal O}_X^{\physics,\flat, {\Bbb R}^\times}$  is expressed
    in terms of its associated sheaf of Lie algebras as $e^{\breve{V}}$
    with $\breve{V}\in {\cal O}_X^{\physics,\flat,\stc}$.
 From the gauge-theoretical aspect,
  ${\cal O}_X^{\physics, \scriptsizech,\times}$ plays the role of
  the {\it sheaf of gauge symmetries} in the problem.
}\end{definition}

\bigskip

For all the discussions below until the last theme
 `{\sl A supersymmetric action functional for $U(1)$ gauge theory with matter on $X$}' of Sec.\;3.5,
 we will set $e_m=1$ so that we don't have to carry the symbol all along.
By replacing $\breve{\Lambda}$ with $e_m\breve{\Lambda}$, we recover the charge $e_m$ case.

\bigskip

\begin{definition} {\bf [${\cal O}_X^{\physics}$ \& $\widehat{\cal O}_X^{\,\widehat{\boxplus}}$
                                             as ${\cal O}_X^{\physics,\scriptsizech,\times}$-modules]}\; {\rm
 The same three operations (1), (2), (3) in Definition~3.1.2
   realize both ${\cal O}_X^{\physics}$ and $\widehat{\cal O}_X^{\,\widehat{\boxplus}}$
   as {\it left} (cf.\;(1)), {\it right} (cf.\;(2)), {\it bi-} (cf.\;(3))
   ${\cal O}_X^{\physics,\scriptsizech,\times}$-modules respectively.
}\end{definition}

\bigskip

\begin{flushleft}
{\bf (Even left) connections and their curvature tensor}
\end{flushleft}
The notion of `connection' in [L-Y1: Sec.\,2.1] can be adapted here.
However, there are two opposing factors ahead of us:
 \begin{itemize}
  \item[($+$)]
   Since physics focuses on ${\cal O}_X^{\physics}$, which is purely even,
     all the complication due to the ${\Bbb Z}/2$-grading that we have to address in ibidem is gone.
   Thus, one {\it only needs to consider even left connections}.
	
  \item[($-$)]	
   Since a connection is a generalization of the exterior differential operator $d$ and
    $\Der_{\Bbb C}(\widehat{X})$ does not leave ${\cal O}_X^{\physics}$ invariant,
   one {\it cannot just consider a connection on an ${\cal O}_X^{\physics}$-module alone}.
 \end{itemize}
Based on the physical  applications in practice, with both of the above two factors taken into account, 
 one is led to consider even left connections $\nabla$
 on a full $\widehat{\cal O}^{\,\widehat{\boxplus}}_X$-module $\widehat{\cal E}$.
It won't necessarily leave an ${\cal O}_X^{\physics}$-submodule ${\cal F}$ of $\widehat{\cal E}$ invariant
 but this is okay as long as we  know where $\nabla\!_\xi s$ is in $\widehat{\cal E}$
 for all $s\in {\cal F}$ and $\xi\in \Der_{\Bbb C}(\widehat{X})$.
 
\bigskip

\begin{definition} {\bf [even left connection on
  $\widehat{\cal O}_X^{\,\widehat{\boxplus}}$-module]}\;
{\rm (Cf.\;[L-Y1:Definition 2.1.2].)
 Let $\widehat{\cal E}$ be an $\widehat{\cal O}_X^{\,\widehat{\boxplus}}$-module.
 An {\it even left connection} $\widehat{\nabla}$ on $\widehat{\cal E}$
    is a ${\Bbb C}$-bilinear pairing	
  $$
    \begin{array}{ccccc}
	 \widehat{\nabla} & : & {\cal T}_{\widehat{X}^{\widehat{\boxplus}}} \times \widehat{\cal E}
	     & \longrightarrow   & \widehat{\cal E}  \\[1.2ex]
    && (\xi, s)              &  \longmapsto    &  \widehat{\nabla}\!_{\xi}s		 	
	\end{array}
  $$
  such that
	\begin{itemize}
	 \item[(1)]  [{\it $\widehat{\cal O}_X$-linearity
	                                   in the ${\cal T}_{\widehat{X}^{\widehat{\boxplus}}}$-argument}]\\[.6ex]	
	  $\mbox{\hspace{1em}}$
	  $\widehat{\nabla}\!_{f_1\xi_1 + f_2\xi_2}s\;
	     =\; f_1 \widehat{\nabla}\!_{\xi_1}s + f_2 \widehat{\nabla}\!_{\xi_2}s$, \hspace{1em}
      for $f_1, f_2 \in \widehat{\cal O}_X^{\,\widehat{\boxplus}}$,
	       $\xi_1, \xi_2 \in {\cal T}_{\widehat{X}^{\widehat{\boxplus}}}$,  and
	       $s\in \widehat{\cal E}$;

     \item[(2)]  [{\it ${\Bbb C}$-linearity in the $\widehat{\cal E}$-argument}]\\[.6ex]
	 $\mbox{\hspace{1em}}$
	 $\widehat{\nabla}\!_\xi(c_1s_1+c_2s_2)\;
	     =\;  c_1 \widehat{\nabla}\!_\xi s_1 + c_2 \widehat{\nabla}\!_\xi s_2$, \hspace{1em}
	  for $c_1, c_2\in {\Bbb C}$, $\xi\in {\cal T}_{\widehat{X}^{\widehat{\boxplus}}}$, and
	       $s_1, s_2\in \widehat{\cal E}$;
	
	 \item[(3)] [{\it ${\Bbb Z}/2$-graded Leibniz rule in the $\widehat{\cal E}$-argument}]\footnote{In
	                                       [L-Y1: Definition 2.1.2 ], a left connection on $\widehat{\cal E}$ is required to satisfy
	                                       the  generalized ${\Bbb Z}/2$-graded Leibniz rule in the $\widehat{\cal E}$-argument:\;
	                                         $\widehat{\nabla}_\xi(fs)\;
	                                            =\; (\xi f)s
	                                              + (-1)^{p(f)p(\xi)}\,f\cdot\,\!^{\varsigma_{\!f}}\!
											                    (\widehat{\nabla})_\xi s$,
                                                for $f\in \widehat{\cal O}_X^{\,\widehat{\boxplus}}$, 
												      $\xi\in {\cal T}_{\widehat{X}^{\widehat{\boxplus}}}$ parity homogeneous
	                                                  and $s\in\widehat{\cal E}$,
                                             where
	                                               $^{\varsigma_{\!f}}\!(\widehat{\nabla})$
                                                      is the parity-conjugation of $\widehat{\nabla}$ induced by $f$;
                                                  i.e., 	
	                                               $^{\varsigma_{\!f}}\!(\widehat{\nabla})
	                                                      = \widehat{\nabla}$,  if $f$ is even, or
		                                           $\,\!^{\varsigma}\widehat{\nabla}
			                                                := \mbox{(even part of $\widehat{\nabla}$)}\,
					                                          -\, \mbox{(odd part of $\widehat{\nabla}$)}$  if $f$ is odd;
                                               (cf.\ [L-Y1: Definition~1.3.1]).
											When $\widehat{\nabla}$ is even,
											    $^{\varsigma_{\!f}}\!(\widehat{\nabla})=\widehat{\nabla}$ always
											 and the general ${\Bbb Z}/2$-graded Leibniz rule reduces to
 											  the ${\Bbb Z}$-graded Leibniz rule.											  
												                                                                     }\\[.6ex]   
	 $\mbox{\hspace{1em}}$
	 $\widehat{\nabla}\!_\xi(fs)\;
	   =\; (\xi f)s
	           + (-1)^{p(f)p(\xi)}\,f\cdot  \widehat{\nabla}\!_\xi s$,\\[.6ex]
      for $f\in \widehat{\cal O}_X^{\,\boxplus}$,
	       $\xi\in {\cal T}_{\widehat{X}^{\widehat{\boxplus}}}$ parity homogeneous
	       and $s\in\widehat{\cal E}$.
	\end{itemize}
  As an operation on the pairs $(\xi, s)$,
   a connection $\nabla$ on $\widehat{\cal E}$ is applied to $\xi$ from the right
   while applied to $s$ from the left;\footnote{In the ${\Bbb Z}/2$-graded world,
                                                                     it is instructive to denote $\widehat{\nabla}_{\!\xi}s$ as
																	  $\xi \widehat{\nabla} s$ or $_{\xi}\!\widehat{\nabla} s$
																	 (though we do not adopt it as a regularly used notation in this work).
																	In particular,
  																	  from $_{f\xi}\!\widehat{\nabla} s$
																	  to $f (\,\!_{\xi}\!\widehat{\nabla} s)$,
                                                                     $f$ and $\widehat{\nabla}$ do {\it not} pass each other.
                                                                      }  
  cf.\ [L-Y1: Lemma~1.3.7 \& Remark~1.3.8].
}\end{definition}

\bigskip

Note that since $\widehat{\nabla}$ is even, the parity of $\widehat{\nabla}\!_\xi$ is the same as that of $\xi$.

\bigskip

\begin{lemma-definition} {\bf [curvature tensor of (even left) connection]}\; {\rm
 (Cf.\ [L-Y1: Lemma/Definition~2.1.9].)
 Continuing Definition ~3.1.4.
 {\it Let $\widehat{\nabla}$ be an even left connection on $\widehat{\cal E}$.
 Then the correspondence
  $$
   F^{\widehat{\nabla}}\;:\;
   (\xi_1,\xi_2;  s)\;\longmapsto\;
        \mbox{\Large $($}
		  [\widehat{\nabla}_{\!\xi_1}, \widehat{\nabla}_{\!\xi_2}\}\,
	         -\,\widehat{\nabla}_{[\xi_1,\xi_2\}}
		\mbox{\Large $)$}\, s\,,
  $$
  for $\xi_1,\xi_2\in \Der_{\Bbb C}(\widehat{X}^{\widehat{\boxplus}})$ parity-homogeneous and
       $s\in \widehat{\cal E}$,
  defines an $\Endsheaf_{\widehat{\cal O}_X^{\,\widehat{\boxplus}}}(\widehat{\cal E})$-valued
   $2$-tensor on $\widehat{X}^{\widehat{\boxplus}}$.}
 {\rm We shall call $F^{\widehat{\nabla}}$ thus defined
     the {\it curvature tensor} on $\widehat{X}^{\widehat{\boxplus}}$
	   associated to the even left connection $\widehat{\nabla}$ on $\widehat{\cal E}$.}			     
}\end{lemma-definition}

\medskip

\begin{proof}
 This is a special case of [L-Y1: Lemma/Definition~2.1.9] with the odd part of $\widehat{\nabla}$ vanishes.
 Using the ${\Bbb Z}/2$-graded Leibniz rule,
  one can show straightforwardly that,
  for  $f$, $\xi_1$, $\xi_2$ parity-homogeneous,
  \begin{eqnarray*}
    \mbox{\Large $($}
     [\widehat{\nabla}\!_{f\xi_1}, \widehat{\nabla}\!_{\xi_2}   \}
        -\widehat{\nabla}\!_{[\xi_1,\xi_2\}}	
    \mbox{\Large $)$}s
     & = &  f\cdot
       \mbox{\Large $($}
         [\widehat{\nabla}\!_{\xi_1}, \widehat{\nabla}\!_{\xi_2}   \}
        -\widehat{\nabla}\!_{[\xi_1,\xi_2\}}	
	   \mbox{\Large $)$}s\,,                  \\[1.2ex]
    \mbox{\Large $($}
     [\widehat{\nabla}\!_{\xi_1}, \widehat{\nabla}\!_{f\xi_2}   \}
        -\widehat{\nabla}\!_{[\xi_1,\xi_2\}}	
    \mbox{\Large $)$}s
     & = &  (-1)^{p(f)p(\xi_1)}\,f\cdot
       \mbox{\Large $($}
         [\widehat{\nabla}\!_{\xi_1}, \widehat{\nabla}\!_{\xi_2}   \}
        -\widehat{\nabla}\!_{[\xi_1,\xi_2\}}	
	   \mbox{\Large $)$}s\,,                  \\[1.2ex]
    \mbox{\Large $($}
     [\widehat{\nabla}\!_{\xi_1}, \widehat{\nabla}\!_{\xi_2}   \}
        -\widehat{\nabla}\!_{[\xi_1,\xi_2\}}	
    \mbox{\Large $)$}(fs)
     & = & (-1)^{p(f)(p(\xi_1)+ p(\xi_2))} f\cdot
       \mbox{\Large $($}
         [\widehat{\nabla}\!_{\xi_1}, \widehat{\nabla}\!_{\xi_2}   \}
        -\widehat{\nabla}\!_{[\xi_1,\xi_2\}}	
	   \mbox{\Large $)$}s\,.
  \end{eqnarray*}
 This proves the lemma.
 
\end{proof}

\bigskip

Since in this work, we only address even left connections, we will simply call them {\it connections}.

\bigskip

\begin{definition} {\bf [connection on ${\cal O}_X^{\physics}$-submodule - abuse]}\; {\rm
 Though in general
  a connection $\widehat{\nabla}$ on an $\widehat{\cal O}_X^{\,\widehat{\boxplus}}$-module
   $\widehat{\cal E}$ does not leave a ${\cal O}_X^{\physics}$-submodule ${\cal F}$ invariant
  and hence does not restrict to a connection on ${\cal F}$,
  for $\xi\in \Der_{\Bbb C}(\widehat{X}^{\widehat{\boxplus}})$  and $s\in {\cal F}$
  one does know where $\widehat{\nabla}\!_\xi s$  goes in $\widehat{\cal E}$.
 Furthermore, in all our applications, $\xi\in \Der_{\Bbb C}(\widehat{X})$ and hence
  $$
   \widehat{\nabla}\!_{\xi}\,:\; {\cal F}\;
     \longrightarrow\;  \widehat{\cal O}_X\cdot{\cal F}
  $$
  in $\widehat{\cal E}$.
 For the convenience of terminology, we will still call $\widehat{\nabla}$ a {\it connection on ${\cal F}$}
   with the understanding that it may not take values in ${\cal F}$ alone.
}\end{definition}

\bigskip

\begin{flushleft}
{\bf Pre-vector superfields and their associated (even left) connections }
\end{flushleft}
\begin{definition} {\bf [pre-vector superfield]}\; {\rm
 A global section
   $$
     \breve{V}\; \in\;  \Gamma({\cal O}_X^{\physics,\flat,\stc})\;
      =:\;  C^\infty(X^{\physics})^{\flat,\stc}
  $$
  is called a {\it pre-vector superfield} on $X^{\physics}$.
}\end{definition}

\bigskip

For physicists working on supersymmetric gauge theories,
 the following class of even left connections (adapted to the current $U(1)$ case) is the major concern. 

\bigskip

\begin{definition} {\bf [(even left) connection associated to pre-vector superfield]}\; {\rm
 With the above setting,
   let $\breve{V}\in C^\infty(X^{\physics})^{\flat,\stc}$ be a pre-vector superfield on $X^{\physics}$.
 Then, one can define
   an (even left) connection $\widehat{\nabla}^{\breve{V}}$
   on $\widehat{\cal O}_X^{\,\widehat{\boxplus}}$ (as a left ${\cal O}_X^{\physics}$-module)
   associated to $\breve{V}$ as follows.
 \begin{itemize}
  \item[(1)]
  Firstly, we acquire the compatibility with the chiral structure on ${\cal O}_X^{\physics}$ by
   setting
     $$
	    \widehat{\nabla}^{\breve{V}}_{e_{\beta^{\prime\prime}}}\;
		 :=\;  e_{\beta^{\prime\prime}}\,.
	 $$
 
  \item[(2)]
  Secondly, we set\footnote{The
                                             	 choice of using whether
												   $ e^{-\breve{V}}\circ e_{\alpha^\prime}\circ e^{\breve{V}}$ or
												   $e^{\breve{V}}\circ e_{\alpha^\prime}\circ e^{-\breve{V}}$
												   as the definition of $\widehat{\nabla}^{\breve{V}}_{e_{\alpha^\prime}}$
												   is dictated by how one would construct the action functional
												   for the gauge-invariant kinetic term for the chiral superfield
												   in the supersymmetric $U(1)$-gauge theory with matter.
												 The former is consistent with the setting in Sec.\,3.5
												   while the latter isn't.
												 Cf.\ Lemma~3.2.6 vs.\;Sec.\,3.5, theme '{\it Explicit computations/formulae}'.
											                                                       } 
     $$
	   \widehat{\nabla}^{\breve{V}}_{e_{\alpha^\prime}}\;
	    :=\;  e^{-\breve{V}}\circ e_{\alpha^\prime}\circ e^{\breve{V}}\;
		  =\; e_{\alpha^\prime}\,+\, e^{-\breve{V}}(e_{\alpha^\prime}e^{\breve{V}})\,.
     $$
  Thus, in a way $V$ is an indication of the twisting of the original antichiral structure of ${\cal O}_X^{\physics}$
    to the one selected by $\widehat{\nabla}_{e_{\alpha^\prime}}$.
   
  \item[(3)]
  Finally, we set
   $$
      \widehat{\nabla}^{\breve{V}}_{e_{\mu}} \; =\;
	   \mbox{\Large $\frac{\sqrt{-1}}{2}$}
	    \sum_{\alpha,\dot{\beta}} \breve{\sigma}_\mu^{\alpha\dot{\beta}}
		   \{\widehat{\nabla}^{\breve{V}}_{e_{\alpha^\prime}},
		        \widehat{\nabla}^{\breve{V}}_{e_{\beta^{\prime\prime}}}\}\,,
   $$
   where
	  $\breve{\sigma}_{\mu}
	     =(\breve{\sigma}_{\mu}^{\alpha\dot{\beta}})_{\alpha\dot{\beta}}$
	  with
	 {\footnotesize
     $$
      \breve{\sigma}_0\;:=\;
       \frac{1}{2}\left[\!\begin{array}{rr} -1 & 0 \\ 0 & -1\end{array}\!\right]\!,\;\;
	  \breve{\sigma}_1\;:=\;
       \frac{1}{2}\left[\!\begin{array}{rr} 0 & 1 \\ 1 & 0\end{array}\!\right]\!,\;\;
      \breve{\sigma}_2\;:=\;
       \frac{1}{2}
	    \left[\!\begin{array}{rr} 0 & \sqrt{-1} \\ -\sqrt{-1} & 0\end{array}\!\right]\!,\;\;	
      \breve{\sigma}_3\;:=\;
       \frac{1}{2}\left[\!\begin{array}{rr} 1 & 0 \\ 0 & -1\end{array}\!\right]\,.
     $$}This 
	 is indeed a flatness condition on the curvature of $\widehat{\nabla}^{\breve{V}}$
	 in the fermionic directions $(e_{\alpha^\prime}, e_{\beta^{\prime\prime}})$.
   (Cf.\;Lemma~3.1.9 for the precise statement.)
 \end{itemize}
 Since $\breve{V}$ is even, $\widehat{\nabla}^{\breve{V}}$ as defined is even as well.
 In this way
   a pre-vector superfield $\breve{V}\in C^\infty(X^{\physics})^{\flat,\stc}$ determines an even left connection
   $\widehat{\nabla}^{\breve{V}}$ on $\widehat{\cal O}_X^{\,\widehat{\boxplus}}$.
 $\widehat{\nabla}^{\breve{V}}$ is called
  the {\it connection} on $\widehat{\cal O}_X^{\,\widehat{\boxplus}}$ {\it associated to $\breve{V}$};   
 cf.\;Definition~3.1.6.
 For simplicity of notations, we often denote $\widehat{\nabla}^{\breve{V}}$ by $\widehat{\nabla}$,
  keeping $\breve{V}$ implicit.
}\end{definition}

\medskip

\begin{lemma} {\bf [flatness of $\nabla^{\breve{V}}$ along fermionic directions]}\;
 Let $\widehat{\nabla}=\widehat{\nabla}^{\breve{V}}$ be the connection
   on $\widehat{\cal O}_X^{\,\widehat{\boxplus}}$
    associated to a pre-vector superfield $\breve{V}$.
 Let $F^{\widehat{\nabla}}$ be the curvature $2$-tensor of $\widehat{\nabla}$ and
 denote
  $F^{\widehat{\nabla}}(e_{\alpha^\prime}, e_{\beta^\prime})$
  (resp.\
      $F^{\widehat{\nabla}}(e_{\alpha^{\prime\prime}}, e_{\beta^{\prime\prime}})$,
	  $F^{\widehat{\nabla}}(e_{\alpha^\prime}, e_{\beta^{\prime\prime}})$)
 by
  $F^{\widehat{\nabla}}_{\alpha^\prime\beta^\prime}$
   (resp.\  $F^{\widehat{\nabla}}_{\alpha^{\prime\prime}\beta^{\prime\prime}}$,
                  $F^{\widehat{\nabla}}_{\alpha^\prime\beta^{\prime\prime}}$).		
 Then
   with respect to the supersymmetrically invariant coframe $(e^I)_I$ on $\widehat{X}$,
  the components of the curvature tensor $F^{\widehat{\nabla}}$ of $\widehat{\nabla}$
  in purely fermionic directions all vanish:
  for $\alpha^\prime, \beta^\prime=1^\prime, 2^\prime$ and
       $\alpha^{\prime\prime}, \beta^{\prime\prime}=1^{\prime\prime}, 2^{\prime\prime}$,
 $$
   F^{\widehat{\nabla}}_{\alpha^{\prime}\beta^{\prime}}\;
   =\; F^{\widehat{\nabla}}_{\alpha^{\prime\prime}\beta^{\prime\prime}}\;
   =\; F^{\widehat{\nabla}}_{\alpha^{\prime}\beta^{\prime\prime}}\;=\; 0\,.
 $$
\end{lemma}

\medskip

\begin{proof}
 $F^{\widehat{\nabla}}_{\alpha^{\prime\prime}\beta^{\prime\prime}}
   = \{e_{\alpha^{\prime\prime}}, e_{\beta^{\prime\prime}}\}=0$.
 $F^{\widehat{\nabla}}_{\alpha^\prime \beta^\prime}
   =    e^{-\breve{V}}\circ
           \{e_{\alpha^{\prime\prime}}, e_{\beta^{\prime\prime}}\}
		   \circ e^{\breve{V} }=0$.
 And\\
   $F^{\widehat{\nabla}}_{\alpha^{\prime}\beta^{\prime\prime}}
      = \{\widehat{\nabla}_{\!e_{\alpha}^{\prime}},
		          \widehat{\nabla}_{\!e_{\beta^{\prime\prime}}}\}
	        - \widehat{\nabla}_{ \{e_{\alpha^{\prime}}, e_{\beta^{\prime\prime}}\}}
	  =0$
   by tautology
  since
   $\widehat{\nabla}_{ \{e_{\alpha^{\prime}}, e_{\beta^{\prime\prime}}\}}
      = -2\sqrt{-1}\,\sum_{\mu}\sigma^{\mu}_{\alpha\dot{\beta}}\,
	          \widehat{\nabla}_{\!e_{\mu}}$
    and the design of $\widehat{\nabla}_{\!e_{\mu}}$
	 as a ${\Bbb C}$-combination of
	 $\widehat{\nabla}_{ \{e_{\alpha^{\prime}}, e_{\beta^{\prime\prime}}\}}$'s	
	 comes exactly from solving $F^{\widehat{\nabla}}_{\alpha^{\prime}\beta^{\prime\prime}}=0$.
 
\end{proof}
 
\bigskip

With the super version of the standard geometry for a $U(1)$ gauge theory $X^{\physics}$ provided,
 we can now redo [W-B: Chap.\,VI] of Wess \& Bagger
 to construct a supersymmetric $U(1)$ gauge theory on $X$.
Readers are referred also to, e.g., [Argu: Sec.\,4.3] of {\sl Riccardo Argurio}
 for  a very detailed physicists' treatment of the topics in the next two subsections.
By comparison, one can see that
  the notion of the physics sector $X^{\physics}$ of $\widehat{X}^{\widehat{\boxplus}}$
	   from (complexified ${\Bbb Z}/2$-graded) $C^\infty$-Algebraic Geometry   and
    the purge-evaluation map ${\cal P}: C^\infty(X^{\physics})\rightarrow C^\infty(\widehat{X})$	   	   
  together
 really fits particle physicists' language of and ways of playing with supersymmetries.

\bigskip

\subsection{Pre-vector superfields in Wess-Zumino gauge}

Most of the discussions in [W-B: Chap.\,VI] of Wess \& Bagger for vector superfields
 hold for pre-vector superfields as well.

\bigskip

\begin{flushleft}
{\bf Gauge transformations of a pre-vector superfield}
\end{flushleft}
Recall from Definition~3.1.2
 that
 under a gauge transformation specified by a chiral superfield
  $\breve{\Lambda}\in {\cal O}_X^{\physics, \scriptsizech}$,
  a pre-vector superfield $\breve{V}$ transforms as
  $$
    \breve{V}\;\longrightarrow\;
	   \breve{V}+\delta_{\breve{\Lambda}}\breve{V}\,
	     :=\,\breve{V}- \sqrt{-1}(\Lambda-\Lambda^\dag)\,.
  $$
Explicitly, let\footnote{From
                                            now on, we keep the $x$-dependence of
						                    $f^{\tinybullet}_{\tinybullet}\in C^\infty(X)^{\Bbb C}$ implicit
						                     to declutter the notations.}
 %
 {\small
 \begin{eqnarray*}
  \breve{V} & =
    &  V^{(0)}_{(0)}
         + \sum_\alpha \theta^\alpha\vartheta_\alpha  V^{(\alpha)}_{(\alpha)}
		 -  \sum_{\dot{\beta}} \bar{\theta}^{\dot{\beta}}\bar{\vartheta}_{\dot{\beta}}
		       \overline{V^{(\beta)}_{(\beta)}}     \\
    && 			
	   +\, \theta^1\theta^2\vartheta_1\vartheta_2 V^{(12)}_{(12)}
	   + \sum_{\alpha,\dot{\beta}} \theta^\alpha \bar{\theta}^{\dot{\beta}}
               \sum_\mu \sigma^\mu_{\alpha\dot{\beta}} V^{(0)}_{[\mu]}
       + \bar{\theta}^{\dot{1}}\bar{\theta}^{\dot{2}}
              \bar{\vartheta}_{\dot{1}}\bar{\vartheta}_{\dot{2}}		
                \overline{V^{(12)}_{(12)}}			\\
    &&
	   +\, \sum_{\dot{\beta}} \theta^1\theta^2 \bar{\theta}^{\dot{\beta}}
              \sum_\alpha \vartheta_\alpha V^{(\alpha)}_{(12\dot{\beta})}	
	   +   \sum_\alpha \theta^\alpha \bar{\theta}^{\dot{1}}\bar{\theta}^{\dot{2}}
       	      \sum_{\dot{\beta}}\bar{\vartheta}_{\dot{\beta}}
                   \overline{V^{(\beta)}_{(12\dot{\alpha})}}	
       +  \theta^1\theta^2\bar{\theta}^{\dot{1}}\bar{\theta}^{\dot{2}}
	          V^{(0)}_{(12\dot{1}\dot{2})}      \\
    & \in & C^\infty(X^{\physics})^{\flat,\stc}
 \end{eqnarray*}}and\footnote{In Sec.\,2.1
                                                   we express the coefficient of the $\theta^1\theta^2\bar{\theta}^{\dot{\beta}}$-term
												   (resp.\
												    the $\theta^\alpha\bar{\theta}^{\dot{1}}\bar{\theta}^{\dot{2}}$-term)
													of a chiral superfield (resp.\ antichiral superfield) in an expanded form.
                                                   Here, it is more convenient to express them in the summation form.
                                           } 
{\small
   \begin{eqnarray*}
     \breve{\Lambda} & = &
	   \Lambda_{(0)}^{(0)}
	   + \sum_\alpha \theta^\alpha\vartheta_\alpha  \Lambda_{(\alpha)}^{(\alpha)}
	   + \theta^1\theta^2\vartheta_1\vartheta_2 \Lambda_{(12)}^{(12)}
       + \sqrt{-1} \sum_{\alpha,\dot{\beta}}
	          \theta^\alpha\bar{\theta}^{\dot{\beta}}	
			   \sum_\mu
			     \sigma^\mu_{\alpha\dot{\beta}}\partial_\mu \Lambda_{(0)}^{(0)}\\
      && \hspace{1em}				
       -\, \sqrt{-1}\sum_{\dot{\beta}, \mu}				
	        \theta^1\theta^2\bar{\theta}^{\dot{\beta}}
			 \sum_{\alpha,\gamma,\mu} \varepsilon^{\alpha\gamma} \sigma^\mu_{\alpha\dot{\beta}}
			       \partial_\mu (\vartheta_\gamma \Lambda^{(\gamma)}_{(\gamma)})
       + \theta^1\theta^2\bar{\theta}^{\dot{1}}\bar{\theta}^{\dot{2}}\,
	      -\,\square \Lambda_{(0)}^{(0)}     \\
    & \in & C^\infty(X^{\physics})^{\scriptsizech}\,,			
   \end{eqnarray*}}where 
  $\square := -\partial_0^2+\partial_1^2+\partial_2^2+\partial_3^2$.
The twisted complex conjugate $\breve{\Lambda}^\dag$ of $\breve{\Lambda}$ is given by
 {\small
   \begin{eqnarray*}
     \breve{\Lambda}^\dag & = &
	   \overline{\Lambda_{(0)}^{(0)}}
	   - \sum_{\dot{\beta}} \bar{\theta}^{\dot{\beta}}\bar{\vartheta}_{\dot{\beta}}
	           \overline{\Lambda_{(\beta)}^{(\beta)}}
	   + \bar{\theta}^{\dot{1}}\bar{\theta}^{\dot{2}}
	      \bar{\vartheta}_{\dot{1}}\bar{\vartheta}_{\dot{2}}
		               \overline{\Lambda_{(12)}^{(12)}}
       - \sqrt{-1} \sum_{\alpha,\dot{\beta}}
	          \theta^\alpha\bar{\theta}^{\dot{\beta}}	
			   \sum_\mu
			     \sigma^\mu_{\alpha\dot{\beta}}
				       \partial_\mu \overline{\Lambda_{(0)}^{(0)}}    \\
      && \hspace{1em}				
       +\, \sqrt{-1}\sum_{\alpha, \mu}				
	        \theta^\alpha\bar{\theta}^{\dot{1}} \bar{\theta}^{\dot{2}}
		     \sum_{\dot{\beta}, \dot{\delta}, \mu} 	
			   \varepsilon^{\dot{\beta}\dot{\delta}} \sigma^\mu_{\alpha\dot{\beta}}
			     \partial_\mu
				   (\bar{\vartheta}_{\dot{\delta}} \overline{\Lambda^{(\delta)}_{(\delta)}})
       -\, \theta^1\theta^2\bar{\theta}^{\dot{1}}\bar{\theta}^{\dot{2}}\,
	        \square \overline{\Lambda_{(0)}^{(0)}}       \\
	 & \in & C^\infty(X^{\physics})^{\scriptsizeach}\,.
   \end{eqnarray*}}  
Then,
 {\small
 \begin{eqnarray*}
  \lefteqn{\breve{V}+\delta_{\breve{\Lambda}}\breve{V}
    := \breve{V} -\sqrt{-1} (\breve{\Lambda}-\breve{\Lambda}^\dag)   }\\
 && =\;
     \mbox{\Large $($}
      V^{(0)}_{(0)} +\delta_{\breve{\Lambda}}V^{(0)}_{(0)})
	 \mbox{\Large $)$}
       + \sum_\alpha \theta^\alpha\vartheta_\alpha
	        \mbox{\Large $($}
	           V^{(\alpha)}_{(\alpha)}
			     +\delta_{\breve{\Lambda}}V^{(\alpha)}_{(\alpha)}
			\mbox{\Large $)$}
	   - \sum_{\dot{\beta}} \bar{\theta}^{\dot{\beta}}\bar{\vartheta}_{\dot{\beta}}
	        \mbox{\Large $($}
		       \overline{V^{(\beta)}_{(\beta)}}
			     +\delta_{\breve{\Lambda}}\overline{V^{(\beta)}_{(\beta)}}
			\mbox{\Large $)$}                \\
    && \hspace{2em}			
	   +\, \theta^1\theta^2\vartheta_1\vartheta_2
	          \mbox{\Large $($}
	           V^{(12)}_{(12)} +\delta_{\breve{\Lambda}}V^{(12)}_{(12)}
			  \mbox{\Large $)$}
	   + \sum_{\alpha,\dot{\beta}} \theta^\alpha \bar{\theta}^{\dot{\beta}}
               \sum_\mu \sigma^\mu_{\alpha\dot{\beta}}
			     \mbox{\Large $($}
				  V^{(0)}_{[\mu]} +\delta_{\breve{\Lambda}}V^{(0)}_{([\mu])} 				 					   
				 \mbox{\Large $)$}       \\
    && \hspace{2em} 				
       +\, \bar{\theta}^{\dot{1}}\bar{\theta}^{\dot{2}}
              \bar{\vartheta}_{\dot{1}}\bar{\vartheta}_{\dot{2}}		
			   \mbox{\Large $($}
                \overline{V^{(12)}_{(12)}}
				+\delta_{\breve{\Lambda}}\overline{V^{(12)}_{(12)}}
			   \mbox{\Large $)$}	     \\[1ex]
    && \hspace{2em} 			
	   +\, \sum_{\dot{\beta}} \theta^1\theta^2 \bar{\theta}^{\dot{\beta}}
	          \sum_\gamma \vartheta_\gamma
	            \mbox{\Large $($}				
				      V^{(\gamma)}_{(12\dot{\beta})}
                       +\delta_{\breve{\Lambda}}V^{(\gamma)}_{(12\dot{\beta})}
  			    \mbox{\Large $)$}			
	   +\, \sum_\alpha \theta^\alpha \bar{\theta}^{\dot{1}}\bar{\theta}^{\dot{2}}	
	         \sum_{\dot{\delta}}
       	      \bar{\vartheta}_{\dot{\delta}}
	           \mbox{\Large $($}	
                     \overline{V^{(\delta)}_{(12\dot{\alpha})}}	
					   +\delta_{\breve{\Lambda}}\overline{V^{(\delta)}_{(12\dot{\alpha})}}					 
			   \mbox{\Large $)$}             \\				
    && \hspace{2em}				
       +\, \theta^1\theta^2\bar{\theta}^{\dot{1}}\bar{\theta}^{\dot{2}}
	         \mbox{\Large $($}
	          V^{(0)}_{(12\dot{1}\dot{2})}
			   +\delta_{\breve{\Lambda}}V^{(0)}_{(12\dot{1}\dot{2})}
             \mbox{\Large $)$}			  \\[1ex]			
   &&  =\;
     \mbox{\Large $($}
      V^{(0)}_{(0)}
	   - \sqrt{-1}(\Lambda^{(0)}_{(0)}-\overline{\Lambda^{(0)}_{(0)}})
	 \mbox{\Large $)$}
       + \sum_\alpha \theta^\alpha\vartheta_\alpha
	        \mbox{\Large $($}
	           V^{(\alpha)}_{(\alpha)}
			   - \sqrt{-1} \Lambda^{(\alpha)}_{(\alpha)}
			\mbox{\Large $)$}
	   - \sum_{\dot{\beta}} \bar{\theta}^{\dot{\beta}}\bar{\vartheta}_{\dot{\beta}}
	        \mbox{\Large $($}
		       \overline{V^{(\beta)}_{(\beta)}}
			  + \sqrt{-1}\,\overline{\Lambda^{(\beta)}_{(\beta)}}
			\mbox{\Large $)$}                \\
    && \hspace{2em}			
	   +\, \theta^1\theta^2\vartheta_1\vartheta_2
	          \mbox{\Large $($}
	           V^{(12)}_{(12)} - \sqrt{-1} \Lambda^{(12)}_{(12)}
			  \mbox{\Large $)$}
	   + \sum_{\alpha,\dot{\beta}} \theta^\alpha \bar{\theta}^{\dot{\beta}}
               \sum_\mu \sigma^\mu_{\alpha\dot{\beta}}
			     \mbox{\Large $($}
				  V^{(0)}_{[\mu]}
				  + \partial_\mu
				       \mbox{\large $($}
					     \Lambda^{(0)}_{(0)} + \overline{\Lambda^{(0)}_{(0)}}
					   \mbox{\large $)$}
				 \mbox{\Large $)$}       \\
    && \hspace{2em} 				
       +\, \bar{\theta}^{\dot{1}}\bar{\theta}^{\dot{2}}
              \bar{\vartheta}_{\dot{1}}\bar{\vartheta}_{\dot{2}}		
			   \mbox{\Large $($}
                \overline{V^{(12)}_{(12)}} + \sqrt{-1}\,\overline{\Lambda^{(12)}_{(12)}}
			   \mbox{\Large $)$}	     \\[1ex]
    && \hspace{2em} 			
	   +\, \sum_{\dot{\beta}} \theta^1\theta^2 \bar{\theta}^{\dot{\beta}}
	          \mbox{\Large $($}	
			   \sum_\gamma\vartheta_\gamma V^{(\gamma)}_{(12\dot{\beta})}
               - \sum_{\alpha, \gamma, \mu}
			        \varepsilon^{\alpha\gamma}\sigma^\mu_{\alpha\dot{\beta}}
					  \partial_\mu (\vartheta_\gamma \Lambda^{(\gamma)}_{(\gamma)})		
			  \mbox{\Large $)$}			   \\
    && \hspace{2em}
	   +\, \sum_\alpha \theta^\alpha \bar{\theta}^{\dot{1}}\bar{\theta}^{\dot{2}}
	          \mbox{\Large $($}	
			   \sum_{\dot{\delta}} \bar{\vartheta}_{\dot{\delta}}
                  \overline{V^{(\delta)}_{(12\dot{\alpha})}}	
			  - \sum_{\dot{\beta}, \dot{\delta}, \mu}
                     \varepsilon^{\dot{\beta}\dot{\delta}}	 \sigma^\mu_{\alpha\dot{\beta}}
					             \partial_\mu
								  (\bar{\vartheta}_{\dot{\delta}}\overline{\Lambda^{(\delta)}_{(\delta)}})				
			  \mbox{\Large $)$}             \\				
    && \hspace{2em}				
       +\, \theta^1\theta^2\bar{\theta}^{\dot{1}}\bar{\theta}^{\dot{2}}
	         \mbox{\Large $($}
	          V^{(0)}_{(12\dot{1}\dot{2})}
			  +\sqrt{-1}\,\square(\Lambda^{(0)}_{(0)}- \overline{\Lambda^{(0)}_{(0)}})
             \mbox{\Large $)$}			  \\
    && \in\;  C^\infty(X^{\physics})^{\flat,\stc}\,.
 \end{eqnarray*}
  } 

A comparison of
    $\delta_{\breve{\Lambda}}V^{(0)}_{12\dot{1}\dot{2}}$
      against
    $\delta_{\breve{\Lambda}}V^{(0)}_{(0)}$
 and
    $\delta_{\breve{\Lambda}}V^{(\alpha)}_{(12\dot{\beta})}$
	  against
	$\delta_{\breve{\Lambda}}V^{(\alpha)}_{(\alpha)}$
 implies that
if one expresses a pre-vector superfield $\breve{V}\in C^\infty(X^{\physics})^{\flat,\stc}$
 in the following shifted form
 (cf.\ [W-B: Chap.\,VI, Eq.\,(6.2)], which can always be made)
{\small
 \begin{eqnarray*}
  \breve{V} & =
    &  V^{(0)}_{(0)}
         + \sum_\alpha \theta^\alpha\vartheta_\alpha  V^{(\alpha)}_{(\alpha)}
		 -  \sum_{\dot{\beta}} \bar{\theta}^{\dot{\beta}}\bar{\vartheta}_{\dot{\beta}}
		       \overline{V^{(\beta)}_{(\beta)}}     \\
    && 			
	   +\, \theta^1\theta^2\vartheta_1\vartheta_2 V^{(12)}_{(12)}
	   + \sum_{\alpha,\dot{\beta}} \theta^\alpha \bar{\theta}^{\dot{\beta}}
               \sum_\mu \sigma^\mu_{\alpha\dot{\beta}} V^{(0)}_{[\mu]}
       + \bar{\theta}^{\dot{1}}\bar{\theta}^{\dot{2}}
              \bar{\vartheta}_{\dot{1}}\bar{\vartheta}_{\dot{2}}		
                \overline{V^{(12)}_{(12)}}			\\
    &&
	   +\, \sum_{\dot{\beta}} \theta^1\theta^2 \bar{\theta}^{\dot{\beta}}
	           \mbox{\Large $($}
			     \sum_\gamma \vartheta_\gamma  V^{(\gamma)}_{(12\dot{\beta})}						
			      -\,\sqrt{-1}\sum_{\alpha, \gamma, \mu}
				        \varepsilon^{\alpha\gamma}\sigma^\mu_{\alpha\dot{\beta}}
						  \partial_\mu (\vartheta_\gamma V^{(\gamma)}_{(\gamma)})		
			   \mbox{\Large $)$}\\
    &&			
	   +   \sum_\alpha \theta^\alpha \bar{\theta}^{\dot{1}}\bar{\theta}^{\dot{2}}
	          \mbox{\Large $($}
			    \sum_{\dot{\delta}}  \bar{\vartheta}_{\dot{\delta}}
                   \overline{V^{(\delta)}_{(12\dot{\alpha})}}					
				    + \sqrt{-1}\sum_{\dot{\beta}, \dot{\delta}, \mu}
					     \varepsilon^{\dot{\beta}\dot{\delta}} \sigma^\mu_{\alpha\dot{\beta}}						  
					      \partial_\mu
						    (\bar{\vartheta}_{\dot{\delta}} \overline{V^{(\delta)}_{(\delta)}})
			  \mbox{\Large $)$}      \\
    &&				
       +\;  \theta^1\theta^2\bar{\theta}^{\dot{1}}\bar{\theta}^{\dot{2}}\,
	           \mbox{\large $($}
			     V^{(0)}_{(12\dot{1}\dot{2})} -\,  \square V^{(0)}_{(0)}
			   \mbox{\large $)$}\,,
 \end{eqnarray*}}then   
 the component-functions
   $V^{(\alpha)}_{(12\dot{\beta})}$, $\overline{V^{(\alpha)}_{(12\dot{\beta})}}$,
   $V^{(0)}_{(12\dot{1}\dot{2})}$, $\alpha=1,2$, $\dot{\beta}=\dot{1},\dot{2}$,
  are invariant under gauge transformations.

\bigskip

\begin{definition} {\bf [pre-vector superfield in the shifted expression]}\; {\rm
 We call a pre-vector superfield expressed in the above shifted form
 a {\it pre-vector superfield in the shifted expression.}
}\end{definition}

\medskip

\begin{lemma}{\bf [$\Bbb R$-linearity]}\;
 An ${\Bbb R}$-linear combination\footnote{However,
                                                               caution that a $C^\infty(X)$-linear combination of pre-vector superfields
															    in general is not directly a pre-vector superfield in the shifted expression.
															   One has to convert it accordingly.
															   }  
  of pre-vector superfields in the shifted expression is also a pre-vector superfield in the shifted expression.
\end{lemma}

\bigskip
  
It follows that, for a given pre-vector superfield $\breve{V}$ in the shifted expression,
 if one chooses $\breve{\Lambda}$ with
  $$
    \Imaginary \Lambda^{(0)}_{(0)}\;
	  =\; -\, \mbox{\large $\frac{1}{2}$}\,V^{(0)}_{(0)}\,,\hspace{2em}
	\Lambda^{(\alpha)}_{(\alpha)}\;
	  =\; -\,\sqrt{-1}\,V^{(\alpha)}_{(\alpha)}\,,\hspace{2em}
    \Lambda^{(12)}_{(12)}\;
	  =\; -\,\sqrt{-1}\,V^{(12)}_{(12)}\,,
  $$
 which always exists,
 then after the gauge transformation specified by $\breve{\Lambda}$,
 $\breve{V}$ becomes
 {\small
 \begin{eqnarray*}
   \breve{V}^\prime
    &= &  			
	  \sum_{\alpha,\dot{\beta}} \theta^\alpha \bar{\theta}^{\dot{\beta}}
               \sum_\mu \sigma^\mu_{\alpha\dot{\beta}}
			     \mbox{\Large $($}
				  V^{(0)}_{[\mu]}
				  + 2\,\partial_\mu \Real \Lambda^{(0)}_{(0)}				 					
				 \mbox{\Large $)$}       \\			
    && \hspace{2em} 			
	   +\, \sum_{\dot{\beta}} \theta^1\theta^2 \bar{\theta}^{\dot{\beta}}			
			    \sum_\alpha \vartheta_\alpha
				      V^{(\alpha)}_{(12\dot{\beta})}					  					
	   +\, \sum_\alpha \theta^\alpha \bar{\theta}^{\dot{1}}\bar{\theta}^{\dot{2}}	   			
       	         \sum_{\dot{\beta}}\bar{\vartheta}_{\dot{\beta}}
                     \overline{V^{(\beta)}_{(12\dot{\alpha})}}						
       +\, \theta^1\theta^2\bar{\theta}^{\dot{1}}\bar{\theta}^{\dot{2}}\, 	
	          V^{(0)}_{(12\dot{1}\dot{2})}\,.		
 \end{eqnarray*}
  } 
 
We summarize the above discussion into the following definition and lemmas:

\bigskip

\begin{definition-lemma} {\bf [pre-vector superfield in Wess-Zumino gauge]}\; {\rm
 We call a pre-vector superfield $\breve{V}\in C^\infty(X^{\physics})^{\flat,\stc}$
  that is in the following form
 {\small
  $$
  \breve{V}\; =\;
     \sum_{\alpha,\dot{\beta}} \theta^\alpha \bar{\theta}^{\dot{\beta}}
               \sum_\mu \sigma^\mu_{\alpha\dot{\beta}} V^{(0)}_{[\mu]}
	   +\, \sum_{\dot{\beta}} \theta^1\theta^2 \bar{\theta}^{\dot{\beta}}
              \sum_\alpha \vartheta_\alpha V^{(\alpha)}_{(12\dot{\beta})}	
	   +   \sum_\alpha \theta^\alpha \bar{\theta}^{\dot{1}}\bar{\theta}^{\dot{2}}
       	      \sum_{\dot{\beta}}\bar{\vartheta}_{\dot{\beta}}
                   \overline{V^{(\beta)}_{(12\dot{\alpha})}}	
       +  \theta^1\theta^2\bar{\theta}^{\dot{1}}\bar{\theta}^{\dot{2}}
	          V^{(0)}_{(12\dot{1}\dot{2})}
  $$}a 
  {\it pre-vector superfield in Wess-Zumino gauge}.
 
 {\it Given any  pre-vector $\breve{V}$, there exists a unique chiral superfield $\breve{\Lambda}$
          depending on $\breve{V}$ with $\Real \Lambda^{(0)}_{(0)}=0$
	 such that
	   the gauge transformation specified by $\breve{\Lambda}$ takes $\breve{V}$
	   to a pre-vector superfield in Wess-Zumino gauge.}
}\end{definition-lemma}

\medskip

\begin{lemma} {\bf [naturality]}\;
 (1)
 The set of pre-vector superfields in Wess-Zumino gauge is a $C^\infty(X)$-submodule of
     $C^\infty(X^{\physics})^{\flat,\stc}$.
 (2)
 If a pre-vector superfield $\breve{V}$ expressed in terms of the standard coordinate functions
     $(x,\theta,\bar{\theta},\vartheta,\bar{\theta})$
	 on $\widehat{X}^{\widehat{\boxplus}}$ is in Wess-Zumino gauge,
 then it remains in Wess-Zumino gauge when re-expressed in terms of
    the chiral coordinate functions $(x^\prime,\theta,\bar{\theta},\vartheta,\bar{\theta})$ or
	the antichiral coordinate functions $(x^{\prime\prime},\theta,\bar{\theta},\vartheta,\bar{\theta})$ 
     on $\widehat{X}^{\widehat{\boxplus}}$.
\end{lemma}

\medskip

\begin{proof}
 Statement (1) is clear. We focus on Statement (2).

 Recall that $x^\prime = x+\sqrt{-1}\theta\boldsigma\bar{\theta}^t$.
 When in Wess-Zumino gauge,
  a pre-vector superfield $\breve{V}$ in terms of the standard coordinate functions
    $(x,\theta,\bar{\theta},\vartheta, \bar{\theta})$
  is written as
 {\small
  \begin{eqnarray*}
   \lefteqn{
    \breve{V}\; =\;
     \sum_{\alpha,\dot{\beta}} \theta^\alpha \bar{\theta}^{\dot{\beta}}
               \sum_\mu \sigma^\mu_{\alpha\dot{\beta}} V^{(0)}_{[\mu]}(x)
	   +\, \sum_{\dot{\beta}} \theta^1\theta^2 \bar{\theta}^{\dot{\beta}}
              \sum_\alpha \vartheta_\alpha V^{(\alpha)}_{(12\dot{\beta})}(x)  }\\
     && \hspace{6em}			
	   +   \sum_\alpha \theta^\alpha \bar{\theta}^{\dot{1}}\bar{\theta}^{\dot{2}}
       	      \sum_{\dot{\beta}}\bar{\vartheta}_{\dot{\beta}}
                   \overline{V^{(\beta)}_{(12\dot{\alpha})}}(x)	
       +  \theta^1\theta^2\bar{\theta}^{\dot{1}}\bar{\theta}^{\dot{2}}
	          V^{(0)}_{(12\dot{1}\dot{2})}(x)
  \end{eqnarray*}}To 
 re-express $\breve{V}$ in terms of
    the chiral coordinate functions $(x^\prime,\theta,\bar{\theta},\vartheta,\bar{\theta})$
	on $\widehat{X}^{\widehat{\boxplus}}$,
 one substitutes $x$ in $V^{\tinybullet}_{\tinybullet}(x)$
   by $x^\prime-\sqrt{-1}\theta\boldsigma\bar{\theta}^t$  and use the $C^\infty$-hull structure of
   $C^\infty(\widehat{X})$ to expand it in $x^\prime$.
 Due to the product structure of $\theta^\alpha$, $\bar{\theta}^{\dot{\beta}}$,
  this will only influence the coefficient of $\theta^1\theta^2\bar{\theta}^{\dot{1}}\bar{\theta}^{\dot{2}}$
   and, hence, keep the pre-vector superfield in Wess-Zumino gauge.
 Explicitly, the result after collecting like terms is
 {\small
  \begin{eqnarray*}
   \lefteqn{
    \breve{V}\; =\;
     \sum_{\alpha,\dot{\beta}} \theta^\alpha \bar{\theta}^{\dot{\beta}}
               \sum_\mu \sigma^\mu_{\alpha\dot{\beta}} V^{(0)}_{[\mu]}(x^\prime)
	   +\, \sum_{\dot{\beta}} \theta^1\theta^2 \bar{\theta}^{\dot{\beta}}
              \sum_\alpha \vartheta_\alpha V^{(\alpha)}_{(12\dot{\beta})}(x^\prime)  }\\
     && \hspace{6em}			
	   +   \sum_\alpha \theta^\alpha \bar{\theta}^{\dot{1}}\bar{\theta}^{\dot{2}}
       	      \sum_{\dot{\beta}}\bar{\vartheta}_{\dot{\beta}}
                   \overline{V^{(\beta)}_{(12\dot{\alpha})}}(x^\prime)	
       +  \theta^1\theta^2\bar{\theta}^{\dot{1}}\bar{\theta}^{\dot{2}}
	         \mbox{\Large $($}
	          V^{(0)}_{(12\dot{1}\dot{2})}(x^\prime)
			   + 2\sqrt{-1}\,\partial^\mu V^{(0)}_{[\mu]}(x^\prime)
			 \mbox{\Large $)$}\,.
  \end{eqnarray*}}

 Similar argument goes
 when re-expressing $\breve{V}$ in terms of the antichiral coordinate functions
    $(x^{\prime\prime},\theta,\bar{\theta},\vartheta,\bar{\theta})$,
         with $x^{\prime\prime} = x-\sqrt{-1}\theta\boldsigma\bar{\theta}^t, \theta,\bar{\theta}$,
    on $\widehat{X}^{\widehat{\boxplus}}$.
 The explicit expression is given by
 {\small
  \begin{eqnarray*}
   \lefteqn{
    \breve{V}\; =\;
     \sum_{\alpha,\dot{\beta}} \theta^\alpha \bar{\theta}^{\dot{\beta}}
               \sum_\mu \sigma^\mu_{\alpha\dot{\beta}} V^{(0)}_{[\mu]}(x^{\prime\prime})
	   +\, \sum_{\dot{\beta}} \theta^1\theta^2 \bar{\theta}^{\dot{\beta}}
              \sum_\alpha \vartheta_\alpha V^{(\alpha)}_{(12\dot{\beta})}(x^{\prime\prime})  }\\
     && \hspace{6em}			
	   +   \sum_\alpha \theta^\alpha \bar{\theta}^{\dot{1}}\bar{\theta}^{\dot{2}}
       	      \sum_{\dot{\beta}}\bar{\vartheta}_{\dot{\beta}}
                   \overline{V^{(\beta)}_{(12\dot{\alpha})}}(x^{\prime\prime})	
       +  \theta^1\theta^2\bar{\theta}^{\dot{1}}\bar{\theta}^{\dot{2}}
	         \mbox{\Large $($}
	          V^{(0)}_{(12\dot{1}\dot{2})}(x^{\prime\prime})
			   - 2\sqrt{-1}\,\partial^\mu V^{(0)}_{[\mu]}(x^{\prime\prime})
			 \mbox{\Large $)$}\,.
  \end{eqnarray*}}

 This completes the proof.
	
\end{proof}

\medskip

\begin{lemma} {\bf [representative in Wess-Zuminio gauge]}\;
 Any pre-vector superfield\\ $\breve{V}\in C^\infty(X^{\physics})^{\flat,\stc}$ can be transformed to
  a pre-vector superfield in Wess-Zumino gauge by a gauge transformation.
\end{lemma}
 
\bigskip

In the shifted expression, once a pre-vector superfield is rendered a pre-vector superfield $\breve{f}$
   in Wess-Zumino gauge,
 a gauge transformation specified by $\breve{\Lambda}$ with
  $$
    \Imaginary \Lambda^{(0)}_{(0)}\;
	=\; \Lambda^{(\alpha)}_{(\alpha)}\;
    =\; \Lambda^{(12)}_{(12)}\; =\; 	0\,,
  $$
 (i.e.\
   $$
     \breve{\Lambda}\;=\;	
	   \Lambda_{(0)}^{(0)}
	   + \sqrt{-1} \sum_{\alpha,\dot{\beta}}
	          \theta^\alpha\bar{\theta}^{\dot{\beta}}	
			   \sum_\mu
			     \sigma^\mu_{\alpha\dot{\beta}}\partial_\mu \Lambda_{(0)}^{(0)}				    			
       -\, \theta^1\theta^2\bar{\theta}^{\dot{1}}\bar{\theta}^{\dot{2}}\,
	        \square \Lambda_{(0)}^{(0)}     	
   $$
   with $\Lambda^{(0)}_{(0)}$ real-valued)
 will send $\breve{f}$ to another $\breve{f}^\prime$ still in Wess-Zumino gauge
  with only the components $V^{(0)}_{[\mu]}$ of $\breve{V}$ transformed,
  by
   $$
     V^{(0)}_{[\mu]}\;\longrightarrow\;
	 V^{(0)}_{[\mu]} \,+\, 2\,\partial_\mu \Lambda^{(0)}_{(0)}\,.
   $$
   
\bigskip

\begin{lemma}
{\bf [restriction $\nabla^{\breve{V}}$ of $\widehat{\nabla}^{\breve{V}}$ to $X^{\Bbb C}$]}\;
 Let $\breve{V}$ be a pre-vector superfield in Wess-Zumino gauge.
 Then the restriction $\nabla^{\breve{V}}$ of $\widehat{\nabla}^{\breve{V}}_\mu$ to $X^{\Bbb C}$
  is given by $\partial_\mu -\frac{\sqrt{-1}}{2}\,V^{(0)}_{[\mu]}$, $\mu=0,1,2,3$\,.
\end{lemma}
   
\medskip

\begin{proof}
 Denote $\widehat{\nabla}^{\breve{V}}$ by $\widehat{\nabla}$.
 Since $X^{\Bbb C}\subset \widehat{X}^{\widehat{\boxplus}}$
   is described by the ideal in $C^\infty(\widehat{X}^{\widehat{\boxplus}})$ generated by
   $\theta$, $\bar{\theta}$, $\vartheta$, $\bar{\vartheta}$,
  the restriction of $\widehat{\nabla}$ to $X^{\Bbb C}$ is simply
   the $(\theta,\bar{\theta},\vartheta,\bar{\vartheta})$-degree-zero part of
   $\widehat{\nabla}$ when expressed in terms of the standard coordinate functions
   $(x,\theta,\bar{\theta},\vartheta,\bar{\vartheta})$ on $\widehat{X}^{\widehat{\boxplus}}$.
 This in turn is a ${\Bbb C}$-linear combination of the
  $(\theta,\bar{\theta},\vartheta,\bar{\vartheta})$-degree-zero part of
  $\{\widehat{\nabla}\!_{e_{\alpha^\prime}},
         \widehat{\nabla}\!_{e_{\beta^{\prime\prime}}}\}$
  in the standard coordinate functions $(x,\theta,\bar{\theta}, \vartheta,\bar{\vartheta})$.
 
 From the definition of $\widehat{\nabla}$ and the ${\Bbb Z}/2$-graded Leibniz rule, one has
   \begin{eqnarray*}
    \{\widehat{\nabla}\!_{e_{\alpha^\prime}}, \widehat{\nabla}\!_{e_{\beta^{\prime\prime}}}\}
	& = &  \{ e_{\alpha^\prime}+ e^{-\breve{V}}(e_{\alpha^\prime}e^{\breve{V}})\,,\,
	                       e_{\beta^{\prime\prime}}\}      \\
	& = & \{e_{\alpha^\prime}, e_{\beta^{\prime\prime}}\}\,
	     +\, (e_{\beta^{\prime\prime}}e^{-\breve{V}})
		         (e_{\alpha^\prime}e^{\breve{V}})
		 +\, e^{-\breve{V}} (e_{\beta^{\prime\prime}}e_{\alpha^\prime}e^{\breve{V}})\,.
   \end{eqnarray*}
 The $(\theta,\bar{\theta},\vartheta, \bar{\vartheta})$-degree-zero terms of
   $\{\widehat{\nabla}\!_{e_{\alpha^\prime}},
          \widehat{\nabla}\!_{e_{\beta^{\prime\prime}}}\}$
   thus come from
     $\{e_{\alpha^\prime}, e_{\beta^{\prime\prime}}\}$,
        which equals\\ $-2\sqrt{-1}\sum_\nu\sigma^\nu_{\alpha\dot{\beta}}\partial_\nu$, and
	the $(\theta,\bar{\theta}, \vartheta,\bar{\vartheta})$-degree-zero terms of the summand
	  $e_{\beta^{\prime\prime}}e_{\alpha^\prime}\breve{V}$
	  from the expansion of
	  $e^{-\breve{V}} (e_{\beta^{\prime\prime}}e_{\alpha^\prime}e^{\breve{V}})
	    = (1-\breve{V}+\frac{1}{2}\breve{V}^2)
		        (e_{\beta^{\prime\prime}}e_{\alpha^\prime}
				   (1+\breve{V}+\frac{1}{2}\breve{V}^2))$,
     which is $-\sum_\nu\sigma^\nu_{\alpha\dot{\beta}}V^{(0)}_{[\nu]}$.
 It follows that
  $$
    \nabla_\mu\;\;
	:=\;\; \widehat{\nabla}_\mu|_{X^{\Bbb C}}\;\;
	 =\;\; \mbox{\Large $\frac{\sqrt{-1}}{2}$}
	         \sum_{\alpha,\dot{\beta}}\breve{\sigma}_\mu^{\alpha\dot{\beta}}			
				 \sum_\nu\sigma^\nu_{\alpha\dot{\beta}}
				   \mbox{\large $($}
				     -2\sqrt{-1}\partial_\nu - V^{(0)}_{[\nu]}
				   \mbox{\large $)$} \;\;
	    =\;\; \partial_\mu - \mbox{\Large $\frac{\sqrt{-1}}{2}$} V^{(0)}_{[\mu]}\,.
  $$
 Here, the identity
   $\sum_{\alpha,\dot{\beta}}
       \breve{\sigma}_\mu^{\alpha\dot{\beta}}\sigma^\nu_{\alpha\dot{\beta}}
	 =\delta_\mu^\nu$ is used.
 This proves the lemma.	

\end{proof}

\bigskip

\begin{flushleft}
{\bf Explicit formulae for $\widehat{\nabla}_\mu^{\breve{V}}$}
\end{flushleft}
The full expression of $\widehat{\nabla}^{\breve{V}}_\mu$ for $\breve{V}$ in Wess-Zumino gauge
  is given here for the completeness of the discussion.
It's a curious feature that the curvature of $\nabla:= \widehat{\nabla}^{\breve{V}}|_{X^{\Bbb C}}$
 is somehow already captured in the $(\theta,\bar{\theta})$-degree-$>0$-terms of
 $\widehat{\nabla}^{\breve{V}}_\mu$'s.
$$
 \widehat{\nabla}_{\mu} \;
   =\; \partial_\mu \,
          +\, \mbox{\Large $\frac{\sqrt{-1}}{2}$}\sum_{\alpha,\dot{\beta}}
		          \breve{\sigma}_\mu^{\alpha\dot{\beta}}\Theta_{\alpha\dot{\beta}}\,,
$$
where
{\small
\begin{eqnarray*}
 \Theta_{\alpha\dot{\beta}}
  & = &
     \mbox{\Large $($}e_{\beta^{\prime\prime}} e^{-\breve{V}} \mbox{\Large $)$}
			    \mbox{\Large $($} e_{\alpha^\prime} e^{\breve{V}}  \mbox{\Large $)$}
				+  e^{-\breve{V}}				
				     \mbox{\Large $($}
                       e_{\beta^{\prime\prime}} e_{\alpha^\prime} e^{\breve{V}}
                     \mbox{\Large $)$}     \\
  & = &  -\sum_\nu \sigma^\nu_{\alpha\dot{\beta}} V^{(0)}_{[\nu]}
              -\sum_\gamma \theta^\gamma \varepsilon_{\alpha\gamma}
			      (\mbox{$\sum$}_{\gamma^\prime} \vartheta_{\gamma^\prime}
				        V^{(\gamma^\prime)}_{(12\dot{\beta})})\,
			  + \sum_{\dot{\delta}}\bar{\theta}^{\dot{\delta}}
			        \varepsilon_{\dot{\beta}\dot{\delta}}
					(\mbox{$\sum$}_{\dot{\delta}^\prime}
					     \bar{\vartheta}_{\dot{\delta}^\prime}
						 \overline{V^{(\delta^\prime)}_{12\dot{\alpha}}})\\
  && + \sum_{\gamma,\dot{\delta}}
              \theta^\gamma \bar{\theta}^{\dot{\delta}}
			    \left(\rule{0ex}{1.2em} \right.\!
				  \varepsilon_{\alpha\gamma} \varepsilon_{\dot{\beta}\dot{\delta}}
				   V^{(0)}_{(12\dot{1}\dot{2})}
				 + \sqrt{-1}\sum_{\dot{\beta}^\prime, \dot{\delta}^\prime, \mu, \nu}
				      \varepsilon_{\dot{\beta}\dot{\delta}}
					  \varepsilon^{\dot{\beta}^\prime\dot{\delta}^\prime}
                      \sigma^\mu_{\alpha\dot{\beta}^\prime}
                      \sigma^\nu_{\gamma\dot{\delta}^\prime}\,
					  \partial_\mu V^{(0)}_{[\nu]}					
				 - \sqrt{-1}\sum_{\mu,\nu}
				       \sigma^\mu_{\gamma\dot{\beta}}\sigma^\nu_{\alpha\dot{\delta}}\,
                       \partial_\mu V^{(0)}_{[\nu]} \\
      && \hspace{6em}	 	
				 + \sum_ \mu
					  \varepsilon_{\alpha\gamma} \varepsilon_{\dot{\beta}\dot{\delta}}
				      V^{(0)}_{(\mu)}V^{[\mu];(0)}
				 - \sum_{\mu,\nu}		
				      \sigma^\mu_{\gamma\dot{\beta}}\sigma^\nu_{\alpha\dot{\delta}}\,
					   V^{(0)}_{[\mu]} V^{(0)}_{[\nu]}
			     + \sum_{\mu,\nu}
				      \sigma^{\mu}_{\alpha\dot{\beta}}\sigma^\nu_{\gamma\dot{\delta}}\,
					  V^{(0)}_{[\mu]} V^{(0)}_{[\nu]}
			    \!\left.\rule{0ex}{1.2em}\right)\\
  && +\sum_{\dot{\delta}} \theta^1\theta^2 \bar{\theta}^{\dot{\delta}}
              \left(\rule{0ex}{1.2em}\right.\!
                \sqrt{-1}\sum_{\dot{\beta}^\prime, \dot{\delta}^\prime, \mu}
				 \varepsilon_{\dot{\beta}\dot{\delta}}
				 \varepsilon^{\dot{\beta}^\prime\dot{\delta}^\prime}
				 \sigma^\mu_{\alpha\dot{\beta}^\prime}\,
				 \partial_\mu(\mbox{$\sum$}_{\gamma^\prime}
				    \vartheta_{\gamma^\prime}V^{(\gamma^\prime)}_{(12\dot{\delta}^\prime)}) \\
      && \hspace{6em}					
                   -\,\sqrt{-1}\sum_\mu \sigma^\mu_{\alpha\dot{\beta}}\,
                         \partial_\mu (\mbox{$\sum$}_{\gamma^\prime}
                             \vartheta_{\gamma^\prime}	V^{(\gamma^\prime)}_{(12\dot{\delta})})
                    + 2 \sum_\mu \sigma^\mu_{\alpha\dot{\delta}}
					    V^{(0)}_{[\mu]}
						(\mbox{$\sum$}_{\gamma^\prime}
						    V^{(\gamma^\prime)}_{(12\dot{\beta})})
              \!\left.\rule{0ex}{1.2em}\right)         \\
  && + \sum_\gamma \theta^\gamma \bar{\theta}^{\dot{1}}\bar{\theta}^{\dot{2}}
               \left(\rule{0ex}{1.2em}\right.\!
			     -\sqrt{-1}\sum_\mu \sigma^\mu_{\gamma\dot{\beta}}\,
				       \partial_\mu (\mbox{$\sum$}_{\dot{\delta}^\prime}
					     \bar{\vartheta}_{\dot{\delta}^\prime}
						 \overline{V^{(\delta^\prime)}_{(12\dot{\alpha})}})
				 + 2 \sum_\mu \sigma^\mu_{\alpha\dot{\beta}}
				       V^{(0)}_{[\mu]}
					   (\mbox{$\sum$}_{\dot{\delta}^\prime}
					       \bar{\vartheta}_{\dot{\delta}^\prime}
						    \overline{V^{(\delta^\prime)}_{(12\dot{\gamma})}})
               \!\left.\rule{0ex}{1.2em}\right)\\			
  &&  +\, \theta^1\theta^2\bar{\theta}^{\dot{1}}\bar{\theta}^{\dot{2}}
                \left(\rule{0ex}{1.2em}\right.\!
				  -\sqrt{-1}\sum_\mu \sigma^\mu_{\alpha\dot{\beta}}\,
				      \partial_\mu V^{(0)}_{(12\dot{1}\dot{2})}
			      -\sum_{\gamma,\gamma^\prime, \dot{\beta}^\prime, \dot{\delta}^\prime, \mu, \mu^\prime, \nu}
				     \varepsilon^{\gamma\gamma^\prime}
					 \varepsilon^{\dot{\beta}^\prime\dot{\delta}^\prime}
					 \sigma^\mu_{\gamma\dot{\beta}}
					 \sigma^{\mu^\prime}_{\alpha\dot{\beta}^\prime}
					 \sigma^\nu_{\gamma^\prime\dot{\delta}^\prime}\,
					 \partial_\mu\partial_{\mu^\prime} V^{(0)}_{[\nu]}\\
      && \hspace{6em}
				  -\,\sqrt{-1}\sum_{\mu,\nu}
                      \sigma^\mu_{\alpha\dot{\beta}}\,
                      \partial_\mu (V^{(0)}_{[\nu]}V^{[\nu]; (0)})		\\
      && \hspace{6em}					
				  + \sqrt{-1}\sum_{\gamma,\gamma^\prime, \dot{\delta}, \dot{\delta}^\prime, \mu,\mu^\prime,\nu}
				        \varepsilon^{\gamma\gamma^\prime}
						\varepsilon^{\dot{\delta}\dot{\delta}^\prime}
						\sigma^\mu_{\gamma\dot{\beta}}
						\sigma^{\mu^\prime}_{\gamma^\prime\dot{\delta}}
						\sigma^\nu_{\alpha\dot{\delta}^\prime}\,
                        V^{(0)}_{[\nu]}\partial_\mu V^{(0)}_{[\mu^\prime]}					\\
      && \hspace{6em}				
				  -\,\sqrt{-1}\sum_{\gamma,\gamma^\prime, \dot{\delta}, \dot{\delta}^\prime, \mu,\mu^\prime,\nu}
				        \varepsilon^{\gamma\gamma^\prime}
						\varepsilon^{\dot{\delta}\dot{\delta}^\prime}
						\sigma^\mu_{\gamma\dot{\delta}}
						\sigma^{\mu^\prime}_{\gamma^\prime\dot{\beta}}
						\sigma^\nu_{\alpha\dot{\delta}^\prime}\,
                        V^{(0)}_{[\mu]}\partial_{\mu^\prime} V^{(0)}_{[\nu]}			\\
      && \hspace{6em}				
				  +\, 4\,(\mbox{$\sum$}_{\gamma^\prime}
				                \vartheta_{\gamma^\prime}
								V^{(\gamma^\prime)}_{(12\dot{\beta})})
				          (\mbox{$\sum$}_{\dot{\delta}^\prime}
						        \bar{\vartheta}_{\dot{\delta}^\prime}
								\overline{V^{(\delta^\prime)}_{(12\dot{\alpha})}})
				\!\left.\rule{0ex}{1.2em}\right)\,.
\end{eqnarray*}
} 
 
Observe that $\Theta_{\alpha\dot{\beta}}$ involves at most terms quadratic in components of $\breve{V}$.
Terms cubic in components of $\breve{V}$ do appear in the immediate steps but they cancel each other in the end.
There are terms that contain the factor $\sigma^\nu_{\alpha\dot{\beta}}$.
They give rise to terms in $\widehat{\nabla}^{\breve{V}}_\mu$ of the following form:
{\small
\begin{eqnarray*}
  \widehat{\nabla}^{\breve{V}}_\mu
    & =  & \partial_\mu
       -\,\mbox{\large $\frac{\sqrt{-1}}{2}$}\, V^{(0)}_{[\mu]}
	   +\, \mbox{\large $\frac{\sqrt{-1}}{2}$}\sum_{\gamma,\dot{\delta}}
              \theta^\gamma \bar{\theta}^{\dot{\delta}}\cdot  		
			    V^{(0)}_{[\mu]}
				 \mbox{$\sum$}_\nu \sigma^\nu_{\gamma\dot{\delta}}\,V^{(0)}_{[\nu]}			
       +\,\mbox{\large $\frac{1}{2}$}
          \sum_{\dot{\delta}} \theta^1\theta^2 \bar{\theta}^{\dot{\delta}}\cdot		
               \partial_\mu (\mbox{$\sum$}_{\gamma^\prime}
                             \vartheta_{\gamma^\prime}	V^{(\gamma^\prime)}_{(12\dot{\delta})})		\\
  && +\, \sqrt{-1} \sum_\gamma \theta^\gamma \bar{\theta}^{\dot{1}}\bar{\theta}^{\dot{2}}
                 V^{(0)}_{[\mu]}
					   (\mbox{$\sum$}_{\dot{\delta}^\prime}
					       \bar{\vartheta}_{\dot{\delta}^\prime}
						    \overline{V^{(\delta^\prime)}_{(12\dot{\gamma})}})
        +\, \mbox{\large $\frac{1}{2}$}\,
               \theta^1\theta^2\bar{\theta}^{\dot{1}}\bar{\theta}^{\dot{2}}
                \left(\rule{0ex}{1.2em}\right.\!
				 \partial_\mu V^{(0)}_{(12\dot{1}\dot{2})}					  			  					 
				  + \partial_\mu (\mbox{$\sum_\nu$}\,V^{(0)}_{[\nu]}V^{[\nu]; (0)})																		
				\!\left.\rule{0ex}{1.2em}\right)               \\[1ex]
   &&  + \, (\,\cdots\cdots\,)\,,	
\end{eqnarray*}}where  
all terms in $(\,\cdots\cdots\,)$ have the total $(\theta,\bar{\theta})$-degree $\ge 1$.

\bigskip

\subsection{Supersymmetry transformations of a pre-vector superfield\\ in Wess-Zumino gauge}

Readers are recommend to read this subsection along with
   [Argu: Sec.\,4.3, pp.\:76, 77] of {\sl Argurio} and
   [G-G-R-S: Sec.\,4.2.1, from before Eq.\,(4.2.7) to after  Eq.\,(4.2.13)]
      of {\it Gates, Grosaru, Ro\u{c}ek, \& Siegel}
 for comparison.
 		
Recall the Grassmann parameter level of $\widehat{X}^{\widehat{\boxplus}}$.
The setup, discussions, and results in Sec.\,3.1 and Sec.\,3.2 can be generalized straightforwardly
 to the case where the Grassmann parameter level is turned on.
We shall use this to understand how supersymmetries act on pre-vector superfields in Wess-Zumino gauge.
 
To preserve the self-twisted-conjugate condition $\breve{V}^\dag=\breve{V}$
 of a pre-vector superfield $\breve{V}$,
 let $(\eta,\bar{\eta}):= (\eta^1,\eta^2; \bar{\eta}^{\dot{1}},\bar{\eta}^{\dot{2}})$
 be a conjugate pair of constant sections of
  ${\cal S}^{\prime\,\vee}_{\parameter}\oplus {\cal S}^{\prime\prime\,\vee}_{\parameter}
      \subset  \widehat{\cal O}_X^{\,\widehat{\boxplus}}$   and
 consider the infinitesimal supersymmetry transformation
  $$
    \delta_{(\eta,\bar{\eta})}\breve{V}\;
     :=\; (\eta Q+ \bar{\eta}\bar{Q})  \breve{V}\;
	 :=\;  \mbox{\Large $($}
	            \mbox{$\sum$}_\alpha \eta^\alpha Q_\alpha
				- \mbox{$\sum$}_{\dot{\beta}}\bar{\eta}^{\dot{\beta}} \bar{Q}_{\dot{\beta}}
		 	 \mbox{\Large $)$}\, \breve{V}			
  $$
  of $\breve{V}$.
Then
  $(\delta_{(\eta,\bar{\eta})}\breve{V})^\dag
     = \delta_{(\eta,\bar{\eta})}\breve{V}$.
However, for $\breve{V}$ in Wess-Zumino gauge,
  $(\eta Q+ \bar{\eta}\bar{Q})  \breve{V}$ remains a pre-vector superfield
     but in general no longer in Wess-Zumino gauge.
This can be remedied by a gauge transformation:
(e.g., [Argu: Sec.\,4.3.1], [G-G-R-S: Sec.\,4.2.a.1], [W-B: Chap.\,VII, Exercise (8)], and
                    [We: Sec.\,15.3, Eq.\,(15.78)])

\bigskip

\begin{lemma} {\bf [uniqueness of correcting gauge transformation]}\;
 Let $\breve{V}$ be a pre-vector superfield in Wess-Zumino gauge.
 Then there is a unique chiral superfield $\Lambda_{(\eta,\bar{\eta}; \breve{V})}$
   depending ${\Bbb C}$-multilinearly on $(\eta,\bar{\eta})$ and $\breve{V}$
   such that
   the gauge transformation
   $$
     (\eta Q+ \bar{\eta}\bar{Q})  \breve{V}  \,
       -\, \sqrt{-1}
	      (\Lambda_{(\eta,\bar{\eta}; \breve{V})}- \Lambda_{(\eta,\bar{\eta}; \breve{V})}^\dag)	
   $$
   of $(\eta Q+ \bar{\eta}\bar{Q})  \breve{V}$
   is in Wess-Zumino gauge.
\end{lemma}

\medskip

\begin{proof}
 When $\breve{V}$ is in Wess-Zumino gauge,
   $(\tinybullet)^{(0)}_{(0)}$-component of
   $(\eta Q+ \bar{\eta}\bar{Q})  \breve{V}$ is always zero.
 It follows from the explicit computation in the previous theme that leads to Lemma~3.2.5
  that there is a unique chiral superfield $\breve{\Lambda}$
   associated to $(\eta Q+ \bar{\eta}\bar{Q})  \breve{V}$ with $\Lambda^{(0)}_{(0)}=0$
  such that
   $(\eta Q+ \bar{\eta}\bar{Q})  \breve{V}  \,-\, \sqrt{-1}(\Lambda- \Lambda^\dag)$
   is in Wess-Zumino gauge.
 The same explicit computation implies also that this unique $\Lambda$ depends ${\Bbb C}$-multilinearly 
   on $(\eta,\bar{\eta})$ and $\breve{V}$.
 This proves the lemma.
 
\end{proof}					

\bigskip

\begin{definition} {\bf [supersymmetry in Wess-Zumino gauge]}\; {\rm
 Set
  $$
   (\eta Q+ \bar{\eta}\bar{Q})  \breve{V}  \,
       -\, \sqrt{-1}
	        (\Lambda_{(\eta,\bar{\eta}; \breve{V})}
		                 - \Lambda_{(\eta,\bar{\eta}; \breve{V})}^\dag)\;	
     =\; \sum_\alpha \eta^\alpha Q_\alpha^{\!\tinyWZ} \breve{V}
	       - \sum_{\dot{\beta}}
		        \bar{\eta}^{\dot{\beta}}\bar{Q}_{\dot{\beta}}^{\!\tinyWZ}\breve{V}\,.
   $$
 This defines {\it (infinitesimal) supersymmetry transformations in Wess-Zumino gauge}
    $Q_\alpha^{\!\tinyWZ}$, $\bar{Q}_{\dot{\beta}}^{\!\tinyWZ}$
  that take a superfield in Wess-Zumino gauge to another in Wess-Zumino gauge.  	
}\end{definition}

\bigskip

Explicitly, let
 $$
   \breve{V}\;=\;
    \sum_{\gamma,\dot{\delta},\nu}\theta^\gamma \bar{\theta}^{\dot{\delta}}
	   \sigma^\nu_{\gamma,\dot{\delta}} V^{(0)}_{[\nu]}
	 + \sum_{\dot{\delta}}\theta^1\theta^2\bar{\theta}^{\dot{\delta}}
	       \sum_{\gamma^\prime}\vartheta_{\gamma^\prime}V^{(\gamma^\prime)}_{(12\dot{\delta})}
	 + \sum_\gamma \theta^\gamma \bar{\theta}^{\dot{1}} \bar{\theta}^{\dot{2}}
	   \sum_{\dot{\delta^\prime}}\bar{\vartheta}_{\dot{\delta}^\prime}
	             \overline{V^{(\delta^\prime)}_{(12\dot{\gamma})}}
     + \theta^1\theta^2\bar{\theta}^{\dot{1}}\bar{\theta}^{\dot{2}}
	       V^{(0)}_{(12\dot{1}\dot{2})}
 $$
 be a pre-vector superfield in Wess-Zumino gauge.
Then,
 {\small
 \begin{eqnarray*}
   \lefteqn{
    \delta_{\eta Q+ \bar{\eta}\bar{Q}}\breve{V}\;
      :=\; (\eta Q+ \bar{\eta}\bar{Q})  \breve{V} } \\
   && =\; \mbox{\Large $($}
	          \sum_\alpha \eta^\alpha \frac{\partial}{\partial\theta^\alpha}
			    -\sqrt{-1}\sum_{\alpha,\dot{\beta},\mu}
				    \eta^\alpha \sigma^\mu_{\alpha\dot{\beta}}\bar{\theta}^{\dot{\beta}}\partial_\mu
	         \mbox{\Large $)$}	    \breve{V}
			+  \mbox{\Large $($}
	          \sum_{\dot{\beta}}
			     \bar{\eta}^{\dot{\beta}} \frac{\partial}{\partial\bar{\theta}^{\dot{\beta}}}
			    + \sqrt{-1}\sum_{\alpha,\dot{\beta},\mu}
				    \theta^\alpha \sigma^\mu_{\alpha\dot{\beta}}\bar{\eta}^{\dot{\beta}}\partial_\mu
	            \mbox{\Large $)$}	    \breve{V}        \\
   && =\; \sum_\gamma \theta^\gamma
                  \sum_{\dot{\beta}, \nu}
                      \bar{\eta}^{\dot{\beta}}	\sigma^\nu_{\gamma\dot{\beta}} V^{(0)}_{[\nu]}
			 + \sum_{\dot{\delta}}\bar{\theta}^{\dot{\delta}}
			      \cdot (-1)\sum_{\alpha,\nu}
				    \eta^\alpha \sigma^\nu_{\alpha\dot{\delta}}V^{(0)}_{[\nu]}
			 + \theta^1\theta^2 \sum_{\dot{\beta}}\bar{\eta}^{\dot{\beta}}
			      (\mbox{$\sum$}_{\gamma^\prime}
				         \vartheta_{\gamma^\prime}V^{(\gamma^\prime)}_{(12\dot{\beta})})\\
    &&  \hspace{2em}
	         + \sum_{\gamma, \dot{\delta}} \theta^\gamma \bar{\theta}^{\dot{\delta}}
			      \mbox{\Large $($}
				    - \sum_{\alpha}\eta^\alpha \varepsilon_{\alpha\gamma}
                        (\mbox{$\sum$}_{\gamma^\prime}\vartheta_{\gamma^\prime}
						        V^{(\gamma^\prime)}_{(12\dot{\delta})})
                    + \sum_{\dot{\beta}}\bar{\eta}^{\dot{\beta}} \varepsilon_{\dot{\beta}\dot{\delta}}	
					    (\mbox{$\sum$}_{\dot{\delta}^\prime}
						    \bar{\vartheta}_{\dot{\delta}^\prime}
							\overline{V^{(\delta^\prime)}_{(12\dot{\gamma})}})
				  \mbox{\Large $)$}
             + \bar{\theta}^{\dot{1}} \bar{\theta}^{\dot{2}}
			      \sum_\alpha \eta^\alpha (\mbox{$\sum$}_{\dot{\delta}^\prime}
				      \overline{V^{(\delta^\prime)}_{(12\dot{\alpha})}} )     \\
	&& \hspace{2em}
	         + \sum_{\dot{\delta}} \theta^1\theta^2 \bar{\theta}^{\dot{\delta}}
			      \mbox{\Large $($}
			       \sum_{\dot{\beta}} \bar{\eta}^{\dot{\beta}}
			         \varepsilon_{\dot{\beta}\dot{\delta}} V^{(0)}_{(12\dot{1}\dot{2})}
				     + \sqrt{-1} \sum_{\alpha, \dot{\beta}, \gamma, \mu, \nu}
				         \bar{\eta}^{\dot{\beta}} \varepsilon^{\alpha\gamma}
					      \sigma^\mu_{\alpha\dot{\beta}} \sigma^\nu_{\gamma\dot{\delta}}\,
					       \partial_\mu V^{(0)}_{[\nu]}
			      \mbox{\Large $)$}    \\
    && \hspace{2em}	
             + \sum_\gamma \theta^\gamma \bar{\theta}^{\dot{1}} \bar{\theta}^{\dot{2}}
                   \mbox{\Large $($}
                     \sum_\alpha \eta^\alpha \varepsilon_{\alpha\gamma} V^{(0)}_{(12\dot{1}\dot{2})}
					  - \sqrt{-1} \sum_{\alpha,\dot{\beta}, \dot{\delta}, \mu, \nu}
					      \eta^\alpha \varepsilon^{\dot{\beta}\dot{\delta}}
						   \sigma^\mu_{\alpha\dot{\beta}} \sigma^\nu_{\gamma\dot{\delta}}\,
						    \partial_\mu V^{(0)}_{[\nu]}
                   \mbox{\Large $)$}				   \\
    && \hspace{2em}
             +\,	\theta^1\theta^2\bar{\theta}^{\dot{1}}\bar{\theta}^{\dot{2}}
			        \cdot(-\sqrt{-1})
			        \mbox{\Large $($}					
					  \sum_{\alpha,\dot{\beta}, \dot{\delta}, \mu}
					       \eta^\alpha \varepsilon^{\dot{\beta}\dot{\delta}} \sigma^\mu_{\alpha\dot{\beta}}\,
						   \partial_\mu (\mbox{$\sum$}_{\gamma^\prime}  \vartheta_{\gamma^\prime}
                                                           V^{(\gamma^\prime)}_{(12\dot{\delta})} )
					  + \sum_{\dot{\beta}, \alpha,\gamma, \mu}
					       \bar{\eta}^{\dot{\beta}} \varepsilon^{\alpha\gamma} \sigma^\mu_{\alpha\dot{\beta}}\,
						    \partial_\mu (\mbox{$\sum$}_{\dot{\delta}^\prime}
							                              \overline{V^{(\delta^\prime)}_{(12\dot{\gamma})}})
					\mbox{\Large $)$}\,.
 \end{eqnarray*}}Let    
 $\breve{\Lambda}$ be the unique chiral superfield in $\widehat{\cal O}_X^{\,\widehat{\boxplus}}$
 with
 $$
    \breve{\Lambda}_{(0)}\;=\; 0\,,\hspace{2em}
	\breve{\Lambda}_{(\gamma)}\;
	  =\; -\sqrt{-1} \sum_{\dot{\beta}, \nu} \bar{\eta}^{\dot{\beta}}
	              \sigma^\nu_{\gamma\dot{\beta}} V^{(0)}_{[\nu]}\,,\hspace{2em}				
    \breve{\Lambda}_{(12)}\;
	  =\; -\sqrt{-1} \sum_{\dot{\beta}} \bar{\eta}^{\dot{\beta}}				
	            (\mbox{$\sum$}_{\gamma^\prime}\vartheta_{\gamma^\prime}
				                    V^{(\gamma^\prime)}_{(12\dot{\beta})})\,.
 $$
I.e.\
 \begin{eqnarray*}
  \breve{\Lambda} & = &
    -\,\sqrt{-1}\sum_\gamma \theta^\gamma \sum_{\dot{\beta}, \nu}
	        \bar{\eta}^{\dot{\beta}}\sigma^\nu_{\gamma\dot{\beta}}V^{(0)}_{[\nu]}
	-\sqrt{-1}\,\theta^1\theta^2 \sum_{\dot{\beta}}
	        \bar{\eta}^{\dot{\beta}}
			  (\mbox{$\sum$}_{\gamma^\prime}\vartheta_{\gamma^\prime}
			      V^{(\gamma^\prime)}_{(12\dot{\beta})}) \\
	&& -\, \sum_{\dot{\delta}} \theta^1\theta^2 \bar{\theta}^{\dot{\delta}}
	             \sum_{\dot{\beta}, \alpha, \gamma, \mu, \nu}
				   \bar{\eta}^{\dot{\beta}} \varepsilon^{\alpha\gamma}
				     \sigma^\mu_{\alpha\dot{\delta}} \sigma^\nu_{\gamma\dot{\beta}}\,
					  \partial_\mu V^{(0)}_{[\nu]}\,.
 \end{eqnarray*}
 Then,
 {\small
 \begin{eqnarray*}
   \lefteqn{
     \delta_{\eta Q+ \bar{\eta}\bar{Q}}\breve{V}\, +\, \delta_{\breve{\Lambda}}\breve{V}\;
      =\; (\eta Q+\bar{\eta}\bar{Q})\breve{V}
	           -\sqrt{-1}\,(\breve{\Lambda}-\breve{\Lambda}^\dag) 	 }\\
   && =\; \sum_{\gamma, \dot{\delta}} \theta^\gamma \bar{\theta}^{\dot{\delta}}
			      \mbox{\Large $($}
				    - \sum_{\alpha}\eta^\alpha \varepsilon_{\alpha\gamma}
                        (\mbox{$\sum$}_{\gamma^\prime}\vartheta_{\gamma^\prime}
						        V^{(\gamma^\prime)}_{(12\dot{\delta})})
                    + \sum_{\dot{\beta}}\bar{\eta}^{\dot{\beta}} \varepsilon_{\dot{\beta}\dot{\delta}}	
					    (\mbox{$\sum$}_{\dot{\delta}^\prime}
						    \bar{\vartheta}_{\dot{\delta}^\prime}
							\overline{V^{(\delta^\prime)}_{(12\dot{\gamma})}})
				  \mbox{\Large $)$}         \\
	&& \hspace{2em}
	         + \sum_{\dot{\delta}} \theta^1\theta^2 \bar{\theta}^{\dot{\delta}} \cdot			
			     \sum_{\dot{\beta}} \bar{\eta}^{\dot{\beta}}
				  \mbox{\Large $($}
			         \varepsilon_{\dot{\beta}\dot{\delta}} V^{(0)}_{(12\dot{1}\dot{2})}
				     + \sqrt{-1} \sum_{\alpha, \gamma, \mu, \nu}
				          \varepsilon^{\alpha\gamma}
					      \sigma^\mu_{\alpha\dot{\beta}} \sigma^\nu_{\gamma\dot{\delta}}\, F_{\mu\nu}
			      \mbox{\Large $)$}    \\
    && \hspace{2em}	
             + \sum_\gamma \theta^\gamma \bar{\theta}^{\dot{1}} \bar{\theta}^{\dot{2}} \cdot
                     \sum_\alpha \eta^\alpha
				   \mbox{\Large $($}
					 \varepsilon_{\alpha\gamma} V^{(0)}_{(12\dot{1}\dot{2})}
					  - \sqrt{-1} \sum_{\dot{\beta}, \dot{\delta}, \mu, \nu}
					      \varepsilon^{\dot{\beta}\dot{\delta}}
						   \sigma^\mu_{\alpha\dot{\beta}} \sigma^\nu_{\gamma\dot{\delta}}\, F_{\mu\nu}
                   \mbox{\Large $)$}				   \\
    && \hspace{2em}
             +\;	\theta^1\theta^2\bar{\theta}^{\dot{1}}\bar{\theta}^{\dot{2}}
			        \cdot(-\sqrt{-1})
			        \mbox{\Large $($}					
					  \sum_{\alpha,\dot{\beta}, \dot{\delta}, \mu}
					       \eta^\alpha \varepsilon^{\dot{\beta}\dot{\delta}} \sigma^\mu_{\alpha\dot{\beta}}\,
						   \partial_\mu (\mbox{$\sum$}_{\gamma^\prime}  \vartheta_{\gamma^\prime}
                                                           V^{(\gamma^\prime)}_{(12\dot{\delta})} )
					  + \sum_{\dot{\beta}, \alpha,\gamma, \mu}
					       \bar{\eta}^{\dot{\beta}} \varepsilon^{\alpha\gamma} \sigma^\mu_{\alpha\dot{\beta}}\,
						    \partial_\mu (\mbox{$\sum$}_{\dot{\delta}^\prime}
							                              \overline{V^{(\delta^\prime)}_{(12\dot{\gamma})}})
					\mbox{\Large $)$}  \\
	 && =:\; \sum_\alpha \eta^\alpha Q_\alpha^{\!\tinyWZ}  \breve{V}
	                - \sum_{\dot{\beta}} \bar{\eta}^{\dot{\beta}}
					        \bar{Q}_{\dot{\beta}}^{\!\tinyWZ} \breve{V} \,,
 \end{eqnarray*}}where    
   $F_{\mu\nu}:= \partial_\mu V^{(0)}_{[\nu]}- \partial_\nu V^{(0)}_{[\mu]}$,
 now resumes in Wess-Zumino gauge.

From this, one reads off
{\small
 \begin{eqnarray*}
   Q_\alpha^{\!\tinyWZ} \breve{V}
    & = & - \sum_{\gamma, \dot{\delta}} \theta^\gamma \bar{\theta}^{\dot{\delta}}
					\varepsilon_{\alpha\gamma}
                        (\mbox{$\sum$}_{\gamma^\prime}\vartheta_{\gamma^\prime}
						        V^{(\gamma^\prime)}_{(12\dot{\delta})})  				 	
             - \sum_\gamma \theta^\gamma \bar{\theta}^{\dot{1}} \bar{\theta}^{\dot{2}} \cdot
				   \mbox{\Large $($}
					 \varepsilon_{\alpha\gamma} V^{(0)}_{(12\dot{1}\dot{2})}
					  - \sqrt{-1} \sum_{\dot{\beta}, \dot{\delta}, \mu, \nu}
					      \varepsilon^{\dot{\beta}\dot{\delta}}
						   \sigma^\mu_{\alpha\dot{\beta}} \sigma^\nu_{\gamma\dot{\delta}}\, F_{\mu\nu}
                   \mbox{\Large $)$}				   \\
    && \hspace{2em}
             +\;	\theta^1\theta^2\bar{\theta}^{\dot{1}}\bar{\theta}^{\dot{2}}
			        \cdot(-\sqrt{-1})					
					  \sum_{\dot{\beta}, \dot{\delta}, \mu}
					     \varepsilon^{\dot{\beta}\dot{\delta}} \sigma^\mu_{\alpha\dot{\beta}}\,
						   \partial_\mu (\mbox{$\sum$}_{\gamma^\prime}  \vartheta_{\gamma^\prime}
                                                           V^{(\gamma^\prime)}_{(12\dot{\delta})} )\,,	
 \end{eqnarray*}}		
{\small
 \begin{eqnarray*}
   \bar{Q}_{\dot{\beta}}^{\!\tinyWZ}\breve{V}
    & = & - \sum_{\gamma, \dot{\delta}} \theta^\gamma \bar{\theta}^{\dot{\delta}}				
                   \varepsilon_{\dot{\beta}\dot{\delta}}	
					    (\mbox{$\sum$}_{\dot{\delta}^\prime}
						    \bar{\vartheta}_{\dot{\delta}^\prime}
							\overline{V^{(\delta^\prime)}_{(12\dot{\gamma})}})
	         + \sum_{\dot{\delta}} \theta^1\theta^2 \bar{\theta}^{\dot{\delta}} \cdot			
				  \mbox{\Large $($}
			         \varepsilon_{\dot{\beta}\dot{\delta}} V^{(0)}_{(12\dot{1}\dot{2})}
				     + \sqrt{-1} \sum_{\alpha, \gamma, \mu, \nu}
				          \varepsilon^{\alpha\gamma}
					      \sigma^\mu_{\alpha\dot{\beta}} \sigma^\nu_{\gamma\dot{\delta}}\, F_{\mu\nu}
			      \mbox{\Large $)$}    \\
    && \hspace{2em}
             +\;	\theta^1\theta^2\bar{\theta}^{\dot{1}}\bar{\theta}^{\dot{2}}
			        \cdot \sqrt{-1} \sum_{\alpha,\gamma, \mu}
					        \varepsilon^{\alpha\gamma} \sigma^\mu_{\alpha\dot{\beta}}\,
						    \partial_\mu (\mbox{$\sum$}_{\dot{\delta}^\prime}
							                              \overline{V^{(\delta^\prime)}_{(12\dot{\gamma})}})\,.
 \end{eqnarray*}}The    
supersymmetry algebra generated by
   $Q_\alpha^{\!\tinyWZ}$'s, $\bar{Q}_{\dot{\beta}}^{\!\tinyWZ}$'s, and $\partial_\mu$'s
  is now closed only up to a gauge transformation.

\bigskip

\subsection{From pre-vector superfields to vector superfields}

The discussion in Sec.\,3.2 is mathematically perfectly fine.
However, when compared to the vector multiplet in representations of $d=3+1$, $N=1$ supersymmetry algebra,
there are two redundant degrees of freedom in a pre-vector superfield $\breve{V}$
 due to that
 for each $\dot{\beta}\in\{\dot{1}, \dot{2}\}$,
  the coefficient of the $\theta^1\theta^2\bar{\theta}^{\dot{\beta}}$-term
  contains
    $\vartheta_1 V^{(1)}_{(12\dot{\beta})}
	    + \vartheta_2 V^{(2)}_{(12\dot{\beta})}$,
    which has two component-functions
       $V^{(1)}_{(12\dot{\beta})},\, V^{(2)}_{(12\dot{\beta})}
	     \in C^\infty(X)^{\Bbb C}$,
    instead of one.
Such redundancies can be removed easily\footnote{Note
                                            that, unlike the set of chiral functions or the set of antichiral functions on $X^{\tinyphysics}$,  											
											  the set of vector superfields on $X^{\tinyphysics}$ is only required to be
											    a $C^\infty(X)$-module, rather than a $C^\infty(X)$-algebra.
											Naively this is what makes it easy.
											{\it However}, gauge transformations act on this module
											  and preferably one wants to remove the redundancy in a gauge-invariant way.
											Usually this may not be always easy.
											That, when in the shifted expressions,
											  these terms are themselves gauge-fixed comes to the rescue.
                                              }  
as follows:
 \begin{itemize}
  \item[\LARGE $\cdot$]
   For a pre-vector superfield $\breve{V}$ {\it in the shifted expression},
     for each $\dot{\beta}=\dot{1}, \dot{2}$,
    introduce an
	{\it ${\Bbb R}$ linear constraint on
	     $\sum_\gamma \vartheta_\gamma V^{(\gamma)}_{(12\dot{\beta})}$},
		 which then induces simultaneously the same ${\Bbb R}$-linear constraint on
           its complex  conjugate
           $\sum_{\dot{\delta}}\bar{\vartheta}_{\dot{\delta}}
              V^{(\dot{\delta})}_{(\beta\dot{1}\dot{2})}
             = \sum_{\dot{\delta}}\bar{\vartheta}_{\dot{\delta}}
                  \overline{V^{(\delta)}_{(12\dot{\beta})}}$,\; 				
	    (\,equivalently,
         for each $\alpha= 1, 2$,
           introduce an
		   {\it ${\Bbb R}$ linear constraint on
	           $\sum_{\dot{\delta}} \bar{\vartheta}_{\dot{\delta}}
			     V^{(\dot{\delta})}_{(\alpha\dot{1}\dot{2})}$},
			which then induces simultaneously the same ${\Bbb R}$-linear constraint on its complex  conjugate
            $\sum_\gamma \vartheta_\gamma
             V^{(\gamma)}_{(12\dot{\alpha})}
            = \sum_\gamma  \vartheta_\gamma
                 \overline{V^{(\dot{\gamma})}_{(\alpha\dot{1}\dot{2})}}$\,)\;
     to remove a redundant degree of freedom.\\

  \item[\LARGE $\cdot$]	
   Denote this set of constrained pre-vector superfields by
    $C^\infty(X^{\physics})^{\flat,\stc}_{(\constrained)}$.
 \end{itemize}

\bigskip

Since for a pre-vector superfield in the shifted expression,
 the component-functions
  $V^{(\alpha)}_{(12\dot{\beta})}$ and $\overline{V^{(\alpha)}_{(12\dot{\beta})}}$
  are fixed under gauge transformations,
the $C^\infty(X)$-submodule $C^\infty(X^{\physics})^{\flat,\stc}_{(\constrained)}$
  of\\     $C^\infty(X^{\physics})^{\flat,\stc}$ is invariant under gauge transformations.
And all the discussion in Sec.\,3.2 on $C^\infty(X^{\physics})^{\flat,\stc}$
 applies to $C^\infty(X^{\physics})^{\flat,\stc}_{(\constrained)}$ as well.
In particular,
 once using a ${\Bbb R}$-linear constraints to remove the redundancies and
  bringing the constrained pre-vector superfield to be in Wess-Zumino gauge,
 the remaining redundancy are the gauge transformations on $V^{(0)}_{[\mu]}$,
 cf. the discussion after Lemma~3.2.5.
Once modding out this last class of gauge symmetries,
the remaining degrees of freedom of a pre-vector superfield in Wess-Zumino gauge match exactly with the vector multiplet
 of the representations of the $d=3+1$, $N=1$ supersymmetry algebra.

\bigskip

To fix the notion and for the simplicity of the notation, before proceeding to the next subsection,
 we choose the following most simple ${\Bbb R}$-linear constraint\footnote{Since
                                                              the underlying topology  ${\Bbb R}^4$ of the space-time $X$ is contractible,
                                                                different choices of the linear constraints would bear no significant consequences
																mathematically and likely so also physically.
                                                                 }  
  on pre-vector superfields $\breve{V}$ in the shifted expression
  (and in standard coordinate functions $(x,\theta,\bar{\theta},\vartheta,\bar{\vartheta})$):
 $$
    V^{(2)}_{(12\dot{\beta})}(x)\;
	 =\;  V^{(\dot{2})}_{(\alpha\dot{1}\dot{2})}(x)\;=\; 0\,.
 $$

\bigskip

\begin{definition} {\bf [vector superfield]}\; {\rm
 A pre-vector superfield $\breve{V}$  that satisfies the above constraint when in the shifted expression is called
    a {\it vector superfield}.
 The set of vector superfields on $X^{\physics}$ is a $C^\infty(X)$-module,
   denoted by $C^\infty(X^{\physics})^{\flat\flat,\stc}$.   	
}\end{definition}

\bigskip

It follows that

\bigskip

\begin{lemma} {\bf [from pre-vector superfield to vector superfield]}\;
 Definition~3.2.1,
 Lemma~3.2.2,
 Definition/Lemma~3.2.3,
 Lemma~3.2.4,
 Lemma~3.2.5,
 Lemma~3.3.1, and
 Definition~3.3.2
   in Sec.\,3.2 and Sec.\,3.3
  remain valid with `pre-vector superfield' replaced by `vector superfield'.
\end{lemma}

\medskip

\begin{lemma} {\bf [independence of coordinate functions chosen when in Wess-Zumino gauge]}\;
 When in Wess-Zumino gauge,
  the conditions
   $\;V^{(2)}_{(12\dot{\beta})}
        =  V^{(\dot{2})}_{(\alpha\dot{1}\dot{2})} = 0\;$
   on a pre-vector superfield $\breve{V}$
   remain to hold
    whether one expresses $\breve{V}$ in terms of the chiral coordinate functions
	   $(x^\prime, \theta, \bar{\theta}, \vartheta, \bar{\vartheta})$
    or antichiral coordinate functions
      $(x^{\prime\prime}, \theta, \bar{\theta}, \vartheta, \bar{\vartheta})$.
\end{lemma}

\medskip

\begin{proof}
 This follows from the proof of Lemma~3.2.4 in Sec.\,3.2.
											
\end{proof}

\bigskip
 
\subsection{Supersymmetric $U(1)$ gauge theory with matter on $X$ in terms of $X^{\physics}$}

With the preparations in Sec.\,3.1, Sec.\,3.2, and Sec.\,3.4, we are now ready to construct
 a supersymmetric $U(1)$ gauge theory with matter on $X$ in terms of $X^{\physics}$.\footnote{{\it Note
                                                               for mathematicians.}\hspace{1em}
                                                             We are in no position to illuminate the original physical insight in the construction
															   of the supersymmetric action functional for $U(1)$ gauge theory with matter.
															 As in Sec.\,2, our attempt here is only to demonstrate that
															  {\it there is a precise
															   (complexified, ${\Bbb Z}/2$-graded) $C^\infty$-Algebraic Geometry
															   tower construction in which everything in Chap.\,VI and the part of Chap.\,VII
															    on the $U(1)$ case of [W-B] by Julius Wess \& Jonathan Bagger is accounted for
																and mathematically harmoniously interpreted.}
														     Nevertheless,
															 mathematicians are highly recommended to try
															    [Argu: Sec.\,4.3.2] of {\sl Riccardo Argurio},
															    [G-G-R-S: Sec.\,4.2.a.1] of
																  {\sl James Gates, Jr., Marcus Grisaru, Martin Ro\u{c}ek, Warren Siegel},
																[We: Sec.\,15.2] of {\sl Peter West}
															   to at least get some sense of why and from where some quantities in this subsection
															    are considered.															
	                                                                     } 

\bigskip

\begin{flushleft}
{\bf Two basic derived\footnote{Here,
                                              we are not using the term `{\it derived}' in any deeper sense.
											  We only mean that such superfields arise from the combination of more basic superfields
											    such as chiral superfields and vector superfields.
											  For example, the superpotential is a polynomial (or more generally holomorphic function) of
											    chiral superfields and thus can be regarded as a "derived" superfield.
											  Caution that these derived superfields may go beyond $C^\infty(X^{\tinyphysics})$
											     and lie only in $C^\infty(\widehat{X}^{\widehat{\boxplus}})$.
												}   
          superfields: gaugino superfield and kinetic-term superfield}
\end{flushleft}
Unlike chiral or antichiral superfields, a vector superfield $\breve{V}$
 contains no components that involve space-time derivatives.
For that reason, to construct a supersymmetric action functional for components of $\breve{V}$,
one needs to work out appropriate derived superfields from $\breve{V}$  first.

\bigskip

\begin{lemma-definition} {\bf [gaugino superfield]}\;
 {\rm ([W-B: Chap.\,VI, Eq.\,(6.7)])}\\
 Let $\breve{V}\in C^\infty(X^{\physics})^{\flat\flat,\stc}$ be a vector superfield.
 Define\footnote{The
                                design here  is made so that\;
								  $W_\alpha= \breve{V}_{(\alpha\dot{1}\dot{2})}\,
								    +\,(\mbox{terms of $(\theta,\bar{\theta})$-degree $\ge 1$})$\;  and\\
								  $\bar{W}_{\dot{\beta}}=\breve{V}_{(12\dot{\beta})}\,
								    +\, (\mbox{terms of $(\theta,\bar{\theta})$-degree $\ge 1$})$.			
                                Caution that, while $e_{\alpha^\prime}= \partial/\partial \theta^\alpha\,+\,\cdots$,\;
                                    $e_{\beta^{\prime\prime}}
									   = -\,\partial/\partial\bar{\theta}^{\dot{\beta}}\,+\,\cdots\,.$
								          } 
   $$
    W_\alpha\;
	  :=\; e_{2^{\prime\prime}}e_{1^{\prime\prime}} e_{\alpha^\prime} \breve{V}
	\hspace{2em}
	(\,\mbox{resp.}\;\;
	\bar{W}_{\dot{\beta}}\;
	  :=\; e_{1^\prime}e_{2^\prime} e_{\beta^{\prime\prime}} \breve{V}
	   \,)\,
   $$
   $\alpha=1,2$, $\dot{\beta}=\dot{1}, \dot{2}$.
 Then
  (1) $W_\alpha$ (resp.\ $\bar{W}_{\dot{\beta}}$) is chiral (resp.\ antichiral).
  (2) $W_\alpha$ and $\bar{W}_{\dot{\beta}}$ are invariant under gauge transformations on $\breve{V}$.

 {\rm
     $W_\alpha$, $\bar{W}_{\dot{\beta}}$
	   are called the {\it gaugino superfields} associated to the vector superfield $\breve{V}$.
      }
\end{lemma-definition}

\medskip

\begin{proof}
 For Statement (1),
 $$
  \begin{array}{llrll}
   e_{1^{\prime\prime}}W_\alpha
    & = & -\,  e_{2^{\prime\prime}}(e_{1^{\prime\prime}})^2 e_{\alpha^\prime}\breve{V}
	       & = & 0\,,          \\
  e_{2^{\prime\prime}}W_\alpha
    & = &   (e_{2^{\prime\prime}})^2 e_{1^{\prime\prime}} e_{\alpha^\prime}\breve{V}
	       &  = & 0
  \end{array}
 $$
  since $(e_{1^{\prime\prime}})^2=(e_{2^{\prime\prime}})^2=0$.
 Similarly for the antichirality of  $\bar{W}_{\dot{\beta}}$.
  
 For Statement (2),
   under a gauge transformation
     $\breve{V}\rightarrow \breve{V} -\sqrt{-1}(\breve{\Lambda}-\breve{\Lambda}^\dag)$
	 on $\breve{V}$ specified by a chiral superfield $\breve{\Lambda}$,
   \begin{eqnarray*}
     W_\alpha  & \rightarrow
	  & e_{2^{\prime\prime}}e_{1^{\prime\prime}}e_{\alpha^\prime}
            \mbox{\large $($}
			  \breve{V} -\sqrt{-1}(\breve{\Lambda}-\breve{\Lambda}^\dag)
			\mbox{\large $)$}\;\; 		
         =\;\; W_\alpha\,
	              -\,\sqrt{-1}\, e_{2^{\prime\prime}}e_{1^{\prime\prime}}e_{\alpha^\prime}
                                              \breve{\Lambda}  \\				
	 && =\;  W_\alpha\,
	              -\,\sqrt{-1}\,
				      \mbox{\large $($}
					   \{e_{1^{\prime\prime}}, e_{\alpha^\prime}\} e_{2^{\prime\prime}}
					   -  e_{2^{\prime\prime}}e_{\alpha^\prime}e_{1^{\prime\prime}}
					  \mbox{\large $ )$}
                                      \breve{\Lambda}\;\;\;=\;\;\;  W_\alpha
   \end{eqnarray*}
 since
   $\Lambda^\dag$ is antichiral (thus, $e_{\alpha^\prime}\breve{\Lambda}^\dag=0$)  and
   $\Lambda$ is chiral
     (thus, $e_{1^{\prime\prime}}\breve{\Lambda}
	                = e_{2^{\prime\prime}}\breve{\Lambda}=0$).
 Similarly for $\bar{W}_{\dot{\beta}}$.
 
\end{proof}
 
\bigskip

It follows that in the construction of a supersymmetric $U(1)$-gauge theory with matter,
 one may assume that the vector superfield $\breve{V}$ is in Wess-Zumino gauge,
  which encodes the component fields
    $V^{(0)}_{[\mu]}$, $V^{(\alpha)}_{(12\dot{\beta})}$, and
    $V^{(0)}_{(12\dot{1}\dot{2})}$ on $X$.
Here, $\mu=0,1,2,3$, $\alpha=1,2$, $\dot{\beta}=\dot{1}, \dot{2}$.
For $\breve{V}$  in Wess-Zumino gauge,
 $\breve{V}^3=0$ and
 its exponential $e^{\breve{V}}$
   is simply the polynomial $1+\breve{V}+\frac{1}{2}\breve{V}^2$ in $\breve{V}$.

\bigskip

\begin{lemma-definition} {\bf [gauge-invariant kinetic term for chiral superfield]}\; Let\\
 $\breve{V}\in C^\infty(X^{\physics})^{\flat\flat,\stc}$
  be a vector superfield and $\breve{\Phi}$ be a chiral superfield on $X^{\physics}$.
 Then the product
   $$
     \breve{\Phi}^\dag  e^{\breve{V}}\breve{\Phi}
   $$
  is gauge-invariant.
 {\rm
  Since the expression of the product in $(x,\theta,\bar{\theta},\vartheta,\bar{\vartheta})$
   involves space-time derivatives ($\partial_\mu$, $\mu=0,1,2,3$) of components of $\breve{\Phi}$,
   this product is called the {\it gauge-invariant kinetic term for the chiral superfield} $\breve{\Phi}$.
   }
\end{lemma-definition}

\medskip

\begin{proof}
   %
   %
 By construction,
 under the gauge transformation specified by a chiral superfield $\breve{\Lambda}$,
    \begin{eqnarray*}
       \breve{\Phi}^\dag  e^{\breve{V}}\breve{\Phi}
	     & \longrightarrow
		 & \mbox{\Large $($}
		      \Phi^\dag  e^{-\,\sqrt{-1}\,\breve{\Lambda}^\dag}
		    \mbox{\Large $)$}\,
		      e^{\breve{V}-\sqrt{-1}(\breve{\Lambda}-\breve{\Lambda}^\dag)}\,
		    \mbox{\Large $($}
		      e^{\sqrt{-1}\,\breve{\Lambda}}\breve{\Phi}
		    \mbox{\Large $)$}                                 \\
    && \hspace{1em}
	        =\; \breve{\Phi}^\dag\,
		        e^{-\,\sqrt{-1}\breve{\Lambda}^\dag
				        + \breve{V}-\sqrt{-1}(\breve{\Lambda}-\breve{\Lambda}^\dag)
						+\sqrt{-1}\,\breve{\Lambda}}\,
			   \breve{\Phi}\;\;
            =\;\;    \breve{\Phi}^\dag  e^{\breve{V}}\breve{\Phi}\,.			
   \end{eqnarray*}
\end{proof}

\bigskip

Note that
  in general $W_\alpha$, $\bar{W}_{\dot{\beta}}$
     only lie in $C^\infty(\widehat{X}^{\widehat{\boxplus}})$,
     not in $C^\infty(X^{\physics})$,
  while
    $\breve{\Phi}^\dag e^{\breve{V}}\breve{\Phi}$ always lies in $C^\infty(X^{\physics})$.

\bigskip

\begin{flushleft}
{\bf Explicit computations/formulae}
\end{flushleft}
The explicit expression of these derived  superfields and some related products can be computed via spinor calculus.
The results are listed below.
 \begin{itemize}
  \item[]  \makebox[6em][l]{\it Notations}
   In the expressions below,
     the Minkowski space-time metric tensor $(-,+,+,+)$ is used to raise or lower the space-time index 
	   $\mu,\, \nu,\, \cdots$
     while the $\varepsilon$-tensor is used to raise or lower the spinor index
        $\alpha, \gamma,\,\cdots, \dot{\beta},\dot{\delta},\,\cdots$;
     $F_{\mu\nu}:= \partial_\mu V^{(0)}_{[\nu]}-\partial_\nu V^{(0)}_{[\mu]}$;
	 $\delta_{\LARGEdot}^{\LARGEdot}$ is the Kronecker $\delta$
	    (when $\delta$ not served as a spinor index);
     $\varepsilon_{\mu\nu\mu^\prime\nu^\prime}$
      	 indicates the standard volume-form on the Minkowski space-time $X$ with $\varepsilon_{0123}= -1$;
     all the summations $\sum_{\cdots}$ are written explicitly.
 \end{itemize}

\bigskip

\noindent $\tinybullet$
$W_\alpha$:
(in chiral coordinate functions $(x^\prime, \theta,\bar{\theta},\vartheta,\bar{\vartheta})$
    on $\widehat{X}^{\widehat{\boxplus}}$)
 {\small
  \begin{eqnarray*}
  W_\alpha
    & = & \bar{\vartheta}_{\dot{1}}\overline{V^{(1)}_{(12\dot{\alpha})}}
	          + \sum_\gamma  \theta_\gamma
			         \mbox{\Large $($}
			       -\,\delta^\gamma_\alpha V^{(0)}_{(12\dot{1}\dot{2})}
				   + \sqrt{-1}\,\sum_{\mu,\nu}
				        (\sigma^\mu\bar{\sigma}^\nu)_\alpha^{\;\;\gamma}
						   \mbox{\large $($}
						     \partial_\mu V^{(0)}_{[\nu]} -\partial_\nu V^{(0)}_{[\mu]}
						   \mbox{\large $)$}
				   \mbox{\Large $)$} \\
      && +\,\theta^1\theta^2
                  \mbox{\Large $($}
				   2\sqrt{-1} \sum_{\dot{\beta},\dot{\delta}}
				        \varepsilon^{\dot{\beta}\dot{\delta}}
						  \sum_\mu \sigma^\mu_{\alpha\dot{\beta}}
						   \partial_\mu (\vartheta_1 V^{(1)}_{(12\dot{\delta})})
                  \mbox{\Large $)$}\,.				
  \end{eqnarray*}
  } 
			
\bigskip

\noindent $\tinybullet$
$(W_1W_2)|_{\theta^1\theta^2}$:  (in coordinate functions $x$ on $X$)
 {\small
 \begin{eqnarray*}
  (W_1W_2)|_{\theta^1\theta^2}
     & = &   -\,\mbox{\Large $\frac{1}{2}$}\,\sum_{\mu,\nu} F^{\mu\nu}F_{\mu\nu}
	     + \mbox{\Large $\frac{\sqrt{-1}}{4}$}
		       \sum_{\mu,\nu,\mu^\prime, \nu^\prime}
			      \varepsilon_{\mu\nu\mu^\prime\nu^\prime}
				  F^{\mu\nu} F^{\mu^\prime\nu^\prime}                \\
	 && -\, \mbox{\Large $\frac{\sqrt{-1}}{2}$}
	             \sum_{\alpha,\dot{\beta},\mu}
				    (\bar{\vartheta}_{\dot{1}} V^{(\dot{1})}_{(\alpha\dot{1}\dot{2})})
					    \bar{\sigma}^ {\mu, \dot{\beta}\alpha}
					     \partial_\mu ( \vartheta_1 V^{(1)}_{(12\dot{\beta})})\;		
		   +\, \mbox{\Large $\frac{1}{4}$}  {V^{(0)}_{(12\dot{1}\dot{2})}}^2\,.
 \end{eqnarray*}
  } 
 
\bigskip

\noindent $\tinybullet$
$\bar{W}_{\dot{\beta}}$:
 (in antichiral coordinate functions $(x^{\prime\prime}, \theta,\bar{\theta}, \vartheta, \bar{\vartheta})$
  on $\widehat{X}^{\widehat{\boxplus}}$)
{\small
 \begin{eqnarray*}
  \bar{W}_{\dot{\beta}}
   & = & \vartheta_1 V^{(1)}_{(12\dot{\beta})}
       + \sum_{\dot{\delta}} \bar{\theta}^{\dot{\delta}}
	       \mbox{\Large $($}
	        -\varepsilon_{\dot{\beta}\dot{\delta}} V^{(0)}_{(12\dot{1}\dot{2})}
			-\sqrt{-1}\sum_{\dot{\delta}^\prime}
			   \varepsilon_{\dot{\beta}\dot{\delta}^\prime}
			     \sum_{\mu,\nu}
				   (\bar{\sigma}^\mu \sigma^\nu)^{\dot{\delta}^\prime}_{\;\;\dot{\delta}}
				    \mbox{\large $($}
					 \partial_\mu V^{(0)}_{[\nu]} -\partial_\nu V^{(0)}_{[\mu]}
					\mbox{\large $)$}
	       \mbox{\Large $)$}\\
   && + \bar{\theta}^{\dot{1}}\bar{\theta}^{\dot{2}}
              \mbox{\Large $($}
			   -2\sqrt{-1}\sum_{\dot{\delta},\gamma}
			    \varepsilon_{\dot{\beta}\dot{\delta}}\bar{\sigma}^{\mu, \dot{\delta}\gamma}
				 \partial_\mu
				 (\bar{\vartheta}_{\dot{1}}\overline{V^{(1)}_{(12\dot{\gamma})}})
			  \mbox{\Large $)$}\,.
 \end{eqnarray*}
} 
 
\bigskip

\noindent $\tinybullet$
$(\bar{W}_{\dot{1}}\bar{W}_{\dot{2}})|
    _{\bar{\theta}^{\dot{1}}\bar{\theta}^{\dot{2}}}$: (in coordinate functions $x$ on $X$)
 {\small
 \begin{eqnarray*}
  (\bar{W}_{\dot{1}}\bar{W}_{\dot{2}})|
       _{\bar{\theta}^{\dot{1}}\bar{\theta}^{\dot{2}}}
  & = &   -\, \mbox{\Large $\frac{1}{2}$}\,\sum_{\mu,\nu} F^{\mu\nu}F_{\mu\nu}
	     - \mbox{\Large $\frac{\sqrt{-1}}{4}$}
		       \sum_{\mu,\nu,\mu^\prime, \nu^\prime}
			      \varepsilon_{\mu\nu\mu^\prime\nu^\prime}
				  F^{\mu\nu} F^{\mu^\prime\nu^\prime}                \\
	 && +\, \mbox{\Large $\frac{\sqrt{-1}}{2}$}
	             \sum_{\alpha,\dot{\beta},\mu}
				     \partial_\mu
					   (\bar{\vartheta}_{\dot{1}} V^{(\dot{1})}_{(\alpha\dot{1}\dot{2})})
					    \bar{\sigma}^{\mu, \dot{\beta}\alpha}
					 ( \vartheta_1 V^{(1)}_{(12\dot{\beta})})\;		
		   +\, \mbox{\Large $\frac{1}{4}$}  {V^{(0)}_{(12\dot{1}\dot{2})}}^2\,.
 \end{eqnarray*}
 } 
 
\bigskip

\noindent $\tinybullet$
$(\breve{\Phi}^\dag e^{\breve{V}}\breve{\Phi})|
  _{\theta^1\theta^2
        \bar{\theta}^{\dot{1}}\bar{\theta}^{\dot{2}}}$:
(in coordinate functions $x$ on $X$)
 {\small
 \begin{eqnarray*}
   \breve{\Phi}^\dag e^{\breve{V}}\breve{\Phi}
   & = & \breve{\Phi}^\dag \breve{\Phi}\,
         +\,  \breve{\Phi}^\dag  \breve{V}\breve{\Phi}\,
		 +\,\mbox{\Large $\frac{1}{2}$}\, \breve{\Phi}^\dag  \breve{V}^2 \breve{\Phi} \\
   & = &\;\;(\mbox{I})\;\; +\;\; (\mbox{II})\;\; +\;\; (\mbox{III})\,.
 \end{eqnarray*}Term}  
 (I) is computed in Sec.\,2.2.
 {\small
  \begin{eqnarray*}
   \mbox{Term (III)} & = &
     \theta^1\theta^2\bar{\theta}^{\dot{1}}\bar{\theta}^{\dot{2}}\cdot
	   \overline{\Phi^{(0)}_{(0)}}\Phi^{(0)}_{(0)}
	     \sum_\mu V^{[\mu];(0)}V^{(0)}_{[\mu]}\,.
  \end{eqnarray*}
  } 
{\small
 \begin{eqnarray*}
  \lefteqn{
   \mbox{Term (II)}|_{\theta^1\theta^2\bar{\theta}^{\dot{1}}\bar{\theta}^{\dot{2}}}}\\[1.8ex]
   && =\;
     \overline{\Phi^{(0)}_{(0)}} \Phi^{(0)}_{(0)}V^{(0)}_{(12\dot{1}\dot{2})}
	 - \sum_{\alpha,\gamma}\varepsilon^{\alpha\gamma}
             \overline{\Phi^{(0)}_{(0)}}
			   (\vartheta_\alpha \Phi^{(\alpha)}_{(\alpha)})
			     (\bar{\vartheta}_{\dot{1}} \overline{V^{(1)}_{(12\dot{\gamma})}})
	 + \sum_{\dot{\beta},\dot{\delta}} \varepsilon^{\dot{\beta}\dot{\delta}}
	       \Phi^{(0)}_{(0)}
		   (\bar{\vartheta}_{\dot{\beta}}\overline{\Phi^{(\beta)}_{(\beta)}})
		   (\vartheta_1 V^{(1)}_{(12\dot{\delta})})\\
  && \hspace{2em}
     +\,2\,\sqrt{-1}\, \sum_\mu
                \mbox{\large $($}
                 \overline{\Phi^{(0)}_{(0)}} \partial_\mu \Phi^{(0)}_{(0)}
				  - \partial_\mu \overline{\Phi^{(0)}_{(0)}}\,\Phi^{(0)}_{(0)}
                \mbox{\large $)$}
				V^{[\mu];(0)}
		- \sum_{\alpha,\dot{\beta},\mu}
		    (\bar{\vartheta}_{\dot{\beta}} \overline{\Phi^{(\beta)}_{(\beta)}})
			   \bar{\sigma}^{\mu, \dot{\beta}\alpha} V^{(0)}_{[\mu]}
			(\vartheta_\alpha \Phi^{(\alpha)}_{(\alpha)})\,.
 \end{eqnarray*}}Altogether 
 {\small
 \begin{eqnarray*}
   \lefteqn{(\breve{\Phi}^\dag e^{\breve{V}}\breve{\Phi})
           |_{\theta^1\theta^2\bar{\theta}^{\dot{1}}\bar{\theta}^{\dot{2}}}  }\\[1.8ex]
    &&=\;
	  -\,\overline{\Phi^{(0)}_{(0)}}\square \Phi^{(0)}_{(0)}
	  - \square \overline{\Phi^{(0)}_{(0)}}\,\Phi^{(0)}_{(0)}
	  + 2\sum_\mu \partial_\mu\overline{\Phi^{(0)}_{(0)}}\,\partial^\mu\Phi^{(0)}_{(0)}
	  + (\vartheta_1\vartheta_2\Phi^{(12)}_{(12)})
	          (\bar{\vartheta}_{\dot{1}}\bar{\vartheta}_{\dot{2}}
			        \overline{\Phi^{(12)}_{(12)}})   \\
	&& \hspace{2em}
	  +\,2\sqrt{-1}\,\sum_\mu
	         V^{[\mu];(0)} \overline{\Phi^{(0)}_{(0)}}\,
			   \partial_\mu \Phi^{(0)}_{(0)}
      -\,2\sqrt{-1}\,\sum_\mu
	         \partial_\mu\overline{\Phi^{(0)}_{(0)}}\,
               V^{[\mu],(0)} \Phi^{(0)}_{(0)}
	  +  \mbox{\large $($}
             \sum_\mu V^{(0)}_{[\mu]}V^{[\mu];(0)}
          \mbox{\large $)$}\,
             \overline{\Phi^{(0)}_{(0)}}\Phi^{(0)}_{(0)}		  \\
    && \hspace{2em}
       -\,\sqrt{-1}\, \sum_{\mu,\alpha,\dot{\beta}}	
	          (\bar{\vartheta}_{\dot{\beta}}\overline{\Phi^{(\beta)}_{(\beta)}})
			      \bar{\sigma}^{\mu,\dot{\beta}\alpha}\,
				  \partial_\mu (\vartheta_\alpha\Phi^{(\alpha)}_{(\alpha)})
	   +\,\sqrt{-1}\sum_{\mu,\alpha,\dot{\beta}}
	       \partial_\mu(\bar{\vartheta}_{\dot{\beta}}\overline{\Phi^{(\beta)}_{(\beta)}})
		     \bar{\sigma}^{\mu,\dot{\beta}\alpha}
			 (\vartheta_\alpha\Phi^{(\alpha)}_{(\alpha)})   \\
    &&\hspace{2em}			
	  -\,\sum_{\mu,\alpha,\dot{\beta}}
	         V^{(0)}_{[\mu]}
			 (\bar{\vartheta}_{\dot{\beta}}\overline{\Phi^{(\beta)}_{(\beta)}})
			  \bar{\sigma}^{\mu,\dot{\beta}\alpha}
			  (\vartheta_\alpha \Phi^{(\alpha)}_{(\alpha)})              \\
    && \hspace{2em}
	  +\,\Phi^{(0)}_{(0)}
	         \sum_{\dot{\beta},\dot{\delta}}
			   \varepsilon^{\dot{\beta}\dot{\delta}}
			     (\vartheta_1 V^{(1)}_{(12\dot{\beta})})
				 (\bar{\vartheta}_{\dot{\delta}}\overline{\Phi^{(\delta)}_{(\delta)}})
      -\, \overline{\Phi^{(0)}_{(0)}}				
	        \sum_{\alpha,\gamma}\varepsilon^{\alpha\gamma}
			  (\bar{\vartheta}_{\dot{1}}\overline{V^{(1)}_{(12\dot{\alpha})}})
			  (\vartheta_\gamma \Phi^{(\gamma)}_{(\gamma)})
      +\, V^{(0)}_{(12\dot{1}\dot{2})}
	         \overline{\Phi^{(0)}_{(0)}} \Phi^{(0)}_{(0)} \\[1.8ex]
   &&=\;
     -\, \sum_\mu \partial_\mu
           \mbox{\large $($}
             \overline{\Phi^{(0)}_{(0)}}\, \partial^\mu \Phi^{(0)}_{(0)}
			 + \partial^\mu \overline{\Phi^{(0)}_{(0)}}\,\Phi^{(0)}_{(0)}
           \mbox{\large $)$}		
	+\,4\, \sum_\mu \nabla_\mu
	      \overline{\Phi^{(0)}_{(0)}}\, \nabla^\mu \Phi^{(0)}_{(0)}\,
    +\,(\vartheta_1\vartheta_2\Phi^{(12)}_{(12)})
	          (\bar{\vartheta}_{\dot{1}}\bar{\vartheta}_{\dot{2}}
			        \overline{\Phi^{(12)}_{(12)}})   \\		
   && \hspace{2em}
       -\,\sqrt{-1}\, \sum_{\mu,\alpha,\dot{\beta}}	
	          (\bar{\vartheta}_{\dot{\beta}}\overline{\Phi^{(\beta)}_{(\beta)}})
			      \bar{\sigma}^{\mu,\dot{\beta}\alpha}\,
				  \nabla_\mu (\vartheta_\alpha\Phi^{(\alpha)}_{(\alpha)})
	   +\,\sqrt{-1}\sum_{\mu,\alpha,\dot{\beta}}
	       \nabla_\mu(\bar{\vartheta}_{\dot{\beta}}\overline{\Phi^{(\beta)}_{(\beta)}})
		     \bar{\sigma}^{\mu,\dot{\beta}\alpha}
			 (\vartheta_\alpha\Phi^{(\alpha)}_{(\alpha)})   \\
   && \hspace{2em}
	  +\,\Phi^{(0)}_{(0)}
	         \sum_{\dot{\beta},\dot{\delta}}
			   \varepsilon^{\dot{\beta}\dot{\delta}}
			     (\vartheta_1 V^{(1)}_{(12\dot{\beta})})
				 (\bar{\vartheta}_{\dot{\delta}}\overline{\Phi^{(\delta)}_{(\delta)}})
      -\, \overline{\Phi^{(0)}_{(0)}}				
	        \sum_{\alpha,\gamma}\varepsilon^{\alpha\gamma}
			  (\bar{\vartheta}_{\dot{1}}\overline{V^{(1)}_{(12\dot{\alpha})}})
			  (\vartheta_\gamma \Phi^{(\gamma)}_{(\gamma)})
      +\, V^{(0)}_{(12\dot{1}\dot{2})}
	         \overline{\Phi^{(0)}_{(0)}} \Phi^{(0)}_{(0)}
 \end{eqnarray*}}where    
 \begin{eqnarray*}
    \nabla_\mu \Phi^{(0)}_{(0)}\;
     :=\; (\partial_\mu - \mbox{\large $\frac{\sqrt{-1}}{2}$}V^{(0)}_{[\mu]})\,
            	 \Phi^{(0)}_{(0)}\,,
	 &
    \nabla_\mu \overline{\Phi^{(0)}_{(0)}}\;
     :=\; (\partial_\mu + \mbox{\large $\frac{\sqrt{-1}}{2}$}V^{(0)}_{[\mu]})\,
            	 \overline{\Phi^{(0)}_{(0)}}\,,     \\[1.2ex]
   \nabla_\mu \Phi^{(\alpha)}_{(\alpha)}\;
     :=\; (\partial_\mu - \mbox{\large $\frac{\sqrt{-1}}{2}$}V^{(0)}_{[\mu]})\,
            	 \Phi^{(\alpha)}_{(\alpha)}\,,
	 &
    \nabla_\mu \overline{\Phi^{(\beta)}_{(\beta)}}\;
     :=\; (\partial_\mu + \mbox{\large $\frac{\sqrt{-1}}{2}$}V^{(0)}_{[\mu]})\,
            	 \overline{\Phi^{(\beta)}_{(\beta)}}
 \end{eqnarray*}
 are the covariant derivatives  defined by $\breve{V}$ on components of $\breve{\Phi}$.
Note that this is consistent with Lemma~3.2.6;
cf.\;footnote~22.									

\bigskip

\begin{flushleft}
{\bf A supersymmetric action functional for $U(1)$ gauge theory with matter on $X$}
\end{flushleft}
Now restore the electric charge $e_m$ in the discussion.
Then the gauge-invariant kinetic term for the matter chiral superfield $\breve{\Phi}$ becomes
 $$
   \breve{\Phi}^\dag e^{e_m \breve{V}}\breve{\Phi}\,.
 $$
Thus, replacing $\breve{\Lambda}$  with $e_m\breve{\Lambda}$ and
                          $\breve{V}$ with $e_m\breve{V}$ in the above discussion and computations,
 we recover the charge $e_m$ case we well.

Recall the standard purge-evaluation/index-contracting map
 ${\cal P}: {\cal O}_X^{\physics}\rightarrow \widehat{\cal O}_X$
  from Sec.\,2.2    and
redenote:
 $$
  \begin{array}{lll}
   \Phi_{(0)}\;:=\;  \Phi_{(0)}\,,\;\;
     &  \Phi_{(\alpha)}\;:=\; {\cal P}(\vartheta_\alpha\Phi^{(\alpha)}_{(\alpha)})\,,\;\;
     & \Phi_{(12)}\;:=\; {\cal P}(\vartheta_1\vartheta_2\Phi^{(12)}_{(12)})\,,\;\; \\[1.2ex]
   V_{[\mu]}\;:=\; V^{(0)}_{([\mu])}\,,\;\;
     & V_{(\alpha\dot{1}\dot{2})}\;:=\;
     {\cal P}(\bar{\vartheta}_{\dot{1}}
                          \overline{V^{(1)}_{(12\dot{\alpha})}})\,,\;\;
     & V_{(12\dot{1}\dot{2})}\;:=\; V^{(0)}_{(12\dot{1}\dot{2})}\,.
  \end{array}	
 $$
Then, it follows from Theorem~1.5.3
 that\footnote{\makebox[11.6em][l]{\it Note for mathematicians}
                                               The coefficients are chosen to make the kinetic term
											     of the complex scalar  field $\Phi_{(0)}$ in the standard/normalized form:
											    $\sum_\mu\partial_\mu\overline{\Phi_{(0)}}
												                       \partial^\mu\Phi_{(0)}$
												and the kinetic term of the gauge field $V_{[\mu]}$
												 in the standard form
												  $-\frac{1}{2g_{\mbox{\tiny\it gauge}}^2}
												      \mbox{\it Tr} \sum_{\mu,\nu}F_{\mu\nu}F^{\mu\nu}$ 
												  of the Yang-Mills theory with gauge coupling $g_{\mbox{\tiny\it gauge}}$.
                                               One may also consider
	                              $$
		                            \mbox{\Large $\frac{\tau}{2}$}\,
                                      \int_{\widehat{X}}d^4x\, d\theta^2d\theta^1	{\cal P}(W_1W_2)\;
                                       +\; \mbox{\Large $\frac{\bar{\tau}}{2}$}
                                      \int_{\widehat{X}}d^4x\, d\bar{\theta}^{\dot{2}}d\bar{\theta}^{\dot{1}}	
			                                {\cal P}(\bar{W}_{\dot{1}}\bar{W}_{\dot{2}}) \\[1ex]
		                          $$
		                                         for the pure gauge term,
												 where
												   $\tau:= \frac{1}{g_{\gaugetiny}} \,-\,\sqrt{-1}\,\frac{\Theta}{8\pi^2}$
												   is the complexified coupling constant.
		                                       This will keep the topological term $\int_X F\wedge F$ to the action functional.}
 {\small
 \begin{eqnarray*}
  \lefteqn{
   S_2(V_{[0]}, V_{[1]}, V_{[2]}, V_{[3]};
                V_{(1\dot{1}\dot{2})}, V_{(2\dot{1}\dot{2})};
			    V_{(12\dot{1}\dot{2})};
               \Phi_{(0)};
			     \Phi_{(1)}, \Phi_{(2)};
				 \Phi_{(12)} ) }\\[1ex]
   && :=\;\;
    \mbox{\Large $\frac{1}{2 g_{\gaugetiny}^2}$}\,
            \int_{\widehat{X}}d^4x\, d\theta^2d\theta^1	{\cal P}(W_1W_2)\;
     +\; \mbox{\Large $\frac{1}{2 g_{\gaugetiny}^2}$}
            \int_{\widehat{X}}d^4x\, d\bar{\theta}^{\dot{2}}d\bar{\theta}^{\dot{1}}	
			     {\cal P}(\bar{W}_{\dot{1}}\bar{W}_{\dot{2}}) \\[1ex]
   && \hspace{2em}
     +\; \mbox{\Large $\frac{1}{4}$}
        \int_{\widehat{X}}
	     d^4x\,    d\bar{\theta}^{\dot{2}}d\bar{\theta}^{\dot{1}} d\theta^2d\theta^1\,
	     {\cal P}\mbox{\Large $($}
		       \breve{\Phi}^\dag e^{e_m\breve{V}} \breve{\Phi}
			               \mbox{\Large $)$}  \\[1ex]								
	&& \hspace{2em}
	  +\; \int_{\widehat{X}}d^4x\,d\theta^2 d\theta^1\,
			    {\cal P}\mbox{\Large $($}
				                     \lambda \breve{\Phi}
				                      + \mbox{\large $\frac{1}{2}$}m \breve{\Phi}^2
                                      + \mbox{\large $\frac{1}{3}$}g \breve{\Phi}^3
								   \mbox{\Large $)$}\,
            +\, \int_{\widehat{X}}d^4x\,d\bar{\theta}^{\dot{2}} d\bar{\theta}^{\dot{1}}\,
			    {\cal P}\mbox{\Large $($}
				                     \bar{\lambda} \breve{\Phi}^\dag
				                      + \mbox{\large $\frac{1}{2}$}\bar{m} (\breve{\Phi}^\dag)^2
                                      + \mbox{\large $\frac{1}{3}$}\bar{g} (\breve{\Phi}^\dag)^3
								   \mbox{\Large $)$}
 \end{eqnarray*}}where   
    $g_{\gaugescriptsize} $ is the gauge coupling constant,
 gives a functional of the component fields\\
  $(\Phi_{(0)}; \Phi_{(\alpha)}; \Phi_{(12)})_{\alpha=1,2}$
    of $\breve{\Phi}$ (cf.\;chiral matter) and
  $(V_{[\mu]}; V_{(\alpha\dot{1}\dot{2})}; V_{(12\dot{1}\dot{2})} )
      _{\mu=0,1,2,3; \alpha=1,2}$
    of $\breve{V}$ (cf.\;gauge field)   on $X$
 that is invariant under supersymmetries up to boundary terms on $X$.

Explicitly, up to boundary terms on $X$, this is the action functional
(cf.\ [W-B: Chap.\,VI, Eq.\,(6.13) \& Chap.\,VII, Eq.\,(7.10)], with mild adjustment, cf.\ footnote~13.) 
 {\small
 \begin{eqnarray*}
  \lefteqn{
   S_2(V_{[0]}, V_{[1]}, V_{[2]}, V_{[3]};
                V_{(1\dot{1}\dot{2})}, V_{(2\dot{1}\dot{2})};
			    V_{(12\dot{1}\dot{2})};
               \Phi_{(0)};
			     \Phi_{(1)}, \Phi_{(2)};
				 \Phi_{(12)} ) }\\
   &&=\;   \int_X d^4x \,\mbox{\Large $($}
     -\,\mbox{\Large $\frac{1}{2 g_{\gaugetiny}^2}$}\,
	      \sum_{\mu,\nu}F_{\mu\nu}F^{\mu\nu}\,
	 +\,\mbox{\Large $\frac{\sqrt{-1}}{2 g_{\gaugetiny}^2}$}\,
	      \sum_{\alpha, \dot{\beta},\mu}
		     \partial_\mu V_{(\alpha\dot{1}\dot{2})}
			   \bar{\sigma}^{\mu, \dot{\beta}\alpha} \overline{V_{(\beta\dot{1}\dot{2})}}\,
	 +\, \mbox{\Large $\frac{1}{4 g_{\gaugetiny}^2}$}\,
	        V_{(12\dot{1}\dot{2})}^{\;2}   \\[1ex]
   && \hspace{2em}
    +\, \sum_\mu \nabla_\mu
	      \overline{\Phi_{(0)}}\, \nabla^\mu \Phi_{(0)}\,
	-\,  \mbox{\Large $\frac{\sqrt{-1}}{2}$}\, \sum_{\mu,\alpha,\dot{\beta}}	
	          \overline{\Phi_{(\beta)}}
			      \bar{\sigma}^{\mu,\dot{\beta}\alpha}\,
				  \nabla_\mu \Phi_{(\alpha)}		
    +\,\mbox{\Large $\frac{1}{4}$}\, \overline{\Phi_{(12)}}\Phi_{(12)}\\
   && \hspace{2em}
	  +\,\mbox{\Large $\frac{1}{4}$}\,\Phi_{(0)}
	         \sum_{\dot{\beta},\dot{\delta}}
			   \varepsilon^{\dot{\beta}\dot{\delta}}
			     V_{(12\dot{\beta})}\overline{\Phi_{(\delta)}}
      -\,\mbox{\Large $\frac{1}{4}$}\, \overline{\Phi_{(0)}}				
	        \sum_{\alpha,\gamma}\varepsilon^{\alpha\gamma}
			  \overline{V_{(12\dot{\alpha})}}\Phi_{(\gamma)}
      +\,\mbox{\Large $\frac{1}{4}$}\, V_{(12\dot{1}\dot{2})}
	         \overline{\Phi_{(0)}} \Phi_{(0)}   \\
	 && \hspace{2em}
       +\, \int_X d^4x
              \mbox{\Large $($}	
                  \lambda \Phi_{(12)}
                     + m\, \mbox{\large $($}
			              \Phi_{(0)} \Phi_{(12)}- \Phi_{(1)} \Phi_{(2)}
			                   \mbox{\large $)$}\,
		          + g\, \mbox{\large $($}
			          \Phi_{(0)}^2 \Phi_{(12)}
			           - 2\, \Phi_{(0)} \Phi_{(1)} \Phi_{(2)}
			          \mbox{\large $)$}  \\[-1ex]
     && \hspace{8em}
                +\, (\,\mbox{complex conjugate}\,)	\,
	      \mbox{\Large $)$}\,.			
 \end{eqnarray*}}where    
 \begin{eqnarray*}
    \nabla_\mu \Phi_{(0)}\;
     :=\; (\partial_\mu - \mbox{\large $\frac{\sqrt{-1}}{2}$}e_m V^{(0)}_{[\mu]})\,
            	 \Phi_{(0)}\,,
	 &
    \nabla_\mu \overline{\Phi_{(0)}}\;
     :=\; (\partial_\mu + \mbox{\large $\frac{\sqrt{-1}}{2}$}e_m V^{(0)}_{[\mu]})\,
            	 \overline{\Phi_{(0)}}\,,     \\[1.2ex]
   \nabla_\mu \Phi_{(\alpha)}\;
     :=\; (\partial_\mu - \mbox{\large $\frac{\sqrt{-1}}{2}$}e_m V^{(0)}_{[\mu]})\,
            	 \Phi_{(\alpha)}\,,
	 &
    \nabla_\mu \overline{\Phi_{(\beta)}}\;
     :=\; (\partial_\mu + \mbox{\large $\frac{\sqrt{-1}}{2}$}e_m V^{(0)}_{[\mu]})\,
            	 \overline{\Phi_{(\beta)}}
 \end{eqnarray*}
 are the covariant derivatives defined by $\breve{V}$ on components of $\breve{\Phi}$.
The index structure of this explicit expression implies that this functional is indeed Lorentz invariant.
 
\bigskip

This is what underlies [W-B: Chap.\,VI \& $U(1)$ part of Chap.\,VII] of Wess \& Bagger
  from the aspect of (complexified ${\Bbb Z}/2$-graded) $C^\infty$-Algebraic Geometry.
Together with Sec.\,2, physicists' two most basic supersymmetric quantum field theories are now recast solidly
 into the realm of (complexified ${\Bbb Z}/2$-graded) $C^\infty$-Algebraic Geometry.

\bigskip
 
The same tower construction can be applied to superspace(-time)s of all other space(-time) dimensions
     with either simple (i.e.\ $N=1$) or extended (i.e.\ $N\ge 2$) supersymmetries,
	 with necessary modifications dictated by the specifics of spinors in each dimension and signature.
A redo of [L-Y1] (D(14.1)) along this new setting should give
 a fundamental (as opposed to solitonic) description of super D-branes
 parallel to Ramond-Neveu-Schwarz fundamental superstrings.


\newpage

\begin{flushleft}
{\large \bf Appendix\;
 Notations, conventions, and identities in spinor calculus}
\end{flushleft}
Notations, conventions, and identities in spinor calculus that are used in the current work are collected below.
See [W-B: Appendix A \& Appendix B] and [Argu] for a more complete list.

\bigskip

\noindent $\bullet$
Minkowski metric:\; $\eta_{\mu\nu}\;=\;(-+++)$.

\bigskip

\noindent $\bullet$
 With the identification $\Spin(3,1;{\Bbb R})\simeq \SL(2;{\Bbb C})$,
  two-component spinors with upper or lower, dotted or undotted indices transform as follows\;
  ($\alpha,\gamma=1,2\,; \dot{\beta}, \dot{\delta} =\dot{1}, \dot{2}$)
 $$
  \begin{array}{lccl}
   \psi^\prime_\gamma \;=\;  \sum_\alpha m_{\gamma}^{\:\:\alpha}\, \psi_\alpha\,,
     &&& \bar{\psi}^\prime_{\dot{\delta}}\;
	      =\; \sum_{\dot{\beta}}
		         \bar{m}_{\dot{\delta}}  ^{\:\:\dot{\beta}}\,\bar{\psi}_{\dot{\beta}} \,, \\[1.2ex]      
  \psi^{\prime\,\gamma}\;=\; \sum_\alpha m^{-1}\,\!_\alpha^{\:\:\gamma}\, \psi^\alpha\,,
     &&& \bar{\psi}^{\prime\, \dot{\delta}} \;
	      =\; \sum_{\dot{\beta}}
         		  \bar{m}^{-1}\,\!_{\dot{\beta}}^{\:\:\dot{\delta}}\, \bar{\psi}^{\dot{\beta}}\,,
  \end{array}
 $$
 where $m\in SL(2;{\Bbb C})$ as $2\times 2$ matrices, $\bar{m}$ the complex conjugate of $m$.

\bigskip

\noindent $\bullet$
The $\varepsilon$-tensor:
 \begin{itemize}
  \item[\LARGE $\cdot$]
  $\varepsilon^{12}= 1 = \varepsilon^{\dot{1}\dot{2}}$,\;
  $\varepsilon^{\alpha\gamma}= -\, \varepsilon^{\gamma\alpha}
     = -\,\varepsilon_{\alpha\gamma}$,\;
  $\varepsilon^{\dot{\beta}\dot{\delta}}
     = -\, \varepsilon^{\dot{\delta}\dot{\beta}} = -\,\varepsilon_{\dot{\beta}\dot{\delta}}$,
   for $\alpha,\gamma=1,2$ and $\dot{\beta},\dot{\gamma}=\dot{1}, \dot{2}$.

  \item[\LARGE $\cdot$]
   Use of $\varepsilon$-tensor to raise and lower spinor indices of the same chirality:
    $$
	 \begin{array}{c}
	  \psi^\gamma\;=\; \sum_\alpha \varepsilon^{\gamma\alpha}\psi_\alpha\,,\hspace{1em}
      \psi_\gamma\;=\; \sum_\alpha \varepsilon_{\gamma\alpha}\psi^\alpha\,,\hspace{1em}	
	
	  \bar{\psi}^{\dot{\delta}}\;
	    =\; \sum_{\dot{\beta}}
		      \varepsilon^{\dot{\delta}\dot{\beta}}\bar{\psi}_{\dot{\beta}}\,,\hspace{1em}
      \psi_{\dot{\delta}}\;
	    =\; \sum_{\dot{\beta}}
		      \varepsilon_{\dot{\delta}\dot{\beta}}\bar{\psi}^{\dot{\beta}}\,.
	 \end{array}
	$$
   
  \item[\LARGE $\cdot$]
  $\theta^\alpha\theta^\gamma= \varepsilon^{\alpha\gamma}\theta^1\theta^2$,\;
  $\theta_\alpha\theta_\gamma=  - \varepsilon_{\alpha\gamma}\theta_1\theta_2$,\;
  $\theta^{\dot{\beta}}\theta^{\dot{\delta}}
     = \varepsilon^{\dot{\beta}\dot{\gamma}}\theta^{\dot{1}}\theta^{\dot{2}}$,\;
  $\theta_{\dot{\beta}}\theta_{\dot{\delta}}
     = -\, \varepsilon_{\dot{\beta}\dot{\gamma}}\theta_{\dot{1}}\theta_{\dot{2}}$.	
\end{itemize}	

\bigskip

\noindent $\bullet$
Pauli matrices $\sigma^\mu$ and $\bar{\sigma}^\mu$:\;
  $\bar{\sigma}^{\mu, \dot{\delta}\gamma}\;
     :=\; \sum_{\alpha, \dot{\beta}}
	              \varepsilon^{\dot{\delta}\dot{\beta}}
                                    \varepsilon^{\gamma\alpha}\,	 \sigma^\mu_{\alpha\dot{\beta}}$,
  where
 \begin{itemize}
  \item[]
   \vspace{-2em}
   $$ \small
     \sigma^0\;=\; \left(\!  \begin{array}{rr}-1 & 0  \\ 0 & -1 \end{array}\!\right)\,,\;\;\;
     \sigma^1\;=\; \left(\!  \begin{array}{rr} 0 & 1  \\ 1 &   0 \end{array}\!\right)\,,\;\;\;
     \sigma^2\;=\; \left(\!  \begin{array}{rr} 0 & -\sqrt{-1}  \\  \sqrt{-1} & 0 \end{array}\!\right)\,,\;\;\;
     \sigma^3\;=\; \left(\!  \begin{array}{rr} 1 & 0  \\ 0 & -1 \end{array}\!\right)\,.
   $$
  
  \item[\LARGE $\cdot$]
   Explicitly,
    $\sigma^0= \bar{\sigma}^0$,  $\sigma^\mu=-\bar{\sigma}^\mu$ for $\mu=1,2,3$;
   in entries,
   $\sigma^0_{\alpha\dot{\beta}}= \bar{\sigma}^{0,\dot{\alpha}\beta}$,
   $\sigma^\mu_{\alpha\dot{\beta}}= - \bar{\sigma}^{\mu, \dot{\alpha}\beta}$.
   
  \item[\LARGE $\cdot$]
   Special relations:
     $\sigma^\mu_{1\dot{1}}= \bar{\sigma}^{\mu, \dot{2}2}$ and
	 $\sigma^\mu_{2\dot{2}}= \bar{\sigma}^{\mu, \dot{1}1}$.
	
  \item[\LARGE $\cdot$]	
    $\theta_\alpha\bar{\theta}_{\dot{\beta}}
	  =\frac{1}{2}
	      \sum_{\mu, \gamma,\dot{\delta}}
		    \theta^\gamma \sigma_{\mu,\gamma\dot{\delta}}\bar{\theta}^{\dot{\delta}}\,
			  \sigma^\mu_{\alpha\dot{\beta}}$\,,\;\;
	$\theta^\alpha\bar{\theta}^{\dot{\beta}}
	  =\frac{1}{2}
	      \sum_{\mu, \gamma,\dot{\delta}}
		    \theta^\gamma \sigma_{\mu,\gamma\dot{\delta}}\bar{\theta}^{\dot{\delta}}\,
			  \bar{\sigma}^{\mu, \dot{\beta}\alpha}$\,.
	  	
  \item[\LARGE $\cdot$]	
   $ \Tr \mbox{\large $($}
     (\sum_\mu p_\mu\sigma^\mu)(\sum_\nu p^\prime_\nu \bar{\sigma}^\nu)
	       \mbox{\large $)$}
      =  -2\, \langle p, p^\prime\rangle$.
 \end{itemize}

\bigskip

\noindent $\bullet$
Let
   $S^\prime$, $S^{\prime}$ be the Weyl-spinor representations   and
   $V$ be the vector/fundamental representation
 of $\SO(3,1)$.
Then,
  $S^{\prime\,\vee}\otimes_{\Bbb C}S^{\prime\prime\,\vee}
     \simeq V\otimes_{\Bbb R}{\Bbb C}$
 as $\SO(3,1)$-modules.

\newpage

\baselineskip 13pt
{\footnotesize

\vspace{1em}

\noindent
chienhao.liu@gmail.com, 
chienliu@cmsa.fas.harvard.edu; \\
yau@math.harvard.edu

}


\begin{thebibliography}{AAaaaa}
%
\marginpar{\raggedright\tiny $\bullet$~\parbox[t]{2cm}{\raggedright
        References\\  for SUSY(1).\\		
		Extended references\\ in D(14.1).}
		}		
%
\bibitem[Argu]{} R.\ Argurio,
 {\sl Introduction to supersymmetry}, lecture notes for PHYS-F-417, October, 2017.

\bibitem[Argy]{} P.\ Argyres,
 {\sl Introduction to supersymmetry},
  Physics 661 lecture notes, Cornell University, fall, 2000.

\bibitem[Bi]{}  A.\ Bilal,
 {\it Introduction to supersymmetry},
  arXiv:hep-th/0101055.

\bibitem[Ch]{} C.\ Chevalley,
 {\sl The algebraic theory of spinors}, Columbia Univ.\ Press, 1954.

\bibitem[De]{} P.\ Deligne,
 {\it Notes on spinors},
   in {\sl Quantum fields and strings: a course for mathematicians}, vol.\ 1,
  P.\ Deligne, P.\ Etingof, D.S.\ Freed, L.C.\ Jeffrey, D.\ Kazhdan, J.W.\ Morgan, and E.\ Witten eds., 99-135,
  American Math.\ Soc., 1999.
 
\bibitem[D-F1]{} P.\ Deligne and D.S.\ Freed,
 {\it Supersolutions},
  in {\sl Quantum fields and strings: a course for mathematicians}, vol.\ 1,
  P.\ Deligne, P.\ Etingof, D.S.\ Freed, L.C.\ Jeffrey, D.\ Kazhdan, J.W.\ Morgan, and E.\ Witten eds., 227-355,
  American Math.\ Soc., 1999.
 
\bibitem[D-F2]{} --------,
 {\it Sign manifesto},
  in {\sl Quantum fields and strings: a course for mathematicians}, vol.\ 1,
  P.\ Deligne, P.\ Etingof, D.S.\ Freed, L.C.\ Jeffrey, D.\ Kazhdan, J.W.\ Morgan, and E.\ Witten eds., 357-363,
  American Math.\ Soc., 1999.

\bibitem[D-K]{} S.K.\ Donaldson and P.B.\ Kronheimer,
 {\sl The geometry of four manifolds},
 Oxford Univ.\ Press, 1990.
    
\bibitem[Ei]{} D.\ Eisenbud,
 {\sl Commutative algebra -- with a view toward algebraic geometry},
 GTM 150, Springer, 1994.

\bibitem[E-H]{} D.~Eisenbud and J.~Harris,
 {\sl The geometry of schemes},
 GTM~197, Springer, 2000.
 
\bibitem[Fr]{} D.S.\ Freed,
 {\sl Five lectures on supersymmetry},
 Amer.\ Math.\ Soc., 1999.

   
\bibitem[F-Z]{} S.\ Ferrara and B.\ Zumino,
 {\it Supergauge invariant Yang-Mills theories},
 {\sl Nucl.\ Phys.}\ {\bf B79} (1974), 413--421.
 
\bibitem[G-G-R-S]{} S.J.\ Gates, Jr., M.T.\ Grisaru, M.\ Roc\u{c}ek, and W.\ Siegel,
 {\sl Superspace -- one thousand and one lessons in supersymmetry}, Frontiers Phys.\ Lect.\ Notes Ser.\ 58,
  Benjamin/Cummings Publ.\ Co., Inc., 1983.

\bibitem[Hai]{} G.\ Hailu, 
 {\sl Quantum field theory III: Supersymmetry},
  course Physics 253cr given at Department of Physcis, Harvard University, fall 2018.  
  
\bibitem[Hart]{} R.\ Hartshorne,
 {\sl Algebraic geometry},
 GTM 52, Springer, 1977.

\bibitem[Harv]{} F.R.\ Harvey,
 {\sl Spinors and calibrations}, 
  Pers.\ Math.\ 9, Academic Press, 1990.

\bibitem[Jo]{} D.\ Joyce,
 {\it Algebraic geometry over $C^{\infty}$-rings},
 arXiv:1001.0023 [math.AG].

\bibitem[Ko]{} S.$\:$Kobayashi,
 {\sl Differential geometry of complex vector bundles},
 Publ.\ Math.\ Soc.\ Japan 15, Princeton Univ.\ Press, 1987.
 
\bibitem[K-N]{} S.$\:$Kobayashi and K.$\:$Nomizu,
 {\sl Foundations of differential geometry}, vol.\:I \& vol.$\:$II,
 Interscience Publ., John Wiley \& Sons, 1963 and 1969.
      
\bibitem[L-Y1] {} C.-H.\ Liu and S.-T.\ Yau,
 {\it $N=1$  fermionic D3-branes in RNS formulation I.  
   $C^\infty$-Algebrogeometric foundations of $d=4$, $N=1$ supersymmetry,
   SUSY-rep compatible hybrid connections, and
   $\widehat{D}$-chiral maps from a $d=4$ $N=1$ Azumaya/matrix superspace}, 
 arXiv:1808.05011 [math.DG]. (D(14.1))  
 
\bibitem[L-Y2]{} --------,
 manuscript in preparation. 
   
\bibitem[Ma]{} Y.I.\ Manin,
 {\sl Gauge field theory and complex geometry},
  Springer, 1988.

\bibitem[P-S]{} M.E.\ Peskin and D.V.\ Schroeder, 
 {\sl An introduction to quantum field theory}, Addison-Wesley Publ.\ Co., 1995.

 
\bibitem[St]{} M.J.\ Strassler, 
 {\it An unorthodox introduction to supersymmetric gauge theory}, 
   lectures given at TASI 2001, arXiv:hep-th/0309149.
  
\bibitem[S-S]{} A.\ Salam and J.A.\ Strathdee, 
 {\it Supersymmetry and superfields}, 
 {\sl Fort.\ Phys.}\ {\bf 26} (1978), 57--142.
 
\bibitem[S-W]{} S.\ Shnider and R.O.\ Wells, Jr.,
 {\sl Supermanifolds, super twister spaces and super Yang-Mills fields},
   S\'{e}minaire Math.\ Sup\'{e}r.\ 106. Press.\ Univ.\ Montr\'{e}al, 1989.

\bibitem[We]{} P.\ West, 
 {\sl Introduction to supersymmetry and supergravity},
  extended 2nd ed., 
 World Scientific, 1990.
  
\bibitem[W-B]{} J.\ Wess and J.\ Bagger,
 {\sl Supersymmetry and supergravity}, 2nd ed.,
  Princeton Univ.\ Press, 1992.

\bibitem[W-Z1]{} J.\ Wess and B.\ Zumino, 
 {\it A Lagrangian model invariant under supergauge transformations},
 {\sl Phys.\ Lett.}\ {\bf 49B} (1974), 52--54.
 
\bibitem[W-Z2]{} --------,
 {\it Supergauge invariant extension of quantum electrodynamics},
 {\sl Nucl.\ Phys.}\ {\bf B78} (1974), 1--13. 
   
  
\end{thebibliography}
\end{document}